\documentclass{amsart}
\usepackage{graphicx}
\graphicspath{{images/}}

\usepackage{a4wide}
\usepackage[hidelinks]{hyperref}
\usepackage{xcolor}
\usepackage{verbatim}
\usepackage{caption}
\usepackage{mathtools}
\usepackage{enumerate}

\usepackage{xcolor}

\theoremstyle{plain}
\newtheorem{thm}{Theorem}

\newtheorem{prop}[thm]{Proposition}
\newtheorem{lemma}[thm]{Lemma}
\newtheorem{cor}[thm]{Corollary}
\theoremstyle{definition}
\newtheorem{definition}[thm]{Definition}
\newtheorem{notation}[thm]{Notation}
\newtheorem{assump}[thm]{Assumptions}
\newtheorem{remark}[thm]{Remark}
\newtheorem{example}[thm]{Example}

\usepackage{amssymb,amsmath,amsthm,mathrsfs}
\usepackage{graphicx}        
\usepackage{upgreek} 

\newcommand{\tn}[1]{\ensuremath{\mathbb{T}^{#1}}}
\newcommand{\hn}[1]{\ensuremath{\mathbb{H}^{#1}}}
\newcommand{\rn}[1]{\ensuremath{\mathbb{R}^{#1}}}

\newcommand{\sn}[1]{\ensuremath{\mathbb{S}^{#1}}}
\newcommand{\nn}[1]{\ensuremath{\mathbb{N}^{#1}}}
\newcommand{\bM}{\bar{M}}
\newcommand{\bR}{\bar{R}}

\newcommand{\bN}{{\bar{N}}}
\newcommand{\bth}{\bar{\theta}}
\newcommand{\g}{\gamma}
\newcommand{\G}{\Gamma}
\newcommand{\de}{\delta}
\newcommand{\dep}[1]{\delta_{#1}\phi}
\newcommand{\deo}[2]{\delta^{#1}_{#2}\omega}
\newcommand{\dee}[2]{\delta_{#1}^{#2}e}
\newcommand{\dek}[1]{\delta_{#1}k}
\newcommand{\vdep}{\vec{\delta}\phi}
\newcommand{\dega}[1]{\delta_{#1}\gamma}
\newcommand{\bge}{\bar{g}}
\newcommand{\fbar}{\bar{f}}

\renewcommand{\d}{\partial}
\newcommand{\rogen}{\mathrm{gen}}
\newcommand{\roLRS}{\mathrm{LRS}}
\newcommand{\iso}{\mathrm{iso}}
\newcommand{\roper}{\mathrm{per}}
\newcommand{\mrI}{\mathrm{I}}

\newcommand{\mrVIII}{\mathrm{VIII}}
\newcommand{\mrIX}{\mathrm{IX}}
\newcommand{\Symn}{\mathrm{Sym}_{n}(\mathbb{R})}
\newcommand{\md}{{\mathrm{d}}}
\newcommand{\rosc}{{\mathrm{sc}}}
\newcommand{\ropar}{{\mathrm{par}}}
\newcommand{\mbg}{{\mathrm{bg}}}
\newcommand{\mfC}{\mathfrak{C}}
\newcommand{\mfD}{\mathfrak{D}}
\newcommand{\mfN}{\mathfrak{N}}
\newcommand{\mfQ}{\mathfrak{Q}}
\newcommand{\mfP}{\mathfrak{P}}
\newcommand{\mfS}{\mathfrak{S}}
\newcommand{\mfK}{\mathfrak{K}}
\newcommand{\mfp}{\mathfrak{p}}
\newcommand{\mfl}{\mathfrak{l}}
\newcommand{\mfg}{\mathfrak{g}}
\newcommand{\mfh}{\mathfrak{h}}
\newcommand{\mfE}{\mathfrak{E}}

\newcommand{\mfI}{\mathfrak{I}}

\newcommand{\mfT}{\mathfrak{T}}
\newcommand{\mfs}{\mathfrak{s}}
\newcommand{\tr}{\mathrm{tr}}

\newcommand{\rodiv}{\mathrm{div}}
\newcommand{\rotot}{\mathrm{tot}}

\newcommand{\roint}{\mathrm{int}}
\newcommand{\roRic}{\mathrm{Ric}}
\newcommand{\roScal}{\mathrm{Scal}}
\newcommand{\roRiem}{\mathrm{Riem}}
\newcommand{\roHess}{\mathrm{Hess}}
\newcommand{\rograd}{\mathrm{grad}}

\newcommand{\refer}{\mathrm{ref}}

\newcommand{\roc}{\mathrm{c}}
\newcommand{\rob}{\mathrm{b}}
\newcommand{\roq}{\mathrm{q}}

\newcommand{\msA}{\mathscr{A}}

\newcommand{\msG}{\mathscr{G}}

\newcommand{\msY}{\mathscr{Y}}

\newcommand{\ldr}[1]{\langle #1\rangle}
\newcommand{\sfN}{\mathsf{N}}

\newcommand{\bna}{\overline{\nabla}}
\newcommand{\mB}{\mathcal{B}}
\newcommand{\mN}{\mathcal{N}}
\newcommand{\mD}{\mathcal{D}}
\newcommand{\mS}{\mathcal{S}}
\newcommand{\mK}{\mathcal{K}}
\newcommand{\mC}{\mathcal{C}}
\newcommand{\mI}{\mathcal{I}}
\newcommand{\mJ}{\mathcal{J}}

\newcommand{\bremI}{\breve{\mathcal{I}}}

\newcommand{\mH}{\mathcal{H}}
\newcommand{\bmK}{\bar{\mathcal{K}}}
\newcommand{\bmH}{\bar{\mathcal{H}}}
\newcommand{\ml}{\mathcal{L}}
\newcommand{\ve}{\varepsilon}

\renewcommand{\b}{\beta}
\renewcommand{\o}{\omega}

\newcommand{\bfI}{\mathbf{I}}

\newcommand{\bfJ}{\mathbf{J}}

\newcommand{\uI}{{\underline{I}}}
\newcommand{\uJ}{{\underline{J}}}

\newcommand{\uK}{{\underline{K}}}

\newcommand{\mbD}{\mathbb{D}}
\newcommand{\mbL}{\mathbb{L}}
\newcommand{\mbH}{\mathbb{H}}

\newcommand{\Q}{\mathcal{Q}}
\newcommand{\Z}{\mathcal{Z}}
\newcommand{\M}{\mathcal{M}}
\newcommand{\ar}{r}
\newcommand{\ear}{\vec{e}}
\newcommand{\cear}{\vec{\ce}}

\newcommand{\abs}[1]{|#1|}
\newcommand{\norm}[1]{\left\|#1\right\|}
\newcommand{\snorm}[1]{\|#1\|}

\renewcommand{\S}{\Sigma}
\newcommand{\s}{\sigma}
\newcommand{\sigmap}{{\sigma_p}}

\newcommand{\n}{\nabla}

\newcommand{\h}{h}
\renewcommand{\k}{k}

\newcommand{\bq}{{\bar p}}
\newcommand{\bh}{\bar \h}
\newcommand{\bk}{\bar \k}
\newcommand{\tk}{\tilde{\k}}
\newcommand{\brek}{\breve{\k}}
\newcommand{\breN}{\breve{N}}

\newcommand{\bga}{{\bar \gamma}}
\newcommand{\tga}{{\tilde{\gamma}}}
\newcommand{\brega}{{\breve{\gamma}}}
\newcommand{\bp}{{\bar \phi}}
\newcommand{\tp}{{\tilde{\phi}}}
\newcommand{\brep}{{\breve{\phi}}}
\newcommand{\bP}{{\bar \Phi}}
\newcommand{\be}{{\bar e}}
\newcommand{\te}{{\tilde{e}}}
\newcommand{\bree}{{\breve{e}}}
\newcommand{\bo}{{\bar \o}}
\newcommand{\breo}{{\breve{\o}}}
\newcommand{\tom}{{\tilde{\o}}}

\newcommand{\ch}{\check \h}
\newcommand{\ck}{\check \k}
\newcommand{\cg}{{\check g}}
\newcommand{\cga}{{\check \gamma}}
\newcommand{\cp}{{\check \phi}}

\newcommand{\ce}{{\check e}}
\newcommand{\co}{{\check \o}}
\newcommand{\cn}{{\check \n}}

\newcommand{\chk}{\check{k}}

\newcommand{\che}{\check{e}}
\newcommand{\chom}{\check{\omega}}
\newcommand{\chphi}{\check{\phi}}
\newcommand{\chga}{\check{\g}}

\newcommand{\mrK}{\mathring{\mK}}
\newcommand{\mrH}{\mathring{\mH}}

\newcommand{\mrPhi}{\mathring{\Phi}}
\newcommand{\mrp}{\mathring{p}}

\newcommand{\q}{p}

\newcommand{\ep}{\ear \phi}
\newcommand{\chep}{\cear \chphi}
\newcommand{\dtp}{e_0(\phi)}
\newcommand{\chdtp}{\d_t \chphi}
\newcommand{\en}{\ear N}
\newcommand{\wg}{{\widetilde \gamma}}
\newcommand{\wzeta}{{\widetilde \zeta}}

\DeclareMathOperator{\Id}{Id}

\begin{document}

\author{Hans Oude Groeniger}
\address{Department of Mathematics\\
  KTH Royal Institute of Technology\\
  100 44 Stockholm, Sweden}
\email[H.~Oude Groeniger]{joog@kth.se}
\author{Oliver Petersen}
\address{Department of Mathematics\\
  Stockholm University\\
  106 91 Stockholm, Sweden}
\email[O.~Petersen]{oliver.petersen@math.su.se}
\author{Hans Ringstr\"{o}m}
\address{Department of Mathematics\\
  KTH Royal Institute of Technology\\
  100 44 Stockholm, Sweden}
\email[H.~Ringstr\"{o}m]{hansr@kth.se}


\title{Formation of quiescent big bang singularities}
\begin{abstract}  
  Hawking's singularity theorem says that cosmological solutions satisfying the strong energy condition and corresponding to initial data with positive
  mean curvature have a past singularity; any past timelike curve emanating from the initial hypersurface has length at
  most equal to the inverse of the mean curvature. However, the nature of the singularity remains unclear. We therefore ask the following question: If
  the initial hypersurface has sufficiently large mean curvature, does the curvature necessarily blow up towards the singularity?

  In case the eigenvalues of the expansion-normalized Weingarten map are everywhere distinct and satisfy a certain algebraic condition
  (which in 3+1 dimensions is equivalent to them being positive), we prove that
  this is indeed the case in the CMC Einstein-non-linear scalar field setting. More specifically, we associate a set of geometric expansion-normalized
  quantities to any initial data set with positive mean curvature. These quantities are expected to converge, in the quiescent setting, in the direction
  of crushing big bang singularities; i.e.~as the mean curvature diverges. Our main result says that if the mean curvature is large enough,
  relative to an appropriate Sobolev norm of these geometric quantities, and if the algebraic condition on the eigenvalues is satisfied,
  then a quiescent (as opposed to oscillatory) big bang singularity with curvature blow-up necessarily forms. This provides a stable regime of big bang
  formation without requiring proximity to any particular class of background solutions.

  An important recent result by Fournodavlos, Rodnianski and Speck demonstrates stable big bang formation for all the spatially flat and spatially
  homogeneous solutions to the Einstein-scalar field equations satisfying the algebraic condition. As an application of our analysis, we obtain
  analogous stability results for
  any solution with induced data at a quiescent big bang singularity, in the sense introduced by the third author. In particular, we conclude stable
  big bang formation of large classes of spatially locally homogeneous solutions, of which the result by Fournodavlos, Rodnianski and Speck is a special
  case. Finally, since we here consider the Einstein-non-linear scalar field setting, we are also, combining the results of this article with an
  analysis of Bianchi class A solutions, able to prove both future and past global non-linear stability of a large class of spatially locally
  homogeneous solutions.
\end{abstract}

\maketitle

\tableofcontents

\section{Introduction}

Singularities are a natural feature of solutions to Einstein's equations in general relativity. 
In the cosmological setting, this follows from Hawking's singularity theorem. 
However, the only conclusion provided by the theorem is the existence of incomplete timelike geodesics. In particular, the nature of the
singularity remains unclear. Is the Kretschmann scalar (the Riemann curvature tensor contracted with itself) unbounded along incomplete geodesics?
What is the causal structure close to the singularity? These and related questions remain unanswered. In a series of articles, see e.g.
\cite{bkl70,bkl82}, Belinski\v{\i}, Khalatnikov and Lifschitz (BKL) suggested a scenario, based on heuristic arguments, concerning the nature of
generic big bang singularities. The scenario has since been refined by many authors; see, e.g., \cite{dhn,dn,hur,hul}. However, depending on the
matter model and the dimension, the idea is that the dynamics are spatially local and either oscillatory
and chaotic (the essential features of the local dynamics being modelled by a specific one-dimensional chaotic dynamical system) or quiescent. The
oscillatory setting is quite complicated to analyze and, to date, mathematical results have only been obtained in the spatially homogeneous setting;
see, e.g., \cite{weaver,cbu,BianchiIXattr,beguin,lhwg,lrt,bad}. In the quiescent setting, there is, however, a rich literature; see, e.g.,
\cite{cim,kar,iak,ren,aarendall,iam,damouretal,sta,RinAsVel,RinSCC,ABIF,klinger,aeta,rasql,rasq,specks3,rsh,GIJ,fal}. Many of the results
concern symmetric settings, often in the absence of a
smallness assumption; it is worth noting that the presence of symmetries can lead to quiescent behaviour even though the expected generic behaviour for
the matter model under consideration is oscillatory. Focusing on results in the absence of symmetries, Andersson and Rendall demonstrated that it is,
in the real analytic setting, possible to specify the asymptotics near the big bang for solutions to the Einstein-scalar field equations (as well as to
the Einstein-stiff fluid equations) in $3+1$-dimensions; see \cite{aarendall}. This result was later generalised to higher dimensions and other
matter models in \cite{damouretal}. An interesting, related, result in the smooth $3+1$-dimensional vacuum setting and in the absence of symmetries is
due to Fournodavlos and Luk; see \cite{fal}. In \cite{andres}, the results of \cite{aarendall} and \cite{fal} are simultaneously generalized in the
smooth, non-degenerate setting. Moreover, \cite{andres} rests upon a geometric notion of initial data on the singularity, see \cite{RinQC}, and yields,
given initial data, a maximal globally hyperbolic development which is unique up to isometry. The important goal of localizing \cite{fal} in space is
achieved in \cite{aaf} by Athanasiou and Fournodavlos.

Until recently, there were no proofs of stable big bang formation. However, this changed with the work of Rodnianski and Speck; see \cite{rasql,rasq,rsh}.
Initially, the authors proved stable big bang formation for initial data close to those of the spatially flat, spatially homogeneous and isotropic
solutions to the Einstein-scalar field and the Einstein-stiff fluid equations (note also the work of Beyer and Oliynyk, \cite{bat,baotwo}, which yields
stable big bang formation using a local gauge, as opposed to the non-local gauge used in \cite{rasql,rasq,rsh}; the results of Speck \cite{specks3}
in the case of $\sn{3}$-spatial geometry; and the work of Fajman and Urban \cite{fau} in the case of closed hyperbolic spatial geometry, which yields
future and past global non-linear stability). They then generalized this result to cover situations with moderate anisotropies and higher dimensions; see
\cite{rsh}. However, for the purposes of the present article, the most relevant previous result is \cite{GIJ}, due to Fournodavlos, Rodnianski and Speck
(note also the recent work of Beyer, Oliynyk and Zheng, \cite{boz}, localizing \cite{GIJ}). In this paper, the authors prove stability, in the
direction of the big bang singularity, of all the spatially homogeneous and spatially flat solutions to the Einstein-scalar field equations
satisfying the algebraic condition on the eigenvalues of the expansion-normalized Weingarten map mentioned in the abstract; see (\ref{eq: the condition})
below.

Many of the ideas introduced in \cite{GIJ} are of central importance in our arguments. However, as opposed to the previous results, i.e.
\cite{rasql,rasq,rsh,GIJ,bat,baotwo,boz,fau}, our primary goal here is not to prove past global non-linear stability of specific solutions, but rather to
identify conditions on initial data leading to quiescent big bang formation. When formulating such conditions, it is natural to single out the mean
curvature $\theta$ associated with the initial data. Due to Hawking's theorem, we know that $1/\theta$ is a measure of the distance to the singularity.
Moreover, singularity formation is signalled by $\theta$ diverging to $\infty$. However, it is also important to isolate quantities that are
complementary to $\theta$. Here we identify the expansion-normalized quantities $\mfN:=(\mK,\mH,\Phi_0,\Phi_1)$; see
Definitions~\ref{def: exp nor Weingarten map}--\ref{def:Phizodef} below. What can be said about $\mfN$? In the quiescent setting, $\mfN$ is
expected to converge; see, e.g., \cite{RinQC}. On the other hand, due to constructions of solutions with prescribed data on the singularity
(see, e.g., \cite{aarendall,damouretal,andres} for results in the case of Gaussian foliations), it is not expected to be possible to impose universal bounds
on $\mfN$; i.e., for any choice of Sobolev norm $\|\cdot\|$ for $\mfN$, and for any choice of constant $C$, we expect that there are solutions such that
$\|\mfN\|\geq C$ in the limit in the direction of the singularity. To summarize: $\theta$ diverges to infinity, and there is no universal bound on $\mfN$.
However, since, for a specific solution, $\mfN$ remains bounded as
$\theta$ tends to infinity, it is meaningful to first impose an upper bound, say $\zeta_0$, on a suitable Sobolev norm, say $\|\cdot\|$, of
$\mfN$, and then to try to determine a lower bound on $\theta$, say $\zeta_1$, depending on $\zeta_0$, such that if $\|\mfN\|\leq\zeta_0$ and
$\theta\geq\zeta_1$, then a singularity with curvature blow up forms. Unfortunately, conditions of this type are not sufficient. Due to arguments
going back to \cite{dhs}, there is one algebraic condition on the eigenvalues, say $p_I$, of the expansion-normalized Weingarten map $\mK$
that has to be satisfied, namely
\begin{equation}\label{eq: the condition}
  p_I+p_J-p_K<1-\s_p
\end{equation}
for all $I$, $J$, $K$ such that $I\neq J$ and some $\s_p>0$. 
A natural formulation one could hope for is then the following. Fix $\s_p>0$ and
$\zeta_0$. Consider $\mfN$'s that satisfy (\ref{eq: the condition}) and $\|\mfN\|\leq\zeta_0$. Prove that there is a $\zeta_1$, depending on
$\zeta_0$ and $\s_p$, such that if (\ref{eq: the condition}), $\|\mfN\|\leq\zeta_0$ and $\theta>\zeta_1$ are satisfied, then the corresponding
solution has a singularity exhibiting curvature blow up. For reasons mentioned above, conditions of this type can be expected to apply to generic
quiescent solutions. We are not able to prove precisely this statement, but almost. In the Einstein-scalar field setting, we only need to include one
additional condition: a strictly positive lower bound on $|\bq_I-\bq_J|$ for $I\neq J$. In the Einstein-non-linear scalar field setting, we also need
to impose conditions on the potential. 

Results of the above type are of interest in their own right. However, they can also be used to derive several corollaries. As a first corollary,
solutions corresponding to initial data on the singularity, in the sense of \cite{RinQC}, satisfying the condition on data on the singularity analogous
to (\ref{eq: the condition}) exhibit stable big bang formation. This corollary can, in its turn, be used to deduce past global non-linear stability
of large classes of spatially locally homogeneous solutions. In fact, since we allow a potential, we are able to prove past and future
global non-linear stability. 

The outline of the remainder of this section is as follows. We begin, in Subsection~\ref{ssection:enlsfeq}, by introducing the Einstein-non-linear
scalar field equations and imposing conditions on the potentials we consider. Next, as is
clear from the above, expansion-normalized quantities are of crucial importance, both in the statements of the results and in the arguments. We therefore
devote Subsection~\ref{ssection:expnormquan} to a discussion of appropriate ways of defining expansion-normalized versions of the initial data. The main
result is then stated in Subsection~\ref{ssection:mainresult}. In preparation for the corollaries of the main result, we introduce the notion of robust
initial data on the singularity in the Einstein-non-linear scalar field setting in Subsection~\ref{ssection:idotsing}. Once this has been done, we are in a
position to demonstrate that solutions arising from robust initial data on the singularity exhibit stable big bang formation. As a final corollary, we
prove that the results can be used to demonstrate stable quiescent singularity formation of large classes of spatially locally homogeneous solutions.
This includes global non-linear stability for recollapsing solutions. It also includes global non-linear stability for solutions with a big bang and an
expanding direction. This is the topic of Subsection~\ref{ssection:stablebbfshs}. In the main result, we make a
non-degeneracy assumption concerning the eigenvalues of the expansion-normalized Weingarten map. However, it is also possible to deduce conclusions in
the degenerate case, see Subsection~\ref{ssection: the degenerate case}.
Finally, in Subsections~\ref{ssection:ProofStrategy} and \ref{ssec:Outline}, we give an outline of the article and describe the strategy of the proof. 



\subsection{The Einstein-non-linear scalar field equations}\label{ssection:enlsfeq}

We are interested in finding solutions $(M,g,\phi)$ to the Einstein-non-linear scalar field equations with a cosmological constant $\Lambda$ and a potential $V\in C^{\infty}(\rn{})$. 
Here $(M,g)$ is a Lorentz manifold of dimension $n + 1 \geq 3$ and $\phi\in C^{\infty}(M)$. 
The equations are given by 
\begin{subequations}
  \begin{align}
    \roRic_g - \tfrac12 \roScal_g g +\Lambda g &=  T,\label{eq:EFE}\\
    \Box_{g}\phi-V'\circ\phi &= 0,\label{eq:scfieldeq}    
  \end{align}
\end{subequations}
where $\roRic_{g}$ and $\roScal_{g}$ are the Ricci tensor and scalar curvature of $(M,g)$; $\Box_{g}$ denotes the wave operator associated with $g$; and
the stress energy tensor $T$ is given by
\begin{equation}\label{eq:setsf}
    T
        = \md \phi \otimes \md \phi-\left[\tfrac{1}{2}| \md \phi|_{g}^{2}+V\circ\phi\right]g.
\end{equation}
Note that by adding a constant to $V$, we can eliminate the cosmological constant. In other words, there is no loss of generality in assuming $\Lambda=0$,
and we do so in what follows. We restrict the analysis to the following class of potentials (see Remark~\ref{remark:spsV} below for a justification):

\begin{definition}\label{def:Vadm}
  Fix $\s_{V} \in (0, 1)$. If $V\in C^{\infty}(\rn{})$ is non-negative and has the property that for each $0\leq k\in \mathbb{Z}$, there is a constant
  $c_{k}>0$ such that
  \begin{equation}
    \textstyle{\sum}_{l \leq k}|V^{(l)}(x)|
        \leq c_{k} e^{2(1-\s_{V})|x|} \label{eq: V assumption}
  \end{equation}
  for all $x\in\rn{}$, then $V$ is said to be a \textit{$\s_V$-admissible potential}. 
\end{definition}
\begin{remark}
  The requirement of non-negativity can be dropped. However, it is then necessary to impose stronger conditions on the initial data;
  see Remark~\ref{remark:negative potential} below. 
\end{remark}
Note that (\ref{eq:EFE}) can be written
\begin{equation}\label{eq:roRicpot}
	\roRic_g
		= \md \phi \otimes \md \phi + \tfrac2{n-1} (V \circ \phi) g.
\end{equation}
Given any spacelike hypersurface $\S \subset M$ with future pointing unit normal vector field $\nu$, (\ref{eq:EFE}) implies
\emph{constraint equations} on the induced first and second fundamental forms $h$, $k$ as well as $\phi_0 := \phi|_\S$ and $\phi_1 := \nu(\phi)|_\S$.
The constraint equations are given by
\begin{align}
    \roScal_h - \abs k^2_{h} + (\tr_h k)^2 
        &= \phi_1^2 + |\md \phi_0|_{h}^{2} + 2 V \circ \phi_0, \label{eq: Hamiltonian constraint} \\
    \mathrm{div}_hk - \md \tr_h k
        &= \phi_1 \md \phi_0. \label{eq: momentum constraint}
\end{align}
Equation \eqref{eq: Hamiltonian constraint} is called the \emph{Hamiltonian constraint equation} and \eqref{eq: momentum constraint} is called the
\emph{momentum constraint equation}.


\subsection{The expansion-normalized quantities}\label{ssection:expnormquan}
As already mentioned, when approaching a quiescent big bang singularity, certain quantities are expected to stay bounded, once the expansion of the
spacetime has been accounted for. The purpose of this section is to introduce these \emph{expansion-normalized} quantities.
As above, let $(M,g)$ be a Lorentz manifold of dimension $n + 1 \geq 3$ and $\phi\in C^{\infty}(M)$. 
Let $\S \subset M$ denote a spacelike hypersurface and $h$ and $k$ the induced first and second fundamental forms on $\S$, respectively.
Our measure of the expansion of $\S$ is the mean curvature:

\begin{definition}\label{def:mc}
The \textit{mean curvature} is defined as $\theta:= \tr_h(k)$.
\end{definition}

\begin{definition} \label{def: exp nor Weingarten map}
Assume that $\theta > 0$.
The \emph{expansion-normalized Weingarten map} is the endomorphism
\[
    \mK(X) := \tfrac{k(X, \cdot)^\sharp}{\theta}
\]
for any $X \in T\S$, where $\sharp : T^*M \to TM$ is the unique map such that $\o = h(\o^\sharp, \cdot)$ for all $\o \in T^*M$.
\end{definition}

Note that $\mK$ is symmetric with respect to $h$ and therefore diagonalizable:

\begin{definition}
The eigenvalues $p_1, \hdots, p_n$ of $\mK$ are said to be the \textit{eigenvalues associated with} $\mK$. 
\end{definition}

Note that $\sum_{j} p_j = 1$.
If $\theta > 0$, then $\theta ^{\mK}$ is a well defined smooth endomorphism on $T\S$, defined by
\[
    \theta^\mK(X)
        := e^{\ln(\theta) \mK}(X)
        = \textstyle{\sum}_{m = 0}^\infty \tfrac{\left( \ln(\theta) \mK\right)^m}{m!}(X),
\]
for any $X \in T\S$.

\begin{definition} \label{def: exp nor first ff}
Assume that $\theta > 0$.
The \emph{expansion-normalized first fundamental form} is the covariant 2-tensor field $\mH$ given by 
\[
    \mH(X, Y) 
        := h\left(\theta ^{\mK}(X), \theta ^{\mK}(Y)\right)
\]
for all $X,Y \in T_p\S$ and all $p \in \S$.
\end{definition}

Since we consider the Einstein-non-linear scalar field equations, we also need expansion-normalized quantities related to the
scalar field.

\begin{definition}\label{def:Phizodef}
Assume that $\theta > 0$.
The \emph{expansion-normalized normal derivative of the scalar field} is given by
\[
    \Phi_1 := \theta^{-1} \nu(\phi),
\]
where $\nu$ is the future unit normal vector field along $\S$.
The \emph{expansion-normalized induced scalar field} is given, along $\S$, by
\begin{equation}\label{eq:Phiz def}
    \Phi_0 
        := \phi + \Phi_1 \ln(\theta).
\end{equation}
\end{definition}
In order to illustrate the advantage of these definitions, we consider spatially homogeneous and spatially flat solutions to the Einstein-non-linear scalar field
equations as a special case:
\begin{example} \label{example: spatially hom flat}
  The \textit{spatially homogeneous and spatially flat model solutions} to the Einstein-scalar field equations, with vanishing potential, on the manifold $M=(0, \infty) \times \tn{n}$, are given by
  \[
  g=-\md t^2 + \textstyle{\sum}_{I} t^{2p_I}\md x^I \otimes \md x^I, 
  \qquad \phi = a \ln(t) + b,
  \]
  where $p_1, \hdots, p_n$, $a$, $b \in \rn{}$ are constants such that
  \[
  \textstyle{\sum}_{I} p_I
  = \sum_{I} p_I^2 + a^2
  = 1.
  \]
  Here $\tn{n}$ denotes the $n$-dimensional torus. For this solution, 
  \begin{align*}
    \theta & = \tfrac1t,\ \ \ \mK=\textstyle{\sum}_{I} p_I \d_{x^I} \otimes \md x^I,\ \ \
    \mH = \textstyle{\sum}_{I} \md x^I \otimes \md x^I, \\
    \Phi_1
        &= a, \ \ \ \Phi_0= b.
  \end{align*}  
  Note that $\mK$, $\mH$, $\Phi_1$ and $\Phi_0$ are independent of $t$. In particular, these quantities are smooth up to $t = 0$.
\end{example}

\subsection{Main result: Formation of quiescent big bang singularities}\label{ssection:mainresult}
Before stating the main result, it is convenient to introduce the notion of expansion-normalized initial data. 
\begin{definition}\label{def:idexpnormid}
Let $\mfI:=(\S,\bh,\bk,\bp_{0},\bp_{1})$ be initial data for the Einstein-non-linear scalar field equations with potential $V\in C^{\infty}(\rn{})$.
In other words, $\mfI$ satisfies the \textit{constraint equations}:
\begin{subequations}\label{seq:con}
\begin{align}
    \roScal_{\bh}-|\bk|_{\bh}^{2}+(\tr_{\bh}\bk)^{2} 
        = & \bp_{1}^{2}+|\md\bp_{0}|_{\bh}^{2}+2V\circ\bp_{0},\label{eq:hamcon}\\
    \rodiv_{\bh}\bk-\md(\tr_{\bh}\bk) 
        = & \bp_{1}\md \bp_{0}.\label{eq:momcon}
\end{align}
\end{subequations}
If $\bth:=\tr_{\bh}\bk$ is constant, $\mfI$ are said to be \textit{constant mean curvature (CMC)} initial data. If $\bth>0$ on $\S$, define the
associated expansion-normalized Weingarten map $\bmK$, expansion-normalized first fundamental form $\bmH$,
expansion-normalized normal derivative of the scalar field $\bP_{1}$ and expansion-normalized induced scalar field $\bP_{0}$ by appealing to
Definitions~\ref{def: exp nor Weingarten map}--\ref{def:Phizodef}. Then $(\S, \bmH,\bmK,\bP_{0},\bP_{1})$ are said to be the
\textit{expansion-normalized initial data associated to $\mfI$}. 
\end{definition}

The conditions in our main theorem below are formulated in terms of expansion-normalized initial data associated to CMC initial data
with $\bth>0$. If $\bq_1, \hdots, \bq_n$ are the eigenvalues associated with $\bmK$, then $\textstyle{\sum}_{j} \bq_j= 1$. However, in addition, the
following condition is expected to be essential in order to obtain a quiescent singularity (see Remarks~\ref{remark:spone}--\ref{remark:spsV} below):

\begin{definition} \label{def: alpha admissible}
  Let $\mfI$ be CMC initial data with $\bth>0$ as in Definition~\ref{def:idexpnormid}. Let $\sigmap \in (0, 1)$. Then $\mfI$ and the
  eigenvalues $\bq_1, \hdots, \bq_n$ are said to be $\sigmap$-\textit{admissible} if
  \begin{equation}
    \bq_I + \bq_J - \bq_K
    < 1 - \sigmap, \label{cond: alternating q}
  \end{equation}
  for all $I, J, K$, such that $I \neq J$.
\end{definition}

\begin{remark}\label{remark:frame intro}
  In the main theorem, we are interested in initial data on a closed manifold $\S$ such that there is a $(1,1)$-tensor field $\mK$ on $\S$ with
  distinct eigenvalues. Then a finite covering space of $\S$ is parallelizable; this follows
  by an argument similar to the proof of \cite[Lemma~A.1, p.~223]{RinWave}. Since we might as well consider the induced initial data on this covering
  space and the corresponding development, we assume $\S$ to be parallelizable. Throughout the paper, we also fix a reference Riemannian metric $h_\refer$
  on $\S$ and a smooth global orthonormal frame $(E_{i})_{i=1}^{n}$ on $(\S,h_{\refer})$. 
\end{remark}

The main result of this article is the following:

\begin{thm}[The main theorem] \label{thm: big bang formation} 
Fix admissibility thresholds $\s_V$, $\sigmap \in (0,1)$ and let
\begin{equation} \label{eq: sigma condition}
    \s
        := \min \left( \tfrac{\s_V}3, \tfrac \sigmap 5 \right).
\end{equation}
Fix $3\leq n\in\nn{}$ and regularity degrees $k_0$, $k_1 \in \nn{}$, such that
\begin{subequations}\label{seq: k_0 k_1 inequality}
  \begin{align}
    k_0 \geq & \left\lceil \tfrac {n+1}2 \right\rceil,\label{eq: k_0 inequality}\\
    k_1 \geq & \tfrac{(2n + 3)(1+2\s)}\s \left( k_0 + 3 + \left\lceil \tfrac {n+1}2 \right\rceil \right).\label{eq: k_1 inequality}
  \end{align}
\end{subequations}
Let $(\S, h_\refer)$ be a closed Riemannian manifold of dimension $n$ with smooth global orthonormal frame $(E_{i})_{i=1}^{n}$, and let
$V\in C^{\infty}(\rn{})$ be a $\s_V$-admissible potential. For any $\zeta_0 > 0$, there is then a $\zeta_1 > 0$ such that:

If $\mfI$ are $\sigmap$-admissible CMC initial data on $\S$ for the Einstein-non-linear scalar field equations with potential $V$, such that the
associated expansion-normalized initial data $(\S, \bmH,\bmK,\bP_{0},\bP_{1})$ satisfy
\begin{equation} \label{eq: Sobolev bound assumption}
    \snorm{\bmH^{-1}}_{C^{0}(\S)}+\snorm{\bmH}_{H^{k_1+2}(\S)} + \snorm{\bmK}_{H^{k_1+2}(\S)} + \snorm{\bP_0}_{H^{k_1+2}(\S)} + \snorm{\bP_1}_{H^{k_1+2}(\S)}
        < \zeta_0;
\end{equation}
$|\bq_I-\bq_J|>\zeta_{0}^{-1}$ for $I\neq J$; and the mean curvature satisfies $\bth > \zeta_1$, then the maximal globally hyperbolic development of
$\mfI$, say $(M, g, \phi)$, with associated embedding $\iota: \S \hookrightarrow M$, has a past crushing big bang singularity in the following sense:

\noindent \textbf{CMC foliation:} 
There is a diffeomorphism $\Psi$ from $(0, t_0] \times \S$ to $J^-\left(\iota(\S)\right)$, i.e.~the causal past of the Cauchy hypersurface
$\iota(\S)$, such that $\Psi (\{t_0 \} \times \S) = \iota(\S)$ and the hypersurfaces $\Psi(\S_t) \subset M$, where $\S_t := \{t\} \times \S$, are
spacelike Cauchy hypersurfaces with constant mean curvature $\theta = \frac1t$, for each $t \in (0, t_0]$.

\noindent \textbf{Asymptotic data:}
There are unique everywhere distinct functions $\mrp_1, \hdots, \mrp_n \in C^{k_0+1}(\S)$ and functions $\mrPhi_0$, $\mrPhi_1 \in C^{k_0+1}(\S)$, satisfying 
\begin{equation} \label{eq: Kasner conditions}
    \textstyle{\sum}_{I} \mrp_I 
        = \textstyle{\sum}_{I} \mrp_I^2 + \mrPhi_1^2 = 1,
\end{equation}
i.e., the generalized Kasner conditions, and, for all $I, J, K$ with $I \neq J$,
\begin{equation} \label{eq: quiescence conclusion}
    \mrp_I + \mrp_J - \mrp_K < 1.
\end{equation}
There is also a constant $C > 0$ such that, for all $t \in (0, t_0]$,
\begin{subequations}\label{seq:pIPhizPhioneasintro}
  \begin{align}
    \| p_I(t, \cdot) - \mrp_I \|_{C^{k_{0}+1}(\S)}  &\leq C t^{\s},\label{seq:pIasintro}\\
    \| \Phi_0(t, \cdot) - \mrPhi_0 \|_{C^{k_{0}+1}(\S)}+ \| \Phi_1(t, \cdot) - \mrPhi_1 \|_{C^{k_{0}+1}(\S)} &\leq C t^{\s},\label{seq:PhizPhioneasintro}
  \end{align}
\end{subequations}
where $\q_1, \hdots, \q_n$ are the eigenvalues of $\mK$, and $\mK$, $\Phi_1$ and $\Phi_0$ are the expansion-normalized quantities induced on
$\S_t$ by appealing to Definitions~\ref{def: exp nor Weingarten map} and \ref{def:Phizodef}.

\noindent \textbf{Curvature blow-up:}
There is a constant $C > 0$ such that $\mathfrak{R}_g := \roRic_{g, \mu \nu} \roRic_{g}^{\mu \nu}$ and
$\mathfrak{K}_g := \roRiem_{g,\mu \nu \xi \rho} \roRiem_g^{\mu \nu \xi \rho}$ satisfy
\begin{subequations}\label{seq:Kretschmann Ricci asymptotics intro}
  \begin{align}
    \left \| t^4 \mathfrak{K}_g(t, \cdot) - 
    4 \left[ \textstyle{\sum}_{I} \mrp_I^2 (1 - \mrp_I^2)
      + \textstyle{\sum}_{I < J} \mrp_I^2 \mrp_J^2 \right]
    \right \|_{C^{k_{0}+1}(\S)}\leq C t^{2\s}, \\
    \| t^4 \mathfrak{R}_g(t, \cdot) - \mrPhi_1^4
    \|_{C^{k_{0}+1}(\S)}\leq C t^{2\s}
  \end{align}
\end{subequations}
for all $t \in (0, t_0]$, so that $(M, g)$ is $C^2$ past inextendible. Moreover, every past directed causal geodesic in $M$ is incomplete
and $\mathfrak{K}_g$ blows up along every past inextendible causal curve. 
\end{thm}
\begin{remark}
  It is also possible to obtain conclusions when the $\bq_I$'s coincide, see Subsection~\ref{ssection: the degenerate case} below.
\end{remark}
\begin{remark}
  Due to Remark~\ref{remark:frame intro}, there is no restriction in assuming $(\S,h_\refer)$ to have a smooth global orthonormal frame. 
\end{remark}
\begin{remark}\label{remark:Constant dependence Main Theorem}
  The constant $\zeta_1$ and the constants $C$ appearing in (\ref{seq:pIPhizPhioneasintro}) and (\ref{seq:Kretschmann Ricci asymptotics intro}) only
  depend on $\zeta_0$, $\s_p$, $\s_V$, $k_0$, $k_1$, $c_{k_{1}+2}$, $(\S,h_{\refer})$ and $(E_i)_{i=1}^n$.
\end{remark}
\begin{remark}\label{remark:negative potential}
  The non-negativity requirement in Definition~\ref{def:Vadm} is only used to prove (\ref{eq:abs Phi1 initial bound}). If one is prepared
  to impose the condition (\ref{eq:abs Phi1 initial bound}) on initial data, with $\rho_0$ replaced by $\zeta_0$, it is not necessary to
  assume $V$ to be non-negative. 
\end{remark}
\begin{remark}
  The Sobolev norms appearing in the statement of the theorem are defined in Appendix~\ref{sec:SobolevInequalities} below.
\end{remark}
\begin{remark}\label{remark:FRaR non-deg}
  The conclusions of the theorem can be improved substantially. In fact, given one of the solutions constructed in the theorem, there are
  smooth robust non-degenerate quiescent initial data on the singularity for the Einstein-non-linear scalar field equations with a potential $V$,
  see Definition~\ref{def:ndsfidonbbs}, such that the solution is a local crushing CMC development corresponding to the initial data on the
  singularity; see Definition~\ref{def:developmentsfLambdacrushing}. For a justification of this statement, see
  \cite[Theorem~11, pp.~5-6]{FGaR}.
\end{remark}

The argument is based on the Fournodavlos--Rodnianski--Speck (FRS) equations introduced in Subsection~\ref{ssection: the FRS equations} below.
For a more detailed statement, including the asymptotics of the FRS variables, we refer to
Theorem~\ref{thm:GlobalExistenceFermiWalker} and Theorem~\ref{thm:Asymptotics} below.

\begin{remark}\label{remark:spone}
Let us consider Condition \eqref{cond: alternating q} in $3 + 1$ dimensional spacetimes.
We order the eigenvalues so that $\bq_1 > \bq_2 > \bq_3$.
Note that \eqref{cond: alternating q} is then equivalent to $\bq_1 < 1 - \sigmap$ and $\bq_1 + \bq_2 - \bq_3 < 1 - \sigmap$. 
Since the eigenvalues sum up to $1$, it follows $\bq_1 + \bq_2 - \bq_3 = 1 - 2 \bq_3$ and the second inequality is therefore equivalent to $\bq_3 > \frac \sigmap2$.
We therefore conclude that, in the $3 + 1$ dimensional case, \eqref{cond: alternating q} is equivalent to requiring, for $I = 1, 2, 3$,
\begin{equation}\label{eigenvcondthreedim}
    \tfrac \sigmap2 
        < \bq_I
        < 1 - \sigmap.
\end{equation}
\end{remark}
\begin{remark}\label{remark:sptwo}
  Condition \eqref{cond: alternating q} is the same as \cite[(1.8), p.~835]{GIJ}. Note, however, that here the $\bq_{I}$ are the eigenvalues
  of the expansion-normalized Weingarten map of the initial data (as opposed to the eigenvalues of the expansion-normalized Weingarten map
  of the background solution, cf. \cite{GIJ}). Moreover, the $\bq_{I}$ are functions in the present paper, not constants. 
  Finally, we are not assuming the variation of these functions to be small.
\end{remark}
\begin{remark}\label{remark:spsV}
  To the best of our knowledge, \eqref{cond: alternating q} originates with \cite{dhs} in the higher dimensional vacuum setting. The
  consistency of this condition in the Einstein-scalar field setting is illustrated by \cite{damouretal,GIJ}. On a heuristic level, the
  necessity is demonstrated in \cite{RinGeo}. In order to justify the condition on the potential introduced in Definition~\ref{def:Vadm},
  it is natural to consider initial data $(h,k,\phi_0,\phi_1)$, induced on a CMC hypersurface, and the associated Hamiltonian constraint; cf. 
  (\ref{eq:hamcon}). We are here interested in non-linear scalar fields, but we restrict
  our attention to potentials that give a subdominant contribution in the direction of the singularity. More specifically, dividing
  the Hamiltonian constraint by $\theta^{2}$, we expect, in the limit, $\roScal_{h}/\theta^{2}$ to be small, $|\md\phi_{0}|_{h}^{2}/\theta^{2}$
  to be small and $2V\circ\phi_{0}/\theta^{2}$ to be small. Using the notation introduced in Definitions~\ref{def: exp nor Weingarten map}
  and \ref{def:Phizodef}, it is then natural to expect the following approximate equality to hold in the asymptotic regime:
  \begin{equation}\label{eq:approxlimhamcon}
    1\approx \tr\mK^{2}+\Phi_{1}^{2};
  \end{equation}
  see also Definition~\ref{def:ndsfidonbbs} below, in which we introduce the notion of initial data on the singularity. In order for this
  argument to be consistent,
  \[
    V(\phi)/\theta^{2}=V(-\Phi_{1}\ln\theta+\Phi_{0})/\theta^{2}
  \]
  has to tend to zero in the direction of the singularity. Since we should here think of $\Phi_{0}$ and $\Phi_{1}$ as being essentially time
  independent, and since $\Phi_{1}^{2}$ is essentially bounded from above by $1-1/n$ due to
  (\ref{eq:approxlimhamcon}) (note that $\tr\mK^{2}\geq 1/n$ since $\tr\mK=1$), we want a bound of the form
  $|V(x)|\leq Ce^{\alpha|x|}$, where $\alpha<2/(1-1/n)^{1/2}$. For convenience, we here assume $\alpha\leq 2$, since we do not impose an upper
  bound on $n$. Moreover, it is convenient to have a margin independent of $n$, which is why we introduce $\s_{V}$. Finally, we also need to
  control derivatives of $V$. This is what leads to Definition~\ref{def:Vadm}. It should, however, be noted that if $n=3$, the natural assumption
  is that $|V(x)|\leq Ce^{\alpha|x|}$, where $\alpha<\sqrt{6}$. This is the condition we impose in many of the results in the spatially locally
  homogeneous setting; see \cite{RinModel}.
\end{remark}
\begin{remark}
  Due to the above remarks, the significance of $\s_{V}$ and $\s_p$ (and therefore $\s$) is clear. Given these constants, $k_{0}$ is the
  number of derivatives we control, asymptotically,  in $C^{0}$, and $k_{1}$ is the number of derivatives we need to control initially, in $L^{2}$.
  Note, in particular, that $k_{1}$ is inversely proportional to $\sigma$. Finally, $\zeta_{0}$
  is an arbitrary constant, quantifying the bound on the expansion-normalized initial data, see (\ref{eq: Sobolev bound assumption}). Moreover,
  $\zeta_{0}^{-1}$ is a lower bound on the distance between the initial eigenvalues. In particular, it is clear that $\zeta_{0}$ should be thought of
  as being large. Due to the above, it is natural to ask: where is the smallness condition? Here the smallness condition comes in the form of
  the proximity to the singularity, measured by the size of the mean curvature. More specifically: $\zeta_{1}$, which should be thought of as
  being large, is a lower bound on the mean curvature of the initial data. 
\end{remark}

\subsection{Application to solutions with induced data on the singularity}\label{ssection:idotsing}
One implication of Theorem~\ref{thm: big bang formation} is the existence of data $(\mrp_1, \hdots, \mrp_n, \mrPhi_0, \mrPhi_0)$ to which the
corresponding expansion-normalized quantities converge as $t \downarrow 0$. In fact, these data are naturally thought of as a subset of the data that
one would like to prescribe when solving the Einstein-scalar field equations with initial data at the singularity. In the case of the Einstein-scalar
field equations, the relevant notion of initial data on the singularity is introduced in \cite[Definition~1.10, p.~13]{RinQC}. This definition extends
verbatim to the Einstein-non-linear scalar field setting. However, we are here only interested in initial data on the singularity such that 
solutions with the corresponding asymptotics can be expected to exhibit stable big bang formation. This leads us to the following definition:

\begin{definition}\label{def:ndsfidonbbs}
  Let $3\leq n\in\nn{}$, $\s_V\in (0,1)$, $V\in C^{\infty}(\rn{})$ be a $\s_{V}$-admissible potential, $(\S,\mrH)$ be a smooth $n$-dimensional
  Riemannian manifold, $\mrK$ be a smooth $(1,1)$-tensor field on $\S$ and $\mrPhi_0$ and $\mrPhi_1$ be smooth functions on $\S$. 
  Then $(\S,\mrH,\mrK,\mrPhi_0,\mrPhi_1)$ are
  \textit{robust non-degenerate quiescent initial data on the singularity for the Einstein-non-linear scalar field equations with a potential $V$} if:
  \begin{enumerate}[(i)]
  \item $\tr\mrK=1$ and $\mrK$ is symmetric with respect to $\mrH$.
  \item $\tr\mrK^2 + \mrPhi_1^2 = 1$ and $\mathrm{div}_{\mrH} \mrK = \mrPhi_1 \md\mrPhi_0$.
  \item The eigenvalues $\mrp_1, \hdots, \mrp_n$ of $\mrK$ are everywhere distinct and satisfy
  \begin{equation} \label{eq: quiescence sing initial data}
    \mrp_I + \mrp_J - \mrp_K
        < 1,
  \end{equation}
  for all $I, J, K = 1, \hdots, n$, such that $I \neq J$.
  \end{enumerate}
\end{definition}
\begin{remark}
  If we let $\mrp_1, \hdots, \mrp_n$ denote the eigenvalues of $\mrK$, given by Definition~\ref{def:ndsfidonbbs}, then Conditions (i) and (iii) in
  Definition~\ref{def:ndsfidonbbs} are the first part of \eqref{eq: Kasner conditions} and \eqref{eq: quiescence conclusion}. Moreover, the first part
  of Condition (ii) is the second part of \eqref{eq: Kasner conditions}.
\end{remark}
\begin{remark}\label{remark:dataonsingKasnerhigherdim}
  Consider Example~\ref{example: spatially hom flat}. Assume that $n\geq 3$, that the $p_{I}$ are distinct and that $p_{I}+p_{J}-p_{K}<1$ for all
  $I, J, K = 1, \hdots, n$, such that $I \neq J$. Then, using the notation of Example~\ref{example: spatially hom flat}, $(\tn{n},\mH,\mK,\Phi_{0},\Phi_{1})$
  are robust non-degenerate quiescent initial data on the singularity for the Einstein-non-linear scalar field equations with a vanishing potential. 
\end{remark}

Next, we clarify what is meant by a CMC development corresponding to robust initial data on the singularity. The definition is very similar to
\cite[Definition~1.17, p.~14]{RinQC}:
\begin{definition}\label{def:developmentsfLambdacrushing}
  Let $3\leq n\in\nn{}$, $\s_V\in (0,1)$, $V\in C^{\infty}(\rn{})$ be a $\s_{V}$-admissible potential and $(\S,\mrH,\mrK,\mrPhi_0,\mrPhi_1)$ be
  robust non-degenerate quiescent initial data on the singularity for the Einstein-non-linear scalar field equations with a potential $V$; see
  Definition~\ref{def:ndsfidonbbs}. A \textit{local crushing CMC development}, corresponding to the initial data, is then a smooth time oriented
  Lorentz manifold $(M,g)$ and a $\phi\in C^{\infty}(M)$, solving the Einstein-non-linear scalar field equations with potential $V$, such that the
  following holds. There is a $0 < t_{+}\in\rn{}$ and a diffeomorphism $\Psi$ from $(0,t_{+}) \times \S$ to an open subset of $(M,g)$ such that the
  hypersurfaces $\Psi\left( \S_t \right) \subset M$, where $\S_t := \{t\} \times \S$, are spacelike Cauchy hypersurfaces with constant mean curvature
  $\theta = \frac1t$ for $t \in (0, t_+)$. Let $\mK$, $\mH$, $\Phi_0$ and
  $\Phi_1$ be the expansion-normalized quantities induced on $\S_t$, see Definitions~\ref{def: exp nor Weingarten map}--\ref{def:Phizodef}.
  Then the following correspondence between the solution and the asymptotic data is required to hold. There is a $\de>0$ and, for every
  $l\in\nn{}$, a constant $C_{l} > 0$ such that, for some $0 < t_1 < t_{+}$ and $t\in (0,t_1]$,
  \begin{equation}\label{eq:mKchhconvratesfN}
    \begin{split}
      \| \mH(t,\cdot)-\mrH \|_{H^{l}(\S)} +\|\mK(t,\cdot)-\mrK\|_{H^{l}(\S)} + \textstyle{\sum}_{i=0}^{1}\|\Phi_i(t, \cdot)-\mrPhi_i \|_{H^{l}(\S)}
      & \leq C_{l}t^\de.
    \end{split}
  \end{equation}  
\end{definition}

Given robust initial data on the singularity in the sense of Definition~\ref{def:ndsfidonbbs}, the expectation is that there should be an associated
development in the sense of Definition~\ref{def:developmentsfLambdacrushing} (which is unique under appropriate conditions). This remains to be
demonstrated, but there are related results for Gaussian foliations (as opposed to CMC foliations) in the real analytic setting (see
\cite{aarendall,damouretal} and the reformulations of these results given in \cite[Subsection~1.5, pp.~15--18]{RinQC}); in the spatially
homogeneous setting (see \cite{RinSHID,RinModel}); and in the smooth $3$-dimensional setting (see \cite{fal,andres}). 

There are several reasons for introducing the notion of initial data on the singularity. The most optimistic hope is that, given such data, one
can prove that there is a corresponding development, and that solutions exhibiting quiescent asymptotics in the direction of the
singularity induce initial data on the singularity. If one is able to prove statements of this nature, it is clear that data on the singularity can
be used to parametrize quiescent solutions. In some spatially homogeneous settings, e.g., this program can be carried out to completion; see
\cite{RinSHID,RinModel}. In fact, in \cite{RinModel}, there are arguments demonstrating that the Einstein flow (i.e., the map which, in an appropriate
foliation (CMC, Gaussian etc.), maps the initial data on one leaf to the initial data on another leaf) defines a global diffeomorphism between
isometry classes of developments and isometry classes of data on the singularity. However, more generally, the question is open. There might also be
complications arising due to features such as spikes. Since this is not the main topic of the present article, we refer the interested reader to
\cite{RinQC} for further discussions.

A more modest question than that of trying to parametrize quiescent solutions by means of initial data on the singularity is the following: Given a
locally crushing CMC development, corresponding to robust initial data on the singularity in the sense of
Definition~\ref{def:ndsfidonbbs}, is it stable under perturbations? In other words, does perturbing regular initial data on $\Psi(\S_t)$, using the
notation of Definition~\ref{def:developmentsfLambdacrushing}, give rise to a maximal globally hyperbolic development with a crushing singularity such
that the curvature blows up in the direction of the singularity? More specifically, does the maximal globally hyperbolic development give rise to
initial data on the singularity and constitute an associated locally crushing CMC development? Here we are able
to answer the first question, but not the final one. Note, however, that in order to obtain a positive answer to the first question, we expect
\eqref{eq: quiescence sing initial data} to be necessary. 

\begin{thm}\label{thm:stabilityofdevelopments}
  Let $\s_V \in (0, 1)$, $(\S, h_\refer)$ be a closed Riemannian manifold of dimension $n \geq 3$ and let
  $V\in C^{\infty}(\rn{})$ be a $\s_V$-admissible potential. Let $(\S,\mrH,\mrK,\mrPhi_0,\mrPhi_1)$
  be robust non-degenerate quiescent initial data on the singularity for the Einstein-non-linear scalar field equations with potential $V$. Assume that
  there is an associated locally crushing CMC development, say $(M,g,\phi)$, and let $t_+>0$ and $\S_t$ be as in
  Definition~\ref{def:developmentsfLambdacrushing}. Assume $t_0\in (0,t_+)$ to be such that
  \begin{equation}\label{eq:goodlapsecoeff}
    |k(t, \cdot)|_{h(t, \cdot)}^{2}+|\phi_1(t, \cdot)|^{2}-2V\circ\phi_0(t, \cdot)/(n-1)>0
  \end{equation}
  for all $t\leq t_{0}$, where $(\S, h, k, \phi_0, \phi_1)$ are the induced initial data on the Cauchy hypersurfaces $\S_t$ by $(M, g, \phi)$. 
  Then there is a $\sigmap \in (0, 1)$, depending only on $\mrK$, such that if $\s$, $k_0$ and $k_1$ are chosen as in
  Theorem~\ref{thm: big bang formation}, then there is an $\ve> 0$ such that if  $\mfI := (\S,\bh,\bk,\bp_{0}, \bp_{1})$ is a solution to the
  constraint equations \eqref{seq:con} with constant mean curvature $1/t_0$ and
  \begin{equation}\label{eq:smallnessdataonsing}
    \begin{split}
      \snorm{\bh - h(t_0, \cdot)}_{H^{k_1 + 3}(\S)} + \snorm{\bk - k(t_0, \cdot)}_{H^{k_1 + 3}(\S)}
       + \textstyle{\sum}_{i=0}^{1}\snorm{\bp_i -  \phi_i(t_0, \cdot)}_{H^{k_1 + 3}(\S)} & < \ve,
    \end{split}
  \end{equation}
  then the maximal globally hyperbolic development of $\mfI$ has a crushing big bang singularity in the sense of Theorem~\ref{thm: big bang formation}.
\end{thm}
\begin{remark}\label{remark:stabplusCauchystab}  
  Note that (\ref{eq:goodlapsecoeff}) is satisfied for all $t\in (0,t_{+})$ in case $V\equiv 0$, since the first term on the left hand side of
  (\ref{eq:goodlapsecoeff}) is bounded from below by $1/(nt^{2})$. Moreover, in general, the left hand side tends to infinity as $t\downarrow 0$;
  see Subsection~\ref{ssection:ProofStrategy} below.
\end{remark}
\begin{remark}
  Combining Remark~\ref{remark:dataonsingKasnerhigherdim}, Theorem~\ref{thm:stabilityofdevelopments} and Remark~\ref{remark:stabplusCauchystab}
  yields the conclusion that if the $p_{I}$ are distinct and if $p_{I}+p_{J}-p_{K}<1$ for all $I, J, K$, such that $I \neq J$, then
  the solutions in Example~\ref{example: spatially hom flat} exhibit stable big bang formation. Moreover, any starting time $t_{0}\in (0,\infty)$
  can be used in the statement of stability. In particular, in the non-degenerate setting, the stability statement in \cite[Theorem~6.1, pp.~905--908]{GIJ}
  follows as a corollary. For the degenerate case, see Example~\ref{example: result of fellows} below. 
\end{remark}
\begin{proof}
  Let $(\S, \mH, \mK, \Phi_0, \Phi_1)$ be the expansion-normalized initial data induced on the Cauchy hypersurface $\S_t$ in the associated locally
  crushing CMC development $(M, g, \phi)$. By compactness of $\S$, there is a $\mathring \sigma_p \in (0, 1)$, such that
  $\mrp_I + \mrp_J - \mrp_K < 1 - \mathring \sigma_p$.
  By choosing $l = \lceil n/2 + 1 \rceil$, the bound \eqref{eq:mKchhconvratesfN} and Sobolev embedding imply that
  \[
    \abs{\mK(t, \cdot) - \mrK}_{h_{\refer}}\leq C t^\de
  \]
  for all $t \in (0, t_+)$.
  Therefore, for a small enough $t_1 \in (0, t_+)$, there is a $\s_p \in (0, \mathring \s_p]$, such that
  \[
    \q_I(t, \cdot) + \q_J(t, \cdot) - \q_K(t, \cdot)
    < 1 - \sigma_p
  \]
  for all $t \in (0, t_1]$, where $\q_1, \hdots, \q_n$ are the eigenvalues associated with $\mK$.

  Let $\s$ be as in \eqref{eq: sigma condition} and choose regularity degrees $k_0$ and $k_1$ as in (\ref{seq: k_0 k_1 inequality}). By the
  smoothness of the initial data on the singularity and the compactness of $\S$, 
  \begin{equation}\label{eq:mrHetcbd}
    \snorm{\mrH^{-1}}_{C^{0}(\S)}+\snorm{\mrH}_{H^{k_1+2}(\S)} + \snorm{\mrK}_{H^{k_1+2}(\S)} +
    \textstyle{\sum}_{i=0}^{1}\snorm{\mrPhi_i}_{H^{k_1+2}(\S)} < \mathring \zeta_0
  \end{equation}
  for some $\mathring \zeta_0 > 0$. Moreover, $|\mrp_I-\mrp_J|>\mathring \zeta_{0}^{-1}$ for $I\neq J$. Combining (\ref{eq:mrHetcbd}) and
  \eqref{eq:mKchhconvratesfN}, with $l = k_1 + 2$, there is a $t_2 \in (0, t_+)$ and a constant $C$ such that 
  \begin{align*}
    &\snorm{\mH^{-1}}_{C^{0}(\S)}+\snorm{\mH}_{H^{k_1+2}(\S_t)} + \snorm{\mK}_{H^{k_1+2}(\S_t)} + \textstyle{\sum}_{i=0}^{1}\snorm{\Phi_i}_{H^{k_1+2}(\S_t)}    
    \leq \mathring \zeta_0 + Ct^\de
  \end{align*}
  for all $t \in (0, t_2]$. Similarly, by \eqref{eq:mKchhconvratesfN}, with $l$ larger than $n/2$, and Sobolev embedding, there is a
  $t_3\in (0, t_+)$ and $\zeta_0>0$ such that, for $t \in (0, t_3]$, $|p_I(t\cdot)-p_J(\cdot,t)|>\zeta_{0}^{-1}$ for $I\neq J$ and
  \[
  \snorm{\mH^{-1}(\cdot,t)}_{C^{0}(\S)}+\snorm{\mH(t, \cdot)}_{H^{k_1+2}(\S)} + \snorm{\mK(t, \cdot)}_{H^{k_1+2}(\S)}
  + \textstyle{\sum}_{i=0}^{1}\snorm{\Phi_i(t, \cdot)}_{H^{k_1+2}(\S)} < \zeta_0
  \]
  for all $t \in (0,t_4]$, where $t_4:=\min\{t_1, t_2,t_3\}$.
    
  Given $\zeta_0$, $\s_p$, $\s_V$, $k_0$, $k_1$, $(\S, h_\refer)$ and $(E_i)_{i=1}^n$ as above, let $\zeta_1 > 0$ be the constant provided by
  Theorem~\ref{thm: big bang formation}.
  By the definition of a locally crushing CMC development, the mean curvature $\theta(t)$ of the hypersurface $\S_t$ satisfies $\theta(t)> \zeta_1$ for all
  $t \in (0, \zeta_1^{-1})$. In conclusion, defining $\tau := \min\{t_4,\zeta_1^{-1}\}$, the induced CMC initial data
  $(\S, h(t, \cdot), k(t, \cdot), \phi_0(t, \cdot), \phi_1(t, \cdot))$ satisfies the assumptions in Theorem~\ref{thm: big bang formation}, with the same
  constants $\zeta_0, \s, k_0, k_1$, for any $t \in (0, \tau)$.
  Let now $t_0 \in (0, \tau)$. Since the conditions in Theorem~\ref{thm: big bang formation} are open and the map
  \begin{align*}
    \left( H^{k_1+2}\left(\S_{t_0}\right) \right)^4
    &\to \left( H^{k_1+2}\left(\S_{t_0}\right) \right)^4;\ \ \
    (\S, \bh, \bk, \bp_0, \bp_1)
    \mapsto (\S, \bmH, \bmK, \bP_0, \bP_1)
  \end{align*}
  is continuous at $(\S, h(t_0, \cdot), k(t_0, \cdot), \phi_0(t_0, \cdot), \phi_1(t_0, \cdot))$, the statement follows for $t_0\in (0,\tau)$ and
  $k_1$ in (\ref{eq:smallnessdataonsing}) replaced by $k_1-1$. However, combining this result with Cauchy stability, see
  Lemma~\ref{lemma:CauchystabEinstein}, yields the desired statement. Note that the loss of one derivative is due to the fact that the Cauchy
  stability argument is based on a second order system for the second fundamental form. Translating the conclusions to first
  order form entails a loss of one derivative of the spatial metric, and therefore the loss of one derivative in the smallness condition. 
\end{proof}

\subsection{The degenerate case}\label{ssection: the degenerate case}
It is also possible to obtain stability in degenerate cases; i.e~when the eigenvalues of the expansion-normalized Weingarten map are not all distinct.
More specifically, we prove past global non-linear stability of the following class of background solutions. 

\begin{definition}\label{def:quiescent model solution}
  Let $\S$ be a closed $n$-dimensional manifold and assume that it has a global frame $(E_{i})_{i=1}^{n}$ with dual co-frame $(\eta^{i})_{i=1}^{n}$.
  Define $h_\refer$ by demanding that $(E_{i})_{i=1}^{n}$ be an orthonormal frame. 
  Fix admissibility thresholds $\s_p$, $\s_V\in (0,1)$. Assume that there is an open interval $\mI$; a $\s_V$-admissible potential $V\in C^{\infty}(\rn{})$;
  $\phi\in C^{\infty}(\mI,\rn{})$; and $a_{i}\in C^{\infty}(\mI,(0,\infty))$, $i=1,\dots,n$, such that if 
  \begin{equation}\label{eq:g form quie mod sol}
    g:=-d\tau\otimes d\tau+\textstyle{\sum}_{i}a_{i}^{2}\eta^{i}\otimes\eta^{i}
  \end{equation}
  and $M:=\mI\times\S$, then $(M,g,\phi)$ is a solution to the Einstein-non-linear scalar field equations with potential $V$. Assume that
  $\mI=(\tau_{-},\tau_{+})$ and that $\theta(\tau)\rightarrow\infty$ as $\tau\downarrow \tau_{-}$, where $\theta(\tau)$ is the mean curvature
  of $\S_\tau$. Let $\tau_a\in\mI$ be such that $\theta(\tau)\geq 1$ for $\tau\leq\tau_a$, and let $\mK(\tau)$, $\mH(\tau)$, $\Phi_0(\tau)$,
  $\Phi_1(\tau)$ denote the expansion-normalized quantities induced on $\S_\tau$ by the solution for $\tau\leq\tau_a$; see
  Definitions~\ref{def: exp nor Weingarten map}--\ref{def:Phizodef}. Assume, in addition that there are real constants $C>0$ and $\delta>0$
  such that 
  \begin{equation}\label{eq: geometric degenerate convergence}
    \|\mK(\tau)-\mrK\|_{C^{0}(\S)}+\|\mH(\tau)-\mrH\|_{C^{0}(\S)}+\textstyle{\sum}_{i=0}^{1}|\Phi_i(\tau)-\mrPhi_i|\leq C[\theta(\tau)]^{-\delta}  
  \end{equation}
  for all $\tau\leq\tau_{a}$, where $\mrK$ (with eigenvalues $\mrp_i$); $\mrH$; and $\mrPhi_i$, $i=0,1$, are a $(1,1)$-tensor field, a Riemannian
  metric and two constants respectively. Assume, finally, that there is a $\s_p>0$ such that $\mrp_{i}+\mrp_{j}-\mrp_{k}<1-2\s_p$ for all $i\neq j$.
  Then $(M,g,\phi)$ is said to be a \textit{quiescent model solution}. 
\end{definition}
\begin{remark}
  The hypersurfaces $\S_\tau$ in a quiescent model solution have constant mean curvature. However, the time coordinate $\tau$ is proper time, not
  $1/\theta$. Changing time coordinate to $1/\theta$ would introduce a lapse function which would typically be different from $1$. 
\end{remark}    
\begin{remark}\label{remark: effective formulation degenerate case}
  Assumption (\ref{eq: geometric degenerate convergence}) can be reformulated to 
  \begin{equation}\label{eq:mKietcconvergence}
    |p_i(\tau)-\mrp_{i}|+|\hat{a}_{i}(\tau)-\mathring{\alpha}_{i}|
    +|\Phi_1(\tau)-\mrPhi_1|
    +|\phi(\tau)+\mrPhi_{1}\ln \theta(\tau)-\mrPhi_0|\leq C[\theta(\tau)]^{-\delta}  
  \end{equation}
  for all $\tau\leq\tau_{a}$ and $0<\mathring{\alpha}_i\in\rn{}$, $i=1,\dots,n$, where $p_i:=\tfrac{1}{\theta a_{i}}\d_{\tau}a_i$,
  $\hat{a}_i:=\theta^{\mrp_{i}}a_{i}$ and $\Phi_1=\tfrac{1}{\theta}\phi_\tau$. 
\end{remark}
\begin{proof}[Proof of Remark~\ref{remark: effective formulation degenerate case}]
  Assume (\ref{eq: geometric degenerate convergence}) to hold and denote the components of $\mK$, $\mH$ etc. with respect to
  $(E_{i})_{i=1}^{n}$ and $(\eta^{i})_{i=1}^{n}$ by $\mK^{i}_{\phantom{i}j}$, $\mH_{ij}$ etc. Then $\mK^{i}_{\phantom{i}j}=0$ for $i\neq j$ and
  $\mH_{ij}=0$ for $i\neq j$. Similar statements must thus hold for their limits. Moreover, since $\mK^{i}_{\phantom{i}j}$ and $\mH_{ij}$
  are independent of the spatial variables, the same must be true of their limits. To summarize,
  \[
  \mrK=\textstyle{\sum}_i\mrp_{i}E_i\otimes\eta^i,\ \ \
  \mrH=\textstyle{\sum}_i\mathring{\alpha}_{i}^{2}\eta^i\otimes\eta^i,\ \ \
  \mK=\textstyle{\sum}_ip_iE_i\otimes\eta^i,\ \ \
  \mH=\textstyle{\sum}_i\theta^{2p_i}a_{i}^{2}\eta^i\otimes\eta^i,
  \]
  for some constants $\mathring{\alpha}_i>0$ and $\mrp_i$. The first term on the left hand side of (\ref{eq:mKietcconvergence}) is thus bounded by the
  right hand side. Next, consider
  \[
  |\hat{a}_{i}(\tau)-\mathring{\alpha}_{i}|=|(\theta^{\mrp_{i}}a_{i})(\tau)-\mathring{\alpha}_{i}|
  \leq \mathring{\alpha}_{i}^{-1}[|(\theta^{2\mrp_{i}}a_{i}^{2})(\tau)-(\theta^{2p_{i}}a_{i}^{2})(\tau)|+|(\theta^{2p_{i}}a_{i}^{2})(\tau)-\mathring{\alpha}_{i}^2|].
  \]
  By assumption, the second term in the parenthesis on the far right hand side decays as desired. Since $p_i-\mrp_i$ has the desired decay, the
  same is true of the first term in the parenthesis (except for a logarithmic factor, which leads to a deterioration of the constant $\delta$).
  Thus the second term on the left hand side of (\ref{eq:mKietcconvergence}) satisfies the desired bound (up to a deterioration of $\delta$).
  The third term on the left hand side of (\ref{eq:mKietcconvergence}) satisfies the desired bound by assumption. To estimate the last term, note
  that
  \[
  |\phi(\tau)+\mrPhi_{1}\ln \theta(\tau)-\mrPhi_0|=|\Phi_0(\tau)-\Phi_1(\tau)\ln\theta(\tau)+\mrPhi_{1}\ln \theta(\tau)-\mrPhi_0|.
  \]
  Due to the assumptions, we obtain, up to a logarithm, the desired conclusion. The proof of the converse statement is similar and left to the reader.
\end{proof}

For quiescent model solutions, the following past global non-linear stability result holds.

\begin{thm}\label{thm:degenerate case}
  Fix a quiescent model solution as in Definition~\ref{def:quiescent model solution}. With terminology as in
  Definition~\ref{def:quiescent model solution}, let $\s$, $k_0$ and $k_1$ as in Theorem~\ref{thm: big bang formation} be given. Let
  $\tau_0\in\mI$ be such that if $h(\tau)$ and $k(\tau)$ are the first and second fundamental forms induced on $\S_\tau$ by the quiescent
  model solution, then 
  \begin{equation}\label{eq:condition for Cauchy stability}
    |k(\tau)|_{h(\tau)}^{2}+|\phi_{\tau}(\tau)|^{2}-2V\circ\phi(\tau)/(n-1)>0
  \end{equation}
  for all $\tau\leq\tau_0$. Then there is an  $\varepsilon>0$ such that if $\mfI := (\S,\bh,\bk,\bp_{0}, \bp_{1})$ is a solution to the constraint
  equations \eqref{seq:con} with constant mean curvature $\tr_{h(\tau_0)}k(\tau_0)$ and (\ref{eq:smallnessdataonsing}) holds with $t_0$ replaced by
  $\tau_0$, $\phi_0$ replaced by $\phi$ and $\phi_1$ replaced by $\phi_\tau$, then the maximal globally hyperbolic
  development of $\mfI$ has a crushing big bang singularity in the sense of Theorem~\ref{thm: big bang formation}, with the following modifications:
  the $\mrp_I$ are only $C^{0}$ and in the estimates (\ref{seq:pIasintro}) and (\ref{seq:Kretschmann Ricci asymptotics intro}), the $C^{k_0+1}$-norm
  has to be replaced by the $C^0$-norm. 
\end{thm}
\begin{remark}
  If $x\in \S$ is such that the $\mrp_I(x)$ are distinct, then conclusions similar to those mentioned in Remark~\ref{remark:FRaR non-deg} hold in
  a neighbourhood of $x$; see \cite[Remark~15, p.~6]{FGaR}. 
\end{remark}
\begin{proof}
  The proof is to be found in Section~\ref{sec: finishing the proof}.
\end{proof}

\begin{example}\label{example: result of fellows}
  Assuming that there is a $\s_p>0$ such that $p_I+p_J-p_K<1-2\s_p$ for $I\neq J$, the solutions described in
  Example~\ref{example: spatially hom flat} are quiescent model solutions. In particular, these solutions are past globally non-linearly stable
  starting at any initial hypersurface, since the condition (\ref{eq:condition for Cauchy stability}) is always satisfied in this case. In
  particular the conclusions in \cite[Theorem~1.6, p.~838]{GIJ} follow (with the exception of the statements concerning polarized
  $\mathrm{U}(1)$-symmetric perturbations). 
\end{example}


\subsection{Perturbing spatially locally homogeneous solutions}\label{ssection:stablebbfshs}
Next, we turn to the question of stability of spatially locally homogeneous solutions. Since we specify solutions via initial data,
it is convenient to recall \cite[Definition~1.1, p.~7]{RinModel} and \cite[Remark~1.2, p.~7]{RinModel} in detail.
\begin{definition}[Definition~1.1, \cite{RinModel}]\label{def:Bianchiid}
  \textit{Bianchi class A initial data for the Einstein non-linear scalar field equations}, with potential $V\in C^{\infty}(\rn{})$, consist of the
  following: a connected $3$-dimensional unimodular Lie group $G$; a left invariant metric $\bh$ on $G$; a left invariant symmetric covariant
  $2$-tensor field $\bk$ on $G$; and two constants $\bp_{0}$ and $\bp_{1}$ satisfying
  \begin{subequations}\label{seq:constraintsBA}
    \begin{align}
      \roScal_{\bh}-|\bk|_{\bh}^{2}+(\tr_{\bh}\bk)^{2} = & \bp_{1}^{2}+2V(\bp_{0}),\\
      \md\tr_{\bh}\bk-\rodiv_{\bh}\bk = & 0.
    \end{align}
  \end{subequations}
  The data are said to be \textit{trivial} if $\bh$ is flat, $3\bk=(\tr_{\bh}\bk)\bh$, $\bp_1=0$ and $V'(\bp_0)=0$. 
\end{definition}
\begin{remark}[Remark~1.2, \cite{RinModel}]\label{remark:unimodular}
  In order to define the notion of unimodularity,
  let $G$ be a Lie group and $\mfg$ the associated Lie algebra.
  Given $X\in \mfg$, define
  $\mathrm{ad}_{X}:\mfg\rightarrow\mfg$ by $\mathrm{ad}_{X}(Y)=[X,Y]$.
  Let $\eta_{G}\in \mfg^{*}$ be defined by $\eta_{G}(X)=\tr\ \mathrm{ad}_{X}$.
  Then $G$ is \textit{unimodular} if $\eta_{G}=0$ and \textit{non-unimodular} if $\eta_{G}\neq 0$. An alternate characterisation is that
  $G$ is unimodular if and only if $\rodiv_{h}X=0$ for every left invariant metric $h$ on $G$ and every left invariant vector field $X$ on $G$. 
\end{remark}
\begin{remark}
  Bianchi class A initial data can be divided into Bianchi types I, II, VI${}_{0}$, VII${}_{0}$, VIII and IX, corresponding to a
  classification of the Lie algebra of $G$; see \cite[Definition~1.5, p.~7]{RinModel} and \cite[Table~1.1, p.~8]{RinModel}. Next, initial data
  as in Definition~\ref{def:Bianchiid} can, beyond a Bianchi type, say $\mfT$,
  have a symmetry type, here denoted $\mfs$. Which symmetry types are allowed depends on the Bianchi type. However, initial data can be
  \textit{isotropic}, written iso, \cite[Definition~1.6, p.~8]{RinModel}; \textit{locally rotationally symmetric}, written LRS,
  \cite[Definition~1.8, p.~8]{RinModel}; \textit{permutation symmetric}, written per, \cite[Definition~1.11, p.~8]{RinModel}; and \textit{generic},
  written gen, meaning they are neither isotropic, locally rotationally symmetric or permutation symmetric. 
\end{remark}
\begin{remark}
  Simply connected initial data of Bianchi type VII${}_0$ which are either isotropic or LRS are isometric to initial data of Bianchi type I; see
  \cite[Lemma~A.6, p.~148]{RinModel}. For this reason, we exclude isotropic and LRS Bianchi type VII${}_0$ initial data in what follows.
\end{remark}
Due to the above remarks, it is convenient to introduce the following notation. 
\begin{definition}
  The set of non-trivial Bianchi class A initial data for the Einstein non-linear scalar field equations with potential $V$, which are neither
  isotropic nor LRS Bianchi type VII${}_0$, is denoted $\mB[V]$. The elements of $\mB[V]$ which are of Bianchi type $\mfT$ and symmetry type $\mfs$
  are denoted $\mB_{\mfT}^\mfs[V]$. The sets ${}^{\rosc}\mB[V]$ (and ${}^{\rosc}\mB_{\mfT}^\mfs[V]$) consist of the initial data in
  $\mB[V]$ ($\mB_{\mfT}^\mfs[V]$) such that the corresponding Lie group is simply connected. 
\end{definition}
\begin{remark}
  Given $V\in C^{\infty}(\rn{})$ and $\mfI\in\mB[V]$, there is a unique (up to translation of the time interval) associated so-called
  \textit{Bianchi class A non-linear scalar field development}, denoted $\mD[V](\mfI)$; see \cite[Definition~1.28, p.~11]{RinModel},
  \cite[Proposition~1.31, p.~11]{RinModel}
  and \cite[Definition~1.34, p.~11]{RinModel}. It is of interest to note that if $V\geq 0$, the only obstruction to global existence is that
  the mean curvature might blow up in finite time; see \cite[Remark~1.32, p.~11]{RinModel}. 
\end{remark}
\begin{remark}\label{remark:trivial id}
  Trivial initial data give rise to developments that do not have a big bang singularity; see \cite[Remark~1.13, p.~9]{RinModel}. They also cause
  problems when endowing the set of isometry classes of initial data with a smooth structure. For these reasons, we exclude trivial initial data. 
\end{remark}
Next, there is the following notion of initial data on the singularity in the Bianchi class A setting.
\begin{definition}[Definition~1.17, \cite{RinModel}]\label{def:ndvacidonbbssh}
  Let $G$ be a connected $3$-dimensional unimodular Lie group, $\mrH$ be a left invariant Riemannian metric on $G$, $\mrK$ be a left invariant
  $(1,1)$-tensor field on $G$ and $(\mrPhi_{0},\mrPhi_{1})\in\rn{2}$. Then $(G,\mrH,\mrK,\mrPhi_{0},\mrPhi_{1})$ are
  \textit{quiescent Bianchi class A initial data on the singularity for the Einstein non-linear scalar field equations}
  if
  \begin{enumerate}
  \item $\tr\mrK=1$ and $\mrK$ is symmetric with respect to $\mrH$.
  \item $\tr\mrK^{2}+\mrPhi_{1}^{2}=1$ and $\mathrm{div}_{\mrH}\mrK=0$.
  \item In case all the eigenvalues of $\mrK$ are $<1$ and there is one eigenvalue, say $\mrp_A$, satisfying $\mrp_A\leq 0$, then the vector subspace of
    $\mfg$, say $\mfh$, perpendicular to the eigenspace of $\mrp_A$ is a subalgebra of $\mfg$.
  \item If $1$ is an eigenvalue of $\mrK$, there is an orthonormal basis $\{e_{i}\}$ of $\mfg$ with respect to $\mrH$ such that
    $\mrK e_{1}=e_{1}$ and such that if $\Psi_t$ is defined by
    \[
      \Psi_{t}e_{1}=e_{1},\ \ \
      \Psi_{t}e_{2}=\cos(t)e_{2}+\sin(t)e_{3},\ \ \
      \Psi_{t}e_{3}=-\sin(t)e_{2}+\cos(t)e_{3},
    \]
    then $\Psi_t$ is a Lie algebra isomorphism for all $t$.
  \end{enumerate}  
\end{definition}
\begin{remark}
  While Definition~\ref{def:ndvacidonbbssh} is more restrictive than Definition~\ref{def:ndsfidonbbs} in that we only allow homogeneous initial data, it
  is more general in the sense that the eigenvalues of $\mrK$ need not be distinct; the condition (\ref{eq: quiescence sing initial data}) need not hold;
  the manifold need not be compact; and Definition~\ref{def:ndvacidonbbssh} even includes Cauchy horizons. 
\end{remark}
\begin{remark}
  It is possible to associate a Bianchi and symmetry type to quiescent Bianchi class A initial data on the singularity for the Einstein non-linear
  scalar field equations; see \cite[Definition~1.21, p.~10]{RinModel}.
\end{remark}
\begin{remark}
  Isotropic and LRS Bianchi type VII${}_0$ initial data on the singularity are of Bianchi type I; see \cite[Lemma~A.7, p.~149]{RinModel}. For
  this reason, we do not consider such data in what follows. 
\end{remark}
Below, we use the following terminology; cf. \cite[Definition~1.25, p.~10]{RinModel}.
\begin{definition}
  The set of quiescent Bianchi class A initial data on the singularity for the Einstein non-linear scalar field equations which are neither of
  isotropic nor of LRS Bianchi type VII${}_0$ is denoted $\mS$. The corresponding set of simply connected initial data on the singularity is denoted
  ${}^{\rosc}\mS$. Given a Bianchi class A type $\mfT$ and a symmetry type $\mfs$, the elements of $\mS$ (${}^{\rosc}\mS$) which are of Bianchi type
  $\mfT$ and symmetry type $\mfs$ are denoted $\mS_\mfT^\mfs$ (${}^{\rosc}\mS_\mfT^\mfs$). Finally, ${}^{\rosc}\mfS_\mfT^\mfs$ denotes the isometry classes
  of elements in ${}^{\rosc}\mS_\mfT^\mfs$; cf. \cite[Definition~1.20, p.~10]{RinModel}. 
\end{definition}

\subsubsection{Stability of developments corresponding to data on the singularity}
Given data as in Definition~\ref{def:ndvacidonbbssh}, there is, under suitable assumptions, a unique corresponding development inducing the given
data. Combining this result with Theorem~\ref{thm:degenerate case} yields the following conclusion. 

\begin{cor}\label{cor:shpastglnlstab}
  Fix an admissibility threshold $\s_{V}\in (0,1)$ and let $V$ be a $\s_{V}$-admissible potential. Let $\mfI=(G,\mrH,\mrK,\mrPhi_1,\mrPhi_0)\in {}^{\rosc}\mS$.
  Assume that the eigenvalues of $\mrK$ are all strictly positive. Then there is a unique associated Bianchi class A non-linear scalar field development,
  say $(M,g,\phi)$, inducing $\mfI$ on the singularity; see \cite[Definition~1.38, p.~12]{RinModel}, \cite[Theorem~1.45, p.~13]{RinModel} and \cite[Remark~1.46, p.~13]{RinModel}. In particular
  $M=(0,t_{+})\times G$, where the mean curvature of $\{t\}\times G$, say $\theta(t)$, satisfies $\theta(t)\rightarrow\infty$ as $t\downarrow 0$.
  Let $\Gamma$ be a co-compact
  subgroup of $G$ and let $\Sigma$ be the quotient of $G$ by $\Gamma$. Taking the quotient of $(M,g,\phi)$ by $\{\mathrm{Id}\}\times\Gamma$ induces a
  solution to the Einstein-non-linear scalar field equations, say $(M_{\roq},g_{\roq},\phi_{\roq})$, with $M_{\roq}=(0,t_{+})\times\Sigma$. Finally, there
  is a $\sigmap\in (0,1)$, depending only on $\mrK$, such that if $\s$, $k_{0}$ and $k_{1}$ are chosen as in the statement of
  Theorem~\ref{thm: big bang formation}, then, for $t_{0}$ small enough that (\ref{eq:goodlapsecoeff}) is satisfied for $t\leq t_{0}$, the
  following holds: There is an $\ve>0$ such that if $\mfI_{0}$ are the initial data induced on $\Sigma_{0}:=\{t_{0}\}\times\Sigma$ by
  $(M_{\roq},g_{\roq},\phi_{\roq})$, then CMC initial data for the Einstein-non-linear scalar field equations with mean curvature $1/t_{0}$ and closer to
  $\mfI_{0}$ than $\ve$ in the $H^{k_{1}+3}$-norm (in the sense that an analogue of (\ref{eq:smallnessdataonsing}) holds) give rise to maximal globally
  hyperbolic developments with the properties stated in Theorem~\ref{thm:degenerate case}.
\end{cor}
\begin{remark}\label{remark:exofcocompsubgroups}
  If $G$ is a unimodular $3$-dimensional Lie group, there are co-compact subgroups $\Gamma$ of $G$; see \cite{rav}.
\end{remark}
\begin{proof}
  Due to the proof of \cite[Theorem~1.45, p.~13]{RinModel}, it is clear that the development $(M_{\roq},g_{\roq},\phi_{\roq})$ is a quiescent model solution
  in the sense of Definition~\ref{def:quiescent model solution}; in order to obtain this conclusion, we used the fact that $\Gamma$ is a subgroup of
  $G$, so that the form (\ref{eq:g form quie mod sol}) which holds for $(M,g)$ (due to the proof of \cite[Theorem~1.45, p.~13]{RinModel})
  decends to the quotient $(M_\roq,g_\roq)$. The desired conclusion therefore follows from Theorem~\ref{thm:degenerate case}.
\end{proof}

\subsubsection{Asymptotics in the direction of the singularity}

Next, it is of interest to start with initial data in $\mB[V]$ and to analyze the asymptotics in the direction of the singularity. In order to illustrate
why this is not straightforward, note that the Bianchi type IX setting includes not only solutions that induce data on the singularity, but also de Sitter
space, which is expanding both to the future and to the past; vacuum solutions which exhibit chaotic dynamcs in the direction of the singularity etc. In
order to exclude solutions similar to de Sitter space or the Einstein static universe (in particular, in order to restrict our attention to solutions that
actually have a big bang singularity), we introduce a notion of \textit{pseudo positive} initial data; see \cite[Definition~1.53, p.~15]{RinModel} (since
the definition is somewhat technical, we refrain from repeating the details here). The set of pseudo positive elements of $\mB_{\mathrm{IX}}^\mfs[V]$ is denoted
$\mB_{\mathrm{IX},\mathrm{pp}}^\mfs[V]$. We also need the following terminology.
\begin{definition}[Definition~1.44, \cite{RinModel}]\label{def:mfP a inf}
  Let $\alpha_V\in [0,\infty)$ and $k\in\nn{}_0$. Then the set of $V\in C^{\infty}(\rn{})$ such that there is a constant $c_k<\infty$ with the property that
  \begin{equation}\label{eq:V k-derivatives exp estimate}
    \textstyle{\sum}_{l=0}^k|V^{(l)}(s)|\leq c_ke^{\sqrt{6}\alpha_V|s|}
  \end{equation}
  for all $s\in\rn{}$ is denoted $\mfP_{\alpha_V}^k$. Moreover, $\mfP_{\alpha_V}^{\infty}:=\cap_{l=0}^{\infty}\mfP_{\alpha_V}^{l}$.
\end{definition}
Combining \cite[Proposition~1.80, p.~20]{RinModel} and \cite[Proposition~1.82, p.~20]{RinModel} then yields the following conclusion.

\begin{prop}\label{prop:dichotomy}
  Let $\mfT$ be a Bianchi class A type, $\mfs\in\{\iso,\roLRS,\roper,\rogen\}$ and $V\in \mfP_{\alpha_V}^1$ be non-negative. Assume that $\alpha_{V}\in (0,1)$ in
  the case of anisotropic Bianchi type I and non-LRS Bianchi type II; and that $\alpha_{V}\in (0,1/3)$ otherwise. Let $\mfI\in \mB_{\mfT}^\mfs[V]$, assume
  that $\tr_{\bge}\bk\geq 0$; that $(\mfT,\mfs)\neq (\mrI,\iso)$; and that $\mfI\in \mB_{\mrIX,\mathrm{pp}}[V]$ in case $\mfT=\mrIX$. Let $(M,g,\phi)=\mD[V](\mfI)$.
  Then the associated existence interval is of the form $(0,t_+)$ and $\theta(t)\rightarrow\infty$ as $t\downarrow 0$. Moreover, there are two
  possibilities. Either there is a $t_{0}>0$ and a $C\in\rn{}$ such that $|\theta(t)\phi_{t}(t)|\leq C$ for all $t\leq t_{0}$; or $\phi_{t}(t)/\theta(t)$
  converges to a non-zero limit as $t\downarrow 0$. 
\end{prop}
\begin{remark}
  It would be desirable to prove the result for $\alpha_V\in (0,1)$. However, the method of proof imposes a, conjecturally artificial, restriction on
  $\alpha_V$. 
\end{remark}
This result naturally leads to the following terminology.
\begin{definition}[Definition~1.83, \cite{RinModel}]\label{def:matter and vacuum dominated}
  A Bianchi class A non-linear scalar field development as in Proposition~\ref{prop:dichotomy} is said to be \textit{matter dominated} if
  $\phi_{t}(t)/\theta(t)$ converges to a non-zero limit as $t\downarrow 0$, and is said to be \textit{vacuum dominated} otherwise. 
\end{definition}
With this terminology at our disposal, we can formulate the main result concerning the asymptotics in the direction of the singularity.
\begin{thm}[Theorem~1.85, \cite{RinModel}]\label{thm:dev inducing data on the sing}
  Let $\mfT$ be a Bianchi class A type, $\mfs\in\{\iso,\roLRS,\roper,\rogen\}$ and $V\in \mfP_{\alpha_V}^1$ be non-negative, where $\alpha_V\in (0,1)$ in case
  of Bianchi type I and non-LRS Bianchi type II; and $\alpha_V\in (0,1/3)$ otherwise. Assume $(\mfT,\mfs)\neq (\mrI,\iso)$ and let
  $\mfI\in\mB_{\mfT}^\mfs[V]$ with $\tr_{\bge}\bk\geq 0$. In case
  $\mfT=\mrIX$ assume, in addition, that $\mfI\in\mB_{\mrIX,\mathrm{pp}}^\mfs[V]$. Then the development $\mD[V](\mfI)$ induces initial data on the singularity
  unless it is vacuum dominated, $\mfs=\rogen$ and $\mfT\in\{\mrVIII,\mrIX\}$. Finally, if $\mfs=\rogen$, $\mfT\in\{\mrVIII,\mrIX\}$ and
  $\mD[V](\mfI)$ is vacuum dominated, then the expansion normalised Weingarten map $\mK$ does not converge. In fact, the $\alpha$-limit set of the
  eigenvalues of $\mK$ contains two distinct points on the Kasner circle and the line connecting them. Moreover, $\roScal_{\bge}/\theta^2$ does not converge
  to zero. 
\end{thm}
This result can be substantially improved to guarantee that the Einstein flow generates a diffeomorphism between isometry classes of developments and
isometry classes of data on the singularity. More specifically, we have the following informal reformulation of
\cite[Corollary~1.88, p.~21]{RinModel}. 
\begin{cor}[Corollary~1.88, \cite{RinModel}]\label{cor:nBIX diffeo}
  Let $\mfT$ be a Bianchi class A type, $\mfs\in\{\iso,\roLRS,\roper,\rogen\}$ and $V\in \mfP_{\alpha_V}^\infty$ be non-negative, where $\alpha_V\in (0,1)$
  in case
  of Bianchi type I and non-LRS Bianchi type II; and $\alpha_V\in (0,1/3)$ otherwise. Assume that $(\mfT,\mfs)\neq(\mrI,\iso)$ and $\mfT\neq\mrIX$.
  Then, if $(\mfT,\mfs)\neq (\mrVIII,\rogen)$, the Einstein flow generates a diffeomorphism between isometry classes of developments $\mD[V](\mfI)$,
  for $\mfI\in {}^{\rosc}\mB_\mfT^\mfs[V]$, and ${}^{\rosc}\mfS_\mfT^\mfs$. Similarly, the Einstein flow generates a
  diffeomorphism between isometry classes of matter dominated developments $\mD[V](\mfI)$, for $\mfI\in {}^{\rosc}\mB_{\mrVIII}^{\rogen}[V]$, and
  ${}^{\rosc}\mfS_{\mrVIII}^{\rogen}$. 
\end{cor}
The statement in the case of Bianchi type IX is slightly different.
\begin{cor}[Corollary~1.92, \cite{RinModel}]\label{cor:BIX diffeo}
  Let $\mfs\in\{\iso,\roLRS,\rogen\}$ and $V\in \mfP_{\alpha_V}^\infty$ be non-negative, where $\alpha_V\in (0,1/3)$. Then, if $\mfs\neq\rogen$, the Einstein
  flow generates a diffeomorphism between isometry classes of developments $\mD[V](\mfI)$, for $\mfI\in {}^{\rosc}\mB_{\mrIX,\mathrm{pp}}^\mfs[V]$, and
  ${}^{\rosc}\mfS_\mfT^\mfs$. Similarly, the Einstein flow generates a diffeomorphism between isometry classes of matter dominated developments
  $\mD[V](\mfI)$, for $\mfI\in {}^{\rosc}\mB_{\mrIX,\mathrm{pp}}^{\rogen}[V]$, and  ${}^{\rosc}\mfS_{\mrIX}^{\rogen}$. 
\end{cor}
Due to the above observations, we have the following conclusions.

\subsubsection{Bianchi types VIII and IX}\label{sssection:BVIII a IX}
If $\mfI\in {}^{\rosc}\mB_{\mrIX,\mathrm{pp}}^{\rogen}[V]$ or $\mfI\in {}^{\rosc}\mB_{\mrVIII}^{\rogen}[V]$, then, under the assumptions of
Theorem~\ref{thm:dev inducing data on the sing}, the corresponding development either induces data on the singularity or it is vacuum dominated and
exhibits oscillations in the direction of the singularity. In the former case, appropriate quotients of the development exhibit stable big bang formation
due to Corollary~\ref{cor:shpastglnlstab}. In the case of LRS Bianchi type VIII or IX, the vacuum dominated developments correspond to a positive
codimension submanifold of the set of isometry classes of developments; see Corollaries~\ref{cor:nBIX diffeo} and \ref{cor:BIX diffeo}. Appropriate
quotients of the matter dominated ones induce data on the singularity such that Corollary~\ref{cor:shpastglnlstab} applies. In the isotropic setting,
$\mfI\in {}^{\rosc}\mB_{\mrIX,\mathrm{pp}}^{\iso}[V]$ always give rise to developments such that stability in the direction of the singularity follows from
Corollary~\ref{cor:shpastglnlstab}. In the Einstein-scalar field setting (i.e., when the potential vanishes), non-vacuum solutions are always matter
dominated. In particular, Corollary~\ref{cor:shpastglnlstab} then always applies. 

\subsubsection{Bianchi types VII${}_0$, VI${}_0$, II and anisotropic Bianchi type I} For these Bianchi types, combined with corresponding symmetry
types, we have a diffeomorphism between isometry classes of developments and isometry classes of data on the singularity; see
Corollary~\ref{cor:nBIX diffeo}. However, the data on the singularity need not be such that the corresponding $\mrK$ has positive eigenvalues. In fact,
there is, for each of these Bianchi types, an open subset such that this condition is violated. On the other hand, there is also an open set such that
it is fulfilled. In the set of isometry classes of developments there is thus an open subset defined by the condition that
Corollary~\ref{cor:shpastglnlstab} applies to appropriate quotients of the corresponding developments.

\subsubsection{Isotropic Bianchi type I}\label{sssection:iso B I}
One complication that arises in the case of isotropic Bianchi type I is that if $s_0\in\rn{}$ is such that
$V'(s_0)=0$, then initial data with $\bp_1=0$ and $\bp_0=s_0$ are trivial. This means that the corresponding development does not have a crushing
singularity; see Remark~\ref{remark:trivial id} and \cite[Remark~1.13, p.~9]{RinModel}. In addition, developments could be such that $\phi_t$
converges to zero and $\phi$ converges to $s_0$. If $V'$ has infinitely many distinct zeros at infinitely many different values of $V$, then
there are infinitely many different solutions without a crushing singularity. This is a rather exotic situation we wish to avoid. For this reason,
we, in the context of isotropic Bianchi type I solutions, introduce additional conditions on the potential.
\begin{definition}[Definition~1.62, \cite{RinModel}]\label{def:mfPdef}
  Let $V\in C^{\infty}(\rn{})$. If $V(s)$ converges to a finite number as $s\rightarrow \pm\infty$, denote the limit by $v_{\infty,\pm}$. If $V(s)$ does
  not converge to a finite number as $s\rightarrow \pm\infty$, let $v_{\infty,\pm}:=0$. Define
  \[
  v_{\max}(V):=\sup(\{v_{\infty,+},v_{\infty,-}\}\cup\{V(s_0)\, |\, s_0\in\rn{},\, V'(s_0)=0\}).
  \]
  Let $\mfP_{\ropar}$
  denote the set of $V\in C^{\infty}(\rn{})$ such that $V(s)\geq 0$ for all $s\in\rn{}$; $V'$ is bounded on every interval on which $V$
  is bounded; $V'(s)$ tends to a limit (finite or infinite) as $s\rightarrow \infty$ and as $s\rightarrow -\infty$; and $v_{\max}(V)<\infty$.
\end{definition}
\begin{remark}[Remark~1.63, \cite{RinModel}]
  The set $\mfP_{\ropar}$ includes, e.g., the following three classes of potentials: non-negative polynomials; non-negative smooth functions such that
  $V>0$ outside a compact set and such that $V'/V$ converges to a non-zero limit as $s\rightarrow\pm\infty$; bounded non-negative smooth functions such
  that $V'(s)\rightarrow 0$ as $s\rightarrow\pm\infty$. 
\end{remark}
In the isotropic Bianchi type I setting, it can be calculated that ${}^{\rosc}\mfS_{\mrI}^{\iso}$ is diffeomorphic to two copies of $\rn{}$; see
\cite[Section~2.3]{RinModel}. However, if we let ${}^{\rosc}\mfD_{\mrI,\roc}^{\iso}[V]$ denote isometry classes of developments with a crushing singularity
arising from isotropic and simply connected Bianchi type I initial data with a potential $V$, then it can be calculated that, depending on the potential
$V$, ${}^{\rosc}\mfD_{\mrI,\roc}^{\iso}[V]$ has one of three possible topologies: two disjoint copies of $\rn{}$; $\rn{}$; and $\sn{1}$. This statement
is justified in the paragraph above the statement of \cite[Theorem~1.98, p.~23]{RinModel}. For this reason,
it is, in general, not possible for the Einstein flow to generate a diffeomorphism between ${}^{\rosc}\mfD_{\mrI,\roc}^{\iso}[V]$ and 
${}^{\rosc}\mfS_{\mrI}^{\iso}$. If $V$ is bounded, these sets have the same topology, and we can hope for a diffeomorphism. In fact, this is what happens,
as is illustrated by the following reformulation of \cite[Theorem~1.98, p.~23]{RinModel}.
\begin{thm}[Theorem~1.98, \cite{RinModel}]
  Assume $V\in C^{\infty}(\rn{})$ to be bounded and to be such that $V\in\mfP_{\ropar}\cap\mfP_{\alpha_V}^\infty$ for some $\alpha_V\in (0,1)$. Then the Einstein
  flow generates a diffeomorphism from ${}^{\rosc}\mfD_{\mrI,\roc}^{\iso}[V]$ to ${}^{\rosc}\mfS_{\mrI}^{\iso}$. 
\end{thm}
\begin{remark}
  Due to the isotropy, the eigenvalues of $\mrK$ all equal $1/3$. This means that Corollary~\ref{cor:shpastglnlstab} applies to appropriate quotients of
  the developments discussed in the theorem. 
\end{remark}
If $V$ is unbounded in one direction and bounded in the other, then ${}^{\rosc}\mfD_{\mrI,\roc}^{\iso}[V]$ is diffeomorphic to one copy of $\rn{}$. The
simplest way to go from one copy of $\rn{}$ to two copies of $\rn{}$ is to remove one point. Naively, one could then hope that there is one unique
solution that does not induce data on the singularity, but that all others do. Similarly, if $V$ is unbounded in both directions, 
${}^{\rosc}\mfD_{\mrI,\roc}^{\iso}[V]$ is diffeomorphic to $\sn{1}$. Removing two points from this set thus yields a set
diffeomorphic to ${}^{\rosc}\mfS_{\mrI}^{\iso}$. Again, one would thus naively expect that there are precisely two solutions that do not induce data on
the singularity. Making slightly stronger assumptions concerning the potential, the above expectations turn out to be justified, as is seen from the
following slight reformulation of \cite[Theorem~101, p.~24]{RinModel}.
\begin{thm}[Theorem~1.101, \cite{RinModel}]\label{thm:asympt as exp pot intro}
  Assume $0\leq V\in C^{\infty}(\rn{})$ and that there are constants $C_V$ and $M$ such that $V(s)>0$ and 
  \begin{equation}\label{eq:ln V biss est intro}
    \left|\left(\ln V\right)''(s)\right|\leq C_V\ldr{s}^{-2}
  \end{equation}
  for all $|s|\geq M$. This means that $(\ln V)'(s)$ converges to limits as $s\rightarrow\pm \infty$. Call the limits $\lambda_\pm$ and assume
  that $-\sqrt{6}<\lambda_-<0$ and that $0<\lambda_+<\sqrt{6}$. Let $\theta\in C^{\infty}(J,(0,\infty))$ and $\phi\in C^{\infty}(J,\rn{})$ be the
  mean curvature and the scalar field of a development corresponding to non-trivial, isotropic Bianchi type I initial data, where $J=(t_-,t_+)$
  is the maximal existence interval. Assuming that $\theta$ is unbounded, there are the following, mutually exclusive, cases:
  \begin{enumerate}[(i)]
  \item The solution is such that
    \begin{equation}\label{eq:phi exc limit intro}
      \lim_{t\rightarrow t_-}\left[3\phi_t(t)/\theta(t)+(\ln V)'[\phi(t)]\right]=0
    \end{equation}
    holds and $\phi(t)\rightarrow\infty$ as $t\rightarrow t_-$. Up to time translation, there is exactly one such solution, and its image is a smooth
    submanifold of the state space. 
  \item The solution is such that (\ref{eq:phi exc limit intro}) holds and $\phi(t)\rightarrow-\infty$ as $t\rightarrow t_-$. Up to
    time translation, there is exactly one such solution, and its image is a smooth submanifold of the state space. 
  \item The solution has a crushing singularity and induces data on the singularity. 
  \end{enumerate}
  Moreover, assuming, in addition, $V\in\mfP_{\ropar}\cap\mfP_{\alpha_V}^\infty$ for some $\alpha_V\in (0,1)$
  and removing the two unique solutions mentioned in (i) and (ii) from the set of isometry classes ${}^{\rosc}\mfD_{\mrI,\roc}^{\iso}[V]$ yields
  a set which is diffeomorphic to ${}^{\rosc}\mfS_\mrI^{\iso}$ via the Einstein flow. 
\end{thm}
\begin{remark}
  Similar conclusions hold if $V$ is bounded in one direction and unbounded in one direction; see \cite[Remarks~1.104--1.106, p.~24]{RinModel}.
\end{remark}
\begin{remark}
  Removing the developments corresponding to (i) and (ii), all other developments are such that Corollary~\ref{cor:shpastglnlstab} applies to appropriate quotients.
\end{remark}

\subsubsection{The hyperbolic setting}\label{sssection:hyp setting}
In the hyperbolic setting, we are interested in the following class of initial data.
\begin{definition}[Definition~1.107, \cite{RinModel}]\label{def:id k minus one}
  \textit{Locally homogeneous and isotropic negative curvature initial data for the Einstein non-linear scalar field equations}, 
  with potential $V\in C^{\infty}(\rn{})$, consist of the following: a complete hyperbolic $3$-manifold $(\bM,\bge)$; a covariant
  $2$-tensor field $\bk$ on $\bM$ which is a non-negative constant multiple of $\bge$; and two constants $\bp_{0}$ and $\bp_{1}$ satisfying:
  \begin{equation}\label{eq:ham con id k minus one}
    \roScal_{\bge}-|\bk|_{\bge}^{2}+(\tr_{\bge}\bk)^{2} = \bp_{1}^{2}+2V(\bp_{0}).
  \end{equation}
  The data are said to be \textit{trivial} if $\bp_1=0$ and $V'(\bp_0)=0$. Let $\mN[V]$
  denote the set of all locally homogeneous and isotropic negative curvature initial data for the Einstein non-linear scalar field
  equations with potential $V$. 
\end{definition}
\begin{remark}
  If $V$ is non-negative and $\mfI\in\mN[V]$, then, due to \cite[Remark~1.112, p.~26]{RinModel}, there is a unique spatially locally homogeneous
  and isotropic non-linear scalar field development of $\mfI$ with a crushing singularity (this terminology is introduced in
  \cite[Definition~1.111, p.~26]{RinModel}) and an existence interval $J$ which can be assumed to equal $(0,\infty)$. We denote this development
  by $\mD[V](\mfI)$. 
\end{remark}
Trivial data lead to solutions to Einstein's vacuum equations with a cosmological constant $\Lambda=V(\bp_0)$. If $V(\bp_0)=0$, the solution is
the Milne model, and if $V(\bp_0)>0$, the solution is a generalization of the Milne model with a positive cosmological constant; see
\cite[Remark~1.108, p.~25]{RinModel}. In fact, if $V\geq 0$ and $\phi_\infty\in\rn{}$, there are unique smooth functions
$a:(0,\infty)\rightarrow (0,\infty)$ and $\phi:(0,\infty)\rightarrow \rn{}$ such that if $(\bM,\bge_-)$ is a complete hyperbolic $3$-manifold
with scalar curvature $-6$, $M=\bM\times (0,\infty)$ and $g$ is defined by $g=-dt\otimes dt+a^2(t)\bge_-$, then $(M,g,\phi)$ is a solution
to the Einstein-non-linear scalar field equations with $\phi(t)\rightarrow\phi_\infty$ and $\theta(t)\rightarrow\infty$ as $t\downarrow 0$;
see \cite[Proposition~1.109, p.~25]{RinModel}. This solution asymptotes to a solution to the Einstein vacuum equations with a cosmological
constant $\Lambda:=V(\phi_\infty)$. Note that in the case of the Einstein-scalar field equations (i.e., when the potential vanishes), then
all of these solutions are the Milne model (since the value of the scalar field is irrelevant if the potential is a constant). The solutions
obtained in \cite[Proposition~1.109, p.~25]{RinModel} do not induce data on the singularity; see \cite[Remark~1.114, p.~26]{RinModel}.

On the other hand, there is a natural notion of initial data on the singularity in this setting.
\begin{definition}[Definition~1.110, p.~26]\label{def:ndvacidonbbssh k minus one}
  Let $(\bM,\mrH)$ be a complete $3$-dimensional hyperbolic manifold, $\mrK$ be the $(1,1)$-tensor field $\mrK=\Id/3$ on $\bM$ and
  $(\mrPhi_{0},\mrPhi_{1})\in\rn{2}$. Then $(\bM,\mrH,\mrK,\mrPhi_{0},\mrPhi_{1})$ are \textit{locally homogeneous and isotropic negative curvature initial
    data on the singularity for the Einstein non-linear scalar field equations}
  if $\mrPhi_1^2=2/3$. The set of such data is denoted $\mS_-^{\iso}$.
\end{definition}
Given initial data on the singularity, there is, again, a corresponding development.
\begin{prop}[Proposition~1.115, \cite{RinModel}]\label{prop:dos ind dev k minus one}
  Let $V\in \mfP_{\alpha_V}^2$ for some $\alpha_V\in (0,1)$ and $\mfI_\infty\in \mS_-^{\iso}$; see Definition~\ref{def:ndvacidonbbssh k minus one}. Then there
  is a unique (up to time translation) development in the sense of \cite[Definition~1.111, p.~26]{RinModel} which induces the data $\mfI_\infty$ on the
  singularity in the sense of \cite[Definition~1.113, p.~26]{RinModel}.
\end{prop}
Making slightly stronger assumptions on the potential, it can be verified that there are only two outcomes possible, given initial data; either
the development induces initial data on the singularity or it asymptotes to a Milne solution (or a generalization thereof with a positive cosmological
constant); see \cite[Proposition~1.117, p.~27]{RinModel}. It only remains to determine the relative frequency of the two outcomes. To this end, it
is convenient to fix a complete hyperbolic $3$-manifold $(\bM,\bge_-)$ with scalar curvature $-6$. If $(\bM,\bge,\bk,\bp_0,\bp_1)\in\mN[V]$ are such
that $\bge=\alpha^2\bge_-$ and $\bk=\alpha\beta\bge_-$, then (\ref{eq:ham con id k minus one}) reads
\begin{equation}\label{eq:ham con al be k minus one}
  6\b^2=\bp_1^2+\tfrac{6}{\alpha^2}+2V(\bp_0);
\end{equation}
see \cite[(1.34), p.~27]{RinModel}. As in \cite[(1.35), p.~27]{RinModel}, we therefore introduce 
\begin{equation}\label{eq:sfN minus V def}
  \sfN_-[V]:=\{(\b,\bp_0,\bp_1)\in [0,\infty)\times\rn{2}\ |\ 6\b^2-\bp_1^2-2V(\bp_0)>0\}.
\end{equation}
Since $\alpha>0$ is uniquely determined by (\ref{eq:ham con al be k minus one}), given $(\b,\bp_0,\bp_1)\in \sfN_-[V]$, $\sfN_-[V]$ parametrises
initial data, given $(\bM,\bge_-)$. Finally, we recall \cite[Proposition~1.118, p.~27]{RinModel}. 
\begin{prop}[Proposition~1.118, \cite{RinModel}]\label{prop:generic k minus one intro}
  Let $0\leq V\in C^\infty(\rn{})$ and $(\bM,\bge_-)$ be a complete hyperbolic $3$-manifold with scalar curvature $-6$. Then the subset of $\sfN_-[V]$
  that gives rise to developments with the property that $\phi(t)$ converges to a finite number is contained in a countable union of codimension one
  submanifolds. In this sense, the outcome represented by \cite[Proposition~1.109, p.~25]{RinModel} corresponds to a set of initial data which
  is both Baire and Lebesgue non-generic. 
\end{prop}
\begin{remark}
  Combining this result with the assumption that $0\leq V\in\mfP_{\alpha_V}^1$, where $\alpha_V\in (0,1/3)$, and \cite[Proposition~1.117, p.~27]{RinModel}
  yields the conclusion that generic solutions induce data on the singularity. This means that Theorem~\ref{thm:degenerate case} applies to closed
  quotients of generic solutions. 
\end{remark}

\subsubsection{Comparison with previous results}

The first stability result in the direction of the singularity, \cite{rasq}, concerns perturbations of isotropic Bianchi type I solutions to the
Einstein-scalar field and the Einstein-stiff fluid equations. The scalar field part of this result follows as a consequence of our work;
see Example~\ref{example: result of fellows}. However, the discussion in Subsubsection~\ref{sssection:iso B I} provides a generalization
to the isotropic Bianchi type I setting with non-trivial potentials. The stability results in \cite{rsh} are special cases of the results in
\cite{GIJ}, and, excluding the results concerning $\mathrm{U}(1)$-symmetric solutions, the stability statements contained in \cite{GIJ} are
special cases of the conclusions obtained here; see Example~\ref{example: result of fellows}. The article \cite{specks3} concerns isotropic
Bianchi type IX solutions. However, the stability results of the present article yield past and future global non-linear stability of all Bianchi
type IX solutions to the Einstein-scalar field equations, assuming the scalar field matter is non-trivial. In fact, we also obtain results more
generally for non-trivial potentials; see Subsubsection~\ref{sssection:BVIII a IX} for details. Finally, our discussion concerning the hyperbolic
setting, see Subsubsection~\ref{sssection:hyp setting}, contains \cite[Theorem~1.1, p.~1616]{fau} as a special case. In fact, we also treat large
classes of potentials. 

\subsubsection{Past and future global non-linear stability}
Combining the past global non-linear stability results of the previous subsubsections with future global non-linear stability results such
as those contained in \cite[Proposition~1.119, p.~28]{RinModel}, \cite[Corollary~1.120, p.~28]{RinModel},
\cite[Proposition~1.124, p.~29]{RinModel} (combined with \cite[Theorem~4, pp.~134--135]{RinInv}) and \cite[Theorem~3, p.~162]{RinPL} yields large
classes of spacetimes which are both past and future globally non-linearly stable. We refrain from writing down the details.

\subsection{Strategy of the proof} \label{ssection:ProofStrategy}

The proof of the main theorem can roughly speaking be divided into three steps.
First, we construct an approximate solution using the assumptions
concerning the expansion-normalized quantities.
We refer to it as the \textit{scaffold}.
Second, we use a bootstrap argument to control the deviation
between the actual solution and the scaffold.
The outcome of the bootstrap argument is past global existence and rough bounds on the solution.
Third, we derive more detailed information concerning the asymptotics.

\textit{The scaffold.} In our setting, we expect the algebraic condition (\ref{cond: alternating q}), the
initial bound on the expansion-normalized quantities and the requirement of sufficiently large initial
mean curvature to imply that the expansion-normalized quantities $(\mH, \mK, \Phi_0, \Phi_1)$ converge to a limit. Moreover,
the smaller the initial time $t_0 = \bth^{-1}$, the smaller one expects the deviation of $(\mH, \mK, \Phi_0, \Phi_1)$
from the initial quantities $(\bmH, \bmK, \bP_0, \bP_1)$ to be.
In Example~\ref{example: spatially hom flat},
the quantities $(\mH, \mK, \Phi_0, \Phi_1)$ are even independent of time. 
By analogy with Example~\ref{example: spatially hom flat},
we therefore construct a spacetime with a CMC foliation,
vanishing shift and lapse equal to $1$, and a scalar function, 
such that the induced expansion-normalized quantities are
equal to $(\bmH, \bmK, \bP_0, \bP_1)$ for all time.
We call this \emph{the scaffold}. 
Note, of course, that the scaffold typically does not satisfy the Einstein-non-linear scalar field equations.
However, we expect it to remain close to the solution with the same initial data, if the initial time $t_0 = \bth^{-1}$ is sufficiently small.
One could therefore view the scaffold as the zeroth order approximation of the actual solution.
Our method to construct the scaffold in the proof of the main theorem, Theorem \ref{thm: big bang formation}, is based on the assumption of
non-degeneracy, i.e., $|\bq_I-\bq_J|>\zeta_0^{-1}$ for $I\neq J$.
The construction comes with natural estimates for the scaffold in terms of $\zeta_0$.
However, if one is interested in proving stability of a specific background solution, that background solution can be used to construct the scaffold, and then it is not necessary to assume non-degeneracy;
see Subsection~\ref{ssection: the degenerate case}.

\textit{The bootstrap argument.}
The rough setup for the bootstrap argument is very similar to the corresponding argument in \cite{GIJ}.
We use a gauge with constant mean curvature (CMC) and vanishing shift vector field, with the time coordinate $t = \theta^{-1}$.
We also use a Fermi-Walker propagated frame.
Moreover, we insist on detailed control of a low number of derivatives in $C^0$, represented by a quantity, say $\mbL$, and rough control over a high number of derivatives in $L^2$, represented by, say, $\mbH$.
We also borrow derivatives from $\mbH$ to control some extra derivatives in $C^0$.
Finally, we separately derive estimates for $\mbL$ and $\mbH$ in order to close the bootstrap,
and the estimates for $\mbL$ are, roughly speaking, of ODE-nature,
while in the estimates for $\mbH$,
we crucially need to use the fact that the momentum constraint is satisfied.
There are, however, significant differences in the analysis in our setting and the analysis of \cite{GIJ}.
In \cite{GIJ}, the authors control the distance to a spatially homogeneous solution.
In our case, the scaffold is typically neither a solution, nor is it spatially homogeneous.
At some points, this represents a significant difficulty and requires a refined analysis.
For example, in the proof of the energy estimates, it is necessary to commute spatial derivatives with an operator of the form $-t^{-1}(\bq_I+\bq_J-\bq_K)-\d_t$.
In the case of \cite{GIJ}, this commutator vanishes.
In our case, the spatial derivatives of the $\bq_I$ can be arbitrarily large.
This causes complications, which necessitate a different definition of the energy $\mbL$, with
weights depending on the number of derivatives and $C^{k}$-bounds on the $\bq_I$. 
In addition, in the energy estimates for $\mbH$, we require tailored bounds, extracting the terms in which the eigenvalues $\bq_I$ are not differentiated. 
We derive the necessary estimates in the appendix.
In fact, when we appeal to the momentum constraint,
we require one more derivative for the definition of the scaffold than we recover from the energy estimate, due to the spatial inhomogeneity of the scaffold. 
Moreover, in this paper, we use the structure coefficients (denoted by $\gamma_{IJK}$) as variables instead of the connection coefficients (which are denoted by $\gamma_{IJK}$ in \cite{GIJ}) in the bootstrap argument and energy estimates.
Another consequence of the spatial inhomogeneity of the scaffold is that extra terms appear in the evolution equations for the structure coefficients.
However, by discriminating between the frame components and the structure coefficients in the bootstrap assumptions, using different powers of $t$, the required estimates for the structure coefficients actually simplify.
Next, in \cite{GIJ}, the authors consider the torus and use the
associated vector fields $\d_i$ as a basis for their arguments.
In our case, this is of course not possible.
Our substitute is the global frame $(E_i)_{i=1}^{n}$.
However, the elements of this frame do not commute, and we need appropriate associated estimates.
This, in principle, increases the length of the arguments.
However, in Section~\ref{sec:MainEstimates} below, we develop a scheme for bounding terms
that appear in the energy estimates.
This scheme allows us to estimate most of the terms of interest by simply inspecting
the number of factors of different types.
Next, as in \cite{GIJ}, the norms with the high number of derivatives are weighted by a factor $t^A$. 
By rewriting the lapse equation using the Hamiltonian constraint (in particular for the control of the high number of derivatives in $L^2$ of the lapse) we may in fact fix the parameter $A$ explicitly, independently of $k_0$, but with dependence only on $\s$ and $n$.
Finally, in \cite{GIJ}, there is a smallness assumption: the initial data should be
close enough to those of the background.
In our case, we do not have such a smallness assumption.
Instead, for a fixed bound of the form
(\ref{eq: Sobolev bound assumption}), we insist that the mean curvature is large enough.
The reason for making this type of assumption is
that it allows us to identify a convergent regime
without any reference to a background solution.

\textit{Deriving the asymptotics.}
As already emphasized, in the quiescent setting, it is natural to expect the expansion-normalized
quantities and the mean curvature to decouple: the expansion-normalized quantities
should converge as the mean curvature tends to infinity.
An extreme example of this is of course Example~\ref{example: spatially hom flat},
in which case the expansion-normalized quantities are independent of time.
We use the assumption of non-degeneracy in order to obtain $C^{k}$-control of the eigenvalues of the expansion-normalized
Weingarten map and the expansion normalized scalar field quantities $\Phi_1$ and $\Phi_0$.
Except for the arguments in the case of the regularity of the eigenvalues, the proof is in this case quite similar to that of
\cite{GIJ}, with the difference that we here also obtain information on the asymptotics of the expansion-normalized quantity $\Phi_0$.
In the degenerate setting, see Subsection~\ref{ssection: the degenerate case},
the argument for the asymptotics is similar to that of \cite{GIJ}.

Next, we highlight two ingredients entering the proof of Theorem~\ref{thm: big bang formation} in further detail:
The approximate satisfaction of the asymptotic Hamiltonian constraint and the double role of the inverse of the mean curvature
- as a measure of distance to the singularity, and as a threshold for an asymptotic regime.

\textit{Asymptotic Hamiltonian constraint.} As explained above, 
we expect the approximate equality \eqref{eq:approxlimhamcon} to hold in the asymptotic regime. 
We refer to it as the asymptotic Hamiltonian constraint.
To illustrate how our assumptions force us to be in the asymptotic regime,
note that for initial mean curvature $\bth >0$, the Hamiltonian constraint may be written as
\begin{equation} \label{eq:HamiltonianConstraintOutline}
    \tr \bmK^2 + \bP_1^2 
    + 2 \bth^{-2} V\circ\bp_0
    = 1 + \bth^{-2} \roScal_{\bh}-\bth^{-2}|\md \bp_0|_{\bh}^2.
\end{equation}
However, the last two terms on the right hand side decay as a negative power of $\bth$. Indeed, if $\be_I$ is an eigenvector field of
$\bmK$ corresponding to $\bq_I$ such that $|\be_I|_{\bh}=1$, then $\hat{e}_I:=\bth^{-\bq_I}\be_I$ (no summation) satisfies a $C^k$-bound
independent of the mean curvature due to the assumptions; see the proof of Proposition~\ref{prop: initial data estimate} for the details.
Moreover, $|\hat{e}_I|_{\bmH}=1$. Since
\begin{equation}\label{eq:exp norm sc field}
  \bth^{-2}|\md \bp_0|_{\bh}^2=\bth^{-2}\textstyle{\sum}_I\be_I(\bp_0)\be_I(\bp_0)=\textstyle{\sum}_I\bth^{2\bq_I-2}\hat{e}_I(\bp_0)\hat{e}_I(\bp_0)
\end{equation}
and since $\bp_0$ grows at worst linearly with $\ln\bth$, see (\ref{eq:Phiz def}), it is clear that the far right hand side of
(\ref{eq:exp norm sc field}) decays (up to a polynomial in $\ln\bth$) as $\bth^{2\bq_I-2}$; putting $J=K$ in (\ref{cond: alternating q})
yields $\bq_I<1-\s_p$. Similarly, the structure coefficient $\bga_{IJK} := \bh([\be_I, \be_J], \be_K)$ can, up to a polynomial in
$\ln\bth$, be bounded by $\bth^{\bq_I + \bq_J - \bq_K}$. By a standard expression of the scalar curvature in terms of the structure coefficients,
combined with (\ref{cond: alternating q}), it follows that (up to a polynomial in $\ln\bth$)
\begin{equation}
    \bth^{-2}|\md \bp_0|_{\bh}^2+\bth^{-2} |\roScal_{\bh}| \lesssim \bth^{-2\s_p}.
\end{equation}
Due to this estimate and the non-negativity of $V$, (\ref{eq:HamiltonianConstraintOutline}) yields the conclusion that
$|\bP_1|$ is essentially bounded from above by $1$. Thus, due to (\ref{eq:Phiz def}), 
$|\bp_0|$ is essentially bounded by $|\ln\bth|+|\bP_0|$. Combining this estimate with (\ref{eq: V assumption}) yields, roughly speaking,
an estimate of the form $\bth^{-2}|V\circ\bp_0|\lesssim \bth^{-2\s_V}$. Summarizing, there are constants $C$ and $k$, independent of $\bth$,
such that 
\begin{equation} \label{eq:BoundHamiltonianConstraintOutline}
  \big|\textstyle{\sum}_I\bq_I^{2}+\bP_1^2-1\big|
  	=\big|\bmK_{IJ} \bmK^{IJ}+\bP_1^2-1\big|\leq C\ldr{\ln(\bth)}^{k} \bth^{-2\min\{\s_p,\s_V\}}.
\end{equation}
In this sense, the asymptotic Hamiltonian constraint is implicitly encoded in the assumptions. 

Next, we use the mean curvature to define the time coordinate: $t = \theta^{-1}$. Moreover, we let $t_0:= \bth^{-1}$. The equation for the
initial lapse $\bN$ (which in fact is equivalent to the Raychaudhuri equation for the normal vector field to the CMC foliation) reads
\[
	(t_0^2 \Delta - 1 ) (\bN-1) + \big( 1 - \bmK_{IJ}\bmK^{IJ} - \bP_1^2 + t_0^2 \tfrac{2}{n-1}V\circ\bp \big) \bN = 0.
\]
Hence \eqref{eq:BoundHamiltonianConstraintOutline} and the bounds on the potential force $\bN \approx 1$ if $t_0 = \bth^{-1}$ is small enough.
It is therefore reasonable to think of $\bth$ as a measure of the initial distance to the singularity.

\textit{The dual role of the mean curvature.}
The proof of past global existence relies on controlling the deviation between the solution and the scaffold.
In the course of the argument, it is essential to isolate the dependence of the constants
that appear in the estimates on the mean curvature.
This allows us to, at the end, insist on a lower bound on the mean curvature in order to close the bootstrap argument
on which the past global existence proof is based.
This perspective makes it possible to consider solutions that are not necessarily close to symmetric background
solutions; note, e.g., that we do not impose any limitations on the spatial derivatives of the eigenvalues of the
initial expansion-normalized Weingarten map (and if the initial data would be close to initial data admitting a
Killing vector field, then the derivative of the eigenvalues along the
Killing vector field (of the background) would be small). 


To summarize, the lower bound on $\bth$, in the form of $\zeta_1$
in the statement of Theorem~\ref{thm: big bang formation}, works as a threshold for the asymptotic regime.
Moreover, the lower bound depends only on the reference geometry, the chosen admissibility thresholds,
the chosen regularity thresholds, and the chosen upper bound $\zeta_0$ in \eqref{eq: Sobolev bound assumption}.
Initial data satisfying these conditions as well as $\bth>\zeta_1$
are then forced to produce a quiescent solution with controlled behaviour to the past.
On the other hand, the inverse of the mean curvature features as a time-coordinate,
telling us in proper time when the big bang singularity will occur.

The logical structure of the proof of Theorem~\ref{thm: big bang formation} is the following:
\begin{enumerate}
\item By Proposition~\ref{prop: initial data estimate},
the assumptions in Theorem~\ref{thm: big bang formation}
    imply initial estimates for the Fermi-Walker quantities.
Moreover, if $\bth$ is large enough
we may deduce bounds for the initial lapse as shown in Proposition~\ref{prop: initial lapse estimate}.
As explained above, the initial bounds for the lapse follow, in essence,
    from the approximate satisfaction of the asymptotic Hamiltonian constraint.
\item Theorem~\ref{thm:GlobalExistenceFermiWalker} can then be applied to conclude global existence and an energy estimate in the CMC gauge with a Fermi-Walker propagated frame.
This is based on a bootstrap argument, where the solution is compared to the scaffold.
As explained above, if $\bth$ is large enough, the estimates for the Fermi-Walker quantities
may be bootstrapped, thus allowing us to conclude past global existence as well as control
relative to the scaffold.
\item By Theorem~\ref{thm:Asymptotics}, again for large enough $\bth$, the energy estimate in Theorem~\ref{thm:GlobalExistenceFermiWalker} suffices to conclude the desired asymptotics.
The demand for $\bth$ here is a result of the need for the algebraic conditions
    to hold all the way up to the initial singularity. 
\item By Proposition~\ref{prop:FieldEquationsFermiWalker}, we get a corresponding solution to the Einstein-non-linear scalar field equations which inherits the properties 
shown for the solution in the chosen gauge and frame.
\end{enumerate}

\subsection{Outline} \label{ssec:Outline}
In Section~\ref{section: Cauchy problem FW Gauge}, we introduce the gauge and formulate the equations we use: \textit{the FRS equations}. In
Proposition~\ref{prop:FieldEquationsFermiWalker}, we show that a solution to these equations yields a solution the Einstein-non-linear scalar
field equations. We define Sobolev norms in Subsection~\ref{sec:sobolev}, and, in Subsection~\ref{ssection: Initial bounds FW}, a notion of
initial data for the FRS variables: \textit{diagonal FRS initial data}.
Subsequently, we define expansion-normalized bounds for diagonal FRS initial data.
Then, in Proposition~\ref{prop: initial data estimate}, we show that
the expansion-normalized bounds assumed in Theorem~\ref{thm: big bang formation} 
imply similar bounds for the diagonal FRS initial data.
In Proposition~\ref{prop: initial lapse estimate} these bounds,
as well as a bound on the initial curvature, are used to deduce an estimate for the initial lapse,
which is required later to initiate the bootstrap argument.

Next, in Section~\ref{section: past global existence}, the past global existence theorem,
Theorem~\ref{thm:GlobalExistenceFermiWalker} is formulated and,
towards the end of the section, also proven,
assuming the bootstrap improvement theorem, Theorem~\ref{thm:BootstrapImprovement}.
We also recall local existence and Cauchy stability of solutions, formulated and demonstrated in \cite{OPR23}.
The rest of the section is dedicated to the formulation
    and setup for the proof of Theorem~\ref{thm:BootstrapImprovement};
in particular, the scaffold is introduced, as well as
the deviation quantities which measure the difference of the solution to the scaffold.
Lastly, a-priori estimates, i.e. estimates not based on the evolution equations,
which are important for the main estimates of Section~\ref{sec:MainEstimates}, are deduced.

Section~\ref{sec:MainEstimates} contains estimates for the lapse, the deviation quantities
and the time derivatives of the deviation quantities. These estimates are required for the
energy estimates and are in part based on the evolution equations.
The section begins with an algorithm, which we refer to as the scheme,
which is used to conveniently deduce estimates which suffice
for many of the results of Section~\ref{sec:MainEstimates}.
Continuing, Section~\ref{sec:EnergyEstimates} contains the two energy estimates required
for the proof of Theorem~\ref{thm:BootstrapImprovement},
which itself is then proven at the end of the section.
Subsequently, in Section~\ref{sec:asymptotics} we encounter Theorem~\ref{thm:Asymptotics},
in which asymptotics for $\mK, \Phi_1$ and $\Phi_0$ are obtained
as well as curvature blow-up.
This completes all the relevant ingredients required for the proof
of Theorem~\ref{thm: big bang formation}, which is the content
    of Section~\ref{sec: finishing the proof}.

Finally, there are two appendices. In Appendix~\ref{sec:SobolevInequalities}, the Sobolev inequalities 
that we require throughout Sections~\ref{section: Cauchy problem FW Gauge}--\ref{sec:asymptotics} are proven,
while Appendix~\ref{sec:SobolevInequalities} contains results concerning the regularity of eigenvalues
which are required in the proof of Theorem~\ref{thm:Asymptotics}.



\subsection{Notation in the paper}\label{ssection:notation}
Throughout this paper, we use round brackets around capital Latin indices for symmetrization, e.g. $k_{(IJ)} = (k_{IJ} + k_{JI})/2$,
and square brackets for anti-symmetrization, e.g. $\g_{[IJ]K} =(\g_{IJK} - \g_{JIK})/2$, etc. Moreover, we use the Einstein summation
convention; i.e., we sum over repeated upstairs and downstairs indices. In the case of capital Latin indices, we also sum over
repeated downstairs and repeated upstairs indices. However, underlined indices are not summed over. Moreover, throughout the paper,
we do not use Einstein's summation convention for expressions where one of the indices appear in an exponent, e.g.\ for expressions
of the form $t^{a_I}b_I$.

Throughout the paper, we use multiindices for frames. The corresponding notation, and the notation for Sobolev norms we use here
is introduced in Appendix~\ref{sec:SobolevInequalities} below. 

\subsection{Acknowledgements}
This research was funded by the Swedish Research Council (Vetenskapsr\aa det), dnr. 2017-03863, 2021-04269 and 2022-03053.

\section{The Cauchy problem for the FRS equations}\label{section: Cauchy problem FW Gauge}

As in \cite{GIJ}, the solutions to the Einstein-non-linear scalar field equations we construct have CMC foliations with a vanishing shift vector field
(i.e., $\d_{t}$ is perpendicular to the leaves of the foliation). In particular, we look for globally hyperbolic solutions of the form
\begin{equation} \label{eq: form of the metric}
    g = -N^2 \md t^2 + h
\end{equation}
on $M= (a, b) \times \S$, where $a<b$, $\S$ is closed, $N>0$ and $h$ is a family of Riemannian metrics on $\S$. Moreover, we express the metric $g$ in
terms of a Fermi-Walker transported frame.

\subsection{The reference frame}\label{ssection: The reference frame}
Let $(\S,h_{\refer})$ be a closed Riemannian manifold with a smooth global orthonormal frame $(E_{i})_{i=1}^{n}$; cf. Remark~\ref{remark:frame intro}. 
There is a canonical way of extending this frame to $M$ by requiring that
\begin{equation} \label{eq: Lie transport}
    [\d_t, E_i]
        = 0
\end{equation}
and that $E_1, \hdots, E_n$ are tangent to the leaves $\S_t := \{t\} \times \S \subset M$. We let $(\eta^i)_{i=1}^{n}$ denote the co-frame dual to
$(E_i)_{i=1}^{n}$. The reference metric, reference frame and the reference co-frame will be fixed throughout the paper.

\subsection{The FRS equations}\label{ssection: the FRS equations}

The equations we solve follow in a straightforward manner from the Einstein-non-linear scalar field equations, assuming a CMC foliation,
a vanishing shift vector field, and using a Fermi-Walker transported frame. The idea of using this combination of gauge conditions and
frame to prove stable big bang formation goes back to the work of Fournodavlos, Rodnianski and Speck, see \cite{GIJ}. We therefore refer
to the resulting system of equations as \textit{the FRS equations}. The purpose of the present section is to show that if we have a solution
to the FRS equations, then we can construct a solution to the Einstein-non-linear scalar field equations. In what follows, we use the conventions
introduced in Subsection~\ref{ssection:notation}. 

\begin{prop} \label{prop:FieldEquationsFermiWalker}
  Let $(\S, h_\refer)$ be a Riemannian manifold with a frame $(E_i)_{i=1}^{n}$ and co-frame $(\eta_i)_{i=1}^{n}$ as above. Let $t_0 > 0$,
  let $V \in C^\infty(\rn{})$ and let 
  \[
  \be_I^i, \bo^I_i, \bk_{IJ}, \bga_{IJK}, \bp_0, \bp_1 : \S \to \rn{},
  \]
  for all $i, I, J, K$, be smooth functions. Define the vector fields and one-forms
  \begin{align*}
    \be_I
    := \be_I^i E_i, \quad
    \bo^I
    := \bo^I_i \eta^i.
  \end{align*}
Assume that 
\begin{itemize}
    \item $(\be_I)_{I=1}^{n}$ is a smooth frame of the tangent space of $\S$,
    \item $(\bo^I)_{I=1}^{n}$ is the dual frame of $(\be_I)_{I=1}^{n}$,
    \item $\bk_{IJ} = \bk_{JI}$ and $\sum_{I} \bk_{II} = \frac1{t_0}$, 
    \item $\bga_{IJK} = \bo^K([\be_I, \be_J])$,
\end{itemize}
for all $I,J, K $. Assume, moreover, that there are smooth functions
\[
    e_I^i, \o^I_i, k_{IJ}, \gamma_{IJK}, \phi, N:
        (a, b) \times \S \to \rn{},
\]
with $N > 0$ and $(a,b)\subseteq (0,\infty)$, for all $i, I, J, K $, solving
\begin{itemize}
    \item the evolution equations for the frame and dual frame:
    \begin{align}
    e_0(e_I^i) 
        &= - k_{IJ}e_J^i, \label{eq: transport frame} \\
    e_0(\o^I_i) 
        &= k_{IJ} \o^J_i, \label{eq: transport co-frame}
\end{align}
    \item the evolution equations for $\gamma$ and $k$:
    \begin{align}
    \begin{split}
    e_{0}\g_{IJK} 
        = & - 2 N^{-1} e_{[I}(Nk_{J]K}) 
        - k_{IL} \g_{LJK} - k_{JL} \g_{ILK} + k_{KL} \g_{IJL},
  \end{split} \label{eq:concoef} \\
  \begin{split}
    e_{0}k_{IJ} 
    	= & N^{-1}e_{(I}e_{J)} (N) - N^{-1} e_{K} (N\g_{K(IJ)})
    	- e_{(I}(\g_{J)KK}) - t^{-1}k_{IJ} \\
    	& +\g_{KLL}\g_{K(IJ)} + \g_{I(KL)}\g_{J(KL)} - \tfrac14 \g_{KLI} \g_{KLJ}\\
    	& + e_I(\phi)e_J(\phi)+\tfrac2{n-1}(V\circ\phi) \delta_{IJ},
  \end{split} \label{eq:sff}
\end{align}
    \item the evolution equations for the derivatives of the scalar field:
    \begin{align}
    e_0 \left( e_I \phi \right)
        &= N^{-1} e_I \left( N e_0 \phi \right) - k_{IJ} e_J(\phi), \label{eq:spatial scalar field derivative} \\
    e_{0}\left( e_{0}(\phi) \right)
        &= e_{I}\left( e_{I}(\phi) \right) -t^{-1}e_{0}(\phi)+N^{-1}e_{I}(N)e_{I}(\phi) -\g_{JII}e_{J}(\phi)-V'\circ\phi, \label{eq:scalarfield}
    \end{align}
    \item the Hamiltonian constraint equation:
    \begin{align}
    \begin{split}
    2e_I(\g_{IJJ})
  	&- \tfrac14\g_{IJK}(\g_{IJK} + 2 \g_{IKJ}) -\g_{IJJ}\g_{IKK} - k_{IJ}k_{IJ} + t^{-2} \\
  	&= (e_{0}\phi)^{2} + e_{I}(\phi)e_{I}(\phi) + 2V \circ \phi,
    \end{split} \label{eq:hamconstraint}
    \end{align}
    \item the momentum constraint equation:
    \begin{equation} \label{eq:MC}
    e_I k_{IJ}
  	=\g_{LII}k_{LJ} + \g_{IJL}k_{IL} + e_{0}(\phi) e_{J}(\phi),
    \end{equation}
    \item the lapse equation:
    \begin{equation}
    \begin{split}
    e_{I}e_{I}(N) = & t^{-2}(N-1)+\g_{JII}e_{J}(N) - \left( e_{I}(\phi)e_{I}(\phi) + \tfrac{2n}{n-1} V \circ \phi \right)N \\
        & + \left( 2e_I(\g_{IJJ}) - \tfrac14\g_{IJK}(\g_{IJK} + 2 \g_{IKJ}) - \g_{IJJ}\g_{IKK} \right) N,
    \end{split}\label{eq:LapseEquation}
    \end{equation}
\end{itemize}
on $M:=(a,b) \times \S$, with $t_0 \in (a,b)$, where $e_0:= N^{-1}\d_t$, $e_I:= e_I^i E_i$, subject to the initial condition
\[
    \left( e_I^i, \o^I_i, k_{IJ}, \g_{IJK}, \phi, e_0 \phi \right)|_{t = t_0} 
        = \left( \be_I^i, \bo^I_i, \bk_{IJ}, \bga_{IJK}, \bp_0, \bp_1 \right),
\]
for all $i,I, J, K $. Define $\o^I := \o^I_i \eta^i$ and $h := \o^I \otimes \o^I$. Then the spacetime metric $g$, defined by
(\ref{eq: form of the metric}), and the scalar field $\phi$, satisfy the Einstein-non-linear scalar field equations with a potential $V$, i.e.\
\begin{align}
	\roRic_g
		&= \md \phi \otimes \md \phi + \tfrac2{n-1} (V \circ \phi) g, \label{eq: field equations conclusion}\\
	\Box_{g}\phi
		&= V'\circ \phi, \label{eq: scalar equation conclusion}
\end{align}
and the hypersurfaces $\S_t := \{t\} \times \S$ are CMC Cauchy hypersurfaces of mean curvature $\frac1t$. Moreover,
\begin{itemize}
    \item $(e_I)_{I=1}^{n}$ is a smooth frame on $\S_{t}$ for each $t \in (a,b)$,
    \item $(\o^I)_{I=1}^{n}$ is the dual frame of $(e_I)_{I=1}^{n}$, 
    \item $k:= k_{IJ}\o^I \otimes \o^J$ is the second fundamental form,
    \item $(e_I)_{I=1}^{n}$ is a Fermi-Walker propagated frame, i.e.\
    \begin{equation} \label{eq: Fermi-Walker equation}
        \n_{e_0}e_I
            = e_I \ln(N)e_0,
    \end{equation}
    \item $\g_{IJK} = \o^K([e_I, e_J])$,
\end{itemize}
for all $I, J, K $. Finally, the initial data induced on $\S_{t_{0}}$ by $(M,g,\phi)$ is $(\S,\bh,\bk,\bp_0, \bp_1)$, where
$\bh:=\bo^{I}\otimes \bo^{I}$ and $\bk:=\bk_{IJ}\bo^{I}\otimes\bo^{J}$. 
\end{prop}

\begin{remark}
  Our sign convention concerning the second fundamental form is opposite to that of \cite{GIJ}. Moreover, $\gamma_{IJK}$ here denotes the structure
  coefficients (with the last index lowered), for the frame $(e_I)_{I=1}^{n}$, for indices $I,J,K = 1, \hdots, n$. This should be contrasted with
  \cite{GIJ} in which $\gamma_{IJK}$ denotes the connection coefficients; see \cite[(2.18), p.~860]{GIJ}. 
\end{remark}

The following lemma will be useful when proving Proposition~\ref{prop:FieldEquationsFermiWalker}:

\begin{lemma} \label{lemma: spatial curvatures}
  Let $(\S,h)$ be a Riemannian manifold with a smooth global orthonormal frame $(e_I)_{I=1}^{n}$. Then the Ricci and scalar curvature are given by
  \begin{align}
    \begin{split}
    \roRic_h(e_I, e_J) 
    =& e_K (\g_{K(IJ)}) + e_{(I} (\g_{J)KK}) - \g_{I(KL)} \g_{J(KL)}\\
    & + \tfrac14 \g_{KLI} \g_{KLJ} - \g_{KLL} \g_{K(IJ)},
    \end{split}\label{eq: SpatialRicciCurvature} \\
    \roScal_h 
    &= 2 e_I (\g_{IJJ}) - \tfrac14 \g_{IJK} (\g_{IJK} + 2 \g_{IKJ}) - \g_{IJJ} \g_{IKK},\label{eq:SpatialScalarCurvature}
  \end{align}
  where $\gamma_{IJK} = h([e_I, e_J], e_K)$.
\end{lemma}

\begin{proof}[Proof of Lemma~\ref{lemma: spatial curvatures}]
Define first $\G_{IJK} := h \left( \nabla^h_{e_I} e_J, e_K \right)$.
The Ricci curvature is given by
\begin{equation}
\begin{split} \label{eq: First spatial Ricci}
    \roRic_h(e_I, e_J) 
        &= h(\n^h_{e_K} \n^h_{e_I}e_J, e_K) - h(\n^h_{\n^h_{e_K} e_I}e_J, e_K) \\
        &\qquad - h(\n^h_{e_I} \n^h_{e_K}e_J, e_K) + h(\n^h_{\n^h_{e_I} e_K}e_J, e_K) \\
        &= e_K (\G_{IJK}) - \Gamma_{IJL} \Gamma_{KKL} - \Gamma_{KIL} \Gamma_{LJK} \\
        &\qquad - e_I (\G_{KJK}) + \G_{KJL}\G_{IKL} + \Gamma_{IKL}\Gamma_{LJK} \\
        &= e_K (\G_{IJK}) - \Gamma_{IJL} \Gamma_{KKL}  - \Gamma_{KIL} \Gamma_{LJK} - e_I (\G_{KJK}),
\end{split}
\end{equation}
where we in the last line used that
\[
    \G_{KJL}\G_{IKL} + \Gamma_{IKL}\Gamma_{LJK}
        = -\G_{KJL}\G_{ILK} + \Gamma_{IKL}\Gamma_{LJK} = 0.
\]
The Koszul formula and the fact that $\g_{IJK}$ is skew-symmetric in $I$ and $J$ give the following three ways of writing $\Gamma_{IJK}$:
\begin{equation} \label{eq: Gamma as gamma}
\begin{split}
    \Gamma_{IJK}
        &= \tfrac12 \left( \gamma_{IJK} + \gamma_{KIJ} - \gamma_{JKI} \right) 
        = \tfrac12 \gamma_{KIJ} - \gamma_{J(KI)}= \tfrac12 \gamma_{IJK} + \gamma_{K(IJ)}.
\end{split}
\end{equation}
Hence, the third term in \eqref{eq: First spatial Ricci} is given by
\begin{align*}
    - \G_{KIL} \G_{LJK}
        &= - \tfrac{1}{4} \big(\g_{LKI} - 2 \g_{I(KL)} \big) \big(\g_{KLJ} - 2 \g_{J(KL)} \big) \\
        &= \tfrac{1}{4} \g_{KLI} \g_{KLJ} - \g_{I(KL)} \g_{J(KL)}.
\end{align*}

Using \eqref{eq: Gamma as gamma} again, and the special cases $\G_{KJK} = \g_{KJK} = - \gamma_{JKK}$ and $\G_{LLK} = \g_{KLL}$, Equation \eqref{eq: First spatial Ricci} becomes
\begin{align} \label{eq: Ricci h main computation}
\begin{split}
\roRic_h(e_I, e_J) 
    = & \tfrac12 e_K (\g_{IJK}) + e_K (\g_{K(IJ)}) + e_I (\g_{JKK}) \\
    & - \g_{I(KL)} \g_{J(KL)} + \tfrac14 \g_{KLI} \g_{KLJ}
     - \tfrac12\g_{KLL}\left( \g_{IJK} + 2 \g_{K(IJ)} \right).
\end{split}
\end{align}
Since the left-hand side is symmetric in $I$ and $J$, each term on the right-hand side should also be symmetrized in $I$ and $J$.
Since $\g_{(IJ)K} = 0$, \eqref{eq: Ricci h main computation} simplifies to \eqref{eq: SpatialRicciCurvature}, as claimed.

Finally, by taking the trace of \eqref{eq: SpatialRicciCurvature},
we obtain \eqref{eq:SpatialScalarCurvature}.
\end{proof}

\begin{proof}[Proof of Proposition~\ref{prop:FieldEquationsFermiWalker}]
  We begin by showing that $k_{IJ} = k_{JI}$. Due to \eqref{eq:sff}, $e_0 k_{IJ} + t^{-1}k_{IJ}$ is symmetric in
  $I$ and $J$. Since, in addition, $k_{IJ}|_{t_0}= \bk_{IJ}= \bk_{JI}= k_{JI}|_{t_0}$, we conclude that the skew-symmetric part of $k_{IJ}$
  vanishes by uniqueness for first order transport equations. Combining this observation with \eqref{eq: transport frame} and
  \eqref{eq: transport co-frame} yields
  \begin{align*}
    \d_t \left( \o^I_i e_J^i -\delta^{I}_{J}\right)
    &= \o^I_i \left( - N k_{JM} e_M^i \right) + N k_{IM} \o_i^M e_J^i \\
    &= - N k_{MJ} \left(\o_i^I e_M^i-\delta^{I}_{M}\right) + N k_{IM} \left(\o^M_i e^i_J-\delta^{M}_{J}\right).
  \end{align*}
  This is a homogeneous system of first order ODE's for $\o^I_i e_J^i -\delta^{I}_{J}$ with vanishing initial data. Thus
  $\o^I_i e_J^i =\delta^{I}_{J}$ on $M$. This shows that $e_I^i$ are the elements of an invertible matrix, which implies that
  $(e_I)_{I=1}^{n}$ are a frame on $\S_t$ for all $t \in (a, b)$, since $(E_i)_{i=1}^{n}$ are. Moreover, this also shows
  that the matrix $\o_i^I$ is the inverse of $e_I^i$, which implies that $(\o^I)_{I=1}^{n}$ is the dual frame of $(e_I)_{I=1}^{n}$.

  Next, we check that $k = k_{IJ} \o^I \otimes \o^J$ equals the second fundamental form associated to the metric $g$, say $\tilde k$.
  Recalling that $[\d_t, E_i] = 0$, we use \eqref{eq: transport frame} to compute
  \begin{equation}\label{eq:mldtoKeI}
    \begin{split}
      \ml_{\d_t} \o^K (e_I) = & -\o^{K}([\d_{t},e_{I}])=\o^{K}(Nk_{IJ}e_{J})=Nk_{IK}.
    \end{split}
  \end{equation}
  Combining this observation with $k_{IJ} = k_{JI}$ yields
  \begin{equation*}
    \begin{split}
      \tilde k(e_I, e_J)
      &= \tfrac1{2N} \ml_{\d_t} \left( \o^K \otimes \o^K \right)(e_I, e_J) \\
      &= \tfrac1{2N} \left( \ml_{\d_t} \o^K (e_I) \de_{KJ} + \ml_{\d_t} \o^K(e_J) \de_{IK} \right) =k_{IJ}
    \end{split}
   \end{equation*}
  as claimed.

  Let us now show that $e_1,\dots,e_n$ satisfy \eqref{eq: Fermi-Walker equation}.
  Due to \eqref{eq: transport frame}, it can be computed that
  \begin{equation}\label{eq:ezeIcomm}
    [e_{0},e_{I}]=-k_{IK}e_{K}+e_{I}(\ln N)e_{0}.
  \end{equation}
  Thus
  \begin{equation}\label{eq:nablaezeIeJ}
    g(\nabla_{e_{0}}e_{I},e_{J})=g([e_{0},e_{I}],e_{J})+k_{IJ}=-g(k_{IK}e_{K},e_{J})+k_{IJ}=0.
  \end{equation}
  Next,
  \begin{equation}\label{eq:nablaezeIez}
    \begin{split}
      g(\nabla_{e_{0}}e_{I},e_{0})=g([e_{0},e_{I}],e_{0})+g(\nabla_{e_{I}}e_{0},e_{0})=-e_{I}(\ln N),
    \end{split}
  \end{equation}
  where we use the fact that $e_{0}$ is a unit vector field and (\ref{eq:ezeIcomm}) in the second step. Combining
  (\ref{eq:nablaezeIeJ}) and (\ref{eq:nablaezeIez}) yields (\ref{eq: Fermi-Walker equation}).  

  We now prove that $\gamma_{IJK} = \o^K([e_I, e_J])$. Let $\wg_{IJK} := \o^K([e_I, e_J])$. The Jacobi identity gives
  \begin{equation}\label{eq:Jacobi identity}
    0= g \left( [e_0, [e_I, e_J]] + [e_I, [e_J, e_0]] + [e_J, [e_0, e_I]], e_K \right).
  \end{equation}
  On the other hand, (\ref{eq:ezeIcomm}) yields
  \begin{equation*}
    \begin{split}
      g([e_{I},[e_{J},e_{0}]],e_{K}) = & g([e_{I},k_{JL}e_{L}-e_{J}(\ln N)e_{0}],e_{K})\\
      = & g(e_{I}(k_{JL})e_{L}+k_{JL}\wg_{ILM}e_{M}-e_{J}(\ln N)[e_{I},e_{0}],e_{K}).
    \end{split}
  \end{equation*}
  Appealing to (\ref{eq:ezeIcomm}) again yields
  \begin{equation}\label{eq:terms two and three}
    g([e_{I},[e_{J},e_{0}]],e_{K})=e_{I}(k_{JK})+k_{JL}\wg_{ILK}-e_{J}(\ln N)k_{IK}.
  \end{equation}
  A similar argument, again appealing to (\ref{eq:ezeIcomm}), yields
  \begin{equation}\label{eq:term one}
    g( [e_0, [e_I, e_J]],e_{K})=e_{0}(\wg_{IJK})-\wg_{IJM}k_{MK}.
  \end{equation}  
  Combining (\ref{eq:Jacobi identity}), (\ref{eq:terms two and three}) and (\ref{eq:term one}) yields
\[
    e_{0}\wg_{IJK} 
        = - 2 N^{-1} e_{[I}(Nk_{J]K}) - k_{IL} \wg_{LJK} - k_{JL} \wg_{ILK} + k_{KL} \wg_{IJL}.
\]
Hence, $\wg_{IJK}$ is a solution to \eqref{eq:concoef} with initial data
\[
    \wg_{IJK}|_{t_0}
        =  \o^K([e_I, e_J])|_{t_0}
        = \bo^K([\be_I, \be_J])
        = \bga_{IJK}
        = \gamma_{IJK}|_{t_0}.
\]
Again by uniqueness of solutions to transport equations, we conclude that 
\[
    \g_{IJK} = \wg_{IJK} = \o^K([e_I, e_J])
\]
as claimed.

Next we show that the mean curvature is given by $\frac1t$. Taking the trace over Equation \eqref{eq:sff} and applying Equation
\eqref{eq:LapseEquation} yields
\begin{equation} \label{eq: trace k ODE}
    \d_t k_{II} + \tfrac N t k_{II}
        = \tfrac {N - 1}{t^2}.
\end{equation}
The unique solution to \eqref{eq: trace k ODE} with initial condition $\k_{II}|_{t_0} = \bk_{II} = \frac1{t_0}$ is given by $\k_{II} = \frac1t$, proving the claim.

We now turn to proving that $g$ satisfies the Einstein-non-linear scalar field equations with potential $V$. For this we use the notation
\begin{align*}
    G_g
        &= \roRic_g - \tfrac12 \roScal_g g, \\
    T_{g, \phi}
        &= \md \phi \otimes \md \phi - \left[\tfrac{1}{2}| \md \phi|_{g}^{2} + V\circ\phi\right]g.
\end{align*}
Combining \cite[(13.5), p.~149]{RinCauchy} with Lemma~\ref{lemma: spatial curvatures} and Equation \eqref{eq:hamconstraint} yields
\begin{equation}\label{eq:GgmTgzz}
  \begin{split}
    \left( G_g - T_{g, \phi} \right)(e_0, e_0)
    = & \tfrac12\left( \roScal_h - k_{IJ}k_{IJ} + t^{-2} 
    - (e_0\phi)^2 - e_I \phi e_I \phi - 2 V \circ \phi \right)\\
    = & \tfrac12 \left( 2 e_I (\g_{IJJ}) - \tfrac14 \g_{IJK} (\g_{IJK} + 2 \g_{IKJ}) - \g_{IJJ} \g_{IKK} \right) \\
    & + \tfrac12 \left( - k_{IJ}k_{IJ} + t^{-2}
    - (e_0\phi)^2 - e_I \phi e_I \phi - 2 V \circ \phi \right)=0.
  \end{split}  
\end{equation}
Due to \cite[(13.6), p.~149]{RinCauchy} and the fact that $\tr_hk$ is constant on $\S_t$, 
\begin{align*}
    \left( G_g - T_{g, \phi} \right)(e_0, e_I)        
        & = \rodiv_h (k)(e_I) - e_0\phi e_I\phi \\
        & = e_J k_{JI} - \Gamma_{JJK} k_{KI} - \Gamma_{JIK}k_{JK}  - e_0\phi e_I \phi \\
        & = e_J k_{JI} - \gamma_{KJJ} k_{KI} + \gamma_{IKJ} k_{KJ} - e_0\phi e_I \phi=0,
\end{align*}
where we appealed to (\ref{eq: Gamma as gamma}), the symmetry of $k_{IJ}$ and the antisymmetry of $\g_{IJK}$ in the first two indices
in the third step; and to (\ref{eq:MC}) in the last step. 

For the remaining components of the Einstein equations, note that (\ref{eq: Fermi-Walker equation}) yields
\begin{align*}
    g(\n_{e_0}e_0, e_I) &= - g(e_0, \n_{e_0}e_I) = e_I \ln(N).
\end{align*}
Since $g(\n_{e_0}e_0, e_0)= 0$, it follows that
\begin{equation} \label{eq: nabla e0 e0}
    \n_{e_0}e_0
        = \rograd_h \left( \ln(N) \right).
\end{equation}
Combining this with (\ref{eq: Fermi-Walker equation}), 
\begin{equation}\label{eq:ezkIJfstep}
  \begin{split}
    e_0 k_{IJ}
    &= e_0g\left(\n_{e_I} e_0, e_J\right) \\
    &= g \left( \n_{e_0} \n_{e_I} e_0, e_J \right) + (e_J\ln(N)) g\left( \n_{e_I}e_0, e_0 \right) \\
    &= \roRiem_g(e_0, e_I, e_0, e_J) + g\left(\n_{e_I} \n_{e_0} e_0, e_J \right) + g\left( \n_{[e_0, e_I]} e_0, e_J \right).
  \end{split}  
\end{equation}
Due to \eqref{eq: nabla e0 e0},
\begin{equation*}
  \begin{split}
    g\left(\n_{e_I} \n_{e_0} e_0, e_J \right) = & \roHess(\ln(N))(e_I, e_J)\\
    = & N^{-1}\roHess(N)(e_I, e_J)-e_{I}(\ln N)e_{J}(\ln N).
  \end{split}
\end{equation*}
Due to \eqref{eq:ezeIcomm} and \eqref{eq: nabla e0 e0}, 
\[
g\left( \n_{[e_0, e_I]} e_0, e_J \right)=-k_{IL}k_{LJ}+e_{I}(\ln N)e_{J}(\ln N). 
\]
Combining the last two equalities yields
\[
g\left(\n_{e_I} \n_{e_0} e_0, e_J \right) + g\left( \n_{[e_0, e_I]} e_0, e_J \right)
= N^{-1} e_{(I} e_{J)} N - \gamma_{K(IJ)} e_K\ln(N) - k_{IL}k_{LJ},
\]
where we appealed to (\ref{eq: Gamma as gamma}) and the symmetry of the Hessian. 
Next,
\begin{align*}
     \roRiem_g(e_0, e_I, e_0, e_J)
        = & \roRic_g(e_I, e_J) + \roRiem_g(e_K, e_I, e_K, e_J) \\
        = & \big( \roRic_g - \md \phi \otimes \md \phi - \tfrac2{n-1} (V \circ \phi) g\big) (e_I, e_J) \\
        & + e_I(\phi) e_J(\phi) + \tfrac2{n-1}( V \circ \phi ) \de_{IJ}\\
        & +\roRiem_h(e_K, e_I, e_K, e_J) - k_{LL}k_{IJ} + k_{IL}k_{LJ},
\end{align*}
where we appealed to the Gau\ss\ equations, \cite[Theorem~5, p.~100]{oneill}, in the last step. Combining the last two equalities with
(\ref{eq:ezkIJfstep}) and $k_{LL}=1/t$ yields
\begin{align*}
    e_0 k_{IJ}
    = & \big( \roRic_g - \md \phi \otimes \md \phi - \tfrac2{n-1} (V \circ \phi) g\big) (e_I, e_J) - t^{-1} k_{IJ} + N^{-1} e_{(I} e_{J)} N\\
    & - N^{-1} \gamma_{K(IJ)} e_K N  - \roRic_h(e_I, e_J) + e_I(\phi) e_J(\phi) + \tfrac2{n-1}\left( V \circ \phi \right) \de_{IJ}.
\end{align*}
Equation \eqref{eq:sff}, together with Lemma~\ref{lemma: spatial curvatures}, therefore implies that
\begin{equation}\label{eq: contracted I J field equation}
     \big( \roRic_g - \md \phi \otimes \md \phi - \tfrac2{n-1} (V \circ \phi) g\big) (e_I, e_J)
        = 0.
\end{equation}
Combining (\ref{eq:GgmTgzz}) and this with the fact that $\left(G_g - T_{g, \phi}\right)(e_0, e_0) = 0$, proven above, we conclude that
\[
    \roScal_g - \abs{\md \phi}_g^2 - \tfrac{2(n+1)}{n-1}V\circ\phi= 0.
\]
This, combined with \eqref{eq: contracted I J field equation}, shows that
\[
    \left( G_g - T_{g, \phi} \right)(e_I, e_J)
        = 0.
\]
We thus have shown that $G_g - T_{g, \phi} = 0$, which is equivalent to \eqref{eq: field equations conclusion}.

Finally, using \eqref{eq: nabla e0 e0}, the scalar wave equation in a Fermi-Walker frame becomes
\begin{align*}
    \Box_{g} \phi
        &= - e_0 e_0 \phi + \left(\n_{e_0}e_0\right) \phi + e_I e_I \phi - \left( \n_{e_I}e_I \right) \phi \\
        &= - e_0 e_0 \phi + e_I\ln(N) e_I \phi + e_I e_I \phi - \Gamma_{IIJ}e_J\phi - k_{II} e_0 \phi \\
        &= - e_0 e_0 \phi + e_I\ln(N) e_I \phi + e_I e_I \phi - \gamma_{JII}e_J\phi - t^{-1} e_0 \phi.
\end{align*}
Equation \eqref{eq:scalarfield} therefore implies Equation \eqref{eq: scalar equation conclusion}. The only thing
that remains to be demonstrated is that the leaves $\S_t$ are Cauchy hypersurfaces. However, this follows from
an argument similar to that given in the proof of \cite[Proposition~1, p.~152]{RinInv}. This concludes the proof.
\end{proof}

\subsection{Sobolev norms} \label{sec:sobolev}

Since the Einstein-non-linear scalar field equations are expressed in terms of the Fermi-Walker frame as in
Proposition~\ref{prop:FieldEquationsFermiWalker}, and the Fermi-Walker frame is expressed in terms of the reference
frame, all involved equations are systems of \emph{scalar} equations. The basic Sobolev norms are therefore
the ones introduced in Appendix~\ref{sec:SobolevInequalities} below, see (\ref{seq:Sobolev norms}). 
The functions $e_I^i, \o_i^I, \gamma_{IJK}, k_{IJ}$ appear in Proposition~\ref{prop:FieldEquationsFermiWalker} as
components, depending on indices $i, I, J, K = 1, \hdots, n$. In such cases, we abbreviate by using the notation
\begin{subequations}\label{seq:norms of indexed quantities}
  \begin{align}
    \norm{e}_{C^{s}(\S)}
    := & \textstyle{\sum}_{I,i}\|e^{i}_{I}\|_{C^{s}(\S)},\\
    \norm{\g}_{H^s(\S)}
    := & \big(\textstyle{\sum}_{I,J,K}\|\g_{IJK}\|_{H^s(\S)}^{2}\big)^{\frac12},
  \end{align}
\end{subequations}
and similarly for other norms. We also use the canonical identification
\[
    \S_t 
        := \{ t \} \times \S
	\cong \S,
\]
in order to define the analogous Sobolev norms on $\S_t$. Consequently, the Sobolev norms on $\S_t$ only depend on the
reference frame $(E_i)_{i=1}^{n}$ and not on the induced geometry on $\S_t$.

\subsection{Initial bounds on the FRS variables}\label{ssection: Initial bounds FW}

The assumptions in Theorem~\ref{thm: big bang formation} are formulated in terms of the expansion-normalized initial data
$(\S, \bmH, \bmK, \bP_0, \bP_1)$. In this section, we show that these assumptions imply the following bounds for the FRS
initial data:

\begin{definition} \label{def:DecomposedData}
  Let $\s_p$, $\s_V$, $\s$, $k_0$, $k_1$, $(\S, h_\refer)$, $(E_{i})_{i=1}^{n}$ and $V$ be as in Theorem~\ref{thm: big bang formation}.
  \emph{Diagonal FRS initial data} is a collection of smooth
  \begin{equation}\label{eq:diagonal Fermi Walker ID}
    \be_I^i, \bq_I, \bp_0, \bp_1: \S \to \rn{},
  \end{equation}
  for $i,I = 1,\dots,n$, such that $(\be_I)_{I=1}^n$ is a global frame on $\S$, where $\be_I := \be_I^i E_i$, and
  such that $\sum_{I=1} \bq_I = 1$. Given any $\rho > 0$, they are said to satisfy the \emph{FRS expansion-normalized bounds
  of regularity $k_1$ for $\rho$ at $t_0$} if  
  \begin{equation}\label{eq: bar pi diagonal conditions}
    \bq_I + \bq_J - \bq_K< 1 - \s_p
  \end{equation}
  for any $I \neq J$ and
  \begin{subequations}\label{eq:FermiWalkerBounds}
    \begin{align}
      \|t_0^{\bq_I}\be_I^i\|_{H^{k_{1}+2}(\S_{t_0})}
      + \|t_0^{-\bq_I}\bo^{I}_i\|_{H^{k_{1}+2}(\S_{t_0})}
      & \leq \rho, \label{eq: e bar o bar assumption}\\
      \ldr{\ln(t_0)}^{-1}\|t_0^{\bq_I + \bq_J - \bq_K }\bga_{IJK}\|_{H^{k_1+1}(\S_{t_0})} 
      + \norm{\bq_I}_{H^{k_1 + 2}(\S_{t_0})}
      &\leq \rho, \label{eq: gamma bar p bar assumption} \\
      %
      \norm{\bp_0 - t_0 \ln(t_0) \bp_1}_{H^{k_1 + 2}} + t_0 \|\bp_1\|_{H^{k_{1}+2}(\S_{t_0})} 
      &\leq \rho\label{eq: scalar field bar assumption}
    \end{align}
  \end{subequations}  
  for all $i,I, J, K $, where we define, for all $I, J, K $,
  \begin{itemize}
  \item the co-frame $(\bo^I)_{I=1}^n$ as the dual frame to $(\be_I)_{I=1}^n$,
    with $\bo^I_i = \bo^I(E_i)$;
  \item the structure coefficients $\bga_{IJK} := \bo^K([\be_I, \be_J])$.
  \end{itemize}
  If, in addition to the above, 
  \begin{equation}\label{eq: bar pi diagonal non-deg conditions}
    \abs{\bq_I - \bq_J}> \rho^{-1},
  \end{equation}
  then the data (\ref{eq:diagonal Fermi Walker ID}) are said to satisfy the \emph{non-degenerate FRS expansion-normalized bounds
  of regularity $k_1$ for $\rho$ at $t_0$}.
\end{definition}

We now show that these bounds follow from the assumptions in Theorem~\ref{thm: big bang formation}.
In the statement, and in what follows, it is convenient to recall the conventions concerning the
Einstein summation convention introduced in Subsection~\ref{ssection:notation}.

\begin{prop} \label{prop: initial data estimate}
  Let $\s_p$, $\s_V$, $\s$, $k_0$, $k_1$, $(\S, h_\refer)$, $(E_i)_{i=1}^{n}$ and $V$ be as in Theorem~\ref{thm: big bang formation}. Let, moreover,
  $\zeta_0 > 0$ and
  assume that $(\S, \bh, \bk, \bp_0, \bp_1)$ are a $\s_p$-admissible solution to the constraint equations \eqref{seq:con} with constant mean curvature
  $1/t_{0}\in (0,\infty)$, such that the corresponding expansion-normalized quantities $(\bmH, \bmK, \bP_0, \bP_1)$
  (see~Definitions~\ref{def: exp nor Weingarten map}-\ref{def:Phizodef}) satisfy (\ref{eq: Sobolev bound assumption})
  and $|\bq_I-\bq_J|>\zeta_{0}^{-1}$ for $I\neq J$. Then there is a $\rho_0 > 0$, depending
  only on $\zeta_0$, $\s$, $k_0$, $k_1$, $(\S,h_{\refer})$ and $(E_i)_{i = 1}^n$, and unique (up to a choice of sign for each $\be_{I}:=\be_{I}^{i}E_{i}$)
  $\be_I^i$, $\bq_I\in C^{\infty}(\S,\rn{})$, $i, I = 1, \hdots, n$, such that $(\be_{I})_{I=1}^{n}$ is an orthonormal frame of $(\S,\bh)$,
  \begin{align*}
    \bh
    &= \bo^I \otimes \bo^I,\ \ \ \bk =  \tfrac{\bq_I}{t_0} \bo^I \otimes \bo^I,
  \end{align*}
  where $(\bo^{I})_{I=1}^{n}$ is the dual frame to $(\be_{I})_{I=1}^{n}$, and such that $(\be_I^i, \bq_I, \bp_0, \bp_1)$, for all $i, I $, are
  smooth diagonal FRS initial data satisfying the non-degenerate FRS expansion-normalized bounds of regularity $k_1$ for $\rho_0$ at $t_0$.
\end{prop}
We use the following notational convention throughout the rest of the paper:
\begin{notation} \label{not:constant}
  Let $\s_p$, $\s_V$, $\s$, $k_0$, $k_1$, $(\S, h_\refer)$, $(E_{i})_{i=1}^{n}$ and $V$ be as in Theorem~\ref{thm: big bang formation} and
  let $\rho_{0}>0$. Then a \textit{standard constant} is a strictly positive constant that only depends on $\rho_0$, $\s_p$, $\s_V$, $k_0$, $k_1$,
  $c_{k_{1}+2}$ (see Definition~\ref{def:Vadm}), $(\S, h_\refer)$ and the
  orthonormal frame $(E_i)_{i=1}^{n}$. Constants called $C$ are from now on tacitly assumed to be standard constants, and their values might
  change from line to line.
\end{notation}
\begin{remark}
  The constant $\rho_{0}$ should be thought of as arising from $\zeta_{0}$ (appearing in the statement of Theorem~\ref{thm: big bang formation}) as
  described in Proposition~\ref{prop: initial data estimate}.
\end{remark}
To prove Proposition~\ref{prop: initial data estimate}, we first need the following lemma:

\begin{lemma} \label{le: first bound on $q_I$}
If $\bq_1, \hdots, \bq_n$ are $\sigmap$-admissible eigenvalues (see~Definition~\ref{def: alpha admissible}), then
\begin{align} 
    \bq_{I}+\bq_{J}-\bq_{K}
        &< 1 - 5 \s, \label{eq:bqboundsa} \\
    |\bq_{I}|
        &< 1 - 5 \s\label{eq:bqboundsb}
\end{align}
for all $I, J, K = 1, \hdots n$ with $I \neq J$.
\end{lemma}
\begin{proof}[Proof of Lemma~\ref{le: first bound on $q_I$}]
  Recalling that $\s \leq \frac \sigmap5$, the first assertion is immediate from \eqref{cond: alternating q}. By choosing $J = K$ in
  \eqref{eq:bqboundsa}, it follows that $\bq_I < 1-5\s$, proving the first inequality in \eqref{eq:bqboundsb}. Assume now, to reach a
  contradiction, that
  \begin{equation} \label{eq: contradiction assumption}
    - \bq_I \geq 1 - 5\s.
  \end{equation}
  Since $\sum_{I} \bq_I = 1$, we conclude that
  \[
  \textstyle{\sum}_{J = 1, J \neq I}^n \bq_J = 1 - \bq_I \geq 2 - 5\s.
  \]
  Since $\bq_J < 1 - 5\s$, this inequality is impossible if $n = 2$.
  If $n \geq 3$, it follows that there are at least two positive terms, say $\bq_{J_1} \text{ and }\bq_{J_2}$.
  By choosing $I = J_1$, $J = J_2$ and $K = I$ in \eqref{eq:bqboundsa}, it follows that
  \[
  - \bq_I 
  < \bq_{J_1} + \bq_{J_2} - \bq_I
  < 1 - 5 \s,
  \]
  which is a contradiction to \eqref{eq: contradiction assumption}. This proves \eqref{eq:bqboundsb}.
\end{proof}

\begin{proof}[Proof of Proposition~\ref{prop: initial data estimate}]
  In this proof, we let $C > 0$ denote a constant that is allowed to depend on $\rho_0$, $\s$, $k_0$, $k_1$, $(\S, h_\refer)$ and
  $(E_i)_{i=1}^n$. The value of $C$ might change from line to line. By the non-degeneracy assumption, we know that $\bh$ and $\bk$ are simultaneously
  diagonalizable, i.e.~there exist smooth functions $\bq_1, \hdots, \bq_n$ and vector fields $\be_1, \hdots, \be_n$ on $\S$, such that
  \[
    \bh = \bo^I \otimes \bo^I, \quad \bk = \tfrac{\bq_I}{t_0} \bo^I \otimes \bo^I,
  \]
  where $\bo^1, \hdots, \bo^n$ are the one-forms dual to $\be_1, \hdots, \be_n$. Moreover, the $\bq_I$ are uniquely determined, and the
  $\be_I$ are uniquely determined up to a sign. A standard implicit function theorem
  argument shows that $\bq_I$ and $\be_I$ are smooth. The estimate \eqref{eq: bar pi diagonal conditions} is immediate by the assumption that the
  eigenvalues $\bq_1, \hdots, \bq_n$ of $\bmK$ are $\s_p$-admissible. Next, \eqref{eq: bar pi diagonal non-deg conditions} is
  immediate if $\rho_0 > \zeta_0$. It therefore remains to prove \eqref{eq:FermiWalkerBounds}.
  By construction of $\bmH$ and $\bmK$, $\bq_I$ and $v_{I}:=t_0^{\bq_I} \be_I$ are distinct eigenvalues and eigenvectors of $\bmK$ and the
  $v_{I}$ are unit vector fields with respect to $\bmH$. In order to deduce bounds on these objects, let us first write
  \begin{align*}
    \bmK_i^j
    &=  \eta^i(\bmK(E_i)), \ \ \
    \bmH_{ij}
    = \bmH(E_i, E_j),\ \ \
    v_I^i
    = t_0^{\bq_I} \be_I^i
  \end{align*}
and hence
\begin{align*}
    \bmK_i^j v^i_I
        &= \bq_I v^j_I, \ \ \
    \bmH_{ik}v_I^i v_J^k
        = \de_{IJ},
\end{align*}
for $j,I, J = 1, \hdots, n$ (no summation on $I$ in the first equality).
Differentiating these equalities for a fixed $I = J$ yields
\begin{equation*}
\begin{split}
    \begin{pmatrix}
        0 \\
        \vdots \\
        0 \\
        0
    \end{pmatrix}
        &= E_l \begin{pmatrix}
        \bmK_i^1 v_I^i - \bq_I v_I^1 \\
        \vdots \\
        \bmK_i^nv_I^i - \bq_I v_I^n \\
        \bmH_{ik}v_I^i v_I^k
    \end{pmatrix} \\
        &= 
        \begin{pmatrix}
            (E_l \bmK_i^1)v^i_I \\
            \vdots \\
            (E_l \bmK_i^n)v^i_I \\
            (E_l \bmH_{ik})v^i_Iv^k_I
        \end{pmatrix}
        +
        \underbrace{\begin{pmatrix}
            \bmK_1^1 - \bq_I & \hdots & \bmK_n^1 & - v_I^1 \\
            \vdots & \ddots & \vdots & \vdots \\
            \bmK_1^n & \hdots & \bmK_n^n - \bq_I & - v_I^n \\
            2 \bmH_{1k} v_I^k & \hdots &  2 \bmH_{nk} v_I^k & 0
        \end{pmatrix}}_{=: A_I}
        E_l
        \begin{pmatrix}
            v_I^1 \\
            \vdots \\
            v_I^n \\
            \bq_I
        \end{pmatrix}.
\end{split}
\end{equation*}
We decompose $A_I$ by noting that
\[
    A_I 
        =   \underbrace{
        \begin{pmatrix}
			v_1^1 & \hdots & v_n^1 & 0 \\
			\vdots & \ddots & \vdots & \vdots \\
			v_1^n & \hdots & v_n^n & 0 \\
			0 & \hdots & 0 & 1
		\end{pmatrix}}_{=: T}
            \underbrace{
            \begin{pmatrix}
			\bq_1 - \bq_I & \hdots & 0 & 0 \\
			\vdots & \ddots & \vdots & -1  \\
			0 & \hdots & \bq_n - \bq_I & 0 \\
			0 & 2 & 0 & 0
		\end{pmatrix}}_{=: B_I}
            \underbrace{
            \begin{pmatrix}
			v^i_1 \bmH_{i1} & \hdots & v_1^i\bmH_{in} & 0 \\
			\vdots & \ddots & \vdots & \vdots \\
			v_n^i \bmH_{i1} & \hdots & v_n^i \bmH_{in} & 0 \\
			0 & \hdots & 0 & 1
		\end{pmatrix}}_{= T^{-1}},
\]
where the entries ``$2$'' and ``$-1$'' in $B_I$ are placed in column $I$ and row $I$, respectively.
It follows, in particular, that $A_I$ is invertible and we conclude that
\begin{equation} \label{eq: matrix equality for initial bounds}
	E_l
        \begin{pmatrix}
            v_I^1 \\
            \vdots \\
            v_I^n \\
            \bq_I
        \end{pmatrix}
    = - A_I^{-1}
    \begin{pmatrix}
        (E_l \bmK_i^1)v^i_I \\
        \vdots \\
        (E_l \bmK_i^n)v^i_I \\
        (E_l \bmH_{ik})v^i_Iv^k_I
    \end{pmatrix}
    = - T B_I^{-1} T^{-1} \begin{pmatrix}
        (E_l \bmK_i^1)v^i_I \\
        \vdots \\
        (E_l \bmK_i^n)v^i_I \\
        (E_l \bmH_{ik})v^i_Iv^k_I
    \end{pmatrix}.
\end{equation}
By the assumption of non-degeneracy of the eigenvalues with margin $\zeta_0^{-1}$, we conclude that
\[
	\left| B_I^{-1}
	\right|
	\leq
	\max \left(\zeta_0, 1 \right),
\]
where $\abs{\cdot}$ denotes the pointwise operator norm on $\rn{n+1}$.
We therefore know that $B_I$ only depends on $\bq_I$ and that its inverse is bounded uniformly.
Moreover, the entries in the matrices $T$ and $T^{-1}$ are products of $\bmH_{ij}$ and $v_j^i$.
Due to Lemma~\ref{le: first bound on $q_I$},
\[
\norm{\bq_I}_{C^{0}(\S_{t_0})}
\leq 1 - 5\s
\leq 1.
\]
Due to the assumed bound on $\bmH^{-1}$, see (\ref{eq: Sobolev bound assumption}), we know that
\begin{equation}\label{eq:bmHinvbd}
  \big(\textstyle{\sum}_{i,j}\bmH^{ij}\bmH^{ij}\big)^{1/2}\leq \zeta_{0},
\end{equation}
where the indices are with respect to the frame $(E_{i})_{i=1}^{n}$ and dual frame $(\eta^{i})_{i=1}^{n}$. If the eigenvalues of
the symmetric matrix with components $\bmH_{ij}$ are $\lambda_{i}>0$, then (\ref{eq:bmHinvbd}) says that $|\mu|\leq \zeta_{0}$,
where $\mu=(\lambda_{1}^{-1},\dots,\lambda_{n}^{-1})$. In particular, $\lambda_{\min}\geq\zeta_0^{-1}$, so that
\[
\bmH\geq \zeta_{0}^{-1}h_{\refer}.
\]
Due to this estimate, 
\begin{align*}
    \snorm{t_0^{\bq_I}\be_I^i}_{C^{0}(\S)}^2
        &\leq \snorm{\textstyle{\sum}_{i} \left(
 t_0^{\bq_I}\be_I^i \right)^2}_{C^{0}(\S)} \\
        &= \snorm{h_\refer(t_0^{\bq_I}\be_I, t_0^{\bq_I}\be_I)}_{C^{0}(\S)} 
        \leq \zeta_0 \snorm{\bmH(t_0^{\bq_I}\be_I, t_0^{\bq_I}\be_I)}_{C^{0}(\S)}= \zeta_0.
\end{align*}
Thus $\|v_{I}^{i}\|_{C^{0}(\S)}\leq \zeta_0^{1/2}$. Therefore, by taking iteratively more derivatives of
\eqref{eq: matrix equality for initial bounds}, Lemma~\ref{lemma: products in L2} yields
\begin{equation}\label{eq: iteration equation}
\begin{split}
    & \textstyle{\sum}_{I}\snorm{\bq_I}_{H^{l+1}(\S)} + \textstyle{\sum}_{i,I}\snorm{t_0^{\bq_I}\be_I^i}_{H^{l+1}(\S)} \\
    \leq & C \textstyle{\sum}_{I}\snorm{\bq_I}_{H^l(\S)} + C \textstyle{\sum}_{i,I}\snorm{t_0^{\bq_I}\be_I}_{H^l(\S)}
    + C \snorm{\bmK}_{H^{l + 1}(\S)} + C \snorm{\bmH}_{H^{l + 1}(\S)}
\end{split}
\end{equation}
for all $l \in \nn{}$, where $C > 0$ only depends on $\zeta_0$, $l$, $(\S,h_{\refer})$ and $(E_{i})_{i=1}^{n}$; note
that we here appeal to \eqref{eq: Sobolev bound assumption} and Sobolev embedding in order to estimate $\bmK_{i}^{j}$
and $\bmH_{ij}$ in $C^{1}$. Iteratively applying \eqref{eq: iteration equation} and \eqref{eq: Sobolev bound assumption} therefore implies the
second part of \eqref{eq: gamma bar p bar assumption} and the first part of \eqref{eq: e bar o bar assumption}. Note that
\begin{align*}
    t_0^{-\bq_I}\bo^I_i 
        &= \bmH(E_i, t_0^{\bq_I}\be_I) 
        = \bmH_{ij} t_0^{\bq_I} \be_I^j.
\end{align*}
From this, Lemma~\ref{lemma: products in L2}, \eqref{eq: Sobolev bound assumption} and the first part of \eqref{eq: e bar o bar assumption},
we conclude the second part of \eqref{eq: e bar o bar assumption}.

In order to prove the first part of \eqref{eq: gamma bar p bar assumption}, we first compute
\begin{equation}\label{eq:bgIJK right frame}
  \begin{split}
    \bga_{IJK} = &\bo^K \left( \left[ \be_I, \be_J \right] \right)
    = \bo^K \left( \left[ t_0^{-\bq_I} t_0^{\bq_I} \be_I, t_0^{-\bq_J} t_0^{\bq_J} \be_J \right] \right) \\
    = & t_0^{\bq_K}t_0^{-\bq_K}\bo^K \left( t_0^{-\bq_I} \left( t_0^{\bq_I} \be_It_0^{-\bq_J} \right) t_0^{\bq_J} \be_J
      - t_0^{-\bq_J} \left( t_0^{\bq_J} \be_Jt_0^{-\bq_I} \right) t_0^{\bq_I} \be_I \right.\\
    & \left. \phantom{t_0^{\bq_K}t_0^{-\bq_K}\bo^K x} + t_0^{-\bq_I - \bq_J} [t_0^{\bq_I}\be_I, t_0^{\bq_J}\be_J]\right) \\
    = & t_0^{- \bq_I - \bq_J + \bq_K} \left( - t_0^{\bq_I}\be_I(\bq_J) \ln(t_0) \de_{JK} + t_0^{\bq_J}\be_J(\bq_I) \ln(t_0) \de_{IK} \right.\\
    & \left. \phantom{t_0^{- \bq_I - \bq_J + \bq_K}x}+ t_0^{-\bq_K} \bo^K([t_0^{\bq_I}\be_I, t_0^{\bq_J}\be_J]) \right).
  \end{split}
\end{equation}
By Lemma~\ref{lemma: products in L2}, \eqref{eq: e bar o bar assumption} and the second part of \eqref{eq: gamma bar p bar assumption},
we conclude that
\begin{align*}
    \norm{t_0^{\bq_I + \bq_J - \bq_K}\bga_{IJK}}_{H^{k_1+1}(\S)}
        &\leq C \ldr{\ln(t_0)},
\end{align*}
proving the first part of \eqref{eq: gamma bar p bar assumption}. Since $t_0 \bp_1 = \bP_1$, and
\[
    \bp - t_0 \bp_1 \ln(t_0)
        = \bp + \bP_1 \ln(\bth)
        = \bP_0,
\]
\eqref{eq: scalar field bar assumption} is immediate from \eqref{eq: Sobolev bound assumption}.
This finishes the proof.
\end{proof}

\subsection{Diagonal FRS initial data arising from background solutions}

Next, we prove that quiescent model solutions, see Definition~\ref{def:quiescent model solution}, give rise to diagonal FRS
initial data satisfying the bounds of Definition~\ref{def:DecomposedData}.

\begin{lemma}\label{lemma: from qms to frs data}
  Let $(M,g,\phi)$ be a quiescent model solution. Using the terminology of Definition~\ref{def:quiescent model solution},
  let $\s$, $k_0$ and $k_1$ be defined as in Theorem~\ref{thm: big bang formation}. Then there is a $\tau_{1}\leq\tau_{a}$, $\tau_1\in\mI$,
  and a $0<\rho\in\rn{}$ such that for any $\tau_0\leq\tau_{1}$, 
  \begin{equation}\label{eq:Fermi Walker ID SH case}
    \be_{I}^{i}:=\tfrac{1}{a_{\uI}(\tau_{0})}\delta_{\uI}^{i},\ \ \
    \bq_I:=\tfrac{1}{(\theta a_\uI)(\tau_0)}a_{\uI}'(\tau_0),\ \ \
    \bp_0:=\phi(\tau_{0}),\ \ \
    \bp_1:=(\d_{\tau}\phi)(\tau_{0})
  \end{equation}
  constitute diagonal FRS initial data satisfying the FRS expansion-normalized bounds of regularity $k_1$ for $\rho$ at
  $t_{0}:=1/\theta(\tau_{0})$, where the $\s_p$ appearing in (\ref{eq: bar pi diagonal conditions}) is the same as the $\s_p$ appearing in
  Definition~\ref{def:quiescent model solution}.
\end{lemma}
\begin{proof}
  Clearly, $(\be_{I})_{I=1}^{n}$ is a global frame and $\sum_{I}\bq_I=1$ by the definition of $\theta$. Assuming
  $\tau_{0}$ to be close enough to $\tau_-$, we can assume $\bq_I+\bq_J-\bq_K<1-\s_p$ in case $I\neq J$. Next, consider
  \begin{equation}\label{eq:frame sh case}
    t_{0}^{\bq_I}\be_{I}^{i}=[\theta(\tau_{0})]^{-\bq_\uI}[a_{\uI}(\tau_{0})]^{-1}\delta_{\uI}^{i}
    =[\theta(\tau_{0})]^{\mrp_\uI-\bq_\uI}[(\theta(\tau_{0}))^{\mrp_\uI}a_{\uI}(\tau_{0})]^{-1}\delta_{\uI}^{i}.
  \end{equation}
  Note that the third factor on the far right hand side is constant and the second factor converges to $\alpha_I^{-1}$ as $\tau_0\downarrow\tau_-$.
  Concerning the first factor, note that
  \[
  \ln[\theta(\tau_{0})]^{\mrp_I-\bq_I}=(\mrp_I-\bq_I)\ln[\theta(\tau_{0})].
  \]
  Due to (\ref{eq:mKietcconvergence}), the right hand side converges to zero as $\tau_{0}\downarrow\tau_-$. Thus the first factor on the far
  right hand side of (\ref{eq:frame sh case}) converges to $1$ as $\tau_{0}\downarrow\tau_-$. By this and a similar argument for $\bo^{I}_{i}$,
  it is clear that (\ref{eq: e bar o bar assumption}) holds for some right hand side independent of $\tau_0\leq\tau_a$. Since $\bq_I$ is independent
  of the spatial variables and converges to $\mrp_I$, it is clear that there is a uniform bound on $\bq_I$ in $H^{k_1+2}$ for $\tau_{0}\leq\tau_a$.
  Since
  \[
  \bga_{IJK}=\tfrac{a_{\uK}(\tau_0)}{a_\uI(\tau_{0})a_\uJ(\tau_0)}\eta^\uK([E_\uI,E_\uJ]),
  \]
  an argument similar to the one concerning $t_{0}^{\bq_I}\be_{I}^{i}$ yields the conclusion that $t_{0}^{\bq_I+\bq_J-\bq_K}\bga_{IJK}$ is uniformly bounded
  in $H^{k_1+2}$ for $\tau_{0}\leq\tau_a$. Since $t_{0}\bp_1$ converges as $\tau_0\downarrow\tau_-$ and is independent of the spatial variables, it
  is uniformly bounded in $H^{k_1+2}$ for $\tau_{0}\leq\tau_a$. Finally, consider
  \begin{equation*}
    \begin{split}
      \bp_0-t_0\ln(t_0)\bp_1 = & \phi(\tau_0)+[\theta(\tau_0)]^{-1}\ln[\theta(\tau_0)](\d_{\tau}\phi)(\tau_0)\\
      = & \Phi_0+[\phi(\tau_0)+\Phi_1\ln[\theta(\tau_0)]-\Phi_0]
      +\ln[\theta(\tau_0)]\big(\tfrac{1}{\theta(\tau_0)}(\d_{\tau}\phi)(\tau_0)-\Phi_1\big).
    \end{split}
  \end{equation*}
  Due to (\ref{eq:mKietcconvergence}), the left hand side converges to $\Phi_0$ as $\tau_0\downarrow\tau_-$. Thus
  $\bp_0-t_0\ln(t_0)\bp_1$ is uniformly bounded in $H^{k_1+2}$ for $\tau_{0}\leq\tau_a$. The lemma follows. 
\end{proof}

\subsection{Initial bounds on the lapse}

Note that Definition~\ref{def:DecomposedData} does not require any bounds for the lapse. The reason is that if $t_0$ is small enough, then
the lapse function is uniquely determined by (\ref{eq:LapseEquation}) and automatically close to $1$.

\begin{remark} \label{rmk: alternative lapse}
  Assuming the Hamiltonian constraint (\ref{eq:hamconstraint}) to hold, (\ref{eq:LapseEquation}) can be reformulated to
  the \emph{alternative lapse equation}:
  \begin{equation} \label{eq: alternative lapse}
    \begin{split}
    e_{I}e_{I}(N) = & t^{-2}(N-1)+\g_{JII}e_{J}(N) \\
    & - \big( t^{-2} - k_{IJ} k_{IJ} - e_0(\phi) e_0(\phi) + \tfrac{2}{n-1}V\circ\phi \big)  N.
  \end{split}  
\end{equation}
\end{remark}
We use this new form to derive an estimate for the initial lapse at $t_0$.

\begin{prop} \label{prop: initial lapse estimate}
  Let $\s_p$, $\s_V$, $\s$, $k_0$, $k_1$, $(\S, h_\refer)$, $(E_{i})_{i=1}^{n}$ and $V$ be as in Theorem~\ref{thm: big bang formation} and let
  $\rho_{0}>0$. Then there are standard constants $\tau_{1}<1$ and $C$ such that the following
  holds: If $t_0 < \tau_1$ and $\be_I^i$, $\bq_I$, $\bp_0$, $\bp_1\in C^{\infty}(\S,\rn{})$ are diagonal FRS initial data, solving
  (\ref{eq:hamconstraint}) and (\ref{eq:MC}), and satisfying the FRS expansion-normalized bounds of regularity $k_1$ for $\rho_0$ at $t_0$
  (see~Definition~\ref{def:DecomposedData}), then there is a unique $\bN\in C^{\infty}(\S,\rn{})$ satisfying the alternative lapse equation,
  \eqref{eq: alternative lapse}, at $t_0$. Moreover, 
  \begin{align*}
    &t_0^2 \textstyle{\sum}_{I}\snorm{\be_I (\bN-1)}_{H^{k_1}(\S)}^2 + \snorm{\bN-1}_{H^{k_1}(\S)}^2 \\
    \leq & C t_0^{9\s/2}\left( t_0 \textstyle{\sum}_{I}\snorm{\be_I (\bN-1)}_{H^{k_1}(\S)}
    + \snorm{\bN - 1}_{H^{k_1}(\S)} + 1 \right) \snorm{\bN-1}_{H^{k_1}(\S)}.
  \end{align*}
\end{prop}

In order to prove this, the following lemma is crucial.

\begin{lemma} \label{lemma:EstimateHamiltonianConstraint}
  Let $\s_p$, $\s_V$, $\s$, $k_0$, $k_1$, $(\S, h_\refer)$, $(E_{i})_{i=1}^{n}$ and $V$ be as in Theorem~\ref{thm: big bang formation} and let
  $\rho_{0}>0$. Then there are standard constants $\tau_{1}<1$ and $C$ such that the following
  holds: If $t_0 < \tau_1$ and $\be_I^i$, $\bq_I$, $\bp_0$, $\bp_1\in C^{\infty}(\S,\rn{})$ are diagonal FRS initial data, solving
  (\ref{eq:hamconstraint}) and (\ref{eq:MC}), and satisfying the FRS expansion-normalized bounds of regularity $k_1$ for $\rho_0$ at $t_0$
  (see~Definition~\ref{def:DecomposedData}), then
  \begin{subequations}\label{seq:preLapseEstimate}
    \begin{align}
      t_0 \abs{\bp_1}
      &< 1 - \tfrac1{2n},\label{eq:absbphioneest} \\
      t_0^2 \norm{V \circ \bp_0}_{H^{k_1}(\S)}
      &\leq C t_0^{6 \s}, \label{eq:Vnormpreest}\\
      \norm{1 - \textstyle{\sum}_{I} \bq_I^2 - t_0^2\bp_1^{2}}_{H^{k_1}(\S)}
      &\leq  C t_0^{6\s}.\label{eq:deviationfromKasner}
    \end{align}
  \end{subequations}
\end{lemma}
\begin{proof}
The Hamiltonian constraint, \eqref{eq:hamconstraint}, at $t_0$ is
\begin{align*}
    & 2\be_{I}(\bga_{IJJ})
  - \tfrac{1}{4}\bga_{IJK}(\bga_{IJK} + 2 \bga_{IKJ}) -\bga_{IJJ}\bga_{IKK}
  - \bk_{IJ}\bk_{IJ}+{t_0}^{-2} \\
  = & \bp_1^{2} + \be_{I}(\bp_0)\be_{I}(\bp_0) + 2V \circ \bp_0.
\end{align*}
Multiplying this by $t_0^2$, using that $\bk_{IJ} = \tfrac{\bq_\uI \de_{\uI J}}{t_0}$ and rearranging yields
\begin{equation}\label{eq:essential terms Ham con}
  \begin{split}
    & \bP_{1}^{2}+\textstyle{\sum}_{I}\bq_I^2 +2t_0^2V \circ \bp_0-1\\
    = & -t_{0}^{2}\be_{I}(\bp_0)\be_{I}(\bp_0)+t_{0}^{2}[2\be_{I}(\bga_{IJJ})- \tfrac{1}{4}\bga_{IJK}(\bga_{IJK} + 2 \bga_{IKJ}) -\bga_{IJJ}\bga_{IKK}];
  \end{split}
\end{equation}
note that $\bP_{1}=t_{0}\bp_{1}$. Next, appealing to Lemma~\ref{lemma: products in L2}, Sobolev embedding
and \eqref{eq:FermiWalkerBounds} yields
\begin{equation}\label{eq:beIbpzsq}
  \begin{split}
    \snorm{t_0^2 \be_{I}(\bp_0)\be_{I}(\bp_0)}_{H^{k_1}(\S)}
    &= \snorm{t_0^{2(1 - \bq_I)}t_0^{\bq_I}\be_{I}(\bp_0)t_0^{\bq_I}\be_{I}(\bp_0)}_{H^{k_1}(\S)} \\
    &\leq C\snorm{t_0^{1 - \bq_I}}_{H^{k_1}(\S)}^2 \snorm{t_0^{\bq_I} \be_I}_{H^{k_1}(\S)}^2 \snorm{\bp_0}_{H^{k_1+1}(\S)}^2 \\
    &\leq C \snorm{t_0^{1 - \bq_I}}_{C^{0}(\S)}^2  \ldr{\ln(t_0)}^{2k_1+2} \left(\snorm{\bq_I}_{H^{k_1}(\S_t)} + 1 \right)^2 \\
    &\leq C t_0^{10\s} \ldr{\ln(t_0)}^{2k_1+2} \leq C t_0^{6\s},
  \end{split}
\end{equation}
where we used that $\abs{\bq_I} < 1 - 5 \s$, by Lemma~\ref{le: first bound on $q_I$}, and that $t_0^{2\s} \ldr{\ln t_0}^{k_1+1}$
is bounded for $t_0 \in [0,1]$ by a constant depending only on $\s$ and $k_1$. Since Lemma~\ref{le: first bound on $q_I$} implies
that $\bq_I<1-5\s$ and $\bq_I + \bq_J - \bq_K < 1 - 5 \s$ if $I \neq J$, we similarly get
\begin{equation}\label{eq:curvatureestimateKasnerHam}
  \norm{- 2 t_0^2 \be_I(\bga_{IJJ}) + \tfrac{1}{4} t_0^2 \bga_{IJK}(\bga_{IJK} + 2 \bga_{IKJ}) + t_0^2 \bga_{IJJ}\bga_{IKK}}_{H^{k_1 + 1}(\S)}
  \leq C t_0^{6\s};
\end{equation}
recall (\ref{eq:bgIJK right frame}). Due to (\ref{eq:essential terms Ham con}), (\ref{eq:beIbpzsq}) and (\ref{eq:curvatureestimateKasnerHam}),
\begin{equation} \label{eq: initial data canonical}
  \|\bP_{1}^{2}+\textstyle{\sum}_{I}\bq_I^2 +2t_0^2V \circ \bp_0-1\|_{H^{k_{1}}(\S)}\leq Ct_{0}^{6}.
\end{equation}
Combining this estimate with Sobolev embedding and the non-negativity of the potential yields
\begin{equation}\label{eq:abs Phi1 initial bound}
  \tfrac1{1 - \s_V} - |\bP_1|\geq \rho_{0}^{-1},
\end{equation}
assuming the standard constant $\tau_{1}$ to be small enough. 

Next we estimate the term involving the potential. Combining \eqref{eq:abs Phi1 initial bound}, \eqref{eq: scalar field bar assumption} and
Sobolev embedding yields, recalling Notation~\ref{not:constant},
\[
    \abs{\bp_0}
        \leq  -\ln(t_0) t_0 \bp_1 + C\rho_0 
        \leq -\big( \tfrac1{1 - \s_V} - \rho_0^{-1} \big) \ln(t_0) + C\rho_0.
\]
Since $V$ is a $\s_V$-admissible potential (see~Definition~\ref{def:Vadm}), it follows that
\begin{equation} \label{eq: V bar phi 0 H k 1 bound}
    \textstyle{\sum}_{l \leq k_1 + 2}\snorm{V^{(l)}\circ\bp_0}_{C^{0}(\S_{t_0})}
        \leq C e^{2(1-\s_{V})|\bp_0|}
        \leq C t_0^{- 2(1 - (1-\s_V)\rho_0^{-1})}.
\end{equation}
Note, by \eqref{eq: scalar field bar assumption}, that 
\begin{equation} \label{eq: bar phi 0 H k 1 plus 2 bound}
    \snorm{\bp_0}_{H^{k_1 + 2}(\S)}
        \leq \abs{\ln(t_0)} t_0 \snorm{\bp_1}_{H^{k_1 + 2}(\S)} + \rho_0 
        \leq C \ldr{\ln(t_0)}.
\end{equation}
Combining \eqref{eq: V bar phi 0 H k 1 bound}, \eqref{eq: bar phi 0 H k 1 plus 2 bound} and Lemma~\ref{lemma: products in L2} yields
\[
  t_0^2 \norm{V \circ \bp_0}_{H^{k_1}(\S)}
  \leq C \ldr{\ln(t_0)}^{k_1} t_0^{2(1-\s_V)\rho_0^{-1}}.
\]      
Combining this estimate with \eqref{eq: initial data canonical}, 
\begin{equation}\label{eq:Rough Hamiltonian Constraint}
  \norm{1 - \textstyle{\sum}_{I} \bq_I^2 - t_0^2\bp_1^{2}}_{H^{k_1}(\S)}
  \leq C \ldr{\ln(t_0)}^{k_{1}} t_0^{\min\{2(1 - \s_V)\rho_0^{-1}, 6\s\}}.
\end{equation}
We now improve this estimate as follows.
First note that if $\bq:=(\bq_{1},\dots,\bq_{n})$ and $\xi_{0}=(1,\dots,1)$, then
$1=\bq\cdot\xi_{0}\leq |\bq|\cdot|\xi_{0}|=\sqrt{n}|\bq|$, so that $|\bq|^{2}\geq 1/n$. Thus
\[
    1 + \left| 1 - |\bq|^{2} - \bP_1^2 \right|
        \geq |\bq|^{2} + \bP_1^2
        \geq \tfrac1n + t_0^2 \bp_1^2.
\]
Combining this estimate with (\ref{eq:Rough Hamiltonian Constraint}) and Sobolev embedding yields
\begin{align*}
    t_0^2 \bp_1^2
        &\leq 1 - \tfrac1n + C \ldr{\ln t_0}^{k_1} t_0^{\min(2(1 - \s_V)\rho_0^{-1}, 6 \s)} 
        < 1 - \tfrac1n + \tfrac 1{4n^2}= \left( 1 - \tfrac1{2n} \right)^2,
\end{align*}
if $\tau_1$ is a small enough standard constant. Thus (\ref{eq:absbphioneest}) holds.
Combining (\ref{eq:absbphioneest}) with \eqref{eq: scalar field bar assumption} implies the improved bound
\[
    \abs{\bp_0}
        \leq - \ln(t_0) \left( 1 - \tfrac1{2n} \right) + C\rho_0.
\]
We thus get the following improvement over \eqref{eq: V bar phi 0 H k 1 bound},
\[
    \textstyle{\sum}_{l \leq k_1 + 2}\snorm{V^{(l)}\circ\bp_0}_{C^{0}(\S_{t_0})}
        \leq C e^{2(1-\s_{V})|\bp_0|}
        \leq C t_0^{- 2(1-\s_V)\left( 1 - 1/(2n) \right)}.
\]
Combining this with \eqref{eq: scalar field bar assumption} and Lemma~\ref{lemma: products in L2} proves
\begin{equation*}
  \begin{split}
    t_0^2 \norm{V \circ \bp_0}_{H^{k_1}(\S)}
    \leq & C \ldr{\ln(t_0)}^{k_{1}}t_0^{2- 2(1-\s_V)\left( 1 - 1/(2n) \right)}
    \leq  C \ldr{\ln(t_0)}^{k_{1}} t_0^{6 \s+(1-\s_V)/n} \leq C t_0^{6 \s},
  \end{split}
\end{equation*}
where we used $\s \leq \frac{\s_V}3$. Thus (\ref{eq:Vnormpreest}) holds.
Combining (\ref{eq:Vnormpreest}) and (\ref{eq: initial data canonical})
yields (\ref{eq:deviationfromKasner}).
\end{proof}

We may now prove Proposition~\ref{prop: initial lapse estimate}.

\begin{proof}[Proof of Proposition~\ref{prop: initial lapse estimate}]
At $t_0$, Equation \eqref{eq: alternative lapse} can be written
\begin{equation}\label{eq:Lapse function N bar}
  \begin{split}
    \be_{I}\be_{I}(\bN - 1) = & t_0^{-2}(\bN-1) +\bga_{KII}\be_{K}(\bN - 1) \\
    & - t_0^{-2} \big( 1 - \textstyle{\sum}_{I} \bq_I^2 - t_0^2 \bp_1^2 + \tfrac{2}{n-1} t_0^2 V\circ \bp_0 \big) \bN.
  \end{split}
\end{equation}
Since all coefficients are smooth and since Lemma~\ref{lemma:EstimateHamiltonianConstraint} yields
\[
    \big| 1 - \textstyle{\sum}_{I} \bq_I^2 - t_0^2 \bp_1^2 + \tfrac{2}{n-1} t_0^2 V\circ \bp_0 \big|
        \leq C t_0^{6\s}, 
\]
standard theory for Laplace type operators ensures the existence of a unique solution $\bN$ if $t_0$ is small enough.
It thus remains to prove the estimate for $\bN$. Recall the conventions concerning multiindices introduced in
Appendix~\ref{sec:SobolevInequalities}. Applying $E_{\bfI}$ to (\ref{eq:Lapse function N bar}) and multiplying
the result with $E_{\bfI}(\bN-1)$ yields
\begin{align*}
    &\be_I E_\bfI \be_I (\bN - 1) E_{\bfI}(\bN-1) + [E_\bfI, \be_I] \be_{I}(\bN - 1) E_{\bfI}(\bN-1) -t_0^{-2}(E_{\bfI}(\bN-1))^2 \\
         = & E_\bfI \big( \bga_{KII}\be_{K}(\bN - 1) 
    - t_0^{-2} \big( 1 - \textstyle{\sum}_{I} \bq_I^2 - t_0^2 \bp_1^2 + \tfrac{2}{n-1}t_0^2V\circ \bp_0 \big) \bN \big) E_{\bfI}(\bN-1).
\end{align*}
Integrating this expressing over $\S$ and summing over all $\abs \bfI \leq k_1$, we integrate the first term by parts and obtain the equality
\begin{align*}
  & - \textstyle{\sum}_{\abs \bfI \leq k_1}\int_\S \rodiv_{h_\refer}(\be_I) E_\bfI \be_I (\bN-1) E_\bfI (\bN-1)\mu_{h_{\refer}}\\
    & - \textstyle{\sum}_{\abs \bfI \leq k_1}\int_\S E_\bfI \be_I(\bN-1) [\be_I, E_\bfI] (\bN-1)  \mu_{h_{\refer}} - \snorm{\be_I (\bN-1)}_{H^{k_1}(\S)}^2\\
  &  + \textstyle{\sum}_{\abs \bfI \leq k_1} \int_\S [E_\bfI, \be_I] \be_{I}(\bN - 1) E_{\bfI}(\bN-1) \mu_{h_{\refer}}
    - t_0^{-2} \norm{\bN-1}_{H^{k_1}(\S)}^2 \\
  = & \big\langle \bga_{KII}\be_{K}(\bN - 1) - t_0^{-2} \big( 1 - \textstyle{\sum}_{I} \bq_I^2 - t_0^2 \bp_1^2
    + \tfrac{2}{n-1}t_0^2V\circ \bp_0 \big) \bN, \bN - 1\big\rangle_{H^{k_1}(\S)}.
\end{align*}
Multiplying this equation with $t_0^2$ and applying the Cauchy-Schwarz inequality, we obtain
\begin{equation}\label{eq:sec step prel elliptic estimate}
  \begin{split}
    &t_0^2 \snorm{\be_I (\bN-1)}_{H^{k_1}(\S)}^2 + \snorm{\bN-1}_{H^{k_1}(\S)}^2 \\
    \leq & C t_0^2 \snorm{\rodiv_{h_\refer}(\be_I)}_{C^{0}(\S)} \snorm{\be_I (\bN-1)}_{H^{k_1}(\S)} \snorm{\bN-1}_{H^{k_1}(\S)} \\
    & + C t_0^2 \snorm{\be_I (\bN-1)}_{H^{k_1}(\S)} \textstyle{\sum}_{\abs \bfI \leq k_1}\snorm{[\be_I, E_\bfI](\bN-1)}_{L^2(\S)} \\
    & + C t_0^2 \textstyle{\sum}_{\abs \bfI \leq k_1}\snorm{[E_\bfI, \be_I]\be_I(\bN-1)}_{L^2(\S)} \snorm{\bN-1}_{H^{k_1}(\S)} \\
    & + C t_0^2 \snorm{\bga_{KII} \be_K(\bN-1)}_{H^{k_1}(\S)} \snorm{\bN-1}_{H^{k_1}(\S)} \\
    & + C \big\|\big( 1 - \textstyle{\sum}_{I} \bq_I^2 - t_0^2 \bp_1^2 + \tfrac{2}{n-1}t_0^2 V\circ \bp_0 \big) \bN\big\|_{H^{k_1}(\S)}
    \snorm{\bN-1}_{H^{k_1}(\S)}.
  \end{split}
\end{equation}
The last term is estimated using Lemma~\ref{lemma: products in L2}, Sobolev embedding and Lemma~\ref{lemma:EstimateHamiltonianConstraint}:
\begin{align*}
  & \big\|\big( 1 - \textstyle{\sum}_{I} \bq_I^2 - t_0^2 \bp_1^2 + \tfrac{2}{n-1}t_0^2V\circ \bp_0 \big) \bN\big\|_{H^{k_1}(\S)} \\
  \leq & C \big\|1 - \textstyle{\sum}_{I} \bq_I^2 - t_0^2 \bp_1^2 + \tfrac{2}{n-1}t_0^2V\circ \bp_0\big\|_{H^{k_1}(\S)}
         \snorm{\bN}_{H^{k_1}(\S)}\\
        \leq &  C t_0^{6\s}\snorm{\bN}_{H^{k_1}(\S)} \leq  C t_0^{6\s}(\snorm{\bN - 1}_{H^{k_1}(\S)} + 1).
\end{align*}
The other terms are estimated by using the fact that $\abs{\bq_I} < 1 - 5 \s$ and $\bq_I + \bq_J - \bq_K < 1 - 5 \s$, for $I \neq J$,
by Lemma~\ref{le: first bound on $q_I$}. To illustrate this, we show how the third term on the right hand side of
(\ref{eq:sec step prel elliptic estimate}) is estimated.
By Lemma~\ref{lemma: products in L2}, Sobolev embedding, \eqref{eq: e bar o bar assumption} and \eqref{eq: gamma bar p bar assumption},
\begin{align*}
  & t_0^2 \textstyle{\sum}_{\abs \bfI \leq k_1}\snorm{[E_\bfI, \be_I]\be_I\left(\bN-1\right)}_{L^2(\S)} \\
  = & \textstyle{\sum}_{\abs \bfI \leq k_1}\snorm{[E_\bfI, t_0^{1-\bq_I} t_0^{\bq_I}\be_I]t_0 \be_I\left(\bN-1\right)}_{L^2(\S)} \\
  \leq & C \snorm{t_0^{1-\bq_I}}_{H^{k_1}(\S)} \snorm{t_0^{\bq_I}\be_I}_{H^{k_1}(\S)} t_0 \snorm{\be_I (N-1)}_{H^{k_1}(\S)} \\
  \leq & C \snorm{t_0^{1 - \bq_I}}_{C^{0}(\S)} \left( \ldr{\ln(t_0)}^{k_1} \snorm{\bq_I}_{H^{k_1}(\S_t)} + 1 \right)
         t_{0}\snorm{\be_I (N-1)}_{H^{k_1}(\S)} \\
  \leq & C t_0^{5\s} \ldr{\ln(t_0)}^{k_1}t_{0}\snorm{\be_I (N-1)}_{H^{k_1}(\S)}  \leq C t_0^{9\s/2}t_{0}\snorm{\be_I (N-1)}_{H^{k_1}(\S)},
\end{align*}
from which we conclude that
\begin{align*}
    &C t_0^2 \textstyle{\sum}_{\abs \bfI \leq k_1}\snorm{[E_\bfI, \be_I]\be_I\left(\bN-1\right)}_{L^2(\S)} \snorm{\bN-1}_{H^{k_1}(\S)} \\
  \leq & C t_0^{9\s/2}t_{0}\snorm{\be_I (N-1)}_{H^{k_1}(\S)} \snorm{N-1}_{H^{k_1}(\S)}.         
\end{align*}
With similar estimates for the other terms, proven analogously, we conclude the assertion.
\end{proof}

\section{Past global existence of solutions to the FRS equations}\label{section: past global existence}

We now turn to the proof of past global existence of solutions to the FRS equations.

\subsection{The scaffold} 
\label{subsec:scaffold}

The first step is to construct an \emph{approximate} solution to the FRS equations, (\ref{eq: transport frame})--(\ref{eq:LapseEquation}).

\begin{definition}\label{def:scaffold variables}
  Given functions $\be_I^i$, $\bq_I$, $\bp_0$, $\bp_1\in C^{\infty}(\S,\rn{})$, that form diagonal FRS initial data
  (see~Definition~\ref{def:DecomposedData}), and a $t_0\in (0,\infty)$, the \textit{scaffold variables} on $(0, \infty) \times \S$ are
  defined as the vector fields
  \begin{equation}
    \ce_0
    := \d_t, \quad
    \ce_I
    := ( \tfrac {t_0}t )^{\bq_I} \be_I, \label{eq:chechomdefa}
  \end{equation}
  with dual frame
\begin{equation}
	\co^0
		:= \md t, \quad 
	\co^I 
		:= ( \tfrac t{t_0} )^{\bq_I} \bo^I, \label{eq:chechomdefb}
\end{equation}
together with the eigenvalues $\bq_1, \hdots, \bq_n$ and the scaffold scalar field
\begin{equation} \label{eq: scaffold scalar field}
    \cp
        := t_0\bp_1 \ln( \tfrac t {t_0} ) + \bp_0.
\end{equation}
We also defined the scaffold second fundamental form
\[
    \ck_{IJ}
        := \tfrac{\bq_I}{t} \de_{IJ} \quad \text{(no summation)}
\]
and scaffold structure coefficients
\[
    \cga_{IJK}
        := \co^K([\ce_I, \ce_J]),
\]
for all $I, J, K $.
\end{definition}
\begin{remark}
Note that $\ce_0, \ce_1, \hdots, \ce_n$ is an orthonormal frame with respect to the metric
\[
    \cg
        := - \md t^2 + \textstyle{\sum}_{I} ( \tfrac t{t_0} )^{2 \bq_I} \bo^I \otimes \bo^I,
\]
and that the components of the first and second fundamental forms of $\cg$ at $t_0$ are $\de_{IJ}$ and $\ck_{IJ}|_{t_{0}}$, for all $I, J $.
\end{remark}

In Example~\ref{example: spatially hom flat}, a family of spatially homogeneous and spatially flat model solutions were discussed.
In that case, the scaffold actually coincides with the solution to the Einstein-non-linear scalar field equations, choosing $t_0 = 1$.
This will in general not be the case, of course.
However, the proof of past global existence is based on bounding the deviation between the scaffold variables and the variables of the
actual solution.

\begin{definition} \label{def: scaffold dynamical variables}
We define the \emph{scaffold dynamical variables} to be
\[
	\ce_I, \co^I, \cga_{I J K}, \ck_{IJ}, \ce_I \cp \text{ and } \d_t \cp,
\]
for all $I, J, K $.
\end{definition}

\subsection{The dynamical variables and the deviation quantities}


Inspired by \cite{GIJ} and the structure of the equations of interest, (\ref{eq: transport frame})--(\ref{eq:LapseEquation}), we make the following
definition.

\begin{definition}\label{def: dynamical variables}
We define the \emph{dynamical variables} to be
the functions 
\[
    e_I, \o^I, \gamma_{I J K}, k_{IJ}, e_I \phi \text{ and } e_0 \phi.
\]
\end{definition}

\begin{remark}
Note that $\gamma_{IJK}$ in this paper denotes the structure coefficients, as opposed to the connection coefficients in \cite{GIJ}.
Here, the structure coefficients are dynamical variables in place of the connection coefficients in \cite{GIJ}.
\end{remark}

Given the scaffold, the proof of global existence will rely on careful bounds on how much the actual solution deviates from the scaffold.
Inspired by \cite{GIJ} and the structure of (\ref{eq: transport frame})--(\ref{eq:LapseEquation}), we choose the following quantities
measuring the deviation from the scaffold.

\begin{definition} \label{def: deviation quantities}
  Fix $\s_p$, $\s_V$, $\s$, $k_0$, $k_1$, $(\S, h_\refer)$, $(E_{i})_{i=1}^{n}$ and $V$ as in Theorem~\ref{thm: big bang formation}.
  Let $A := 2(n+1)(1+2\s)$ and
\[
    e_I^i, \o^I_i, k_{IJ}, \gamma_{IJK}, \phi, N:
        (a, b) \times \S \to \rn{}
\]
be a solution to (\ref{eq: transport frame})--(\ref{eq:LapseEquation}) with $N > 0$ as in
Proposition~\ref{prop:FieldEquationsFermiWalker}. Assume the solution to arise from diagonal FRS initial data at $t_{0}\in (a,b)$,
and define the scaffold variables as in Definition~\ref{def:scaffold variables}. Define
\begin{subequations}\label{eq:deltavariables}
  \begin{align}
    \dep{} := & \phi-\cp,\ \ \
                \dep{I} := e_{I}\phi-\che_{I}\cp,\ \ \
                \dep{0} := e_{0}\phi-\d_{t}\cp,\\
    \deo{}{} := & \omega-\chom,\ \ \
                  \deo{I}{} := \omega^{I}-\chom^{I},\ \ \
                  \deo{I}{i} := \omega^{I}_{i}-\chom^{I}_{i},\\
    \dee{}{} := & e-\che,\ \ \
                  \dee{I}{} := e_{I}-\che_{I},\ \ \
                  \dee{I}{i} := e_{I}^{i}-\che_{I}^{i},\\    
    \dek{} := & k-\chk,\ \ \
                \dek{IJ} := k_{IJ}-\chk_{IJ},\\
    \dega{} := & \g-\chga,\ \ \
                 \dega{IJK} := \g_{IJK}-\chga_{IJK}.
  \end{align}
\end{subequations}
Define, moreover,
\begin{equation}\label{eq:vectorvariables}
  \vdep:=(\dep{1},\dots,\dep{n}),\ \ \
  \ear N:=(e_{1}N,\dots,e_{n}N).
\end{equation}
Define the \emph{deviation quantities} (cf. \cite[Definition~3.1, p.~33]{GIJ}) by
\begin{subequations}\label{eq:mbLdef}
  \begin{align}
    \mbL_{(N)}(t) 
        &:= t^{-\s}\|N-1\|_{C^{k_{0}+1}(\S_t)} + t^{1-4\s}\|\ear N\|_{C^{k_{0}}(\S_t)},\label{eq:mbLNdef}\\
    \mbL_{(e,\omega)}(t) 
        & := t^{1-3\s}\|\dee{}{}\|_{C^{k_{0}+1}(\S_t)}
            +t^{1-3\s}\|\deo{}{}\|_{C^{k_{0}+1}(\S_t)},\label{eq:mbLeomdef}\\
    \mbL_{(\g, k)}(t) 
        & := t^{1-2\s}\|\dega{}\|_{C^{k_{0}}(\S_t)} + t\|\dek{}\|_{C^{k_{0}+1}(\S_t)},\label{eq:mbLgkdef}\\
    \mbL_{(\phi)}(t) 
        & := t^{1-2\s}\|\vdep\|_{C^{k_{0}}(\S_t)} + t\|\dep{0}\|_{C^{k_{0}+1}(\S_t)},
    \label{eq:mbLphikz}
  \end{align}
\end{subequations}
and
\begin{subequations}\label{eq:mbHdef}
  \begin{align}
    \mbH_{(N)}(t) 
        & := t^{A+1} \big( t^{-2} \|N-1\|_{H^{k_{1}}(\S_t)}^2
            + \|\ear N\|_{H^{k_{1}}(\S_t)}^2 \big)^{1/2}, \label{eq:mbHNdef}\\
    \mbH_{(e,\omega)}(t) 
        & := t^{A+1-2\s} \big(\|\dee{}{}\|_{H^{k_{1}}(\S_t)}^2
            +\|\deo{}{}\|_{H^{k_{1}}(\S_t)}^2\big)^{1/2},\label{eq:mbHeomdef}\\
    \mbH_{(\g, k)}(t) 
        & := t^{A+1} \big( \tfrac{1}{2} \|\dega{}\|_{H^{k_{1}}(\S_t)}^2
            + \|\dek{}\|_{H^{k_{1}}(\S_t)}^2 \big)^{1/2},\label{eq:mbHgkdef}\\
    \mbH_{(\phi)}(t) 
        & := t^{A+1} \big( \|\vdep\|_{H^{k_{1}}(\S_t)}^2
            + \|\dep{0}\|_{H^{k_{1}}(\S_t)}^2 \big)^{1/2}. \label{eq: H phi bootstrap}
  \end{align}
\end{subequations}
As in \cite{GIJ}, 
\begin{equation}\label{eq:mbDdefinition}
    \mbD(t)
        :=\mbL_{(e,\omega)}(t) + \mbL_{(\gamma,k)}(t) + \mbL_{(\phi)}(t) + \mbH_{(e,\omega)}(t) + \mbH_{(\gamma,k)}(t) + \mbH_{(\phi)}(t).
\end{equation}
\end{definition}
\begin{remark}
  In (\ref{eq:mbLdef}) and (\ref{eq:mbHdef}) we use conventions similar to (\ref{seq:norms of indexed quantities}).
\end{remark}
\begin{remark}
Following \cite{GIJ}, we work with energy estimates weighted by a factor $t^A$, which motivates the extra weight multiplying the higher Sobolev norms.
However, we fix the weight throughout to be $A := 2(n+1)(1+2\s)$.
This is the reason that we may choose the higher regularity $k_1$ explicitly, with a lower bound only depending on $n$, $k_0$ and $\s$.
\end{remark}

\subsection{The global existence statement}

Theorem~\ref{thm: big bang formation} is proven in Section~\ref{sec: finishing the proof} and is a consequence of the following theorem:

\begin{thm}[Global existence, FRS equations] \label{thm:GlobalExistenceFermiWalker}
  Let $\s_p$, $\s_V$, $\s$, $k_0$, $k_1$, $(\S, h_\refer)$, $(E_{i})_{i=1}^{n}$ and $V$ be as in Theorem~\ref{thm: big bang formation}.
  For every $\rho_0 > 0$, there are standard constants $\tau_1<1$ and $\mathcal{C}$, such
  that the following holds: If $t_0 < \tau_1$; if there are smooth functions
  \begin{equation}\label{eq:beiIetcindata}
    \be_I^i, \bq_I, \bp_0, \bp_1 : \S \to \rn{}
  \end{equation}
  that form diagonal FRS initial data, satisfying the FRS expansion-normalized bounds of regularity $k_1$ for
  $\rho_0$ at $t_0$ (see~Definition~\ref{def:DecomposedData}) as well as (\ref{eq:hamconstraint}) and (\ref{eq:MC}); and if there are smooth
  initial data
  \begin{equation}\label{eq: hat initial data}
    \hat{e}_{I}^{i},\hat{k}_{IJ},\hat{\phi}_0,\hat{\phi}_1:\S\to\rn{}
  \end{equation}
  to (\ref{eq: transport frame})--(\ref{eq:LapseEquation}) satisfying $\hat{k}_{II}=1/t_0$ and $\mbD(t_0)\leq t_0^\s$, then there is a unique smooth solution 
  \[
    (e_I^i, \o_i^I, k_{IJ}, \gamma_{IJK}, \phi, N) : (0, t_+) \times \S \to \rn{},
  \]
  to (\ref{eq: transport frame})--(\ref{eq:LapseEquation}), with $t_0\in (0,t_{+})$, satisfying the initial condition
  \begin{equation}\label{eq:hatted initial data}
    \left( e_I^i, \o^I_i, k_{IJ}, \g_{IJK}, \phi, e_0 \phi \right)|_{t = t_0}
    = \big( \hat{e}_I^i, \hat{\o}^I_i, \hat{k}_{IJ}, \hat{\gamma}_{IJK}, \hat{\phi}_0, \hat{\phi}_1 \big).
  \end{equation}
  Moreover, this solution satisfies the following bound for all $t \in (0,t_0]$:
  \begin{equation}\label{eq:AposterioriBoundsGlobalExistence}
    \mbD(t) + \mbL_{(N)} + \mbH_{(N)} \leq \mathcal{C} t_0^{\s}.
  \end{equation}
\end{thm}
\begin{remark}
  When we say that the functions (\ref{eq:beiIetcindata}) solve (\ref{eq:hamconstraint}) and (\ref{eq:MC}), it is understood that the $\be_I^i$
  define a frame via $\be_I:=\be_I^iE_i$; that $(\bo^I)_{I=1}^{n}$ is the dual frame of $(\be_I)_{I=1}^{n}$; that $\bga_{IJK}:=\bo^{K}([\be_I,\be_J])$;
  and that $\bk_{IJ}:=\tfrac{\bq_\uI}{t_{0}}\de_{\uI J}$. Moreover, $e_I$, $\gamma_{IJK}$, $k_{IJ}$, $e_0\phi$, $\phi$ and $t$ appearing in
  (\ref{eq:hamconstraint}) and (\ref{eq:MC}) should be replaced by $\be_I$, $\bga_{IJK}$, $\bk_{IJ}$, $\bp_1$, $\bp_0$ and $t_0$ respectively.
  The statement that the functions in (\ref{eq: hat initial data}) satisfy (\ref{eq:hamconstraint}) and (\ref{eq:MC}) should be interpreted similarly.
  Moreover, in (\ref{eq:hatted initial data}), $\hat{e}_I:=\hat{e}_I^iE_i$; $(\hat{\o}^{I})_{I=1}^n$ is the frame dual to $(\hat{e}_I)_{I=1}^{n}$;
  and $\hat{\gamma}_{IJK}:=\hat{\o}^{K}([\hat{e}_I,\hat{e}_J])$.
\end{remark}
\begin{remark}
  When we say that (\ref{eq: hat initial data}) satisfy $\mbD(t_0)\leq t_0^\s$, we take it for granted that $t_{0}$ and the functions appearing in
  (\ref{eq:beiIetcindata}) are used to define the scaffold and that the following replacements are made in (\ref{eq:deltavariables}):
  $\phi$ and $e_0\phi$ are replaced by $\hat{\phi}_0$ and $\hat{\phi}_1$ respectively; $e$ and $\o$ are replaced by $\hat{e}$ and
  $\hat{\o}$ respectively; and $k$ and $\gamma$ are replaced by $\hat{k}$ and $\hat{\gamma}$ respectively.  
\end{remark}
\begin{remark}
  In the proof of the main theorem, (\ref{eq:beiIetcindata}) and (\ref{eq: hat initial data}) coincide
  (with $\bk_{IJ}:=\tfrac{\bq_\uI}{t_{0}}\de_{\uI J}$). However, in the proof of Theorem~\ref{thm:degenerate case}, they are different. 
\end{remark}
The proof is to be found at the end of the present section.

\subsection{The bootstrap improvement statement}

The strategy to prove Theorem~\ref{thm:GlobalExistenceFermiWalker} is via a bootstrap argument.

\begin{definition}[The bootstrap inequality]\label{def:BootstrapInequality}
  Let $\s_p$, $\s_V$, $\s$, $k_0$, $k_1$, $(\S, h_\refer)$, $(E_{i})_{i=1}^{n}$ and $V$ be as in Theorem~\ref{thm: big bang formation}.
  Let, moreover, smooth diagonal FRS initial data as in (\ref{eq:beiIetcindata}) be given, satisfying the FRS
  expansion-normalized bounds of regularity $k_1$ for $\rho_0$ at $t_0$, as well as (\ref{eq:hamconstraint}) and (\ref{eq:MC}).
  Finally, let $\ar \in \left(0, \tfrac{1}{6n}\right]$ and $0 < t_{\rob} < t_0 $ be given. Then a solution to
  (\ref{eq: transport frame})--(\ref{eq:LapseEquation}), consisting of smooth functions 
  \begin{equation}
    (N, e_I^i, \o^I_i, \g_{IJK}, k_{IJ}, \phi) 
    : \S \times [t_{\rob}, t_0] \to \rn{},
  \end{equation}
  is said to satisfy the \emph{bootstrap inequality for $\ar$ on $[t_{\rob}, t_0]$}, if the following inequality holds for all
  $t \in [t_{\rob}, t_0]$:
  \begin{equation} \label{eq:BootstrapInequality}
    \mbD(t)+\mbL_{(N)}(t)+\mbH_{(N)}(t) \leq \ar.
  \end{equation}
\end{definition}
\begin{remark}\label{remark:origin of the scaffold}
  In (\ref{eq:BootstrapInequality}), it is understood that the diagonal FRS data as in (\ref{eq:beiIetcindata}) together
  with $t_{0}$ are used to construct the scaffold; see Definition~\ref{def:scaffold variables}. 
\end{remark}

The key step in the proof of past global existence is the following statement, which tells us
that if the bootstrap inequality holds and $\ar$ and $t_0$ are small enough, then a strictly better inequality holds on the same interval.

\begin{thm}[Bootstrap improvement]\label{thm:BootstrapImprovement}
  Let $\s_p$, $\s_V$, $\s$, $k_0$, $k_1$, $(\S, h_{\refer})$, $(E_{i})_{i=1}^{n}$ and $V$ be as in Theorem~\ref{thm: big bang formation}
  and let $\rho_0 >0$. Then there are standard constants $\mathcal{C}$, $\ar_\rob$ and $\tau_\rob<1$,
  such that the following holds: if smooth diagonal FRS initial data as in (\ref{eq:beiIetcindata}) are given, satisfying the FRS
  expansion-normalized bounds of regularity $k_1$ for $\rho_0$ at $t_0$, as well as (\ref{eq:hamconstraint}) and (\ref{eq:MC});
  if $[t_{\rob},t_0] \subset \left(0,\tau_\rob\right]$; if there is a smooth solution to
  (\ref{eq: transport frame})--(\ref{eq:LapseEquation}), consisting of
  \[
  N, e_I^i, \o^I_i, \g_{IJK}, k_{IJ}, \phi : \S \times [t_{\rob}, t_0] \to \rn{},
  \]
  satisfying the bootstrap inequality for $\ar_\rob$ on $[t_{\rob},t_0]$ (see~Definition~\ref{def:BootstrapInequality}); and if
  $\mbD(t_{0})\leq t_{0}^{\s}$, then, for all $t \in [t_{\rob}, t_0]$,
  \begin{equation} \label{eq: global existenece estimate}
    \mbD(t) + \mbL_{(N)}(t) + \mbH_{(N)}(t) \leq \mathcal{C} t_0^{\s}.
  \end{equation}
\end{thm}
\begin{remark}
  An observation similar to Remark~\ref{remark:origin of the scaffold} is equally relevant here.
\end{remark}

The proof of Theorem~\ref{thm:BootstrapImprovement} is presented at the end of Section~\ref{sec:EnergyEstimates}, see in particular
Subsection~\ref{sec:proofs}. In the rest of this section and in Sections~\ref{sec:MainEstimates} and \ref{sec:EnergyEstimates}, it will
be of use to have the standing assumption that we have a solution satisfying a bootstrap inequality.

\begin{assump}[The bootstrap assumption]\label{ass:Bootstrap}
  Let $\s_p$, $\s_V$, $\s$, $k_0$, $k_1$, $(\S, h_{\refer})$, $(E_{i})_{i=1}^{n}$ and $V$ be as in Theorem~\ref{thm: big bang formation}
  and let $\rho_{0}>0$. Let, moreover, smooth diagonal FRS initial data as in (\ref{eq:beiIetcindata}) be given, satisfying
  the FRS expansion-normalized bounds of regularity $k_1$ for $\rho_0$ at $t_0$, as well as (\ref{eq:hamconstraint}) and (\ref{eq:MC}).
  Assume that there are $0 < t_\rob < t_0$, $r \in \left(0, \tfrac{1}{6n}\right]$ and smooth functions
  \[
  N, e_I^i, \o^I_i, \g_{IJK}, k_{IJ}, \phi : \S \times [t_{\rob}, t_0] \to \rn{},
  \]
  which solve (\ref{eq: transport frame})--(\ref{eq:LapseEquation}) on $[t_{\rob},t_0]$; satisfy the bootstrap inequality for $r$ on $[t_{\rob},t_0]$
  (see~Definition~\ref{def:BootstrapInequality}); and are such that $\mbD(t_0)\leq t_0^{\s}$. 
\end{assump}

As a preparation to prove the bootstrap improvement, Theorem~\ref{thm:BootstrapImprovement},
we need a priori estimates on the scaffold dynamical variables.

\subsection{Estimates on the scaffold}
The scaffold dynamical variables satisfy the following estimates:

\begin{lemma}[A priori estimates on the scaffold dynamical variables] \label{lemma:estimatingthebackground}
  Let $\s_p$, $\s_V$, $\s$, $k_0$, $k_1$, $(\S, h_\refer)$, $(E_{i})_{i=1}^{n}$ and $V$ be as in Theorem~\ref{thm: big bang formation}
  and let $\rho_{0}>0$. Then there is a standard constant $C$ such that following holds: If $t_0 \leq 1$ and
  if there are $\be_I^i$, $\bq_I$, $\bp_0$, $\bp_1\in C^{\infty}(\S,\rn{})$ that form diagonal FRS initial data satisfying the FRS
  expansion-normalized bounds of regularity $k_1$ for $\rho_0$ at $t_0$, then the corresponding scaffold
  dynamical variables satisfy the following estimates for all $t\in (0,t_{0}]$:
  \begin{subequations}\label{eq:chchomchgachphiestkz}
    \begin{align}
      t^{1-4\s}\|\che\|_{C^{k_{0}+2}(\S_t)}
      + t^{1-4\s}\|\chom\|_{C^{k_{0}+2}(\S_t)}
      &\leq  C \label{eq:chechomtqkzest}\\
      t^{1-4\s}\|\chga\|_{C^{k_0+2}(\S_t)}
      +
      t \|\ck\|_{C^{k_0+2}(\S_t)}
      &\leq  C,\label{eq:chgaestkz}\\
      t^{1-3\s}\|\cear \chphi\|_{C^{k_0+2}(\S_t)}+t\|\d_{t}\chphi\|_{C^{k_{0}+2}(\S_t)}
      & \leq  C,\label{eq:chphiestkz}
    \end{align}
  \end{subequations}
  and
\begin{subequations}\label{eq: higher estimates scaffold}
    \begin{align}
      t^{1-4\s}\|\che\|_{H^{k_{1}+2}(\S_t)}
      +t^{1-4\s}\|\chom\|_{H^{k_{1}+2}(\S_t)} 
      & \leq  C,\label{eq: check e omega high}\\
      t^{1-4\s}\|\chga\|_{H^{k_1+1}(\S_t)} 
      +
      t \norm{\ck}_{H^{k_1 + 2}(\S_t)}
      &\leq  C,\label{eq: check gamma and k high}\\
      t^{1-3\s}\|\cear \chphi\|_{H^{k_{1}+1}(\S_t)} + t \|\d_{t}\chphi\|_{H^{k_{1}+2}(\S_t)} 
      &\leq  C.\label{eq: check scalar fields high}
    \end{align}
  \end{subequations}  
\end{lemma}

\begin{proof}
The estimates \eqref{eq:chchomchgachphiestkz} follow by \eqref{eq: higher estimates scaffold} and Sobolev embedding. Since
\[
	\che^{i}_{I}
		= ( \tfrac{t_0}{t} )^{\bq_I} \bar e_{I}^{i}, \quad
	\chom^{I}_{i} 
		= ( \tfrac{t}{t_0} )^{\bq_I} \bar \omega_i^{I},
\]
Lemma~\ref{lemma: products in L2}, Sobolev embedding, \eqref{eq: e bar o bar assumption}, \eqref{eq: gamma bar p bar assumption} and
Lemma~\ref{le: first bound on $q_I$} imply that
\begin{align*}
    \snorm{\ce_I^i}_{H^{k_1 +2}(\S_t)}
        &\leq C \snorm{t^{-\bq_I}}_{H^{k_1 + 2}(\S_t)} \snorm{t_0^{\bq_I} \be_I^i}_{H^{k_1 + 2}(\S_t)} \\
        &\leq C t^{-\snorm{\bq_I}_{C^{0}(\S_t)}} \ldr{\ln(t)}^{k_1 + 2} \snorm{\bq_I}_{H^{k_1 + 2}(\S_t)} \\
        &\leq C t^{-1 + 5\s} \ldr{\ln(t)}^{k_1 + 2}\leq C t^{-1 + 4\s},
\end{align*}
since $t^{\s}\ldr{\ln t}^{k_1+2}$ is bounded for $t \in [0,1]$ by a constant depending only on $\s$ and $k_1$.
This, together with the analogous estimate for $\snorm{\chom^I_i}_{H^{k_1 +2}(\S_t)}$, proven the same way, yields \eqref{eq: check e omega high}.
Since $\ck_{IJ} = \tfrac{\bq_\uI}t \de_{\uI J}$, the second part of \eqref{eq: check gamma and k high} is immediate from
\eqref{eq: gamma bar p bar assumption}. In order to prove the first part of \eqref{eq: check gamma and k high}, compute
\begin{equation}\label{eq:formula for chgaIJK}
  \begin{split}
    \cga_{IJK}
    = & \chom^K \left( \left[ t^{-\bq_I} t_0^{\bq_I} \be_I, t^{-\bq_J} t_0^{\bq_J} \be_J \right] \right) \\
    = & t^{\bq_K}t_0^{-\bq_K}\bo^K \left( t^{-\bq_I} \left( t_0^{\bq_I} \be_It^{-\bq_J} \right) t_0^{\bq_J} \be_J
    - t^{-\bq_J} \left( t_0^{\bq_J} \be_Jt^{-\bq_I} \right) t_0^{\bq_I} \be_I\right.\\
    & \left. \phantom{t^{\bq_K}t_0^{-\bq_K}\bo^Ki} + t^{-\bq_I - \bq_J} [t_0^{\bq_I}\be_I, t_0^{\bq_J}\be_J]\right) \\
    = & t^{- \bq_I - \bq_J + \bq_K} \left( - t_0^{\bq_I}\be_I(\bq_J) \ln(t) \de_{JK} + t_0^{\bq_J}\be_J(\bq_I) \ln(t) \de_{IK}\right.\\  
    &  \left.\phantom{t^{- \bq_I - \bq_J + \bq_K}i}+ t_0^{-\bq_K} \bo^K([t_0^{\bq_I}\be_I, t_0^{\bq_J}\be_J]) \right).
  \end{split}
\end{equation}
Lemma~\ref{lemma: products in L2}, Sobolev embedding, \eqref{eq: e bar o bar assumption}, \eqref{eq: gamma bar p bar assumption}
and Lemma~\ref{le: first bound on $q_I$} then yield
\begin{align*}
    \snorm{\cga_{IJK}}_{H^{k_1 + 1}(\S_t)} &\leq C t^{-1 + 5 \s} \ldr{\ln(t)}^{k_1+2}\leq C t^{-1 + 4\s}.
\end{align*}
Recalling (\ref{eq: scaffold scalar field}),  $t \d_t \cp = t_0 \bp_1$, so that the second part of \eqref{eq: check scalar fields high} is
immediate from \eqref{eq: scalar field bar assumption}. By Lemma~\ref{lemma: products in L2}, Sobolev embedding, \eqref{eq: scalar field bar assumption},
(\ref{eq: scaffold scalar field}) and \eqref{eq: check e omega high},
\begin{align*}
    \norm{\ce_I \cp}_{H^{k_1 + 1}(\S_t)}
        &\leq \norm{\ce_I}_{H^{k_1 + 1}(\S_t)} \norm{\cp}_{H^{k_1 + 2}(\S_t)} \\
        &\leq C t^{-1 + 4 \s} \left( \abs{\ln(t)}t_0 \norm{\bp_1}_{H^{k_1 + 2}(\S_t)} + \norm{\bp_0 - t_0 \ln(t_0) \bp_1}_{H^{k_1 + 2}(\S_t)} \right) \\
        &\leq C t^{-1 + 4 \s} \ldr{\ln(t)}\leq C t^{-1 + 3 \s},
\end{align*}
proving the first part of \eqref{eq: check scalar fields high}, which finishes the proof.
\end{proof}


%

\subsection{Estimating the dynamical variables in terms of the deviation quantities}

In the proof of our main estimates, it is useful to have direct control of the dynamical variables in terms of the deviation quantities
introduced in Definition~\ref{def: deviation quantities}.

\begin{lemma}
  [A priori estimates for the dynamical variables] \label{lemma:estimatingthedynamicalvariables}
  Let $\s_p$, $\s_V$, $\s$, $k_0$, $k_1$, $(\S, h_\refer)$, $(E_{i})_{i=1}^{n}$ and $V$ be as in Theorem~\ref{thm: big bang formation}
  and let $\rho_{0}>0$. Then there are standard constants $C$ and $\tau_1<1$ such that if
  Assumption~\ref{ass:Bootstrap} is satisfied for some $\ar \in (0, \frac{1}{6n}]$ and $[t_{\rob}, t_0] \subseteq (0, \tau_{1}]$, then, using the
  notation introduced in \eqref{eq:mbLdef}, the following holds on $[t_{\rob},t_{0}]$:
\begin{subequations}\label{eq:estdynvarmfltot}
\begin{align}
      t^{1-3\s}\|e\|_{C^{k_{0}+1}(\S_t)}
      + t^{1-3\s}\|\omega\|_{C^{k_{0}+1}(\S_t)} 
      &\leq \mbL_{(e,\omega)}(t) + Ct^{\s},\label{eq:tqeomegaWkzest}\\
      t^{1-2\s}\|\g\|_{C^{k_{0}}(\S_t)}  
      &\leq \mbL_{(\g, k)}(t) + Ct^{2\s},\label{eq:tqgaWinfkz}\\
	  t \|k\|_{C^{k_{0}+1}(\S_t)} 
	  	&\leq \mbL_{(\g, k)}(t) + C, \label{eq: low k bound}\\
      t^{1-2\s}\|\ear\phi\|_{C^{k_{0}}(\S_t)} 
      &\leq \mbL_{(\phi)}(t) + Ct^{\s},
      \label{eq:tqpsigWinfphiest}\\
      t \|e_0 \phi\|_{C^{k_{0}+1}(\S_t)} 
      &\leq \mbL_{(\phi)}(t) + C,
      \label{eq: low norm e 0 phi}\\
      t^{2}\|V\circ\phi\|_{C^{k_{0}+1}(\S_t)}
            +t^{2}\|V'\circ\phi\|_{C^{k_{0}+1}(\S_t)}
      &\leq C t^{5\s},
      \label{eq:VcircphiWinfkzest}\\
      t^{1-3\s}\|\rodiv_{h_{\refer}}e_{I}\|_{C^{0}(\S_t)}
      &\leq C \left( \mbL_{(e,\omega)}(t)+t^{\s} \right). \label{eq:diveIest}
\end{align}
\end{subequations}
Given the notation introduced in \eqref{eq:mbHdef}, the following holds on $[t_{\rob},t_{0}]$:
\begin{subequations}\label{eq:estdynvarmfhtot}
\begin{align}
      t^{A+1-2\s}\|e\|_{H^{k_{1}}(\S_t)}
        +t^{A+1-2\s}\|\omega\|_{H^{k_{1}}(\S_t)} 
      &\leq \sqrt{2} \mbH_{(e,\omega)}(t)
        +C t^{A+2\s},
      \label{eq:tqeomegaHkoest}\\
      t^{A+1}\|\g\|_{H^{k_{1}}(\S_t)} 
      &\leq \sqrt 2 \mbH_{(\g, k)}(t)
        +Ct^{A+4\s},\label{eq:AprioriEstimateSC}\\
    t^{A + 1}
        \norm{k}_{H^{k_1}(\S_t)}
        &\leq \mbH_{(\g, k)}(t) + C t^A, \label{eq: k high estimate} \\
    t^{A+1}\|\ear\phi\|_{H^{k_{1}}(\S_t)} 
      &\leq \mbH_{(\phi)}(t) + Ct^{A+3\s}, \label{eq:tqpsigHkophiest} \\
    t^{A + 1}\norm{e_0\phi}_{H^{k_1}(\S_t)}
        &\leq \mbH_{(\phi)}(t) + Ct^A, \label{eq: high norm e 0 phi} \\
    t^{A+2}\|V\circ\phi\|_{H^{k_{1}}(\S_t)}+t^{A+2}\|V'\circ\phi\|_{H^{k_{1}}(\S_t)} 
      &\leq C t^{5\s}. \label{eq:VcircphiHkoest}
    \end{align}
  \end{subequations}
\end{lemma}
\begin{proof}
  The estimates \eqref{eq:tqeomegaWkzest}--\eqref{eq: low norm e 0 phi} are immediate consequences of the estimates on the scaffold quantities,
  \eqref{eq:chchomchgachphiestkz}, and the definition of the deviation quantities, \eqref{eq:mbLdef}. Similarly, the estimates
  \eqref{eq:tqeomegaHkoest}--\eqref{eq: high norm e 0 phi} are immediate consequences of the estimates on the scaffold quantities,
  \eqref{eq: higher estimates scaffold}, and the definition of the deviation quantities, \eqref{eq:mbHdef}.
  It remains to prove \eqref{eq:VcircphiWinfkzest}, \eqref{eq:diveIest} and \eqref{eq:VcircphiHkoest}.
  To derive \eqref{eq:VcircphiWinfkzest}, we first need to control $\phi$. 
  Note, to this end, that for any $t \in [t_{\rob}, t_0]$ and $m \leq k_0 + 1$,
  \begin{equation}\label{eq: scalar field inequality}
    \begin{split}
      \|\phi-\cp\|_{C^{m}(\S_t)}
      = & \big\|\textstyle{\int}_{t}^{t_{0}}\d_{t}(\phi-\chphi)(\cdot, s) \md s\big\|_{C^{m}(\S_t)}  \\
      \leq & \textstyle{\int}_{t}^{t_{0}}s^{-1} \norm{s N\dep{0} + s ( N - 1 ) \d_t \chphi}_{C^{m}(\S_s)} \md s.
    \end{split} 
  \end{equation}
  Combining this estimate for $m = 0$ with \eqref{eq:BootstrapInequality} and Lemma~\ref{lemma:EstimateHamiltonianConstraint} yields
  \begin{equation*}
    \begin{split}
      &\snorm{\phi-\cp}_{C^{0}(\S_t)} \\
      \leq & \sup_{s \in [t, t_0]} \big( ( \snorm{N-1}_{C^{0}(\S_s)} + 1 ) \mbL_{(\phi)}(s)
        + \norm{N - 1}_{C^{0}(\S_s)} t_0 \norm{\bp_1}_{C^{0}(\S_s)} \big) \textstyle{\int}_t^{t_{0}}s^{-1} \md s \\
      \leq & \ar \left( \ar+2 \right) \left( \ln t_{0} - \ln t \right)     
      \leq  - \tfrac1{6n} \left( \tfrac1{6n} + 2 \right) \ln(t) \leq - \tfrac1{2n} \ln(t),
    \end{split}    
  \end{equation*}  
where we used that $\ln(t_0) < 0$.
On the other hand, by (\ref{eq: scaffold scalar field}), Lemma~\ref{lemma:EstimateHamiltonianConstraint} and \eqref{eq: scalar field bar assumption},
\begin{align*}
  \snorm{\cp}_{C^{0}(\S_t)}
  &\leq t_0\snorm{\bp_1}_{C^{0}(\S_t)}(-\ln(t)) + \snorm{\bp_0 - t_0 \bp_1 \ln(t_0)}_{C^{0}(\S_t)} \\
  &\leq - \left( 1 - \tfrac1{2n} \right) \ln(t) + C.
\end{align*}
Combining the last two estimates yields
\begin{equation}\label{eq: L infty phi} 
    \|\phi\|_{C^{0}(\S_t)}
	\leq - \ln t + C
\end{equation}
for all $ t \in [t_{\rob},t_{0}]$.
In particular, this means that
\begin{equation}\label{eq:tsqVkcircphiest}
	t^{2}\textstyle{\sum}_{k\leq k_{0}+2}\|V^{(k)}\circ\phi\|_{C^{0}(\S_t)}
		\leq C t^{2\s_{V}}
            \leq Ct^{6\s}
\end{equation}
on $[t_{\rob},t_{0}]$, where we appealed to \eqref{eq: V assumption} and \eqref{eq: sigma condition}.
Next, \eqref{eq: scalar field inequality} with $m = k_0 + 1$ yields
\begin{align*}
  \snorm{\phi-\cp}_{C^{k_{0}+1}(\S_t)} 
  \leq & C\sup_{s \in [t, t_0]} ( ( \snorm{N-1}_{C^{k_{0}+1}(\S_s)} + 1 ) \mbL_{(\phi)}(s)\\
  & \phantom{C\sup_{s \in [t, t_0]}i} +\snorm{N - 1}_{C^{k_{0}+1}(\S_s)} t_{0}\snorm{\bp_1}_{C^{k_{0}+1} (\S_s)}) \textstyle{\int}_t^{t_{0}}s^{-1} \md s \\
  \leq & C \left( \ln t_0 - \ln t \right)\leq C \ldr{\ln(t)},
\end{align*}
since $\ln(t_0) < 0$.
On the other hand, Sobolev embedding and \eqref{eq: scalar field bar assumption} imply that
\begin{equation} \label{eq: Sobolev bounds on check phi}
\begin{split}
    \snorm{\cp}_{C^{k_{0}+1} (\S_t)}
        &\leq C\snorm{\cp}_{H^{k_1 + 2} (\S_t)} \\
        &\leq - C\ln(t) \snorm{t_0 \bp_0}_{H^{k_1 + 2} (\S_t)} + C\snorm{\bp_0 - t_0 \bp_0 \ln(t_0)}_{H^{k_1 + 2} (\S_t)} 
        \leq C \ldr{\ln(t)}.
\end{split}
\end{equation}
To conclude,
\begin{equation}\label{eq:ldrlntphiWkzest}
	\|\phi\|_{C^{k_{0}+1}(\S_t)}
		\leq C\ldr{\ln t}.
\end{equation}
Combining \eqref{eq:tsqVkcircphiest} and \eqref{eq:ldrlntphiWkzest} yields \eqref{eq:VcircphiWinfkzest}. 
  
In order to estimate the divergence of $e_I$ with respect to $h_{\refer}$, note that
\[
    \rodiv_{h_{\refer}}e_I
        = \rodiv_{h_{\refer}}(e_I^i E_i)
        = E_i(e_I^i) + e_I^i \rodiv_{h_\refer}(E_i).
\]
Combining this with \eqref{eq:tqeomegaWkzest} yields \eqref{eq:diveIest}. 

In order to prove (\ref{eq:VcircphiHkoest}), we first need to estimate $\phi$ in $H^{k_{1}}$. However,
\begin{equation*}
  \begin{split}
     \|\phi - \cp \|_{H^{k_{1}}(\S_t)}
    = & \big\| \textstyle{\int}_{t}^{t_{0}} \d_{t}( \phi - \cp )(\cdot, s) \md s \big\|_{H^{k_1}(\S_t)} \\
    \leq & \textstyle{\int}_{t}^{t_{0}}s^{-1} \snorm{s N\dep{0} + s (N-1) \d_t \chphi}_{H^{k_1}(\S_s)} \md s
  \end{split}
\end{equation*}
On the other hand, Lemma~\ref{lemma: products in L2} and \eqref{eq:BootstrapInequality} yield
\begin{equation*}
  \begin{split}
    \snorm{s N\dep{0}}_{H^{k_1}(\S_s)} \leq & C\|N\|_{C^{0}(\S_s)}s\|\dep{0}\|_{H^{k_1}(\S_s)}+Cs\|\dep{0}\|_{C^{0}(\S_s)}\|N\|_{H^{k_1}(\S_s)}
    \leq Cs^{-A}.
  \end{split}
\end{equation*}
Similarly, appealing to \eqref{eq:BootstrapInequality}, \eqref{eq:chphiestkz} \eqref{eq: check scalar fields high} and
Lemma~\ref{lemma: products in L2} yields
\begin{equation*}
  \begin{split}
    \snorm{s (N-1)\d_t \chphi}_{H^{k_1}(\S_s)} \leq & C\|N-1\|_{C^{0}(\S_s)}s\|\d_t \chphi\|_{H^{k_1}(\S_s)}\\
    & +Cs\|\d_t \chphi\|_{C^{0}(\S_s)}\|N-1\|_{H^{k_1}(\S_s)} \leq Cs^{-A}.
  \end{split}
\end{equation*}
Combining the last three estimates yields
\[
  \|\phi - \cp \|_{H^{k_{1}}(\S_t)}\leq Ct^{-A}.
\]
On the other hand, by (\ref{eq: scalar field bar assumption}) and (\ref{eq: scaffold scalar field}),
$\snorm{\cp}_{H^{k_1}(\S_t)} \leq C \ldr{\ln(t)}$. Thus
\begin{equation}\label{eq: only phi high estimates}
	\|\phi\|_{H^{k_1}(\S_t)}\leq C t^{-A}. 
\end{equation}
Combining \eqref{eq: L infty phi}, \eqref{eq:tsqVkcircphiest} and \eqref{eq: only phi high estimates} yields \eqref{eq:VcircphiHkoest}.
\end{proof}

\subsection{Trading decay for derivatives} \label{subsec: trading decay derivatives}
It is sometimes of interest to trade decay for control of a larger number of derivatives.
Similarly to \cite[Lemma~4.2, p.~34]{GIJ}, we need the following estimate:

\begin{lemma} \label{lemma:extraderivatives}
  Let $\s_p$, $\s_V$, $\s$, $k_0$, $k_1$, $(\S, h_\refer)$, $(E_{i})_{i=1}^{n}$ and $V$ be as in Theorem~\ref{thm: big bang formation}
  and let $\rho_{0}>0$. If Assumption~\ref{ass:Bootstrap} is satisfied for some $\ar \in (0, \frac{1}{6n}]$ and
  $[t_{\rob}, t_0] \subseteq (0, 1]$, then, using the notation introduced in Definition~\ref{def: deviation quantities}, there is a standard constant
  $C$ such that 
  \begin{subequations} \label{eq: extra derivatives psi chpsi}
    \begin{align} 
    t^{1-3\s} \|\dee{}{}\|_{C^{k_{0}+2}(\S_t)} 
    + t^{1-3\s} \|\deo{}{}\|_{C^{k_{0}+2}(\S_t)} 
        &\leq Ct^{-\s}\mbD(t), \label{eq:extraderivativeseche}\\
    t^{1-2\s} \|\dega{}\|_{C^{k_{0}+2}(\S_t)} + t\|\dek{}\|_{C^{k_{0}+2}(\S_t)} 
        &\leq  Ct^{-\s}\mbD(t), \label{eq:extraderivativesgammachgamma}\\
    t^{1-2\s} \|\vdep\|_{C^{k_{0}+2}(\S_t)}
    + t\|\dep{0}\|_{C^{k_{0}+2}(\S_t)}
        &\leq  Ct^{-\s}\mbD(t), \label{eq:extraderivativesphichphi} \\
      \begin{split}
    t^{-\s} \|N-1\|_{C^{k_{0}+3}(\S_t)} & \\
        + t^{1-4\s}\|\ear N\|_{C^{k_{0}+2}(\S_t)}
        &\leq  Ct^{-\s} \big( \mbL_{(N)}+\mbH_{(N)} \big)(t),
      \end{split}\label{eq:extraderivativesNchN}
    \end{align}
  \end{subequations}
for all $t\in [t_{\rob}, t_0]$. Moreover, for all $t\in [t_{\rob}, t_0]$,
  \begin{subequations} \label{eq: extra derivatives dynamical}
    \begin{align}
    t^{1-3\s} \norm{e}_{C^{k_{0}+2}(\S_t)} 
    + t^{1-3\s}\norm{\o}_{C^{k_{0}+2}(\S_t)}
        \leq &~ C\left( t^{-\s}\mbD(t) + t^{\s} \right), \label{eq:extraderivativese} \\
    t^{1-2\s} \norm{\gamma}_{C^{k_{0}+2}(\S_t)}
        \leq &~ C\left( t^{-\s}\mbD(t) + t^{2\s} \right), \label{eq:extraderivativesgamma} \\
    t \norm{k}_{C^{k_{0}+2}(\S_t)}
        \leq &~ C\left( t^{-\s}\mbD(t) + 1 \right), \label{eq:extraderivativesk} \\
    t^{1-2\s} \norm{\ep}_{C^{k_0+2}(\S_t)}
        \leq &~ C\left( t^{-\s}\mbD(t) + t^{\s} \right), \label{eq:extraderivativeseIphi} \\
    t \norm{e_0(\phi)}_{C^{k_0 + 2}(\S_t)}
        \leq &~ C\left( t^{-\s}\mbD(t) + 1 \right). \label{eq:extraderivativesezerophi}
    \end{align}
  \end{subequations}
  \end{lemma}
\begin{proof}
  Let $\kappa_{0}$ denote the smallest integer strictly larger than $n/2$.
  For a smooth function $\Psi$;  positive numbers $\b, B$; and a positive integer $s$, Sobolev embedding and interpolation, see
  Lemma~\ref{le: interpolation}, yields
  \begin{equation} \label{eq: trading der main}
    \begin{split}
      t^\b \snorm{\Psi}_{C^{s}(\S_t)}
      \leq & C t^\b \snorm{\Psi}_{H^{s + \kappa_0}(\S_t)} \\
      \leq & C t^\b \snorm{\Psi}_{L^2(\S_t)}^{1-\frac{s + \kappa_0}{k_1}} \snorm{\Psi}_{H^{k_1}(\S_t)}^{\frac{s + \kappa_0}{k_1}} \\
      \leq & C t^{-B\frac{s + \kappa_0}{k_1}} \left( t^\b \snorm{\Psi}_{L^2(\S_t)}\right)^{1-\frac{s + \kappa_0}{k_1}}
      \left( t^{\b+B} \snorm{\Psi}_{H^{k_1}(\S_t)}\right)^{\frac{s + \kappa_0}{k_1}}.
    \end{split}
  \end{equation}
  Applying this estimate with $\Psi = \dee{I}{i}$, $\b = 1 - 3 \s$, $s = k_0 + 2$, $B = A + \s$, and recalling \eqref{eq:mbLeomdef} and
  \eqref{eq:mbHeomdef} gives the first part of \eqref{eq:extraderivativeseche}, since 
  \[
  k_1
  \geq \tfrac{(A + \s)(k_0 + 2 + \kappa_0)}{\s},
  \]
  which implies that $t^{\frac{-B(s + \kappa_0)}{k_1}} \leq t^{-\s}$. Applying \eqref{eq: trading der main} with $\Psi = \deo{I}{i}$, $\b = 1 - 3 \s$,
  $s = k_0 + 2$, $B = A + \s$; $\Psi = \dega{IJK}$, $\b = 1 - 2 \s$, $s = k_0 + 2$, $B = A + 2\s$; $\Psi = \dek{IJ}$,  $\b = 1$, $s = k_0 + 2$,
  $B = A$; $\Psi = \dep{I}$, $\b = 1 - 2 \s$, $s = k_0 + 2$, $B = A + 2\s$; $\Psi = \dep{0}$, $\b = 1$, $s = k_0 + 2$, $B = A$; $\Psi = N-1$,
  $\b = -\s$, $s = k_0 + 3$, $B = A + \s$; and $\Psi = e_I (N)$, $\b = 1-4\s$, $s = k_0 + 2$, $B = A + 4\s$
respectively, similarly yields the remaining estimates in \eqref{eq: extra derivatives psi chpsi}, recalling  \eqref{eq:mbLdef} and \eqref{eq:mbHdef}.
The estimates \eqref{eq: extra derivatives dynamical} are now immediate consequences of \eqref{eq: extra derivatives psi chpsi} and
\eqref{eq:chchomchgachphiestkz}.
\end{proof}

\subsection{Local existence and Cauchy stability}
In the next subsection, we prove that the global existence theorem
for solutions to the FRS equations, i.e. Theorem~\ref{thm:GlobalExistenceFermiWalker},
follows from Theorem~\ref{thm:BootstrapImprovement} and local existence. In
the present subsection, we state the results we need in this paper concerning
local existence and Cauchy stability. The results essentially follow immediately
from the theory developed in \cite{OPR23}. In order to formulate
the continuation criterion associated with the local existence result, it
is convenient to introduce
\begin{equation}\label{eq:zetadef}
  \zeta:=k_{IJ}k_{IJ}+(e_{0}\phi)^{2}-2V\circ\phi/(n-1)
\end{equation}
and, assuming $\zeta(t,\cdot)>0$, 
\begin{equation}\label{eq:mfCdef}
  \begin{split}
    \mfC(t) := & \|e \|_{C^{3}(\S_t)} + \|\o \|_{C^{3}(\S_t)} + \|N\|_{C^{3}(\S_t)} +\|k  \|_{C^{2}(\S_t)} \\
    & + \| \phi \|_{C^{3}(\S_t)} + \|\dtp \|_{C^{2}(\S_t)}+\|1/N\|_{C^{0}(\S_t)}+\|1/\zeta\|_{C^{0}(\S_t)}.
  \end{split}
\end{equation}

\begin{lemma}[Local existence, FRS equations]\label{le:LocalExistenceFermiWalker}
  Let $(\S,h_{\refer})$ be a closed, connected and oriented Riemannian manifold with a smooth global orthonormal
  frame $(E_{i})_{i=1}^{n}$ (with dual frame $(\eta^{i})_{i=1}^{n}$). Let $V \in C^\infty(\rn{})$,
  \begin{equation} \label{eq: initial functions local existence}
    \be_I^i, \bo^I_i, \bk_{IJ}, \bga_{IJK}, \bp_0, \bp_1 : \S \to \rn{}
  \end{equation}
  be smooth functions and define $\be_I:= \be_I^i E_i$ and $\bo^I:= \bo^I_i \eta^i$. Assume that
  \begin{itemize}
  \item $(\be_I)_{I=1}^{n}$ is a smooth frame of $\S$,
  \item $(\bo^I)_{I=1}^{n}$ is the dual frame of $(\be_I)_{I=1}^{n}$,
  \item $\bk_{IJ} = \bk_{JI}$ and $\bk_{II}$ is a strictly positive real number, say $1/t_0$,
  \item $\bga_{IJK} = \bo^K([\be_I, \be_J])$,
  \end{itemize}
  that the functions \eqref{eq: initial functions local existence} satisfy \eqref{eq:hamconstraint} and \eqref{eq:MC} (with bars added to
  $\gamma$, $e$ and $k$; $t=t_0$; $e_0\phi$ replaced by $\bp_1$; and $\phi$ replaced by $\bp_1$) and that
  \[
    \bk_{IJ}\bk_{IJ}+\bp_{1}^{2}-2V\circ\bp_{0}/(n-1)>0.
  \]
  Then there exists an open interval $\mI=(t_{-},t_{+})\subseteq (0,\infty)$, with $t_{0}\in\mI$, and a unique smooth solution 
  \begin{equation}\label{eq:maximal solution}
    (e_I^i, \o_i^I, k_{IJ}, \gamma_{IJK}, \phi, N) : \mI\times \S \to \rn{},
  \end{equation}
  to (\ref{eq: transport frame})--(\ref{eq:LapseEquation}) satisfying the initial condition
  \begin{equation*}
    \left( e_I^i, \o^I_i, k_{IJ}, \g_{IJK}, \phi, e_0 \phi \right)|_{t = t_0}
    = \left( \be_I^i, \bo^I_i, \bk_{IJ}, \bga_{IJK}, \bp_0, \bp_1 \right)
  \end{equation*}
  and such that if $e_I := e_I^i E_i$ and $\o^I := \o^I_i \eta^i$, then the following holds for $t\in\mI$:
  \begin{itemize}
  \item $\zeta>0$ and $N>0$,    
  \item $ e_1, \hdots, e_n$ is a smooth frame of the tangent space of $\S_{t}$,
  \item $\o^1, \hdots, \o^n$ is the dual frame of $e_1, \hdots, e_n$,
  \item $k_{IJ} = k_{JI}$ and $k_{II} = \frac1{t}$, 
  \item $\g_{IJK} = \o^K([e_I, e_J])$.
  \end{itemize}
  Finally, either $t_{-}=0$ or $\limsup_{t\downarrow t_{-}}\mfC(t)=\infty$. Similarly, either
  $t_{+}=\infty$ or $\limsup_{t\uparrow t_{+}}\mfC(t)=\infty$.
\end{lemma}
\begin{proof}
  The statement follows by adapting \cite[Theorem~10, pp.~6--7]{OPR23} to the present setting. Note, first of all, that
  the initial data (\ref{eq: initial functions local existence}) give rise to
  \[
    \bge:=\o^{I}\otimes\o^{I},\ \ \
    \bk:=\bk_{IJ}\o^{I}\otimes\o^{J}.
  \]
  Due to \eqref{eq:hamconstraint}, \eqref{eq:MC}, (\ref{eq:SpatialScalarCurvature}) and the fact that
  $\tr_{\bge}\bk=\bk_{II}=1/t_{0}$, it is clear that $(\S,\bge,\bk,\bp_{0},\bp_{1})$ are CMC initial data for the
  Einstein-non-linear scalar field equations with potential $V$. Moreover, $\zeta>0$, where $\zeta$ is
  introduced in (\ref{eq:zetadef}), and $\tr_{\bge}\bk=1/t_{0}\in (0,\infty)$. Since $\bM=\S$ is closed, connected,
  oriented and parallelizable, all the conditions of \cite[Theorem~10]{OPR23} are met. Next, let $\xi_{I}:=\be_{I}$,
  $\rho^{I}:=\bo^{I}$ and fix $a_{I}^{J}=0$, using the notation of \cite[Theorem~10]{OPR23}. Then
  $a_{I}^{J}$ satisfies the conditions of \cite[Theorem~10]{OPR23}. Next, $e_{I}^{i}|_{t_{0}}$, $\o^{I}_{i}|_{t_{0}}$,
  $\phi|_{t_{0}}$, $(\d_{t}\phi)|_{t_{0}}$, $N|_{t_{0}}$ and $k_{IJ}|_{t_{0}}$ are specified as in the statement
  of \cite[Theorem~10]{OPR23}. Due to \cite[Theorem~10]{OPR23}, there is then an open interval $\mI\subseteq (0,\infty)$,
  containing $t_0$, and a unique solution to \cite[(13), p.~5]{OPR23} on $M:=\mI\times \S$, corresponding to these
  initial data, such that $N>0$ and $\zeta>0$; $k_{II}=1/t$; and $k_{IJ}=k_{JI}$. In \cite[(13)]{OPR23}, $\bk_{IJ}$ is
  used to denote what we here call $k_{IJ}$, and in what follows we tacitly reformulate \cite[(13)]{OPR23} by
  replacing $\bk_{IJ}$ with $k_{IJ}$. Moreover, if $\mI=(t_{-},t_{+})$, then either $t_{-}=0$ or $\mC(t)$ tends to
  infinity as $t\downarrow t_{-}$, where $\mC$ is defined in \cite[(20), p.~6]{OPR23}.
  There is a similar statement concerning $t_{+}$. What remains to be done is to verify that the solution to
  \cite[(13)]{OPR23} yields a solution to (\ref{eq: transport frame})--(\ref{eq:LapseEquation}); that the
  corresponding solution to (\ref{eq: transport frame})--(\ref{eq:LapseEquation}) is unique; and that
  the continuation criterion from \cite[Theorem~10]{OPR23} yields the continuation criterion of the present
  lemma.

  Note, to begin with, that since $a_{I}^{J}=0$, $f_{I}^{J}=-Nk_{IJ}$, see \cite[(22), p.~6]{OPR23}. This means that
  \cite[(13a), p.~5]{OPR23} and \cite[(13b), p.~5]{OPR23} are equivalent to (\ref{eq: transport frame}) and
  (\ref{eq: transport co-frame}) respectively; note that $e_{0}=N^{-1}\d_{t}$. If we define $\bga_{IJK} = \bo^K([\be_I, \be_J])$,
  then (\ref{eq:concoef}) follows from the Jacobi identity; see \cite[Lemma~38, p.~31]{OPR23}. Next, note that
  \cite[(13c)]{OPR23}, with our choice of $f_{I}^{J}$, reads
  \begin{equation}\label{eq:dtbkIJLE}
    \begin{split}
      \d_{t}k_{IJ} = & e_{(I}e_{J)}(N)-\g_{K(IJ)}e_{K}(N)-Nt^{-1}k_{IJ}+Ne_{I}(\phi)e_{J}(\phi)\\
      &+\tfrac{2N}{n-1}(V\circ\phi)\de_{IJ}-N\bR_{IJ},
    \end{split}
  \end{equation}
  where we appealed to (\ref{eq: Gamma as gamma}) and the symmetry of $\bna_{I}\bna_{J}N$; here $\bna$ and $\bR$
  denote the Levi-Civita connection and the Ricci tensor associated with metric induced on the leaves of the
  foliation. Combining (\ref{eq:dtbkIJLE}) with (\ref{eq: SpatialRicciCurvature}) yields (\ref{eq:sff}).
  Next, (\ref{eq:spatial scalar field derivative}) is a consequence of (\ref{eq: transport frame}) and
  (\ref{eq:scalarfield}) is a consequence of \cite[(13e)]{OPR23}; cf. \cite[(166), p.~36]{OPR23} and the
  adjacent text. Due to \cite[Theorem~10]{OPR23}, the
  solution to \cite[(13)]{OPR23} is a solution to the Einstein-non-linear scalar field equations. This
  means that the constraint equations are satisfied. In particular \eqref{eq:hamconstraint} and \eqref{eq:MC}
  hold (keeping in mind that the constant-$t$ hypersurfaces have constant mean curvature $1/t$). Finally,
  (\ref{eq:LapseEquation}) follows by combining \cite[(13d)]{OPR23}, \cite[(13g)]{OPR23} and
  (\ref{eq:SpatialScalarCurvature}).

  To conclude, we obtain a solution to (\ref{eq: transport frame})--(\ref{eq:LapseEquation}) on $M$, inducing the
  correct initial data. Moreover, we have the above continuation criterion. What remains is to demonstrate uniqueness
  of the solution and the continuation criterion. In order to prove uniqueness, it is sufficient to note that, by
  arguments similar to the above, solutions to (\ref{eq: transport frame})--(\ref{eq:LapseEquation}) yield
  solutions to \cite[(13)]{OPR23}, and that solutions to \cite[(13)]{OPR23} are unique due to \cite[Theorem~10]{OPR23}.
  In order to prove the continuation criterion, assume that $\mfC$ is bounded on $(t_{-},t_{0}]$. We then want to prove
  that $\mC$ is bounded on $(t_{-},t_{0}]$. To be able to prove this, we need to control 
  \begin{equation}\label{eq:norms to be controlled}
    \|\d_{t}k\|_{C^{1}(\S_{t})},\ \ \
    \|\d_{t}\phi\|_{C^{2}(\S_{t})},\ \ \
    \|\d_{t}e_{0}\phi\|_{C^{1}(\S_{t})}.    
  \end{equation}
  However, combining the assumed bound on $\mfC$ with (\ref{eq:sff}) yields a bound on the first norm in (\ref{eq:norms to be controlled}).  
  Similarly, the assumed bound on $\mfC$ immediately yields a bound on the second norm in (\ref{eq:norms to be controlled}). Finally,
  combining the assumed bound on $\mfC$ with (\ref{eq:scalarfield}) yields a bound on the third norm in (\ref{eq:norms to be controlled}).
  To conclude, if $\mfC$ is bounded on $(t_{-},t_{0}]$, then the same is true of $\mC$. The argument to the future is the same. The lemma follows. 
\end{proof}

Next, we state the Cauchy stability result we need. It follows from \cite[Theorem~13, p.~7]{OPR23}.

\begin{lemma}[Cauchy stability, FRS equations]\label{lemma:CauchystabEinstein}
  Given assumptions, conclusions and notation as in the statement of Lemma~\ref{le:LocalExistenceFermiWalker}, let (\ref{eq:maximal solution})
  denote the solution obtained in the conclusions. Given $t_{1}\in\mI$, $m>n/2+1$ and $\epsilon>0$, there is then a $\delta>0$ such that the following
  holds. Let 
  \begin{equation} \label{eq:perturbed initial functions local existence}
    \te_I^i, \tom^I_i, \tk_{IJ}, \tga_{IJK}, \tp_0, \tp_1 : \S \to \rn{}
  \end{equation}
  be smooth functions, satisfying the conditions on initial data described in Lemma~\ref{le:LocalExistenceFermiWalker} (including $\tk_{II}=1/t_{0}$),
  and define $\te_I:= \te_I^i E_i$ and $\tom^I:= \tom^I_i \eta^i$. Let $\bremI$ be the existence interval and 
  \begin{equation}\label{eq:maximal perturbed solution}
    (\bree_I^i, \breo_i^I, \brek_{IJ}, \brega_{IJK}, \brep, \breN) : \bremI\times \S \to \rn{},
  \end{equation}
  be the solution obtained by appealing Lemma~\ref{le:LocalExistenceFermiWalker} to the initial data given by
  (\ref{eq:perturbed initial functions local existence}). Let, moreover, $\bree_{0}:=\breN^{-1}\d_{t}$. Then, if
  \begin{subequations}\label{seq:Cauchycriterion}
    \begin{align}
      \textstyle{\sum}_{i,I}\|\bree_{I}^{i}-e_{I}^{i}\|_{H^{m+1}(\bM_{t_{0}})}+\textstyle{\sum}_{i,I}\|\breo_{i}^{I}-\omega_{i}^{I}\|_{H^{m+1}(\bM_{t_{0}})} & < \delta,\\
      \textstyle{\sum}_{I,J}\|k_{IJ}-\brek_{IJ}\|_{H^{m+1}(\bM_{t_{0}})}+\textstyle{\sum}_{I,J}\|\d_{t}k_{IJ}-\d_{t}\brek_{IJ}\|_{H^{m}(\bM_{t_{0}})} & <\delta,\\
      \begin{split}
      \|\phi-\brep\|_{H^{m+2}(\bM_{t_{0}})}+\|\d_{t}\phi-\d_{t}\brep\|_{H^{m+1}(\bM_{t_{0}})}& \\
      +\|e_{0}\phi-\bree_{0}\brep\|_{H^{m+1}(\bM_{t_{0}})}+\|\d_{t}e_{0}\phi-\d_{t}\bree_{0}\brep\|_{H^{m}(\bM_{t_{0}})} & < \delta,
      \end{split}      
    \end{align}
  \end{subequations}  
  the interval $\bremI$ contains $t_{1}$ and (\ref{seq:Cauchycriterion}) holds with $t_{0}$ replaced by $t_{1}$ and $\delta$ replaced by $\epsilon$.  
\end{lemma}
\begin{remark}
  The conditions appearing in (\ref{seq:Cauchycriterion}) are unfortunate in that they involve quantities that are not immediately expressible in
  terms of initial data. However, since (\ref{eq:maximal solution}) and (\ref{eq:maximal perturbed solution}) are both solutions to
  (\ref{eq: transport frame})--(\ref{eq:LapseEquation}), all the quantities appearing in (\ref{seq:Cauchycriterion}) can indirectly be expressed in
  terms of the initial quantities. 
\end{remark}
\begin{proof}
  Given the relations between \cite[(13)]{OPR23} and (\ref{eq: transport frame})--(\ref{eq:LapseEquation}), described in the proof of
  Lemma~\ref{le:LocalExistenceFermiWalker}, the statement is an immediate consequence of \cite[Theorem~13, p.~7]{OPR23}. 
\end{proof}

\subsection{Global existence as a consequence of bootrap improvement}

Sections~\ref{sec:MainEstimates} and \ref{sec:EnergyEstimates} are devoted
to proving the bootstrap improvement, Theorem~\ref{thm:BootstrapImprovement}.
In the present subsection, we prove that the global existence theorem
for solutions to the FRS equations, i.e. Theorem~\ref{thm:GlobalExistenceFermiWalker},
follows from Theorem~\ref{thm:BootstrapImprovement} and the local existence
result of the previous subsection.

\begin{proof}[Proof of Theorem~\ref{thm:GlobalExistenceFermiWalker}, assuming Theorem~\ref{thm:BootstrapImprovement}]
  Given the constants $\sigma$, $k_0$, $k_1$, $\rho_0$, let $\mathcal C, \tau_{\rob}, r_{\rob}$ be the constants provided by
  Theorem~\ref{thm:BootstrapImprovement}. Due to Proposition~\ref{prop: initial lapse estimate}, there is a $0<\tau_{1}\leq\tau_{\rob}$,
  depending only on $\sigma$, $k_0$, $k_1$, $\rho_0$ and $(E_{i})_{i=1}^{n}$, such that if $t_{0}<\tau_{1}$, then 
  \begin{equation*}
    \begin{split}
      &t_0^2 \textstyle{\sum}_{I}\snorm{\be_I (\bN-1)}_{H^{k_1}(\S)}^2 + \snorm{\bN-1}_{H^{k_1}(\S)}^2 \\
      \leq & C t_0^{9\s/2}\big( t_0^2 \textstyle{\sum}_{I}\snorm{\be_I (\bN-1)}_{H^{k_1}(\S)}^2
      + \snorm{\bN - 1}_{H^{k_1}(\S)}^2\big)
      +\tfrac{1}{2}\snorm{\bN-1}_{H^{k_1}(\S)}^{2}+C t_{0}^{9\sigma},
    \end{split}
  \end{equation*}
  where we appealed to Young's inequality. By decreasing $\tau_1>0$, with the same dependence, 
  \[
  t_{0}\textstyle{\sum}_{I}\snorm{\be_I (\bN-1)}_{H^{k_1}(\S)}+\snorm{\bN-1}_{H^{k_1}(\S)}\leq Ct_{0}^{9\sigma/2}
  \]
  for $t_0 < \tau_1$. Combining this estimate with Sobolev embedding, we conclude that by decreasing $\tau_{1}>0$ again, with the same dependence, 
  \[
    \mbL_{(N)}(t_0) + \mbH_{(N)}(t_0)
    \leq \tfrac{r_\rob}{2}
  \]
  for $t_0 < \tau_1$. Since $\mbD(t_0) = 0$, we may thus conclude that 
  \begin{equation}\label{eq:Improved bootstrap initially}
    \mbD(t_0) + \mbL_{(N)}(t_0) + \mbH_{(N)}(t_0)
    \leq \tfrac{r_\rob}2.
  \end{equation}
  On the other hand, by decreasing $\tau_1>0$ further, if necessary, with the same dependence, we can also make sure that
  \begin{equation}\label{eq:bootstrap improvement condition on tz}
    \mathcal C t_0^\s\leq\tfrac{r_\rob}{2}.
  \end{equation}
  We first show that the bootstrap assumption is satisfied to the past of $t_{0}$ in the existence interval, say $\mI=(t_{-},t_{+})$. Let
  \[
  \msA:=\{t\in (t_{-},t_{0}]:\mbD(s) + \mbL_{(N)}(s) + \mbH_{(N)}(s)\leq r_\rob\ \forall s\in [t,t_{0}]\}.
  \]
  Due to (\ref{eq:Improved bootstrap initially}), we know that there is a $t<t_{0}$ such that $t\in\msA$. Thus $\msA$ is non-empty.
  By definition, $\msA$ is connected and closed. It remains to be demonstrated that $\msA$ is open. Let, to this end,
  $t_{1}\in \msA$. Then the conditions of Theorem~\ref{thm:BootstrapImprovement} are satisfied with $t_\rob$ replaced by $t_{1}$.
  This means that (\ref{eq: global existenece estimate}) is satisfied for all $t\in [t_{1},t_{0}]$. Combining this
  estimate with our restrictions on $t_{0}$, guaranteeing (\ref{eq:bootstrap improvement condition on tz}), it is clear that
  \[
  \mbD(t_1) + \mbL_{(N)}(t_1) + \mbH_{(N)}(t_1)\leq \tfrac{r_\rob}{2}.
  \]
  Due to the smoothness of the solution, it is thus clear that there is a $t<t_{1}$ such that $t\in\msA$. This means that $\msA$ is
  open. To conclude $\msA=(t_{-},t_{0}]$. What remains is to prove that $t_{-}=0$. According to Lemma~\ref{le:LocalExistenceFermiWalker},
  either $t_{-}=0$ or $\mfC$, defined in (\ref{eq:mfCdef}), is unbounded on $(t_{-},t_{0}]=\msA$. Assume, to this end, that $t_{-}>0$.
  Then $\mfC$ is unbounded on $\msA$. To obtain a contradiction, our next goal is to prove that $\mfC$ is bounded on $\msA$. We start
  by proving that $\zeta$, introduced in (\ref{eq:zetadef}), is bounded from below by a strictly positive constant on $\msA$. Note, to
  this end, that the Hamiltonian constraint \eqref{eq:hamconstraint} and Lemma~\ref{lemma:estimatingthedynamicalvariables} imply that
  \begin{equation*}
    \begin{split}
      \left| 1 - t^2\zeta\right| 
      \leq & t^2 | - 2e_I(\g_{IJJ})
      + \tfrac14\g_{IJK}(\g_{IJK} + 2 \g_{IKJ}) +\g_{IJJ}\g_{IKK} + e_I(\phi) e_I(\phi) + \tfrac{2n}{n-1} V \circ \phi | \\
      \leq & C t^{4\s}
    \end{split}
  \end{equation*}  
  for all $t \in \msA$. Hence, if $\tau_1>0$ is small enough, with the same dependence as before, then $1-t^{2}\zeta\leq 1/2$,
  so that $\zeta\geq 1/(2t^{2})$. In particular, there is a uniform positive lower bound on $\zeta$ on $\msA$. Next, we need a uniform
  positive lower bound on $N$. However, by the definition of $\msA$,
  \[
  \|N-1\|_{C^{0}(\S_t)}\leq r_{\rob}t^{\s}.
  \]
  Since we can assume $r_\rob\leq 1/2$ and $t\leq 1$, it is clear that $N\geq 1/2$ on $\msA$. Since $t_{-}>0$, a bound on the remaining
  norms in $\mfC$ follows from Lemma~\ref{lemma:estimatingthedynamicalvariables} and the definition of $\msA$. This leads to a contradiction,
  and we conclude that $t_{-}=0$. This finishes the proof.
\end{proof}

\section{The main estimates}\label{sec:MainEstimates}

\subsection{Scheme for systematic estimates}
\label{sec:scheme}
Many estimates derived in this section share qualitative features. Here, we therefore begin by presenting a way to quickly compute suitable
upper bounds which suffice for most terms. In practice, there are two schemes, the lower-order one, which we demonstrate first, 
and the higher-order one, which makes use of the lower-order one. We note that neither the lower-order nor the higher-order scheme 
makes use of the improved estimates of the lapse that are the subject of Subsection~\ref{sec:Lapse}. 
\subsubsection{Scheme for estimating lower-order terms}

For the lower-order estimates,
typically we have a sum of terms consisting of two to four factors,
with up to $k_0+1$ derivatives falling on the entire product.
As described in Appendix~\ref{sec:SobolevInequalities}, in particular 
Corollary~\ref{cor:SobolevAlgebra}, $C^{k}(\S)$ is up to a multiplicative
constant a Banach algebra, allowing us to estimate each factor individually.
To illustrate the idea, we use the example of the (implicit) sum
$(N-1) k_{IM} e_M^i$ measured in the $C^{k_{0}+1}$-norm 
and with a given time-dependent factor, i.e.
\begin{equation*}
t^{1-3\s} \| (N-1) k_{IM} e_M^i \|_{C^{k_{0}+1}(\S_t)}.
\end{equation*}
By Corollary~\ref{cor:SobolevAlgebra},
\begin{equation*}
  \begin{split}
    & t^{1-3\s} \| (N-1) k_{IM} e_{M}^i \|_{C^{k_{0}+1}(\S_t)}\\
    \leq & C t^{1-3\s} \|N-1 \|_{C^{k_{0}+1}(\S_t)} 
    \textstyle{\sum}_{M} ( \|k_{IM} \|_{C^{k_{0}+1}(\S_t)}
    \|e_M^i \|_{C^{k_{0}+1}(\S_t)} ) \\
    \leq & C t^{1-3\s} \|N-1 \|_{C^{k_{0}+1}(\S_t)}
    \|k \|_{C^{k_{0}+1}(\S_t)} 
    \|e \|_{C^{k_{0}+1}(\S_t)};
  \end{split}
\end{equation*}
in the present section, we take it for granted that all constants $C$ are standard constants (see Notation~\ref{not:constant}). 
At this stage, we may use the a-priori estimates in Lemmata~\ref{lemma:estimatingthebackground},
\ref{lemma:estimatingthedynamicalvariables} and \ref{lemma:extraderivatives},
as well as the bootstrap inequality from Assumption~\ref{ass:Bootstrap},
which will be assumed to hold for any estimate shown in Section~\ref{sec:MainEstimates}.
Lemma~\ref{lemma:extraderivatives}, based on interpolation estimates,
is used to bound terms that have one or two more derivatives
than can be estimated using a-priori estimates of 
Lemmata~\ref{lemma:estimatingthebackground} and \ref{lemma:estimatingthedynamicalvariables}.
As several terms have the same qualitative
(i.e. the same up to a multiplicative constant) a-priori estimates,
we only need to count the number of factors which have the same qualitative a-priori estimate.

We group the factors that may appear in a term according
to their qualitative estimates as follows:
We define the integer-valued, non-negative parameters
$l_j, j \in \{1,\dots,11\}$,
by counting the number of factors in each group,
where for simplicity we ignore any indices of the factors, e.g. the factor $e_I^i$
is represented by $e$, and the factor $e_I(N)$ by $\en$.
In particular:
\begin{itemize}
\item denote by $l_1$
the number of times the factor $N-1$ appears,
by $l_{2}$ the number of times the factor $N$ appears,
and by $l_3$ the number of times a factor of the form $ \ear(N)$
appears;
\item denote by $l_{4}$ 
the number of times a factor of the form 
$\che$ or $\chom$ appears,
and by $l_5$  the number of times a factor of the form
$e$, $\dee{}{}$, $\o$ or $\deo{}{}$ appears;
\item denote by $l_{6}$
the number of times a factor of the form 
$\chga$ or $\chep$ appears,
and by $l_7$ the number of times a factor of the form $\g$, $\dega{}$, $\ep$ or $\vdep$ appears;
\item denote by $l_{8}$ the number of times a factor of the form $\dek{}$ or $\dep{0}$ appears,
by $l_{9}$ the number of times a factor of the form $\chk$ or $\chdtp$ appears,
and by $l_{10}$ the number of times a factor of the form $k$ or $\dtp$
appears;
\item denote by $l_{11}$ the number of times a factor of the form $V \circ \phi$ or $V' \circ \phi$ appears. 
\end{itemize}
Moreover, denote by $l_{\roint}$ the number of factors that require the use of the interpolation estimates from Lemma~\ref{lemma:extraderivatives},
due to having more spatial derivatives than the a-priori estimates allow for.
Note that the upshot of the interpolation estimates is an additional factor $C t^{-\s}$
as well as the appearance of the higher-order deviation quantities in the upper bound
in exchange for up to two more derivatives. However, the appearance of the higher-order
deviation quantities is not material for the quality of the estimates presented here,
as both the lower-order and higher-order deviation quantities are bounded, either by $\mbD$
in case of the dynamical variables or by 1 in case of terms involving the lapse,
for the purposes of the scheme.

We note that for the terms with counters $l_3$ and $l_7$,
we can handle up to $k_0$ derivatives without Lemma~\ref{lemma:extraderivatives},
while for the remainder we can handle up to $k_0+1$ derivatives.

We have summarized the definition of $l_{\roint}$ and $l_1,\dots,l_{11}$
in Table~\ref{tab:CountingLowerOrderScheme1}.
Observe that the estimates for $l_1,\dots,l_{11}$ are worse the higher the index goes.
In particular, if we prove an estimate for a given tuple of non-zero $l_j$,
say $l_{j_1}, \dots, l_{j_p}$, then the same estimate or a better one holds
if we replace one count of any $l_{j_i}$ by any other $l_j$ for which ${j_i} \geq j$.
Also observe that we here have actually deteriorated the estimate
for factors of the form $\dek{}$ and $\dep{0}$ (corresponding to counter $l_8$)
in order to fit well with this monotonocity. For the purposes of the scheme, this deterioration is not material.
On the other hand, the estimates corresponding to the counters $l_3$ and $l_4$ are the same,
and the same goes for $l_9$ and $l_{10}$;
the reason we differentiate them is for the higher-order scheme.
\begin{remark} \label{rmk:LapseDoesNotDeterioriteEstimate}
The presence of one of the factors $l_1$ or $l_2$
does not deteriorate an estimate and in fact leads to either the same or a better
estimate than if the counter were not present.  
\end{remark}

\begin{table}[ht]
    \centering
    \begin{tabular}{|c|c|c|} \hline
    Counter & Factors & Contribution per count \\[1pt] \hline \hline
    $l_1$       & $N-1$     & $t^{\s}$ \\[1pt] \hline
    $l_2$       & $N$       & $1$ \\[1pt] \hline
    $l_3$       & $\ear(N)$ & $t^{-1+4\s}$ \\[1pt] \hline
    $l_4$       & $\che, \chom$ & $t^{-1+4\s}$ \\[1pt] \hline
    $l_5$       & $e, \dee{}{}, 
                \o, \deo{}{}$ 
                & $t^{-1+3\s}(\mbD + t^{\s})$ \\[1pt] \hline
    $l_6$       & $\chga , \chep$ & $t^{-1+3\s}$ \\[1pt] \hline
    $l_7$       & $\g, \dega{},
                \ep, \vdep$ 
                & $t^{-1+2\s}(\mbD + t^{\s})$ \\[1pt] \hline
    $l_8$       & $\dek{},\dep{0}$       
                & $t^{-1} (\mbD + t^{\s})$ \\[1pt] \hline
    $l_9$       & $\chk, \chdtp$
                & $t^{-1}$ \\[1pt] \hline 
    $l_{10}$    & $\k, \dtp$
                & $t^{-1}$ \\[1pt] \hline    
    $l_{11}$    & $V \circ \phi, V' \circ \phi$ 
                & $t^{-2 + 5 \s}$ \\[1pt] \hline \hline
    $l_{\roint}$  & Each application of Lemma~\ref{lemma:extraderivatives}
                & $t^{-\s}$ \\[1pt] \hline 
    \end{tabular}
    \caption{
    This table offers an overview of the definition
    of the counters $l_1,\dots, l_{11}$ and $l_{\roint}$,
    as well as their contribution per count to the estimate.
    Note that the contributions are different than the a-priori estimates,
    as some estimates have been deteriorated for convenience
    and multiplicative constants are not important for the desired estimates.}
    \label{tab:CountingLowerOrderScheme1}
\end{table}

Then, the term measured in the $C^{k}$-norm, $k = k_0,k_0+1$ may be estimated by
\begin{align*}
&C \big( t^{\s} \big)^{l_1}
\cdot (1)^{l_2}
\cdot \big( t^{-1+4\s} \big)^{l_3 + l_4}
\cdot \big(t^{-1+3\s} (\mbD + t^{\s}) \big)^{l_5}
\cdot \big(t^{-1+3\s} \big)^{l_6} \\
&\cdot \big(t^{-1+2\s} (\mbD + t^{\s}) \big)^{l_7}
\cdot \big(t^{-1} (\mbD + t^{\s}) \big)^{l_8}
\cdot \big(t^{-1} \big)^{l_9 + l_{10}}
\cdot \big(t^{-2+5\s} \big)^{l_{11}}
\cdot \big(t^{-\s}\big)^{l_{\roint}},
\end{align*}
which may be written as
\begin{equation} \label{eq:LowerOrderSchemeUpperBound}
 C t^{-m_1 + m_{\s} \s} 
    \big(\mbD + t^{\s} \big)^{m_D},
\end{equation}
where we have defined the parameters $m_1$, $m_{\s}$ and $m_D$ by
\begin{align*}
m_1 &:= l_3 + l_4 + l_5 + l_6 + l_7 + l_8 + l_9 + l_{10} + 2 l_{11} \\
m_{\s} &:= l_1 + 4 l_3 + 4 l_4 + 3 l_5 + 3 l_6 + 2 l_7 + 5 l_{11} - l_{\roint}, \\
m_D &:=l_5+ l_7 + l_8.
\end{align*}
We have summarized the definition of $m_1$, $m_{\s}$ and $m_D$
in Table~\ref{tab:CountingLowerOrderScheme2}.
\begin{table}[ht]
    \centering
    \begin{tabular}{|c|c|c|c|c|c|c|c|c|c|c|c|c|c|} \hline
    Counter & $l_1$ & $l_2$ & $l_3$ & $l_4$ & $l_5$ & $l_6$
    & $l_7$ & $l_8$ & $l_9$ & $l_{10}$ & $l_{11}$ & $l_{\roint}$
            & Contribution per count \\[1pt] \hline \hline
    $m_1$   &  &  & 1 & 1 & 1 & 1 & 1 & 1 & 1 & 1 & 2 & &$t^{-1}$ \\[1pt] \hline 
    $m_\s$  & 1 &  & 4 & 4 & 3 & 3 & 2 &  &  & & 5 & -1 & $t^{\s}$ \\[1pt] \hline 
    $m_D$   &  &  &  &  & 1 & & 1 & 1 &  & & & & $\mbD + t^{\s}$ \\[1pt] \hline 
    \end{tabular}
    \caption{
    This table offers an overview of the definition of the counters
    $m_1$, $m_{\s}$ and $m_{D}$ in terms of $l_1,..,l_{11}$ and $l_i$.}
    \label{tab:CountingLowerOrderScheme2}
\end{table}
In our example case from before we have $l_1 = 1$, $l_5 = 1$, $l_{10} = 1$
and all other parameters zero, which means that $m_1 = 2$, $m_\s = 4$, $m_D = 1$,
and thus we establish the following upper bound: 
\begin{align*}
t^{1-3\s} &\| (N-1) k_{IM} e_M^i \|_{C^{k_{0}+1}(\S_t)}\leq t^{1-3\s} \cdot C t^{-2 + 4\s}(\mbD + t^{\s})= C t^{-1+\s} (\mbD + t^{\s}).
\end{align*}

\subsubsection{Scheme for estimating higher-order terms}
Most of the higher-order estimates are established using the Moser estimates from Lemma~\ref{lemma: products in L2}.
To illustrate how to proceed systematically, we estimate the same term  as above, but in the $H^{k_1}$-norm and with a different
time-dependent factor, namely
\begin{equation*}
t^{A+1-2\s} \| (N-1) k_{IM} e_M^i \|_{H^{k_1}(\S_t)}.
\end{equation*}
The term is again a sum of products with factors $(N-1), k_{IM}$ and $e_M^i$.

In this case, because of the triangle inequality
and the Moser-type inequalities of Lemma~\ref{lemma: products in L2},
we can estimate the example term by
\begin{equation*}
  \begin{split}
     t^{A+1-2\s} \| (N-1) k_{IM} e_M^i \|_{H^{k_1}(\S_t)}
    \leq & t^{A+1-2\s} \textstyle{\sum}_{M}
    \| (N-1) k_{IM} e_{M}^i \|_{H^{k_1}(\S_t)}\\
    \leq & Ct^{A+1-2\s} 
    \big( \|N-1 \|_{H^{k_1}(\S_t)}
    \|k \|_{C^{0}(\S_t)} 
    \|e \|_{C^{0}(\S_t)} \\
    & \phantom{Ct^{A+1-2\s}\big(}+ \|N-1 \|_{C^{0}(\S_t)} 
    \|k \|_{H^{k_1}(\S_t)}
    \|e \|_{C^{0}(\S_t)} \\
    & \phantom{Ct^{A+1-2\s}\big(} + \|N-1 \|_{C^{0}(\S_t)}
    \|k \|_{C^{0}(\S_t)} 
    \|e \|_{H^{k_1}(\S_t)} \big).
  \end{split}
\end{equation*}
Observe that the number of different terms inside the parenthesis on the far right-hand side
of the inequality above equals the number of different factors in the original term.
For each term on the right-hand side, we estimate it
similarly to the one in the lower-order scheme, using the a-priori estimates from
Lemmata~\ref{lemma:estimatingthebackground} and \ref{lemma:estimatingthedynamicalvariables},
as well as the bootstrap inequality, 
as Assumption~\ref{ass:Bootstrap} will be assumed to hold for any estimate 
appearing in this section. On the other hand, we do not need to appeal to
Lemma~\ref{lemma:extraderivatives}.
In this scheme, however, we need to take care of the fact that one factor appears
measured in the $H^{k_1}$-norm instead of some $C^{k}$-norm.
This can be done as follows:
Let $j_1, \dots, j_p$ denote those $j \in \{1,\dots,11\}$
such that $l_j \geq 1$.
Then, take the estimate (\ref{eq:LowerOrderSchemeUpperBound})
from the scheme for the lower-order terms, with $l_{\roint} = 0$,
and multiply it by $t^{-\max \{ S(j_1), \dots , S(j_p)\}}$, 
where 
\begin{equation}
S(j):= \begin{cases} 
0 & \text{if}~ j = 4,6,9, \\
A & \text{if}~ j = 2,8,10,11, \\
A+\s & \text{if}~ j = 1,5, \\
A + 2\s & \text{if}~ j = 7, \\
A + 4\s & \text{if}~ j = 3.
\end{cases}
\end{equation}
Observe that $t^{-S(j)}$ bounds the quotient between the a-priori higher-order and the a-priori lower-order bounds,
up to a multiplicative constant.
For example, we know from Lemma~\ref{lemma:estimatingthedynamicalvariables}
that on the one hand,
\begin{equation*}
\|e\|_{C^{k_{0}+1}(\S_t)} \leq C t^{-1+3\s} (\mbD + t^{\s}),
\end{equation*}
while on the other hand 
\begin{equation*}
\|e\|_{H^{k_1}(\S_t)} \leq C t^{-A-1+2\s} (\mbD + t^{A+2\s}).
\end{equation*}    
Hence, in order to replace an instance of the lower-order norm
by an instance of the higher-order norm,
we include a multiplicative factor $t^{-A-\s}$ in the estimate, as
\begin{equation*} 
\frac{Ct^{-A-1+2\s}(\mbD + t^{A+2\s})}{Ct^{-1+3\s}(\mbD + t^{\s})}
\leq Ct^{-A-\s}.
\end{equation*}

In Table~\ref{tab:CountingHigherOrderScheme},
we present the information regarding the expression 
    $t^{- S(j)}$ slightly differently,
using the factors present in the to-be-estimated term,
instead of the counters. 
Moreover, note that $1 \leq t^{-A} \leq t^{-A-\s} \leq t^{-A-2\s} \leq t^{-A-4\s}$
as $ t \leq 1$ and $A,\s > 0$.
Hence one only needs to check the contribution of the lowest entry of the table
of the factors that are present.

\begin{table}[ht]
    \centering
    \begin{tabular}{|c|c|} \hline
    Factor & Contribution if factor is present \\[1pt] \hline \hline
    $\che, \chom, \chga, \chep, \chk, \chdtp$   &  $1$ \\[1pt] \hline
    $N, k,\dek{},\dtp, \dep{0}, 
    V \circ \phi, V'\circ \phi$                 &  $t^{-A}$ \\[1pt] \hline 
    $N-1, e, \dee{}{}, \o, \deo{}{}$          &  $t^{-A-\s}$ \\[1pt] \hline 
    $\g, \dega{}, \ep, \vdep$          &  $t^{-A-2\s}$ \\[1pt] \hline 
    $\ear(N)$                                   &  $t^{-A-4\s}$ \\[1pt] \hline 
    \end{tabular}
    \caption{
    This table offers an overview of the multiplicative factor 
    that we need to include for the scheme for higher-order estimates,
    depending on which factors are present.
    Note that one only needs to take into account the lowest entries
    that are relevant for the term that is to be estimated.}
    \label{tab:CountingHigherOrderScheme}
\end{table}

Applying the scheme to our example,
we note that as $l_1, l_5, l_{10}$ are the only non-zero counters, 
we have $\max_{i} S(j_i) = A+\s$.
Hence we end up with the estimate
\begin{equation*}
  \begin{split}
    & t^{A+1-2\s} \| (N-1) k_{IM} e_M^i \|_{H^{k_1}(\S_t)} \\
    \leq & Ct^{A+1-2\s} \cdot t^{-A-\s} \cdot C t^{-2+4\s}(\mbD + t^{\s}) 
    \leq C t^{-1+\s} (\mbD + t^{\s}).
  \end{split}
\end{equation*}

\subsection{The curvature and the energy-momentum tensor}
As a more detailed example of how to use the scheme developed in the previous
subsection, and as the corresponding estimates are needed later on in this article, 
we present here some estimates for the spatial Riemann curvature tensor,
for the spacetime Riemann curvature tensor with one time-like entry,
and for certain terms appearing in the energy momentum tensor.

\begin{lemma} \label{lemma:SpatialRiemann}
  Let $\s_p$, $\s_V$, $\s$, $k_0$, $k_1$, $(\S, h_{\refer})$, $(E_{i})_{i=1}^{n}$ and $V$ be as in Theorem~\ref{thm: big bang formation}
  and let $\rho_0 >0$. Assume that there are $t_{0}\leq 1$, $t_\rob<t_0$ and $\ar \in (0, \tfrac{1}{6n}]$ such that
  Assumption~\ref{ass:Bootstrap} is satisfied for this choice of parameters. Then the following 
  estimates hold for any $I,J,K,L$ and $t \in [t_{\rob},t_0]$:
  \begin{align}
    \label{eq:EstimateSpatialRiemann}
    t^2 \|\roRiem_{h}(e_I, e_J, e_K, e_L) \|_{C^{k_{0}+1}(\S_t)}
    &\leq Ct^{2\s}(\mbD + t^{\s})^2, \\
    \label{eq:EstimateNormalRiemann}
    t^2 \|\roRiem_{g}(e_0, e_I, e_J, e_K) \|_{C^{k_{0}+1}(\S_t)}
    &\leq Ct^{\s}(\mbD + t^{\s}), \\ 
    \label{eq:EstimateEnergyMomentumSF}
    t^2 \| e_I(\phi) e_J (\phi) \|_{C^{k_{0}+1}(\S_t)}
    &\leq Ct^{2\s} (\mbD + t^{\s})^2.
  \end{align}  
\end{lemma}
\begin{proof}
  Note that 
  \begin{align*}
    \roRiem_{h}(e_I, e_J, e_K, e_L) 
    = & e_I(\G_{JKL}) - e_J(\G_{IKL})\\
    & + \G_{JKM} \G_{IML} - \G_{IKM} \G_{JML} - \g_{IJM} \G_{MKL}.
  \end{align*}
  As the connection coefficients satisfy (\ref{eq: Gamma as gamma}),  
  it thus suffices to estimate terms of the form $e(\gamma)$ and $\g \cdot \g$ in $C^{k_{0}+1}$.
  However, for the former type we have $l_5 = 1$, $l_7 = 1$ and $l_{\roint} = 1$.
  Thus $m_1 = 2$, $m_{\s} = 4$, $m_D = 2$ and
  \begin{equation} \label{eq:CurvatureFirstbound}
    t^2 \|e_I (\g_{JKL}) \|_{C^{k_{0}+1}(\S_t)} \leq C t^{4\s} (\mbD(t) + t^{\s})^2.
  \end{equation}
  For terms of the latter type, $l_7 = 2$, $l_{\roint} =2$, $m_1 = 2$, $m_{\s} = 2$ and $m_D = 2$. Thus
  \begin{equation} \label{eq:CurvatureSecondBound}
    t^2 \|\g_{JKM} \g_{IML} \|_{C^{k_{0}+1}(\S_t)} \leq C t^{2\s} (\mbD(t) + t^{\s})^2.
  \end{equation} 
  Next, by the Gauss-Codazzi equation,
  \begin{equation*}
    \begin{split}
      \roRiem_g(e_0, e_I, e_J,e_K) &= (\nabla^h_{e_J} k)_{IK}) - (\nabla^h_{e_K} k)_{IJ} \\
      &= e_J(k_{KI}) - e_K(k_{IJ}) - k_{KM} \G_{JIM} + k_{JM} \G_{KIM} - k_{IM} \g_{JKM}.
    \end{split}
  \end{equation*} 
  Hence it suffices to estimate terms of the form $e(k)$ and $k \cdot \gamma$. However,  
  \begin{equation} \label{eq:CurvatureThirdBound}
    t^2 \| e_J (k_{IK}) \|_{C^{k_{0}+1}(\S)} \leq Ct^{2\s}(\mbD + t^{\s})
  \end{equation}
  as, in this case, $l_5 = 1$, $l_{10} = 1$, $l_{\roint} = 1$, $m_{1} =2$, $m_{\s} = 2$ and $m_{D} =1$.
  On the other hand, e.g.,
  \begin{equation} \label{eq:CurvatureFourthBound}
    t^{2}\| k_{KM} \g_{IJM} \|_{C^{k_{0}+1}(\S)} \leq Ct^{\s}(\mbD + t^{\s})
  \end{equation}
  as in this case $l_7 =1$, $l_{10} =1$, $l_{\roint} = 1$, $m_1 = 2$, $m_{\s} = 1$ and $m_D = 1$. Next, 
  \begin{equation} \label{eq:CurvatureLastBound}
    t^2 \| e_I(\phi) e_J(\phi) \|_{C^{k_{0}+1}(\S_t)} 
    \leq C t^{2\s} (\mbD + t^{\s})^2,
  \end{equation}
  for any $I,J \in \{1,\dots,n\}$, as $l_7 = 2$, $l_{\roint} = 2$, $m_1 =2$, $m_{\s} = 2$ and $m_D = 2$.  
\end{proof}

The following corollary of the above lemma will also be of use:
\begin{lemma} \label{lemma:LowerOrderHamiltonianConstraint}
  Let $\s_p$, $\s_V$, $\s$, $k_0$, $k_1$, $(\S, h_{\refer})$, $(E_{i})_{i=1}^{n}$ and $V$ be as in Theorem~\ref{thm: big bang formation}
  and let $\rho_0 >0$. Assume that there are $t_{0}\leq 1$, $t_\rob<t_0$ and $\ar \in (0, \tfrac{1}{6n}]$ such that
  Assumption~\ref{ass:Bootstrap} is satisfied for this choice of parameters. Then the following estimate holds
  for all $t \in [t_{\rob},t_0]$:
  \begin{equation}\label{eq:almost asymptotic Hcon}
    t^2 \|t^{-2} - k_{IJ} k_{IJ} - \dtp \dtp \|_{C^{k_{0}+1}(\S_t)}
    \leq Ct^{2\s}(\mbD + t^{\s})^2.
  \end{equation}  
\end{lemma}
\begin{proof}
  Recall that by the Hamiltonian constraint, i.e. \eqref{eq:hamconstraint},
  \begin{equation} \label{eq:LowerOrderHamiltonianConstraint}
    t^{-2} - k_{IJ} k_{IJ} - \dtp \dtp 
    = - \roScal_h + e_I(\phi) e_I(\phi) + 2 V \circ \phi.
  \end{equation}
  Hence it suffices to estimate the right-hand side in $C^{k_{0}+1}$.
  However, a sufficient estimate for the potential follows directly from Lemma~\ref{lemma:estimatingthedynamicalvariables}
  while for the other terms we obtain a sufficient estimate
  from Lemma~\ref{lemma:SpatialRiemann} above.
\end{proof}

\subsection{The lapse} \label{sec:Lapse}
We require improved estimates for the lapse in the lower- and higher-order norms.
They follow from the elliptic nature of the lapse equation \eqref{eq:LapseEquation}:
\begin{equation} \label{eq:LapseRepeated}
  \begin{split}
    e_{I}e_{I}(N)-t^{-2}(N-1) &= \g_{JII}e_{J}(N) 
    + N \big \{ \roScal_h  - e_{I}(\phi)e_{I}(\phi) 
    - \tfrac{2n}{n-1}V\circ\phi \big\}.
  \end{split}
\end{equation}
We commute this equation with $E_{\bfI}$ in order to control spatial derivatives of $N$.

\subsubsection{Improving the lower-order estimates for the lapse}
We begin with the improved lower-order estimates.
The goal here is to control $\mbL_{(N)}$ in terms of the dynamical variables,
so as to be able to close the bootstrap argument.
Moreover, any additional time-decay makes estimates for the dynamical variables
stronger as well, and is required down the line.
The proof of such estimates relies on a simple version of the maximum principle.
However, we need to estimate the rest terms arising from commutators.

\begin{lemma} \label{lemma:EstLapseLowJunk}
  Let $\s_p$, $\s_V$, $\s$, $k_0$, $k_1$, $(\S, h_{\refer})$, $(E_{i})_{i=1}^{n}$ and $V$ be as in Theorem~\ref{thm: big bang formation}
  and let $\rho_0 >0$. Assume that there are $t_{0}\leq 1$, $t_\rob<t_0$ and $\ar \in (0, \tfrac{1}{6n}]$ such that
  Assumption~\ref{ass:Bootstrap} is satisfied for this choice of parameters. Then, for any $t \in [t_{\rob},t_{0}]$,
  \begin{equation}\label{eq:EstLapseLowJunk}
    t^{2} \|e_I e_I (N) - t^{-2} (N-1) \|_{C^{k_{0}+1}(\S_t)}
    \leq C t^{2\s}(\mbD(t)+t^{\s})^2.
  \end{equation}
\end{lemma}
\begin{proof}
  We proceed by estimating all the terms on the right-hand side of (\ref{eq:LapseRepeated}), making use of the scheme of Subsection~\ref{sec:scheme}.
  First, 
  \begin{equation} \label{eq:EstimateLapseJunkFirstBound}
    t^2 \| \g_{JII} e_J(N) \|_{C^{k_{0}+1}(\S_t)} 
    \leq C t^{4\s} (\mbD + t^{\s}),
  \end{equation}
  as $l_3 = 1$, $l_7 = 1$, $l_{\roint} = 2$, $m_1 = 2$, $m_{\s} = 4$ and $m_D =1$. Next, combining the fact that $N$ is bounded in $C^{k_0+1}$
  with (\ref{eq:EstimateSpatialRiemann}) and (\ref{eq:EstimateEnergyMomentumSF}) yields
  \begin{equation*}
    t^2 \| N \roScal_h \|_{C^{k_{0}+1}(\S_t)}+t^2 \| N e_I(\phi) e_I(\phi) \|_{C^{k_{0}+1}(\S_t)}  \leq C t^{2\s} (\mbD + t^{\s})^2.
  \end{equation*}
  Finally,
  \begin{equation} \label{eq:EstimateLapseJunkLastBound}
    t^2 \|N V \circ \phi \|_{C^{k_{0}+1}(\S_t)} \leq Ct^{5\s},
  \end{equation}
  as $l_2 = 1$, $l_{11} = 1$, $m_1 = 2$, $m_{\s} = 5$ and $m_D = 0$. The lemma follows. 
\end{proof}

We make use of this lemma in order to establish an improvement on the bootstrap assumptions for the lower-order norms:

\begin{lemma}\label{lemma:LowerOrderLapseEstimates}
  Let $\s_p$, $\s_V$, $\s$, $k_0$, $k_1$, $(\S, h_{\refer})$, $(E_{i})_{i=1}^{n}$ and $V$ be as in Theorem~\ref{thm: big bang formation}
  and let $\rho_0 >0$. Assume that there are $t_{0}\leq 1$, $t_\rob<t_0$ and $\ar \in (0, \tfrac{1}{6n}]$ such that
  Assumption~\ref{ass:Bootstrap} is satisfied for this choice of parameters. Then, for any $t \in [t_{\rob},t_{0}]$,
  \begin{equation}\label{eq:LapseControlLow}
    \mbL_{(N)} \leq C t^{\s} \left(\mbD(t)+t^{\s} \right)^{2}.
  \end{equation}  
\end{lemma}
\begin{proof}
  We start by estimating $\|N-1||_{C^{k_0+1}(\S_t)}$. The idea of the proof is to commute $E_{\bfI}$ through the left hand
  side of \eqref{eq:LapseRepeated}:
  \begin{equation} \label{eq:LapseCommuted}
    e_I e_I E_{\bfI}(N-1) - t^{-2} E_\bfI (N-1) 
    = E_{\bfI} \big(e_I e_I (N-1) - t^{-2} (N-1)\big) + [e_I e_I, E_\bfI] (N-1).
  \end{equation}
  for $|\bfI| \leq k_0+1$. On the other hand, the following maximum principle argument holds: Fix $t \in [t_{\rob}, t_0]$.
  If the function $E_\bfI(N-1)$ attains its maximum at some point, say, $p_{\max} \in \S_t$,
  then the function $e_I e_I E_{\bfI}(N-1)$ is non-positive at $p_{\max}$.
  (Note that it is assumed that $N$ is a smooth function.)
  This holds as, for each $M \in \{1,\dots,n\}$, we have $e_M E_{\bfI}(N-1)|_{p_{\max}} = 0$,
  and so the sum $e_I e_I E_{\bfI}(N-1)|_{p_{\max}}$ is the sum of the eigenvalues
  of the Hessian of $E_{\bfI}(N-1)$ at $p_{\max}$, which must all be non-positive.
  We conclude that for any $p \in \S_t$
  \begin{equation*}
    \begin{split}
      t^{-2} E_\bfI(N-1)|_p
      &\leq t^{-2}E_\bfI(N-1)|_{p_{\max}} \\
      &\leq \left(-e_I e_I E_{\bfI}(N-1) + t^{-2}E_\bfI(N-1) \right)|_{p_{\max}} \\
      &\leq \norm{-e_I e_I E_{\bfI}(N-1) + t^{-2}E_\bfI(N-1)}_{C^{0}(\S_t)}
    \end{split}
  \end{equation*}  
  A similar argument applies with $E_\bfI(N-1)$ replaced by $-E_\bfI(N-1)$
  and considering some point $p_{\min}$ where the minimum value of that is attained.
  Summing up,
  \begin{equation*}
    \norm{t^{-2} E_\bfI(N-1)}_{C^{0}(\S_t)}
    \leq \norm{e_I e_I E_{\bfI}(N-1) - t^{-2}E_\bfI(N-1)}_{C^{0}(\S_t)}.
  \end{equation*}
  Thus, by \eqref{eq:LapseCommuted}, we obtain the estimate
  \begin{equation}\label{eq:N m one prel est low}
    \begin{split}
      \| N-1\|_{C^{k_{0}+1}(\S_t)}
      \leq & t^{2}\big( \| e_I e_I (N-1) - t^{-2} (N-1) \|_{C^{k_{0}+1}(\S_t)} \\
      & \phantom{t^{2}\Big(} + \textstyle{\sum}_{|\bfI|\leq k_0+1}\|[e_I e_I, E_\bfI] (N-1)\|_{C^{0}(\S_t)} \big).
    \end{split}
  \end{equation}  
  Lemma~\ref{lemma:EstLapseLowJunk} applies to the first term.
  For the term involving the commutator, note that $[e_I e_I, E_{\bfI}]$
  is a differential operator of order $|\bfI| +1 \leq k_0 + 2$,
  with at most $|\bfI|+1$ derivatives acting on the frame components $e_I^i$,.
  \begin{equation}\label{eq:commutator estimates low N m one}
    \begin{split}
      & t^2 \textstyle{\sum}_{|\bfI|\leq k_0+1} \|[e_I e_I, E_\bfI] (N-1)\|_{C^{0}(\S_t)} \\
      \leq & C t^2 \|e\|_{C^{k_{0}+1}(\S_t)}\|e\|_{C^{k_0+2}(\S_t)}
      \|N-1\|_{C^{k_0+2}(\S_t)} \leq C t^{5\s} (\mbD + t^{\s})^2,
    \end{split}
  \end{equation}
  as we have a term with $l_1 = 1$, $l_5 = 2$, $l_{\roint} = 2$, $m_1 = 2$, $m_{\s} = 5$ and $m_D = 2$.
  Combining (\ref{eq:EstLapseLowJunk}), (\ref{eq:N m one prel est low}) and (\ref{eq:commutator estimates low N m one})
  yields
  \begin{equation*}
     t^{-\s}\|N-1 \|_{C^{k_{0}+1}(\S_t)} 
     \leq C t^{\s} \left(\mbD + t^{\s} \right)^2.
  \end{equation*}
  Combining this estimate with (\ref{eq:tqeomegaWkzest}) yields
  \begin{equation*}
    \begin{split}
      t^{1-4\s} \| \ear N \|_{C^{k_{0}}(\S_t)}
      & \leq Ct^{1-4\s} \norm{e}_{C^{k_{0}}(\S_t)}
      \| N-1 \|_{C^{k_{0}+1}(\S_t)}\\
      & \leq C t^{-\s} \left(\mbL_{(e,\omega)}(t) + t^\s \right)
      \| N-1 \|_{C^{k_{0}+1}(\S_t)} \leq C t^{\s} (\mbD + t^{\s})^2.
    \end{split}
  \end{equation*}
  This concludes the proof.
\end{proof}

\subsubsection{Improving the higher-order estimates for the lapse}
Recall from Remark~\ref{rmk: alternative lapse} that the lapse equation can be reformulated to (\ref{eq: alternative lapse}),
the \emph{alternative lapse equation}. By estimating the right-hand side of (\ref{eq: alternative lapse}), in particular
$t^2 - k_{IJ} k_{IJ} - \dtp \dtp$, we obtain sharper estimates for its left-hand side than shown in \cite{GIJ}.
This is the purpose of the following lemma.

\begin{lemma} \label{lemma:AuxiliaryLapseEstimate}
  Let $\s_p$, $\s_V$, $\s$, $k_0$, $k_1$, $(\S, h_{\refer})$, $(E_{i})_{i=1}^{n}$ and $V$ be as in Theorem~\ref{thm: big bang formation}
  and let $\rho_0 >0$. There is a standard constant $\tau_H <1$ such that if $t_{0}\leq \tau_H$, $t_\rob<t_0$, $\ar \in (0, \tfrac{1}{6n}]$
  and Assumption~\ref{ass:Bootstrap} is satisfied for this choice of parameters, then
  \begin{equation}\label{eq:AuxiliaryLapseEstimate}
    \begin{split}
      & t^{A+1} \| t^{-2} - k_{IJ} k_{IJ} - e_0(\phi) e_0(\phi)  \|_{H^{k_1}(\S_t)} \\
      \leq & \big(C\ar + 2 + \s \big) t^{A} \big(\| \dek{} \|_{H^{k_1}(\S_t)}^2
      + \| \dep{0}\|_{H^{k_1}(\S_t)}^2\big)^{1/2}+ C t^{-1+\s} (\mbD + t^{\s})
    \end{split}
  \end{equation}
  for any $t \in [t_{\rob}, t_0]$. Moreover, if $t_0 \leq \tau_{H}$, 
  \begin{equation} \label{eq:BoundEigenvalues}
    \left\|\textstyle{\sum}_{I} \bq_I^2 + \bP_1^2 \right\|_{C^{0}(\S)}^{1/2}
    \leq 1 + \s /4.
  \end{equation}
\end{lemma}
\begin{remark}
  In the statement of the lemma, and in the proof below, we use the notation $\bP_{1}:=t_{0}\bp_{1}$.
\end{remark}
\begin{proof}
We start with the observation that
\begin{equation*}
  \begin{split}
    & t^{-2} - k_{IJ} k_{IJ} - e_0(\phi) e_0 (\phi)\\
    = & t^{-2} \big( 1 - \textstyle{\sum}_{I} \bq_I^2 - \bP_1^2 \big)
    - \dek{IJ}\cdot\dek{IJ}-(\dep{0})^{2}- 2 \chk_{IJ}\cdot\dek{IJ} - 2 \d_t \chphi \cdot \dep{0}.
  \end{split}  
\end{equation*}
Multiplying both sides by $t^{A+1}$, 
we can estimate the first term on the right-hand side in $H^{k_1}$
by $C t^{A-1}$, which in its turn is bounded by $Ct^{-1+4\s}$;
note that $1 - \sum_{I} \bq_I^2 - \bP_1^2$ is bounded in $H^{k_1}$
as a consequence of (\ref{eq:deviationfromKasner}),
assuming that $\tau_H \leq \tau_1$, where $\tau_{1}$ is the constant
appearing in the statement of Lemma~\ref{lemma:EstimateHamiltonianConstraint}.

On the other hand, from the triangle inequality, Lemma~\ref{lemma: products in L2}
and Assumption~\ref{ass:Bootstrap}, 
\begin{equation*}
  \begin{split}
    t^{A+1}\|\dek{IJ}\cdot\dek{IJ}\|_{H^{k_1}(\S_t)} \leq & C t^{A+1} \textstyle{\sum}_{I,J} \| \dek{IJ} \|_{H^{k_1}(\S_t)}
    \| \dek{IJ} \|_{C^{0}(\S_t)}\\
    \leq & C \ar t^{A}\|\dek{}\|_{H^{k_{1}}(\S_t)}.
  \end{split}
\end{equation*}
The term $t^{A+1}\|(\dep{0})^{2}\|_{H^{k_1}(\S_t)}$ can similarly be bounded by $C \ar t^{A}\|\dep{0}\|_{H^{k_{1}}(\S_t)}$.
Next, we make use of the first statement in Lemma~\ref{lemma:noGNS} with $\eta = \s/4$ and  
\begin{itemize}
    \item $\varphi_I = \bq_I$ for $I \in \{1,..,n\}$ and $\varphi_{n+1} = \bP_1$,
    \item $\psi_I = \dek{II}$ (no summation over $I$) for $I \in \{1,\dots,n\}$ and $\psi_{n+1}=\dep{0}$.
\end{itemize} 
This yields
\begin{equation*}
  \begin{split}
    & t^{A+1} \|\chk_{IJ}\cdot\dek{IJ}
    + \d_t \chphi\cdot \dep{0} \|_{H^{k_1}(\S_t)} \\[2pt]
    = & t^{A} \big\| \textstyle{\sum}_{I} \bq_I\cdot \dek{II}
    + \bP_1\cdot \dep{0} \big\|_{H^{k_1}(\S_t)} \\
    \leq & t^{A} \Big( \s/4 +  \big\| \textstyle{\sum}_{I} 
    \bq_I^2 + \bP_1^2 \big \|_{C^{0}(\S)}^{1/2} \Big)
    \big(\| \dek{} \|_{H^{k_1}(\S_t)}^2 
    + \|\dep{0}\|_{H^{k_1}(\S_t)}^2 \big)^{1/2}\\
    & + C \langle \s^{-1} \rangle^{k_1 -1} t^{A} \left \langle \textstyle{\sum}_{I} 
    \|\bq_I\|_{H^{k_1}(\S)}
    + \|\bP_1\|_{H^{k_1}(\S)} \right \rangle^{k_1} 
     \big(\| \dek{} \|_{C^{k_{0}}(\S_t)}
    + \| \dep{0} \|_{C^{k_{0}}(\S_t)} \big) \\
    \leq & \Big( 1 + \s/2 \Big) t^{A}
    \Big(\| \dek{} \|_{H^{k_1}(\S_t)}^2 
    + \| \dep{0} \|_{H^{k_1}(\S_t)}^2\Big)^{1/2}
    + C t^{-1 + 4\s}.
  \end{split}
\end{equation*}
In the last step, we appealed to (\ref{eq: gamma bar p bar assumption}), (\ref{eq: scalar field bar assumption}),
Lemma~\ref{lemma:EstimateHamiltonianConstraint} and Assumption~\ref{ass:Bootstrap}; note that by
(\ref{eq:deviationfromKasner}) and Sobolev embedding, there is a standard constant $\tau_{H}<1$,
such that if $t_0 \leq \tau_{H}$, then (\ref{eq:BoundEigenvalues}) holds. This concludes the proof.
\end{proof}

The next lemma yields the main estimate for the higher-order norms of the lapse.
\begin{lemma} \label{lemma:HigherOrderLapseEstimate}
  Let $\s_p$, $\s_V$, $\s$, $k_0$, $k_1$, $(\S, h_{\refer})$, $(E_{i})_{i=1}^{n}$ and $V$ be as in Theorem~\ref{thm: big bang formation}
  and let $\rho_0 >0$. Assume that there are $t_{0}\leq \tau_{H}$ (see Lemma~\ref{lemma:AuxiliaryLapseEstimate}), $t_\rob<t_0$ and
  $\ar \in (0, \tfrac{1}{6n}]$ such that Assumption~\ref{ass:Bootstrap} is satisfied for this choice of parameters.
  Then, for any $t \in [t_{\rob},t_{0}]$,
  \begin{equation}\label{eq:EstLapseHigh}
    \begin{split}
      \mbH_{(N)} &\leq \big(C \ar + 2 + \s \big)
      t^{A+1} \big(\| \dek{} \|_{H^{k_1}(\S_t)}^2
      + \| \dep{0} \|_{H^{k_1}(\S_t)}^2\big)^{1/2} + C t^{2\s} (\mbD + t^{\s}).
    \end{split}
  \end{equation}  
\end{lemma}

\begin{proof}
  Applying $E_{\bfI}$ to (\ref{eq: alternative lapse}) and multiplying the result with $E_{\bfI}(N-1)$ yields
  \begin{equation*}
    \begin{split}
      & e_I E_\bfI e_I (N-1) E_{\bfI} (N-1) + [E_\bfI, e_I] e_I (N-1) E_{\bfI}(N-1)
      - t^{-2} \big(E_{\bfI}(N-1)\big)^2 \\
      = &  E_{\bfI} \Big( \g_{JII} e_{J} (N) - N \big\{ t^{-2} - k_{IJ} k_{IJ} - \dtp \dtp + \tfrac{2}{n-1} V \circ \phi \big\} \Big)
        E_{\bfI}(N-1).
    \end{split}
  \end{equation*}  
  Integrating this expression over $\S_t$, partially integrating the first term, using (\ref{eq:Xfint}),
  multiplying by $t^{2(A+1)}$ and summing over $|\bfI| \leq k_1$ yields
  \begin{equation*}
    \begin{split}
       -\mbH_{(N)}^2 (t)
      = & t^{2(A+1)} \langle \g_{JII} e_J (N)
      - N\big \{t^{-2} - k_{IJ} k_{IJ} - \dtp \dtp + \tfrac{2}{n-1} V \circ \phi \big\},
      N-1 \rangle_{H^{k_1}(\S_t)} \\
      & - t^{2(A+1)} \textstyle{\sum}_{|\bfI| \leq k_1} \textstyle{\int}_{\S_t}
      \big([E_\bfI, e_I] e_I (N-1) E_{\bfI}(N-1) 
      + E_{\bfI} e_I (N-1) [E_{\bfI}, e_I] (N-1) \\
      &\phantom{- t^{2(A+1)} \textstyle{\sum}_{|\bfI| \leq k_1} \textstyle{\int}_{\S_t}\big(}
      - E_{\bfI} e_I (N-1) E_\bfI(N-1) \rodiv_{h_\refer} (e_I)
      \big)\mu_{h_{\refer}}.
    \end{split}
  \end{equation*}
  Considering the first term on the right hand side, it suffices, using the Cauchy-Schwarz inequality for the $H^{k_1}$-inner product,
  to establish an upper bound for 
  \begin{equation*}
    t^{A+2} \| \g_{JII} e_{J} (N) - N\big \{t^{-2} - k_{IJ} k_{IJ} - \dtp \dtp + \tfrac{2}{n-1} V \circ \phi \big\} \|_{H^{k_1}(\S_t)}.
  \end{equation*}
  Appealing to Lemmata~\ref{lemma:LowerOrderHamiltonianConstraint},
  \ref{lemma:LowerOrderLapseEstimates} and \ref{lemma:AuxiliaryLapseEstimate} and Moser estimates,
  \begin{equation*}
    \begin{split}
      t^{A+2} &\|  N \big( t^{-2} - k_{IJ} k_{IJ} - \dtp \dtp \big) \|_{H^{k_1}(\S_t)} \\
      &\leq t^{A+2} \big(1+ C \|N-1\|_{C^{0}(\S_t} \big)
      \| t^{-2} - k_{IJ} k_{IJ} - \dtp \dtp \|_{H^{k_1}(\S_t)} \\
      &\quad + C t^{A+2} \| N-1 \|_{H^{k_1}(\S_t)}
      \| t^{-2} - k_{IJ} k_{IJ} - \dtp \dtp \|_{C^{0}(\S_t)} \\
      &\leq \big(C \ar + 2 + \s \big) t^{A+1}
      \big(\| \dek{} \|_{H^{k_1}(\S_t)}^2 
      + \| \dep{0} \|_{H^{k_1}(\S_t)}^2\big)^{1/2} + C t^{\s} (\mbD + t^{\s})
    \end{split}
  \end{equation*}
  for $t_0 \leq \tau_{H}$. The term involving the potential we can simply bound as
  \begin{align*}
    t^{A+2} \tfrac{2n}{n-1} \|N V \circ \phi \|_{H^{k_1}(\S_t)} 
    &\leq Ct^{5\s}
  \end{align*}
  as $l_2 = 1$, $l_{11} =1$, $m_{1}=2, m_{\s} =5, m_D = 0$ and $\max_{i} S(j_i) = A$. Next, 
  \begin{align*}
    t^{A+2} \| \g_{JII} e_J (N) \|_{H^{k_1}(\S_t)} 
    &\leq C t^{2\s} (\mbD + t^{\s}),
  \end{align*}
  as $l_3 = 1$, $l_7 = 1$, $m_{1} = 2$, $m_{\s} = 6$, $m_D = 2$ and $\max_{i} S(j_i) = S(3) = A+4\s$. Next,
  \begin{equation*}
    \begin{split}
      & t^{A+2} \| [E_\bfI, e_I] e_I (N-1) \|_{L^{2}(\S_t)} \\
      \leq & C t^{A+2} \big( \| e \|_{H^{k_1}(\S_t)} \| \ear(N) \|_{C^{1}(\S_t)}
      + \| e \|_{C^{1}(\S_t)}  \| \ear(N) \|_{H^{k_1}(\S_t)} \big)\leq C t^{3\s} (\mbD + t^{\s}).
    \end{split}
  \end{equation*}  
  Here we use the fact that $[E_{\bfI},e_I]$ is a differential operator of order $\max\{|\bfI|,1\}$,
  with at most $|\bfI|$ derivatives falling on the frame components $e^i_I$. Along with the Moser-type estimates
  from Lemma~\ref{lemma: products in L2}, Assumption~\ref{ass:Bootstrap} and Lemma~\ref{lemma:estimatingthedynamicalvariables},
  this yields the desired estimate. We may use this to bound
  \begin{equation*}
    \begin{split}
      & t^{2(A+1)}\big| \textstyle{\sum}_{|\bfI|\leq k_1} \textstyle{\int}_{\S_t}
      [E_{\bfI},e_I] e_I(N-1) E_{\bfI}(N-1)  \mu_{h_{\refer}} \big| \\
    \leq & t^{A+2} \big(\textstyle{\sum}_{|\bfI|\leq k_1}\|[E_{\bfI},e_I] e_I(N-1)\|_{L^2(\S_t)}^2\big)^{1/2} t^{A} \|N-1\|_{H^{k_1}(\S_t)} \\
    \leq & Ct^{3\s}(\mbD + t^{\s}) \mbH_{(N)}.
    \end{split}
  \end{equation*}
  Similarly,
  \begin{equation*}
    \begin{split}
      & t^{A+1} \| [E_\bfI, e_I] (N-1) \|_{L^{2}(\S_t)} \\
      \leq & C t^{A+1} \big( \| e \|_{H^{k_1}(\S_t)} \| N-1 \|_{C^{1}(\S_t)}
      + \| e \|_{C^{1}(\S_t)}  \| N-1 \|_{H^{k_1}(\S_t)} \big) 
      \leq C t^{3\s} (\mbD + t^{\s}),
    \end{split}
  \end{equation*}  
  so that 
  \begin{equation*}
    \begin{split}
      & t^{2(A+1)} \big| \textstyle{\sum}_{|\bfI|\leq k_1} \textstyle{\int}_{\S_t}E_{\bfI} e_I(N-1) [E_{\bfI}, e_I] (N-1)  \mu_{h_{\refer}} \big| \\
      \leq & t^{A+1} \big(\textstyle{\sum}_{|\bfI|\leq k_1}\sum_{I} \|[E_{\bfI},e_I] (N-1)\|_{L^2(\S_t)}^2\big)^{1/2}t^{A+1} \| \ear(N) \|_{H^{k_1}(\S_t)}\\
      \leq & Ct^{3\s}(\mbD + t^{\s}) \mbH_{(N)}.
    \end{split}
  \end{equation*}
  Appealing to Lemma~\ref{lemma:estimatingthedynamicalvariables} yields
  \begin{equation*}
    \begin{split}
      & t^{2(A+1)} \big| \textstyle{\sum}_{|\bfI|\leq k_1} \textstyle{\int}_{\S_t}   E_{\bfI} e_I (N-1) E_{\bfI}(N-1) \rodiv_{h_\refer}(e_I)  \mu_{h_{\refer}} \big| \\
      \leq & t^{2(A+1)} \textstyle{\sum}_{I} \|\rodiv_{h_{\refer}}(e_I)\|_{C^{0}(\S_t)}\|N-1\|_{H^{k_1}(\S_t)} \|\en \|_{H^{k_1}(\S_t)} 
      \leq Ct^{3\s}(\mbD + t^{\s}) \mbH_{(N)}^2.
    \end{split}
  \end{equation*}
  Combining the above estimates yields
  \begin{equation*}
    \begin{split}
      \mbH_{(N)}(t)^2
      \leq &  \big(C \ar + 2 + \s \big)t^{A+1}\big(\| \dek{} \|_{H^{k_1}(\S_t)}^2+ \| \dep{0} \|_{H^{k_1}(\S_t)}^2\big)^{1/2} \mbH_{(N)}(t) \\
      & + C t^{2\s} (\mbD + t^{\s})\mbH_{(N)}(t) ,
    \end{split}
  \end{equation*}  
  since $t^{A} \| N-1 \|_{H^{k_1}(\S_t)}\leq \mbH_{(N)}(t)$. The lemma follows. 
\end{proof}

\subsection{The components of the frame and co-frame}
For the frame component $e_I^i$ the relevant differential operator to consider is 
$\partial_t + t^{-1} \bq_I$. Indeed, the evolution equations for the components of the frame
and co-frame can be written
\begin{subequations}\label{seq:EveECLow}
  \begin{align}
    \big( -\tfrac{\bq_{\uI}}{t} - \d_t \big) \big(\dee{\uI}{i}\big)
    = & (N-1) k_{IM} e^{i}_{M} + \dek{IM}\cdot \che^{i}_M + \dek{IM}\cdot\dee{M}{i},\label{eq:EveFCeLow}\\
    \big(\tfrac{\bq_{\uI}}{t} - \d_t \big)
    \big(\deo{\uI}{i}\big) = & - (N-1) k_{IM} \o_{i}^{M} - \dek{IM}\cdot\deo{M}{i}- \dek{IM}\cdot\deo{M}{i};\label{eq:EveFComLow}
  \end{align}
\end{subequations}
recall Subsection~\ref{ssection:notation}: we do not sum over underlined indices.
\subsubsection{Lower-order estimates for the frame coefficients}
We continue by establishing lower-order
estimates for the time derivative of the components of the frame and of the co-frame.
Here we require pointwise estimates as opposed to estimates in the $C^{k_{0}+1}$-norm,
due to the nature of the argument for the lower-order energy estimate.
\begin{lemma}\label{lemma:EstFCLow}
  Let $\s_p$, $\s_V$, $\s$, $k_0$, $k_1$, $(\S, h_{\refer})$, $(E_{i})_{i=1}^{n}$ and $V$ be as in Theorem~\ref{thm: big bang formation}
  and let $\rho_0 >0$. Assume that there are $t_{0}\leq 1$, $t_\rob<t_0$ and $\ar \in (0, \tfrac{1}{6n}]$ such that
  Assumption~\ref{ass:Bootstrap} is satisfied for this choice of parameters. Then
  for any $i,I$, $|\bfI| \leq k_0+1$, $t \in [t_{\rob}, t_0]$ and $x \in \S$,
  \begin{subequations}
    \begin{align}
      & t^{1-3\s}
      \big| E_{\bfI} \big(\big(-\tfrac{ \bq_{\uI}}{t} - \d_t \big)
      (\dee{\uI}{i}) \big) (t,x) \big|\\
      \leq & C\ar t^{-1 +(1- 3\s)} \textstyle{\sum}_{M}  \sum_{|\bfJ| \leq |\bfI|} 
      \left|E_{\bfJ} (\dee{M}{i}) (t,x) \right|
      + C t^{-1+\s} \left( \mbD(t) + t^{\s} \right), \nonumber \\[2pt] 
      & t^{1-3\s}
      \big| E_{\bfI} \big(\big(\tfrac{\bq_{\uI}}{t} - \d_t \big)
      (\deo{\uI}{i}) \big) (t,x) \big| \\
      \leq & C\ar t^{-1 +(1- 3\s)} \textstyle{\sum}_{M}  \textstyle{\sum}_{|\bfJ| \leq |\bfI|}
      \big|E_{\bfJ} (\deo{M}{i}) (t,x) \big|+ C t^{-1+\s} \left( \mbD(t) + t^{\s} \right). \nonumber
    \end{align}
\end{subequations}
\end{lemma}
\begin{proof}
  Since the proofs of the two estimates are very similar, we here only bound the three terms on the right-hand side of
  \eqref{eq:EveFCeLow}. It suffices to bound the first two terms on the right-hand side of (\ref{eq:EveFCeLow}) in $C^{k_{0}+1}$.
  However, 
  \begin{align} \label{eq:LowerOrderFrameFirstBound}
    t^{1-3\s}\|(N-1)k_{IM}e^{i}_M\|_{C^{k_{0}+1}(\S_t)}
    &\leq Ct^{-1+\s}(\mbD + t^{\s}),
  \end{align}
  as $l_1=1$, $l_5=1$, $l_{10}=1$, $m_1 = 2$, $m_{\s} = 4$ and $m_D =1$. Next
  \begin{align} \label{eq:LowerOrderFrameSecondBound}
    t^{1-3\s}\|\dek{IM}\cdot\che^{i}_M\|_{C^{k_{0}+1}(\S_t)}
    &\leq Ct^{-1+\s}(\mbD + t^{\s}),
  \end{align}
  as $l_4 =1$, $l_8=1$, $m_1 = 2$, $m_{\s} = 4$ and $m_D =1$. For the last term, we have to be more careful and obtain pointwise estimates.
  However, due to \eqref{eq:BootstrapInequality}, we may simply estimate that $|E_{\bfJ}(\dek{IM})| \leq t^{-1} \mbL_{(\g,k)} \leq t^{-1} \ar$
  for any $|\bfJ| \leq k_0+1$. In particular, for $|\bfI|\leq k_0+1$, 
  \begin{align*}
    t^{1-3\s}
    \big|E_{\bfI}\big(\dek{IM}\cdot\dee{M}{i}\big) (t,x) \big|
    &\leq C\ar t^{-1+(1-3\s)} \textstyle{\sum}_{M}  \textstyle{\sum}_{|\bfJ| \leq |\bfI|}
    \big|E_{\bfJ}\big(\dee{M}{i}\big) (t,x) \big|,
  \end{align*}
  where we appealed to Lemma~\ref{lemma:EstimateCommutators} in the Appendix. This concludes the proof.
\end{proof}

\subsubsection{Higher-order estimates for the frame coefficients}

Let us continue with the higher-order estimates.
\begin{lemma}\label{lemma:EstFChigh}
  Let $\s_p$, $\s_V$, $\s$, $k_0$, $k_1$, $(\S, h_{\refer})$, $(E_{i})_{i=1}^{n}$ and $V$ be as in Theorem~\ref{thm: big bang formation}
  and let $\rho_0 >0$. Assume that there are $t_{0}\leq 1$, $t_\rob<t_0$ and $\ar \in (0, \tfrac{1}{6n}]$ such that
  Assumption~\ref{ass:Bootstrap} is satisfied for this choice of parameters. Then, for any $i,I$ and $t \in [t_{\rob}, t_0]$,
  \begin{subequations}
    \begin{align} 
      t^{A+1-2\s} \big \|\big(\tfrac{\bq_{\uI}}{t} + \d_t\big)
      \big(\dee{\uI}{i}\big) \big\|_{H^{k_1}(\S_t)}
      & \leq C \ar  t^{-1} \mbH_{(e,\o)} (t)
      + C t^{-1+\s} \big( \mbD(t) + t^{\s} \big), \label{eq:EstFChighe} \\[2pt]
      t^{A+1-2\s} \big \|\big(\tfrac{\bq_{\uI}}{t} - \d_t\big) 
      \big(\deo{\uI}{i}\big) \big\|_{H^{k_1}(\S_t)}
      & \leq C \ar  t^{-1} \mbH_{(e,\o)} (t)
      + C t^{-1+\s} \big( \mbD(t) + t^{\s} \big).
    \end{align}
  \end{subequations}  
\end{lemma}
\begin{proof}
  Again, the proofs are similar, and we here only bound the three terms on the right-hand side of (\ref{eq:EveFCeLow}).
  For the first and second term, we again make use of the scheme of Subsection~\ref{sec:scheme}. In fact, we may simply
  make use of (\ref{eq:LowerOrderFrameFirstBound}) and (\ref{eq:LowerOrderFrameSecondBound}), multiplied on the left by $t^{A+\s}$
  and on the right by $t^{A+\s-\max_i S(j_i)}$. In particular, 
\begin{align*}
t^{A+1-2\s} \| (N-1) k_{IM} e^{i}_M \|_{H^{k_1}(\S_t)}  
    &\leq C t^{-1+\s} (\mbD + t^{\s}),
\intertext{ as $l_1 = 1, l_5 = 1, l_{10} = 1$, so $\max_i S(j_i) = S(1) = A + \s$, and}
t^{A+1-2\s} \| (k_{IM} - \chk_{IM}) \che^i_{M} \|_{H^{k_1}(\S_t)} 
    &\leq C t^{-1+2\s}(\mbD + t^\s),
\end{align*}
as $l_4 = 1, l_8 = 1$, so here $\max_{i} S(j_i) = S(10) = A$. Finally,
\begin{equation*}
  \begin{split}
    & t^{A+1-2\s} \| \dek{IM}\cdot\dee{M}{i} \|_{H^{k_1}(\S_t)} \\[2pt]
    \leq & C t^{A+1-2\s} \big( \| \dek{} \|_{H^{k_1}(\S_t)} 
        \| \dee{}{} \|_{C^{0}(\S_t)}  + \| \dek{} \|_{C^{0}(\S_t)} \| \dee{}{} \|_{H^{k_1}(\S_t)} \big) \\[2pt]
    \leq & C t^{-1+ \s} \mbH_{(\g, k)} \cdot \mbL_{(e,\o)}  
        + C t^{-1} \mbL_{(\g, k)} \cdot \mbH_{(e,\o)} \leq C  t^{-1 + \s} \mbD + C \ar t^{-1} \mbH_{(e,\o)},
  \end{split}
\end{equation*}
since $\mbL_{(\g,k)} \leq \ar$ and $\mbD \leq 1$ due to \eqref{eq:BootstrapInequality}.
This concludes the proof.
\end{proof}

\begin{remark}
Contrasting our estimates with those of \cite{GIJ},
we note that we do not need to appeal to the estimates for the lapse,
only to the bootstrap assumptions.
This is due to the different decay assumptions on the frame components at low order
compared to the bootstrap assumptions appearing in \cite{GIJ}.
For that reason we do not get a term of the form $c_* \mbH_{(\g,k)}$.
\end{remark}

\subsubsection{The energy estimate for the frame and co-frame coefficients}
With the higher-order estimates in hand,
we can proceed with proving the relevant energy estimate.

\begin{prop} [Energy estimate for $e$ and $\o$] \label{prop:EnergyFC}
  Let $\s_p$, $\s_V$, $\s$, $k_0$, $k_1$, $(\S, h_{\refer})$, $(E_{i})_{i=1}^{n}$ and $V$ be as in Theorem~\ref{thm: big bang formation}
  and let $\rho_0 >0$. Assume that there are $t_{0}\leq 1$, $t_\rob<t_0$ and $\ar \in (0, \tfrac{1}{6n}]$ such that
  Assumption~\ref{ass:Bootstrap} is satisfied for this choice of parameters. Then, for any $t \in [t_{\rob}, t_0]$,
  \begin{align}
    \begin{split}
      \label{eq:EnergyFC}
      \mbH_{(e,\o)}(t)^2
      & \leq \mbH_{(e,\o)}(t_0)^2
      + \big( C \ar - 4 \s - 2 A\big) 
      \textstyle{\int}_{t}^{t_0} s^{-1} \mbH_{(e,\o)}(s)^2 \md s \\
      &\quad + C \textstyle{\int}_{t}^{t_0} s^{-1+\s} \mbD(s) (\mbD(s) + s^{\s}) \md s.
    \end{split}
  \end{align}  
\end{prop}
\begin{proof}
We start by noting that 
\begin{equation*}
  \begin{split}
    & -\d_t \big(\| \dee{}{} \|_{H^{k_1}(\S_t)}^2
    + \| \deo{}{} \|_{H^{k_1}(\S_t)}^2 \big) \\
    = & 2 \textstyle{\sum}_{i,I}  \big[
    \big \langle \big(-\tfrac{\bq_I}{t}-\d_t\big)
    (\dee{I}{i}), \dee{I}{i}\big\rangle_{H^{k_1}(\S_t)} 
    +  \big \langle \tfrac{\bq_I}{t}
    \dee{I}{i}, \dee{I}{i} \big\rangle_{H^{k_1}(\S_t)} \big] \\
    & + 2 \textstyle{\sum}_{i,I} \big[ 
    \big \langle \big(\tfrac{\bq_I}{t}-\d_t\big)(\deo{I}{i}),
    \deo{I}{i} \big \rangle_{H^{k_1}(\S_t)}
   - \big \langle \tfrac{\bq_I}{t}(\deo{I}{i}),\deo{I}{i} \big \rangle_{H^{k_1}(\S_t)} \big].
  \end{split}
\end{equation*}
The first terms in the parentheses can be estimated by appealing to Lemma~\ref{lemma:EstFChigh}.
The second terms can be bounded by using Lemma~\ref{lemma:noGNS} (for the sum over $I$ only) with $\eta = \s$.
In particular, 
\begin{equation*}
  \begin{split}
    & \textstyle{\sum}_{i}
    \big| \textstyle{\sum}_{I} \big  \langle \tfrac{\bq_I}{t}
    (\dee{I}{i}), \dee{I}{i} \big\rangle_{H^{k_1}(\S_t)} \big| \\
    \leq & t^{-1} \big( \s + \textstyle{\max}_{I} \| \bq_I \|_{C^{0}(\S_t)} \big)
    \| \dee{}{} \|_{H^{k_1}(\S_t)}^2 \\
    & + C \langle \s^{-1} \rangle^{k_1 -1} t^{-1} \big\langle 
    \textstyle{\sum}_{I} \| \bq_I \|_{H^{k_1}(\S_t)} \big \rangle^{k_1}
    \|\dee{}{} \|_{C^{k_{0}}(\S_t)} \|\dee{}{} \|_{H^{k_1}(\S_t)} \\
    \leq & t^{-1} \big(1 - 4 \s \big) \|\dee{}{} \|_{H^{k_1}(\S_t)}^2 
    + C t^{-1} \| \dee{}{} \|_{C^{k_{0}}(\S_t)} \|\dee{}{} \|_{H^{k_1}(\S_t)}
  \end{split}
\end{equation*}
due to (\ref{eq: gamma bar p bar assumption}) and Lemma~\ref{le: first bound on $q_I$}. Thus
\begin{equation*}
  \begin{split}
    & - t^{2(A+1-2\s)} \d_t \big(\| \dee{}{} \|_{H^{k_1}(\S_t)}^2 + \| \deo{}{} \|_{H^{k_1}(\S_t)}^2 \big) \\
    \leq & \big(C \ar + 2(1-4\s) \big) t^{-1} \mbH_{(e,\o)}^2
    + C t^{-1+\s} \mbD(\mbD + t^{\s}).
  \end{split}
\end{equation*}
Commuting $t^{2(A+1-2\s)}$ with $\d_t$ and integrating in time from $t$ to $t_0$ yields the result.
\end{proof}

\subsection{The structure coefficients and the second fundamental form}
In this section, we derive estimates for the structure coefficients and the second fundamental form.
The inhomogeneity of the initial data leads to additional complications in comparison to the estimates
shown in \cite{GIJ}. These complications require additional regularity assumptions for their resolution.

\subsubsection{Lower-order estimates for the structure coefficients}
For the structure coefficient $\g_{IJK}$ the relevant differential operator is
$\partial_t + t^{-1}(\bq_I + \bq_J - \bq_K)$.
We start by writing the evolution equation for $\chga$ as 
\begin{align} \label{eq:dtChga}
   \big(-\tfrac{ \bq_{\uI} + \bq_{\uJ} - \bq_{\uK}}{t} - \d_t \big) 
        ( \chga_{\uI \uJ \uK} ) 
    &= \che_I(\chk_{JK}) - \che_J(\chk_{IK}),
\end{align}
see (\ref{eq:formula for chgaIJK}). Combining this observation with \eqref{eq:concoef} yields
\begin{equation}\label{eq:EveSCLow}
  \begin{split} 
    & \big( -\tfrac{\bq_{\uI} + \bq_{\uJ} - \bq_{\uK}}{t} - \d_t \big)(\dega{\uI \uJ \uK}) \\
    = &  (N-1)  \left( k_{IM} \g_{MJK} + k_{JM} \g_{IMK} - k_{KM} \g_{IJM} \right) + \dek{IM}\cdot\dega{MJK} \\
    &+ \dek{JM}\cdot\dega{IMK}- \dek{KM}\cdot\dega{IJM} - \dek{KM}\cdot\chga_{IJM}+ \dek{IM}\cdot \chga_{MJK} \\
    & + \dek{JM}\cdot \chga_{IMK}+ 2 e_{[I} (N) k_{J]K} + 2 N e_{[I}(k_{J]K}) - 2 \che_{[I} (\chk_{J]K}).
  \end{split}
\end{equation}

We continue by establishing pointwise estimates.
\begin{lemma}\label{lemma:EstSCLow}
  Let $\s_p$, $\s_V$, $\s$, $k_0$, $k_1$, $(\S, h_{\refer})$, $(E_{i})_{i=1}^{n}$ and $V$ be as in Theorem~\ref{thm: big bang formation}
  and let $\rho_0 >0$. Assume that there are $t_{0}\leq 1$, $t_\rob<t_0$ and $\ar \in (0, \tfrac{1}{6n}]$ such that
  Assumption~\ref{ass:Bootstrap} is satisfied for this choice of parameters. Then,
  for any $I,J,K$, $|\bfI| \leq k_0$,  $t \in [t_{\rob}, t_0]$ and $x \in \S$,
  \begin{equation}\label{eq:EstSCLow}
    \begin{split} 
      & t^{1-2\s}
      \big| E_{\bfI} \big(\big( -\tfrac{ \bq_{\uI} + \bq_{\uJ} - \bq_{\uK}}{t} - \d_t \big)
      (\dega{\uI \uJ \uK}) \big)  (x,t) \\
      \leq & C \ar t^{-1 + (1-2\s)} \textstyle{\sum}_{M,L,N} \textstyle{\sum}_{|\bfJ| \leq |\bfI|} 
      \left|{E_{\bfJ} (\dega{MLN}}) (x,t) \right|
      + C t^{-1+\s} \left( \mbD(t) + t^{\s} \right).
    \end{split}
  \end{equation}  
\end{lemma}
\begin{proof}
  We begin by estimating, in $C^{k_0}$, the first appearance of each of the following five types of terms,
  found on the right-hand side of (\ref{eq:EveSCLow}): $(N-1) \cdot k \cdot \g$, $\dek{}\cdot  \chga$,
  $\ear(N) \cdot k$, $N \ear(k)$ and $\cear(\chk)$.
  Terms of a given type satisfy the same estimate due to the scheme of Subsection~\ref{sec:scheme}.
  On the other hand, the terms of the form $\dek{}\cdot \dega{}$ are handled differently, using a pointwise estimate.

  Before we continue, note that due to Assumption~\ref{ass:Bootstrap} and the norms appearing in (\ref{eq:BootstrapInequality}),
  we can bound up to $k_0+1$ derivatives of $e$, $\che$, $k$, $\chk$ and $N-1$. As $|\bfI| \leq k_0$ in (\ref{eq:EstSCLow}), we do not
  need to appeal to Lemma~\ref{lemma:extraderivatives} here in order to deal with the additional derivatives
  due to the frame vector fields  $e_I$ and $\che_I$.
  In the language of Subsection~\ref{sec:scheme}, in the five estimates that follow it suffices
  if $m_1 = 2, m_\s \geq 3$ and $m_\s + m_D \geq 4$. First,
  \begin{equation} \label{eq:EstSCFirstBound}
    t^{1-2\s} \|(N-1) k_{IM} \g_{MJK})\|_{C^{k_{0}}(\S_t)}
    \leq Ct^{-1+\s}(\mbD(t) + t^{\s}),
  \end{equation}
  as $l_1 = 1$, $l_7 = 1$, $l_{10} = 1$, $m_1 = 2$, $m_{\s} = 3$ and $m_D = 1$. Second,
  \begin{equation} \label{eq:EstSCSecondBound}
    t^{1-2\s} \|\dek{IM}\cdot \chga_{MJK} \|_{C^{k_{0}}(\S_t)}
    \leq Ct^{-1+\s}(\mbD(t) + t^{\s}),
  \end{equation}
  as $l_6 = 1$, $l_8 = 1$, $m_1 =2$, $m_{\s} = 3$ and $m_D = 1$. Third,
  \begin{align*}
    t^{1-2\s} \|e_I(N) k_{JK} \|_{C^{k_{0}}(\S_t)}
    \leq Ct^{-1+2\s},
  \end{align*}
  as $l_3 = 1$, $l_{10} = 1$, $m_1 = 2$, $m_{\s} = 4$ and $m_D = 0$. Next, 
  \begin{equation} \label{eq:EstSCThirdBound}
    t^{1-2\s} \| N e_I(k_{JK})\|_{C^{k_{0}}(\S_t)}
    \leq Ct^{-1+\s}(\mbD(t) + t^{\s}),
  \end{equation}
  as $l_2=1$, $l_5 = 1$, $l_{10} = 1$, $m_1 = 2$, $m_{\s} = 3$ and $m_D = 1$. Finally,
  \begin{equation}\label{eq:EstSCFourthBound}
    t^{1-2\s}\|\che_I(\chk_{JK})\|_{C^{k_{0}}(\S_t)}\leq Ct^{-1+2\s},
  \end{equation}
  as $l_4 = 1$, $l_9 = 1$, $m_1 = 2$, $m_{\s} = 4$ and $m_D = 0$.

  The only remaining type of term to be estimated is of the form $\dek{}\cdot\dega{}$.
  Similarly to the end of the proof of Lemma~\ref{lemma:EstFCLow}, we estimate such
  terms pointwise by noting that, due to (\ref{eq:BootstrapInequality}),
  $|E_{\bfI}(\dek{IM})| \leq t^{-1} \ar$ for any $|\bfI| \leq k_0$. In particular, for $|\bfI|\leq k_0$
  \begin{equation*}
    \begin{split}
      t^{1-2\s} \left|E_{\bfI}(\dek{IM}\cdot\dega{MJK}) \right|(t,x)
      \leq & C \ar t^{-1+(1-2\s)} \textstyle{\sum}_{M,J,K}\sum_{|\bfJ| \leq |\bfI|}
      \left|E_{\bfJ}\left(\dega{MJK}\right) \right|(t,x),
    \end{split}    
  \end{equation*}
  see Lemma~\ref{lemma:EstimateCommutators} in the appendix. This concludes the proof.
\end{proof}

\subsubsection{Higher-order estimates for the structure coefficients}
Next, we establish higher-order estimates for the structure coefficients.
However, we cannot handle terms of the form $N \ear(k)$ using the a-priori estimates
of Lemma~\ref{lemma:estimatingthedynamicalvariables}; they have too many derivatives
falling on $k$. This issue is resolved in tandem with the energy estimates themselves.
\begin{lemma}\label{lemma:EstimateSCHigh}
  Let $\s_p$, $\s_V$, $\s$, $k_0$, $k_1$, $(\S, h_{\refer})$, $(E_{i})_{i=1}^{n}$ and $V$ be as in Theorem~\ref{thm: big bang formation}
  and let $\rho_0 >0$. Assume that there are $t_{0}\leq \tau_{H}$ (see Lemma~\ref{lemma:AuxiliaryLapseEstimate}), $t_\rob<t_0$ and
  $\ar \in (0, \tfrac{1}{6n}]$ such that Assumption~\ref{ass:Bootstrap} is satisfied for this choice of parameters.
  Then, for any $I,J,K$ and any $t \in [t_{\rob}, t_0]$,
  \begin{equation}\label{eq:EstSChigh}
    \begin{split}
      & t^{A+1} \big \| \big( -\tfrac{\bq_{\uI} + \bq_{\uJ} - \bq_{\uK}}{t}
      - \d_t \big) (\dega{\uI \uJ \uK}) 
      - 2 e_{[I}(N) \chk_{J]K} - 2 N e_{[I} k_{J]K} \big\|_{H^{k_1}(\S_t)} \\[2pt]
      \leq & C \ar t^{-1} 
      \big(\mbH_{(\g,k)} (t) + \mbH_{(\phi)}(t) \big)
      + C t^{-1+\s} \left( \mbD(t) + t^{\s} \right).
    \end{split}
  \end{equation}
\end{lemma}
\begin{proof}
  The logic of the proof is to subtract the terms $2e_{[I}(N)\chk_{J]K}$ and  $2Ne_{[I}(k_{J]K})$
  from both sides of (\ref{eq:EveSCLow}) and to estimate the resulting right hand side. 
  As we do so, the term $2e_{[I}(N)\dek{J]K}$ appears on the right-hand side.
  However, this term can be bounded using Lemma~\ref{lemma:HigherOrderLapseEstimate}
  and Assumption~\ref{ass:Bootstrap}. For example,
  \begin{equation*}
    \begin{split}
      & t^{A+1} \| e_I(N) \dek{JK} \|_{H^{k_1}(\S_t)} \\
      \leq & C t^{A+1} \big( \|\en \|_{H^{k_1}(\S_t)} 
      \| \dek{} \|_{C^{0}(\S_t)} + \| \en \|_{C^{0}(\S_t)}\| \dek{} \|_{H^{k_1}(\S_t)} \\
      \leq & C \ar t^{-1} \big(\mbH_{(\g,k)} + \mbH_{(\phi)} \big) + Ct^{-1+2\s}(\mbD + t^{\s}).
    \end{split}
  \end{equation*}
  In order to bound terms of the form $(N-1) \cdot k \cdot \g$, $\dek{} \cdot \chga$ and $\cear(\chk)$ in $H^{k_1}$,
  we use the scheme of Subsection~\ref{sec:scheme}, alongside (\ref{eq:EstSCFirstBound})--(\ref{eq:EstSCFourthBound})
  multiplied on the left-hand side by $t^{A+2\s}$ and on the right-hand side by $t^{A+2\s - \max_{i} S(j_i)}$. To begin,
  \begin{align*}
    t^{A+1} \| (N-1) k_{IM} \g_{MJK} \|_{H^{k_1}(\S_t)}
    &\leq C t^{-1+\s} (\mbD + t^{\s})
  \end{align*}
  due to \eqref{eq:EstSCFirstBound} and $\max_{i} S(j_i)= S(7) = A+2\s$. Next,
  \begin{align*}
    t^{A+1} \|  \dek{IM}\cdot \chga_{MJK} \|_{H^{k_1}(\S_t)} &\leq C t^{-1+3\s} (\mbD + t^{\s})
  \end{align*}
  due to (\ref{eq:EstSCSecondBound}) and $\max_{i} S(j_i) = S(8) = A$. Finally,
  \begin{align*}
    t^{A+1} \|\che_I (\chk_{JK}) \|_{H^{k_1}(\S_t)}
    &\leq C t^{A-1+4\s},
  \end{align*}
  due to (\ref{eq:EstSCFourthBound}) and $\max_i S(j_i)= 0$; note that bounding up to $k_1+1$ derivatives of
  $\chk$ in $L^{2}$ is allowed in the a-priori estimates of Lemma~\ref{lemma:estimatingthebackground}.

  Finally, making use of (\ref{eq:BootstrapInequality}) and Lemma~\ref{lemma: products in L2}, 
  \begin{equation*}
    \begin{split}
      & t^{A+1} \| \dek{IM}\cdot \dega{MJK} \|_{H^{k_1}(\S_t)} \\
      \leq & C t^{A+1} \big(\| \dek{} \|_{H^{k_1}(\S_t)} \| \dega{} \|_{C^{0}(\S_t)}
      + \| \dek{} \|_{C^{0}(\S_t)} \| \dega{} \|_{H^{k_1}(\S_t)} \big) 
      \leq C \ar t^{-1} \mbH_{(\g,k)}.
    \end{split}
  \end{equation*}
  Estimates for similar terms are similar. This concludes the proof.
\end{proof}

\subsubsection{Lower-order estimates for the second fundamental form}
The relevant differential operator for the components of $k$ is $\partial_t + \frac{1}{t}$. We thus write the evolution equations as
\begin{equation}\label{eq:EveFF}
  \begin{split}
    \big( - \tfrac{1}{t} - \d_t \big)  (\dek{IJ})
    = & t^{-1}(N-1) k_{IJ} - e_{(I} e_{J)} (N)
    + \g_{K(IJ)} e_K(N) \\
    &   + N \big( e_K (\g_{K(IJ)}) + e_{(I} (\g_{J)KK})
    -  e_I (\phi) e_J (\phi) 
    - \tfrac{2}{n-1} (V \circ \phi)\delta_{IJ}  \\
    & - \g_{KLL} \g_{K(IJ)} - \g_{I(KL)} \g_{J(KL)} 
    + \tfrac{1}{4} \g_{KLI} \g_{KLJ} \big).
  \end{split}
\end{equation}
\begin{lemma} \label{lemma:EstFFLow}
  Let $\s_p$, $\s_V$, $\s$, $k_0$, $k_1$, $(\S, h_{\refer})$, $(E_{i})_{i=1}^{n}$ and $V$ be as in Theorem~\ref{thm: big bang formation}
  and let $\rho_0 >0$. Assume that there are $t_{0}\leq 1$, $t_\rob<t_0$ and $\ar \in (0, \tfrac{1}{6n}]$ such that
  Assumption~\ref{ass:Bootstrap} is satisfied for this choice of parameters. Then, for any $I,J$ and $t \in [t_{\rob}, t_0]$,
  \begin{equation}\label{eq:EstimateFFLow}
    \begin{split}
      t \left\| \big( - \tfrac{1}{t} -\d_t \big) (\dek{IJ}) \right\|_{C^{k_{0}+1} (\S_t)}
      \leq C t^{-1+2\s} \left( \mbD(t) + t^{\s} \right)^2.
    \end{split}
  \end{equation}  
\end{lemma}
\begin{proof}  
  Following the scheme of Subsection~\ref{sec:scheme}, 
  \begin{align*}
    t \left\|e_I e_J (N-1) \right\|_{C^{k_{0}+1} (\S_t)} 
    & \leq C t^{-1 + 6\s} (\mbD(t) + t^\s),
  \end{align*}
  as $l_3 = 1$, $l_5 = 1$, $l_{\roint} = 1$, $m_1 =2$, $m_{\s} = 6$ and $ m_D = 1$.
  To estimate $(N-1) k$, we appeal to Lemmata~\ref{lemma:estimatingthedynamicalvariables} and \ref{lemma:LowerOrderLapseEstimates}
  (the improved lower-order lapse estimates):
  \begin{align*}
    \| (N-1) k_{IJ} \|_{{C^{k_{0}+1} (\S_t)}}
    &\leq \| N-1 \|_{C^{k_{0}+1}(\S_t)}
    \| k \|_{C^{k_{0}+1}(\S_t)} \\
    &\leq C t^{ 2\s}(\mbD + t^{\s})^2 \cdot t^{-1} ( \mbL_{(\g,k)} + C) 
    \leq C t^{-1 + 2\s} (\mbD(t) + t^{\s})^2.
  \end{align*}
  Next, 
  \begin{align*}
    t \left\|N e_I (\g_{JKK}) \right\|_{C^{k_{0}+1}(\S_t)}
    & \leq C t^{-1 + 4\s} (\mbD(t) + t^\s)^2,
  \end{align*}
  as $l_2 =1$, $l_5 = 1$, $l_7 = 1$, $l_{\roint} = 1$, $m_1 = 2$, $m_{\s} = 4$ and $m_D = 2$;
  \begin{align*}
    t \left \| \g_{K(IJ)} e_K(N) \right \|_{C^{k_{0}+1}(\S_t)}
    &\leq C t^{-1 + 4\s} (\mbD(t) + t^{\s}),
  \end{align*}
  as $l_3 = 1$, $l_7 =1$, $l_{\roint} = 2$, $m_1 = 2$, $m_{\s} = 4$ and $m_D = 1$; and 
  \begin{align*}
    t \left\|N \g_{KLL} \g_{KIJ} \right\|_{C^{k_{0}+1} (\S _t)} 
    & \leq C t^{-1 + 2\s} (\mbD(t) + t^\s)^2,
  \end{align*}
  as $l_2 = 1$, $l_7 = 2$, $l_{\roint} = 2$, $m_1 = 2$, $m_{\s} = 2$ and $m_D = 2$.
  All terms of the form $N \cdot \g \cdot \g$ or $N \cdot \ep \cdot \ep$
  have the same upper bound, as they have the same counters: $l_2 = 1$ and $l_7 = 2$.
  Finally,
  \begin{align*}
    t \| N \cdot V \circ \phi \|_{C^{k_{0}+1}(\S_t)}
    \leq C t^{-1 + 5\s}
  \end{align*}
  as $l_2 = 1$, $l_{11} = 1$, $m_1 = 2$, $m_{\s} = 5$ and $m_D = 0$. The lemma follows. 
\end{proof}

\subsubsection{Higher-order estimates for the second fundamental form}
\begin{lemma}\label{lemma:EstimateFFHigh}
  Let $\s_p$, $\s_V$, $\s$, $k_0$, $k_1$, $(\S, h_{\refer})$, $(E_{i})_{i=1}^{n}$ and $V$ be as in Theorem~\ref{thm: big bang formation}
  and let $\rho_0 >0$. Assume that there are $t_{0}\leq \tau_{H}$ (see Lemma~\ref{lemma:AuxiliaryLapseEstimate}), $t_\rob<t_0$ and
  $\ar \in (0, \tfrac{1}{6n}]$ such that Assumption~\ref{ass:Bootstrap} is satisfied for this choice of parameters.
  Then, for any $I,J$ and $t \in [t_{\rob}, t_0]$,
  \begin{equation}
    \begin{split}
      & t^{A+1}  \big\|
      \big( - \tfrac{1}{t} - \d_t \big) (\dek{IJ}) 
      + e_{(I} e_{J)} (N) - t^{-1}(N-1) \chk_{IJ} \\[2pt]
      &\phantom{t^{A+1}  \big\|} - N \big(e_K (\g_{K(IJ)}) + e_{(I} (\g_{J)CC}) \big)
      \big\|_{H^{k_1}(\S_t)} \\[2pt]
      \leq & C \ar t^{-1} \big( \mbH_{(\g,k)} (t) + \mbH_{(\phi)} (t) \big)      
      + C t^{-1+\s} \left( \mbD(t) + t^{\s} \right). \label{eq:EstFFhigh}
    \end{split}
  \end{equation}  
\end{lemma}

\begin{proof}
  In analogy with the proof of Lemma~\ref{lemma:EstimateSCHigh}, the idea is to subtract well chosen terms from both sides
  of (\ref{eq:EveFF}). Then $t^{-1}(N-1)\dek{IJ}$ appears, which we bound by making use of Assumption~\ref{ass:Bootstrap},
  Lemma~\ref{lemma:HigherOrderLapseEstimate} and Lemma~\ref{lemma: products in L2}:
  \begin{align*}
    & t^{A+1} \| t^{-1} (N-1)\dek{IJ} \|_{H^{k_1}(\S_t)} \\
    \leq & C t^{A} \big( \|N-1 \|_{H^{k_1}(\S_t)} \| \dek{} \|_{C^{0}(\S_t)} 
        +  \| N-1 \|_{C^{0}(\S_t)} \| \dek{} \|_{H^{k_1}(\S_t)} \\
    \leq & C \ar t^{-1} \big( \mbH_{(\g,k)} + \mbH_{(\phi)} \big)
    + Ct^{-1+\s}(\mbD + t^{\s}). 
\end{align*}
As an example of terms of the form $N\cdot\ep\cdot\ep$ and $N \cdot \g \cdot \g$,
\begin{align*}
t^{A+1} \|N e_I(\phi) e_J(\phi) \|_{H^{k_1}(\S_t)} 
    &\leq C t^{-1+2\s} (\mbD + t^{\s})^2
\end{align*}
as $l_2=1$, $l_7=2$, $m_1 = 2$, $m_{\s} = 4$, $m_{D} =2$ and $\max_{i} S(j_i) = A+2\s$.
Next,
\begin{align*}
t^{A+1} \|\g_{KIJ} e_K(N) \|_{H^{k_1}(\S_t)} 
    & \leq C t^{-1+2\s} (\mbD + t^{\s})
\end{align*}
as $l_3 = 1$, $l_7 = 1$, $m_1 = 2$, $m_{\s} = 6$, $m_D = 1$ and $\max_i S(j_i) = A+4\s$. Finally, 
\begin{align*}
t^{A+1} &\|N \cdot V \circ \phi \|_{H^{k_1}(\S_t)} \leq C t^{-1+5\s},
\end{align*}
as $l_2 = 1$, $l_{11}=1$, $m_1 = 2$, $m_{\s} = 5$, $m_{D} = 0$ and $\max_i S(j_i) = A$.
\end{proof}

\subsubsection{Additional estimates for the derivatives of the metric}
In order to establish energy estimates, it is of use to have an estimate of the following (scalar) quantity:
\begin{equation}\label{eq:Zdefinition}
\begin{split}
\Z(t)
:= & \langle \dega{JKK}, N e_I (\dek{IJ})\rangle_{H^{k_1}(\S_t)} 
- \langle \dega{IJK}, N e_{I} (\chk_{JK}) \rangle_{H^{k_1}(\S_t)} \\
& - \langle e_J (N), e_I (\dek{IJ}) \rangle_{H^{k_1}(\S_t)}
- \langle \dek{IJ}, N (e_I (\chga_{JKK}) 
    + e_K(\chga_{KIJ})) \rangle_{H^{k_1}(\S_t)}.
\end{split}
\end{equation}
As a first step, it is of use to estimate the first and third term by appealing to (\ref{eq:MC}). 

\begin{lemma} \label{lemma:EstimateMomentumConstraints}
  Let $\s_p$, $\s_V$, $\s$, $k_0$, $k_1$, $(\S, h_{\refer})$, $(E_{i})_{i=1}^{n}$ and $V$ be as in Theorem~\ref{thm: big bang formation}
  and let $\rho_0 >0$. Assume that there are $t_{0}\leq \tau_{H}$ (see Lemma~\ref{lemma:AuxiliaryLapseEstimate}), $t_\rob<t_0$ and
  $\ar \in (0, \tfrac{1}{6n}]$ such that Assumption~\ref{ass:Bootstrap} is satisfied for this choice of parameters.
  Then, for any $t \in [t_{\rob}, t_0]$,
  \begin{equation}\label{eq:Z estimate first step}
    \begin{split}
      & t^{2(A+1)}\big|\langle N e_I (k_{IJ}),
      \dega{JKK} \rangle_{H^{k_1}(\S_t)} - 
      \langle e_I (k_{IJ}), e_J(N) \rangle_{H^{k_1}(\S_t)}\big| \\
      \leq & \big(C \ar + 2n(1+\s) \big) t^{-1} \big(\mbH_{(\g,k)}^2 + \mbH_{(\phi)}^2 \big)
      + Ct^{-1 + \s} (\mbD + t^{\s})(\mbD + t^{3\s}).
    \end{split}
  \end{equation}  
\end{lemma}
\begin{proof}
  Recall (\ref{eq:MC}), 
  which we write as follows:
  \begin{equation}\label{eq:alternate MC}
    \begin{split}
      e_I k_{IJ}
      = & \dek{IJ}\cdot\dega{IMM}+ \dek{IM}\cdot\dega{IJM} + \dep{0}\cdot\dep{J}+ \dek{IJ}\cdot \chga_{IMM}\\
      &  + \dek{IM}\cdot\chga_{IJM} + \dep{0}\cdot\che_J(\chphi)+ \chk_{IJ} \cdot \dega{IMM} + \chk_{IM}\cdot \dega{IJM} \\
      &  + \chdtp\cdot\dep{J} + \chk_{IJ} \chga_{IMM} + \chk_{IM} \chga_{IJM} + \chdtp\cdot \che_J(\chphi).
    \end{split}
  \end{equation}  
  It is convenient to introduce the notation
  \begin{align*}
    \M_J := & \chk_{IJ} \cdot \dega{IMM} + \chk_{IM}\cdot \dega{IJM} + \chdtp\cdot\dep{J},\\
    X_J := & \dega{JKK} - e_J(N).
  \end{align*}
  In particular, $\M_J$ denotes terms in (\ref{eq:alternate MC}) that have to be treated differently from the rest.
  To bound $e_I(k_{IJ}) - \M_J$, note that Lemma~\ref{lemma: products in L2} and Assumption~\ref{ass:Bootstrap} yield
  \begin{equation*}
    \begin{split}
      & t^{A+1} \|\dek{IJ}\cdot\dega{IMM}\|_{H^{k_1}(\S_t)}+t^{A+1}\|\dek{IM}\cdot\dega{IJM} \|_{H^{k_1}(\S_t)}  \\
      \leq & C t^{A+1} \big( \| \dek{} \|_{H^{k_1}(\S_t)} \| \dega{} \|_{C^{0}(\S_t)}
      + \| \dek{} \|_{C^{0}(\S_t)} \| \dega{} \|_{H^{k_1}(\S_t)} \big)\leq C \ar t^{-1} \mbH_{(\g,k)}.
    \end{split}
  \end{equation*}
  Similarly,
  \begin{equation*}
    t^{A+1} \| \dep{0}\cdot\dep{J}\|_{H^{k_1}(\S_t)} \leq C \ar t^{-1} \mbH_{(\phi)}.
  \end{equation*}
  As an example of how to bound terms of the form $\dek{}\cdot\chga$ and $\dep{0}\cdot\che_J(\chphi)$,
  \begin{align*}
    t^{A+1} \| \dek{IJ}\cdot\chga_{IMM} \|_{H^{k_1}(\S_t)}\leq Ct^{-1+3\s}(\mbD + t^{\s}),
  \end{align*}
  as $l_6 = 1$, $l_8 = 1$, $m_{1} = 2$, $m_{\s} =3$, $m_D = 1$ and $\max_{i} S(j_i) = A$. Next, 
  \begin{equation*}
    t^{A+1} \| \chk_{IJ} \chga_{IMM}+\chk_{IM} \chga_{IJM} + \chdtp\cdot \che_J(\chphi)\|_{H^{k_1}(\S_t)}  \leq C t^{A-1+3\s},
  \end{equation*}
  as $l_{6}=1$, $l_9=1$, $m_1 = 2$, $m_{\s} = 3$ and $\max_{i} S(j_i) = 0$. Thus
  \begin{equation*}
    t^{A+1} \| e_I (k_{IJ}) - \M_J \|_{H^{k_1}(\S_t)} 
    \leq C \ar t^{-1} \big( \mbH_{(\g,k)} + \mbH_{(\phi)} \big) + C t^{-1+\s} (\mbD + t^{\s}).
  \end{equation*}
  Combining this estimate with Lemma~\ref{lemma:HigherOrderLapseEstimate} yields
  \begin{equation}\label{eq:eI kIJ minus MJ}
    \begin{split}
      & t^{2(A+1)} \big| \langle e_I (k_{IJ}) - \M_J , e_J(N) \rangle_{H^{k_1}(\S_t} \big| \\
      \leq & C \ar t^{-1} \big( \mbH_{(\g,k)}^2 + \mbH_{(\phi)}^2 \big)
      + C t^{-1+\s} (\mbD + t^{\s})(\mbD + t^{3\s}).
    \end{split}
  \end{equation}  
  On the other hand, by Lemma~\ref{lemma: products in L2},
  \begin{equation*}
    \begin{split}
      & t^{A+1} \| N \big(e_I (k_{IJ}) - \M_J\big) \|_{H^{k_1}(\S_t)} \\
      \leq & t^{A+1} (1+ C \|N-1\|_{C^{0}(\S_t)})\| e_I (k_{IJ}) - \M_J \|_{H^{k_1}(\S_t)} \\
      & + C t^{A+1}\| N-1\|_{H^{k_1}(\S_t)}\| e_I (k_{IJ}) - \M_J \|_{C^{0}(\S_t)} \\
      \leq & C \ar t^{-1} \big( \mbH_{(\g,k)} + \mbH_{(\phi)} \big)+ C t^{-1+\s} (\mbD + t^{\s}).    
    \end{split}
  \end{equation*}  
  Here we make use of the fact that, following the scheme of Subsection~\ref{sec:scheme},
  \begin{equation}
    t \| e_I (k_{IJ}) \|_{C^{0}(\S_t)} \leq C t^{-1+3\s}(\mbD + t^{\s})
  \end{equation}
  as $l_5 = 1$, $l_{10} = 1$, $m_1 = 2$, $m_{\s} = 3$ and $m_D = 1$. Similarly, 
  \begin{equation} \label{eq:BoundednessMJ}
    t \| \M_J \|_{C^{0}(\S_t)} \leq C t^{-1+2\s}(\mbD + t^{\s})
  \end{equation}
  as $l_7 =1$, $l_9 = 1$, $m_1 = 2$, $m_{\s} = 2$ and $m_D = 1$. Thus 
  \begin{equation}\label{eq:N eI kIJ minus MJ}
    \begin{split}
      & t^{2(A+1)} | \langle N( e_I (k_{IJ}) - \M_J ),
      \dega{JKK}\rangle_{H^{k_1}(\S_t} | \\
      \leq & C \ar t^{-1} \big( \mbH_{(\g,k)}^2 + \mbH_{(\phi)}^2 \big)
      + C t^{-1+\s} (\mbD + t^{\s})(\mbD + t^{3\s}).
    \end{split}
  \end{equation}
  Our next goal is to prove that
  \begin{equation}\label{eq:est to be proven MC}
    \begin{split}
      & t^{2(A+1)} | \langle \M_J, X_J \rangle_{H^{k_1}(\S_t)} | \\
      \leq & (C \ar + 2n(1+\s) ) t^{-1}\big( \mbH_{(\g,k)}^2 + \mbH_{(\phi)}^2 \big)
      + C t^{-1+\s} \mbD (\mbD + t^{\s}).
    \end{split}    
  \end{equation}
  In order to prove that the statement follows from this estimate, note that
  \begin{equation}\label{eq:MJ XJ decomp}
    \begin{split}
      & \langle N \M_J,\dega{JKK} \rangle_{H^{k_1}(\S_t)} - \langle \M_J, e_J(N) \rangle_{H^{k_1}(\S_t)}\\
      = & \langle (N-1) \M_J,\dega{JKK} \rangle_{H^{k_1}(\S_t)} + \langle \M_J, X_J \rangle_{H^{k_1}(\S_t)}.
    \end{split}
  \end{equation}
  On the other hand, Lemma~\ref{lemma: products in L2} yields
  \begin{equation}\label{eq:N minus one MJest}
    \begin{split}
      & t^{A+1}\|(N-1)\M_J\|_{H^{k_{1}}(\S_t)}\\
      \leq & Ct^{A+1}(\|N-1\|_{H^{k_{1}}(\S_t)}\|\M_J\|_{C^{0}(\S_t)}+\|N-1\|_{C^{0}(\S_t)}\|\M_J\|_{H^{k_{1}}(\S_t)}).
    \end{split}
  \end{equation}
  Next, 
  \begin{equation}\label{eq:MJ crude Hkopo est}
    t^{A+1} \| \M_J \|_{H^{k_1}(\S_t)} \leq C t^{-1}(\mbD + t^{\s})    
  \end{equation}
  as $l_7 =1$, $l_9 = 1$, $m_1 = 2$, $m_{\s} = 2$, $m_D = 1$ and $\max_i S(j_i) = A + 2\s$.
  Combining \eqref{eq:BoundednessMJ}, 
  (\ref{eq:N minus one MJest}), (\ref{eq:MJ crude Hkopo est}) with the fact that $\|N-1\|_{C^{0}(\S_t)} \leq t^{\s} r$
  yields
  \[
  t^{2(A+1)}|\langle (N-1) \M_J,\dega{JKK} \rangle_{H^{k_1}(\S_t)} |\leq Ct^{-1+\s}\mbD(\mbD+t^{\s}).
  \]
  Combining this estimate with (\ref{eq:eI kIJ minus MJ}), (\ref{eq:N eI kIJ minus MJ}),  (\ref{eq:est to be proven MC})
  and (\ref{eq:MJ XJ decomp}) yields the conclusion of the lemma.

  To prove (\ref{eq:est to be proven MC}), note that $|X|:=\big(\textstyle{\sum}_{J} \| X_J \|_{H^{k_1}(\S_t)}^2\big)^{1/2}$
  satisfies
  \begin{equation}\label{eq:EstimateXJ}
    \begin{split}
      t^{A+1} |X|
      \leq & t^{A+1} \|\ear(N)\|_{H^{k_1}(\S_t)} + t^{A+1} \big( \textstyle{\sum}_{J}
      \| \dega{JKK} \|_{H^{k_1}(\S_t)}^2 \big)^{1/2}\\
      \leq & \big(C \ar + 2 + \s \big)t^{A+1}\big( \| \dek{} \|_{H^{k_1}(\S_t)}^2
      + \| \dep{0} \|_{H^{k_1}(\S_t)}^2 \big)^{1/2} \\
      & + \big(\tfrac{n-1}{2}\big)^{1/2} t^{A+1}\| \dega{} \|_{H^{k_1}(\S_t)}
      + C t^{2\s} (\mbD + t^{\s}),
    \end{split}
  \end{equation}
  where we appeal to Lemma~\ref{lemma:HigherOrderLapseEstimate} and the fact that
  \begin{equation*}
    \textstyle{\sum}_{J} \big\| \textstyle{\sum}_{K} \dega{JKK} \big\|_{H^{k_1}(\S_t)}^2
    \leq \tfrac{n-1}{2} \| \dega{} \|_{H^{k_1}(\S_t)}^2.
  \end{equation*}
  Next, (\ref{eq: gamma bar p bar assumption}), (\ref{eq: scalar field bar assumption}), (\ref{eq:BoundEigenvalues}) and
  (\ref{eq:noGNSSum}), with $\eta = \s$, yield
  \begin{equation*}
    \begin{split}
      & \big( \textstyle{\sum}_{J}
      \| \chk_{IM} \cdot\dega{IJM}
      + \chdtp\cdot \dep{J} \|_{H^{k_1}(\S_t)}^2 \big)^{1/2}\\
      \leq & t^{-1}\big(\s +
      \big\| \textstyle{\sum}_{I} \bq_I^2
      + \bP_1^2 \big\|_{C^{0}(\S_{t})}^{1/2} \big)
      \big( \tfrac{1}{2} \|\dega{} \|_{H^{k_1}(\S_t)}^2 
      + \|\vdep \|_{H^{k_1}(\S_t)}^2 \big)^{1/2} \\
      & + Ct^{-1}\big(\| \dega{} \|_{C^{k_{0}}(\S_t)}
      + \| \vdep\|_{C^{k_{0}}(\S_t)} \big) \\
      \leq &   t^{-1}\big(1 + 2 \s \big) \big( \tfrac{1}{2} 
      \|\dega{} \|_{H^{k_1}(\S_t)}^2
      + \|\vdep \|_{H^{k_1}(\S_t)}^2 \big)^{1/2}  
      + C t^{-2+2\s}\mbD(t). 
    \end{split}
  \end{equation*}
  Similarly, appealing to (\ref{eq: gamma bar p bar assumption}), (\ref{eq:bqboundsb}) and
  (\ref{eq:noGNSInnerProduct}), with $\eta = \s$, yields
  \begin{equation*}
    \begin{split}
      & \big| \langle \chk_{IJ}\cdot\dega{IMM},X_J \rangle_{H^{k_1}(\S_t)} \big| \\
      \leq & t^{-1}\big[ \big(\s +  \textstyle{\max}_J \|\bq_J\|_{C^{0}(\S_{t})} \big)
      \big(\tfrac{n-1}{2}\big)^{1/2}\|\dega{}\|_{H^{k_1}(\S_t)}
      + C\|\dega{}\|_{C^{k_{0}}(\S_t)} \big]\cdot |X| \\
      \leq & t^{-1} \big[ (1 - 4\s )\sqrt{n-1}
      \big( \tfrac{1}{2} \|\dega{} \|_{H^{k_1}(\S_t)}^2+ \|\vdep \|_{H^{k_1}(\S_t)}^2 \big)^{1/2}
      + C t^{-1+2\s}\mbD(t) \big]\cdot |X|.
    \end{split}
  \end{equation*}
  Combining the last two estimates yields
  \begin{equation*}
    \begin{split}
      & t^{2(A+1)}  \big| \langle \M_J, X_J \rangle_{H^{k_1}(\S_t)} \big| \\
      \leq & t^{2A+1} \big[ \big(1+\sqrt{n-1} \big)
      \big( \tfrac{1}{2} \|\dega{} \|_{H^{k_1}(\S_t)}^2
      + \|\vdep \|_{H^{k_1}(\S_t)}^2 \big)^{1/2} 
      + C t^{-1+2\s}\mbD(t) \big]\cdot |X|.
    \end{split}
  \end{equation*}
  Combining this estimate with \eqref{eq:EstimateXJ} yields (\ref{eq:est to be proven MC}); note that
  \begin{equation*}
    \begin{split}
      & (1 + \sqrt{n-1})(\sqrt{n-1} a^2 + (2+\s) ab)\\
      \leq & (1+\sqrt{n-1})(1 + \s + \sqrt{n-1}) (a^2 + b^2)\\
      \leq & (1 + \sqrt{n-1})^2(1+\s)(a^2 + b^2)\leq 2n(1+\s) (a^2 + b^2).
    \end{split}
  \end{equation*}  
  The lemma follows. 
\end{proof}

\begin{lemma} \label{lemma:EstimateForZ}
  Let $\s_p$, $\s_V$, $\s$, $k_0$, $k_1$, $(\S, h_{\refer})$, $(E_{i})_{i=1}^{n}$ and $V$ be as in Theorem~\ref{thm: big bang formation}
  and let $\rho_0 >0$. Assume that there are $t_{0}\leq \tau_{H}$ (see Lemma~\ref{lemma:AuxiliaryLapseEstimate}), $t_\rob<t_0$ and
  $\ar \in (0, \tfrac{1}{6n}]$ such that Assumption~\ref{ass:Bootstrap} is satisfied for this choice of parameters.
  Then, for any $t \in [t_{\rob}, t_0]$,
  \begin{align}
    \begin{split}
      t^{2(A+1)} | \Z(t) |
      \leq (C \ar + 2n(1+\s)) t^{-1} \big(\mbH_{(\g,k)}^2 + \mbH_{(\phi)}^2 \big)
      + C t^{-1 + \s} (\mbD + t^{\s})(\mbD + t^{3\s}).
    \end{split}
  \end{align}
\end{lemma}
\begin{proof}
  Given Lemma~\ref{lemma:EstimateMomentumConstraints}, the terms in (\ref{eq:Zdefinition}) that remain to
  be bounded, can be estimated by appealing to Cauchy-Schwarz and $H^{k_{1}}$-estimates for terms of the form
  $Ne(\chk)$, $e(\chk)$ and $Ne(\chga)$. For all these terms, we can take up to $k_1+1$ derivatives in the a-priori
  estimates, due to the assumed regularity on the initial data.

  For terms of the form $Ne(\chk)$ and $e(\chk)$, we have, e.g.
  \begin{align*}
    t^{A+1} \| N e_I (\chk_{JK})\|_{H^{k_1}(\S_t)}
    &\leq C t^{-1+2\s} (\mbD + t^{\s}),
  \end{align*}
  as  $l_{2} = 1$, $l_{5}=1$, $l_9 = 1$, $m_{1}=2$, $m_{\s} = 3$, $m_{D} = 1$ and $\max_{i} S(j_i) =A+\s$. Next, 
  \begin{align*}
    t^{A+1} \| N e_I(\chga_{JKK}))\|_{H^{k_1}(\S_t)}
    \leq C t^{-1+5\s} (\mbD + t^{\s}),
  \end{align*}
  as $l_{2}=1$, $l_{5}=1$, $l_{6}=1$, $m_{1} = 2$, $m_{\s} = 6$, $m_{D} = 1$, and $\max_{i} S(j_i) = A+\s$. Combining these
  estimates with (\ref{eq:Z estimate first step}) and Lemma~\ref{lemma:HigherOrderLapseEstimate}, the lemma follows. 
\end{proof}

\subsection{The scalar field}
In this section, we write down estimates arising from the evolution equations for the scalar field,
\eqref{eq:spatial scalar field derivative} and \eqref{eq:scalarfield}, which we here write as 
\begin{subequations}\label{seq:SFLow}
  \begin{align}     
      \big(-\tfrac{\bq_{\uI}}{t} - \d_t \big) (\dep{\uI})
      = & - N e_I e_0 (\phi) - e_I (N) e_0 (\phi) 
      - \dee{I}{} (\d_t \chphi)
      + e_I (\chdtp) \label{eq:EveSFeILow}\\
      & + (N - 1) k_{IM} e_M(\phi)  
      + \dek{IM}\cdot \che_M (\chphi)+ \dek{IM}\cdot\dep{M},\nonumber \\
      \big(-\tfrac{1}{t}-\d_t \big) ( \dep{0})
      = & t^{-1}(N-1) e_0(\phi) - N e_I e_I (\phi) 
      - e_I (N) e_I(\phi) \label{eq:EveSFe0Low}\\
      & + N \gamma_{JII} e_{J}(\phi) + N V'\circ\phi.\nonumber    
  \end{align}
\end{subequations}

\subsubsection{Lower-order estimates for the spatial gradient of the scalar field}
Next, we establish lower-order estimates for the spatial gradient of the scalar field.

\begin{lemma} \label{lemma:EstSFeILow}
  Let $\s_p$, $\s_V$, $\s$, $k_0$, $k_1$, $(\S, h_{\refer})$, $(E_{i})_{i=1}^{n}$ and $V$ be as in Theorem~\ref{thm: big bang formation}
  and let $\rho_0 >0$. Assume that there are $t_{0}\leq 1$, $t_\rob<t_0$ and $\ar \in (0, \tfrac{1}{6n}]$ such that
  Assumption~\ref{ass:Bootstrap} is satisfied for this choice of parameters. Then, for any $I$, $|\bfI| \leq k_0$, $t \in [t_{\rob}, t_0]$
  and $x \in \S$,
\begin{equation*}
\begin{split}
  & t^{1-2\s}
  \big| E_{\bfI} \big(\big(-\tfrac{ \bq_{\uI}}{t} - \d_t \big)
  (\dep{I}) \big) (t,x) \big|\\
  \leq & C\ar t^{-1 +(1-2\s)} \textstyle{\sum}_{M}  \textstyle{\sum}_{|\bfJ| \leq |\bfI|}
  |E_{\bfJ} (\dep{M}) (t,x) |+ C t^{-1+\s} ( \mbD(t) + t^{\s} ).
\end{split}
\end{equation*}
\end{lemma}

\begin{proof}
  We begin by bounding all the terms on the right-hand side of (\ref{eq:EveSFeILow}) but the last one
  using the scheme in Subsection~\ref{sec:scheme}. First,
  \begin{equation*}
    t^{1-2\s} \| N e_I e_0(\phi)\|_{C^{k_{0}}(\S_t)} 
    \leq C t^{-1+\s} (\mbD + t^{\s}),
  \end{equation*}
  as $l_2 =1$, $l_5=1$, $l_{10}=1$, $m_1 = 2$, $m_{\s} = 3$ and $m_D =1$. Second,
  \begin{equation*} 
    t^{1-2\s} \| e_I (N) e_0(\phi)\|_{C^{k_{0}}(\S_t)}    \leq C t^{-1+2\s}.
  \end{equation*}
  as $l_3 =1$, $l_{10} =1$, $m_1 = 2$, $m_{\s} = 4$ and $m_D =0$. Third,
\begin{equation} \label{eq:EstSpGrScFdFirstBound}
t^{1-2\s} \| \dee{I}{} (\d_t \chphi)\|_{C^{k_{0}}(\S_t)} +t^{1-2\s} \| e_{I}(\d_t \chphi)\|_{C^{k_{0}}(\S_t)} 
    \leq C t^{-1+\s} (\mbD + t^{\s}),
\end{equation}
as $l_5 = 1$, $l_9 =1$, $m_1 = 2$, $m_{\s} = 3$ and $m_D =1$. Next, 
\begin{equation} \label{eq:EstSpGrScFdSecondBound}
t^{1-2\s} \| (N-1) k_{IM} e_M(\phi)\|_{C^{k_{0}}(\S_t)}
    \leq C t^{-1+\s} (\mbD + t^{\s}),
\end{equation}
as $l_1 = 1$, $l_7 =1$, $l_{10} =1$, $m_1 = 2$, $m_{\s} = 3$ and $m_D =1$. Finally,
\begin{equation} \label{eq:EstSpGrScFdLastBound}
t^{1-2\s} \| \dek{IM}\cdot \che_M(\chphi)\|_{C^{k_{0}}(\S_t)}
\leq C t^{-1+\s} (\mbD + t^{\s}), 
\end{equation}
as $l_6 = 1$, $l_8 =1$, $m_1 = 2$, $m_{\s} = 3$ and $m_D =1$.

For the remaining term, $|E_{\bfJ}(\dek{IM})| \leq t^{-1} \mbL_{(\g,k)} \leq t^{-1} \ar$
for any $|\bfJ| \leq k_0$. Thus
\begin{align*}
t^{1-2\s}  |E_{\bfI}(\dek{IM}\cdot\dep{M}) (t,x) |
&\leq C\ar t^{-1+(1-2\s)} \textstyle{\sum}_{M}  \textstyle{\sum}_{|\bfJ| \leq |\bfI|}|E_{\bfJ} (\dep{M}) (t,x) |.
\end{align*}
This concludes the proof.
\end{proof}

\subsubsection{Higher-order estimates for the spatial gradient of the scalar field}
We proceed with the higher order estimates for the spatial gradient of the scalar field.
\begin{lemma}\label{lemma:EstimateSFeIHigh}
  Let $\s_p$, $\s_V$, $\s$, $k_0$, $k_1$, $(\S, h_{\refer})$, $(E_{i})_{i=1}^{n}$ and $V$ be as in Theorem~\ref{thm: big bang formation}
  and let $\rho_0 >0$. Assume that there are $t_{0}\leq \tau_{H}$ (see Lemma~\ref{lemma:AuxiliaryLapseEstimate}), $t_\rob<t_0$ and
  $\ar \in (0, \tfrac{1}{6n}]$ such that Assumption~\ref{ass:Bootstrap} is satisfied for this choice of parameters.
  Then, for any $I$ and $t \in [t_{\rob}, t_0]$,
\begin{equation} 
\begin{split} \label{eq:EstSFeIHigh}
    & t^{A+1}  \big\| \big(-\tfrac{\bq_{\uI}}{t} - \d_t \big) (\dep{I})
        + e_{I}(N) \chdtp + N e_I \dtp \big\|_{H^{k_1}(\S_t)} \\[2pt]
    \leq & C \ar  t^{-1} \big(\mbH_{(\phi)} (t) 
        + \mbH_{(\g,k)} (t) \big) 
        + C t^{-1+\s} \left( \mbD(t) + t^{\s} \right).
\end{split}
\end{equation}
\end{lemma}

\begin{proof}
  Adding $e_{I}(N) \chdtp + N e_I \dtp$ to both sides of (\ref{eq:EveSFeILow}), a term of the form $-e_I(N)\dep{0}$ appears.
  However, appealing to (\ref{eq:BootstrapInequality}) and Lemmata~\ref{lemma:HigherOrderLapseEstimate} and \ref{lemma: products in L2},
  \begin{equation*}
    \begin{split}
      t^{A+1} \| e_I(N)\dep{0} \|_{H^{k_1}(\S_t)} 
      \leq & C t^{A+1} \big( \|\en \|_{H^{k_1}(\S_t)} 
      \|\dep{0} \|_{C^{0}(\S_t)}
      +  \| \en \|_{C^{0}(\S_t)}
      \|\dep{0} \|_{H^{k_1}(\S_t)} \big) \\
      \leq & C \ar t^{-1} \big(\mbH_{(\g,k)} + \mbH_{(\phi)} \big) + Ct^{-1+\s}(\mbD + t^{\s}). 
    \end{split}
  \end{equation*}
  Similarly,
  \begin{equation*}
    \begin{split}
       t^{A+1} \| \dek{IM}\cdot\dep{M}\|_{H^{k_1}(\S_t)} 
      \leq & C t^{A+1} \big( \| \dek{} \|_{H^{k_1}(\S_t)}
      \| \vdep \|_{C^{0}(\S_t)} 
      +  \| \dek{} \|_{C^{0}(\S_t)}
      \| \vdep \|_{H^{k_1}(\S_t)} \big) \\
      \leq & C \ar t^{-1} \mbH_{(\phi)} + Ct^{-1+\s}(\mbD + t^{\s}).
    \end{split}
  \end{equation*}
  To obtain $H^{k_{1}}$-estimates analogous to (\ref{eq:EstSpGrScFdFirstBound})--(\ref{eq:EstSpGrScFdLastBound}),
  all we have to do, due to the scheme of Subsection~\ref{sec:scheme}, is to multiply the left hand
  sides by $t^{A+2\s}$ and the right-hand sides by $t^{A+2\s - \max_{i} S(j_i)}$. Since $\max_i S(j_i) \leq A + 2\s$
  for all the terms of interest, the desired estimates follow.
\end{proof}

\subsubsection{Lower-order estimates for the normal derivative of the scalar field}
Next, there are the lower order estimates for the normal derivative.
\begin{lemma} \label{lemma:EstSFe0Low}
  Let $\s_p$, $\s_V$, $\s$, $k_0$, $k_1$, $(\S, h_{\refer})$, $(E_{i})_{i=1}^{n}$ and $V$ be as in Theorem~\ref{thm: big bang formation}
  and let $\rho_0 >0$. Assume that there are $t_{0}\leq 1$, $t_\rob<t_0$ and $\ar \in (0, \tfrac{1}{6n}]$ such that
  Assumption~\ref{ass:Bootstrap} is satisfied for this choice of parameters. Then, for any $t \in [t_{\rob}, t_0]$,
  \begin{align} \label{eq:EstSFe0Low}
    t \left\|  \big( - \tfrac{1}{t} -\d_t \big) (\dep{0}) \right\|_{C^{k_{0}+1} (\S_t)}
    \leq C t^{-1+2\s} \left( \mbD(t) + t^{\s} \right)^2.
  \end{align}  
\end{lemma}
\begin{proof}
  Throughout the proof, note that since we estimate the right-hand side of (\ref{eq:EveSFe0Low}) in $C^{k_{0}+1}$, we sometimes need to appeal to
  Lemma~\ref{lemma:extraderivatives}. Due to Lemma~\ref{lemma:LowerOrderLapseEstimates},
  \begin{align*}
    \| (N-1) e_0(\phi)\|_{C^{k_{0}+1}(\S_t)} 
    & \leq \|N-1\|_{C^{k_{0}+1}(\S_t)}
    \cdot \|e_0(\phi)\|_{C^{k_{0}+1}(\S_t)} \\
    & \leq C t^{2\s} (\mbD + t^{\s})^2 \cdot t^{-1}(\mbL_{(\phi)} + C)
    \leq C t^{-1+2\s} (\mbD + t^{\s})^2.
  \end{align*}
  Next, due to the scheme of Subsection~\ref{sec:scheme},
  \begin{align*}
    t & \| N e_I e_I (\phi)\|_{C^{k_{0}+1}(\S_t)}
    \leq C t^{-1+4\s} (\mbD + t^{\s})^2,
  \end{align*}
  as $l_2 =1$, $l_5 =1$, $l_7= 1$, $l_{\roint} = 1$, $m_{1} = 2$, $m_{\s} = 4$ and $m_D = 2$. Second,
  \begin{equation}\label{eq:eINeIphiezest}
    t \| e_I(N) e_I (\phi)\|_{C^{k_{0}+1}(\S_t)} 
    \leq C t^{-1+4\s} (\mbD + t^{\s}),
  \end{equation}
  as $l_3 =1$, $l_7= 1$, $l_{\roint} = 2$, $m_{1} = 2$, $m_{\s} = 4$ and $m_D = 1$. Third,
  \begin{equation}\label{eq:ngammaeJphiezest}
    t \| N \g_{JII} e_{J}(\phi)\|_{C^{k_{0}+1}(\S_t)} 
    \leq C t^{-1+2\s} (\mbD + t^{\s})^2,
  \end{equation}
  as $l_2 =1$, $l_7=2$, $l_{\roint} = 2$, $m_{1} = 2$, $m_{\s} = 2$ and $m_D = 2$. Finally,
  \begin{equation}\label{eq:NVprimecircphiezest}
    t \| N V'\circ\phi\|_{C^{k_{0}+1}(\S_t)} 
    \leq C t^{-1+5\s},
  \end{equation}
  as $l_2 =1$, $l_{11}=1$, $m_1 = 2$, $m_{\s} = 5$ and $m_D =0$. This concludes the proof.
\end{proof}

\subsubsection{Higher-order estimates for the normal derivative of the scalar field}
Next, we prove the higher-order estimates for the normal derivative of the scalar field.
\begin{lemma}\label{lemma:EstimateSFe0High}
  Let $\s_p$, $\s_V$, $\s$, $k_0$, $k_1$, $(\S, h_{\refer})$, $(E_{i})_{i=1}^{n}$ and $V$ be as in Theorem~\ref{thm: big bang formation}
  and let $\rho_0 >0$. Assume that there are $t_{0}\leq \tau_{H}$ (see Lemma~\ref{lemma:AuxiliaryLapseEstimate}), $t_\rob<t_0$ and
  $\ar \in (0, \tfrac{1}{6n}]$ such that Assumption~\ref{ass:Bootstrap} is satisfied for this choice of parameters.
  Then, for any $t \in [t_{\rob}, t_0]$,    
  \begin{align}
    \begin{split}\label{eq:EstSFe0High}
      & t^{A+1} \big \| \big(-\tfrac{1}{t} -\d_t \big) (\dep{0})
      - t^{-1} (N-1) \chdtp + N e_I e_I(\phi)
      \big\|_{H^{k_1}(\S_t)} \ \\[2pt]
      \leq & C \ar t^{-1} \big(\mbH_{(\phi)} (t)
      + \mbH_{(\g,k)} (t) \big) 
      + C t^{-1+\s} \left( \mbD(t) + t^{\s} \right).
    \end{split}
  \end{align}
\end{lemma}
\begin{proof}
  By Lemmata~\ref{lemma:HigherOrderLapseEstimate} and \ref{lemma: products in L2},
  \begin{align*}
    t^{A+1}\| t^{-1} (N-1) 
    (\dep{0}) \|_{H^{k_1}(\S_t)} 
    \leq & C t^{A} \big( \|N-1 \|_{H^{k_1}(\S_t)} 
    \| \dep{0} \|_{C^{0}(\S_t)}
    +  \| N-1 \|_{C^{0}(\S_t)}
    \| \dep{0} \|_{H^{k_1}(\S_t)}\big) \\
    \leq & C \ar t^{-1} \big( \mbH_{(\g,k)} + \mbH_{(\phi)} \big)
    + Ct^{-1+\s}(\mbD + t^{\s}),
  \end{align*}
  The remaining terms can be estimated by the scheme. First,
  \begin{align*}
    t^{A+1} \|  e_I(N) e_I(\phi) \|_{H^{k_1}(\S_t)} 
    &\leq C t^{-1+2\s}(\mbD + t^{\s}),
  \end{align*}
  as $l_3 =1$, $l_7=1$, $m_1 =2$, $m_{\s} = 6$, $m_D = 1$ and $\max_j S(j) =  A + 4\s$. Second,
  \begin{align*}
    t^{A+1} \| N \g_{JII} e_{J} (\phi) \|_{H^{k_1}(\S_t)} 
    &\leq C t^{-1+2\s} (\mbD + t^{\s})^2,
  \end{align*}
  as $l_2 =1$, $l_7 =2$, $m_1 = 2$, $m_{\s} =4$, $m_D = 2$ and $\max_i S(j_i)= A + 2 \s$. Finally,
  \begin{align*}
    t^{A+1}\| N V'\circ\phi \|_{H^{k_1}(\S_t)}
    &\leq C t^{-1+ 5\s},
  \end{align*}
  as $l_2 = 1$, $l_{11} = 1$, $m_1 = 2$, $m_{\s} = 5$, $m_D = 0$ and $\max_i S(j_i) = A$.
\end{proof}

\subsubsection{Additional estimates for the scalar field}
To derive energy estimates, we also need to bound the following quantity:
\begin{align}
\begin{split}
\Xi(t)
&:= \langle N e_I (\d_t \chphi), \dep{I}\rangle_{H^{k_1}(\S_t)} 
+ \langle \dep{0}, N e_I (\che_I (\chphi)) \rangle_{H^{k_1}(\S_t)}.
\end{split}
\end{align}
\begin{lemma} \label{lemma:EstimateForXi}
  Let $\s_p$, $\s_V$, $\s$, $k_0$, $k_1$, $(\S, h_{\refer})$, $(E_{i})_{i=1}^{n}$ and $V$ be as in Theorem~\ref{thm: big bang formation}
  and let $\rho_0 >0$. Assume that there are $t_{0}\leq 1$, $t_\rob<t_0$ and $\ar \in (0, \tfrac{1}{6n}]$ such that
  Assumption~\ref{ass:Bootstrap} is satisfied for this choice of parameters. Then, for any $t \in [t_{\rob}, t_0]$,
  \begin{align}
    \begin{split}
      t^{2(A+1)} \big| \Xi(t) \big|
      \leq C t^{-1 + \s} \mbD(t) (\mbD(t) + t^{\s}).
    \end{split}
  \end{align}
\end{lemma}
\begin{proof}
  Due to the scheme of Subsection~\ref{sec:scheme}, 
  \begin{equation}\label{eq:NeIcheIchephiest}
    t^{A+1} \| N e_I (\che_I (\chphi))\|_{H^{k_1}(\S_t)}
    \leq C t^{-1+5\s} (\mbD + t^{\s}),
  \end{equation}
  as $l_2 =1$, $l_5 = 1$, $l_6 =1$, $m_1=2$, $m_{\s} = 6$, $m_D = 1$ and $\max_i S(j_i) = A + \s$;
  we control $\che_{I}(\chphi)$ in $H^{k_{1}+1}$ due to Lemma~\ref{lemma:estimatingthebackground}.
  Next, 
  \begin{equation}\label{eq:NeIchdtpest}
    t^{A+1} \| N e_I (\chdtp)\|_{H^{k_1}(\S_t)}
    \leq C t^{-1+2\s} (\mbD + t^{\s}),
  \end{equation}
  as $l_2 =1$, $l_5 = 1$, $l_9 =1$, $m_1=2$, $m_{\s} = 3$, $m_D = 1$ and $\max_i S(j_i) = A + \s$;
  we control $\chdtp$ in $H^{k_1 +1}$ due to Lemma~\ref{lemma:estimatingthebackground}.
  Combining (\ref{eq:NeIcheIchephiest}) and (\ref{eq:NeIchdtpest}) with (\ref{eq: H phi bootstrap}),
  (\ref{eq:mbDdefinition}) and the Cauchy-Schwarz inequality for $H^{k_{1}}$ yields the desired conclusion. 
\end{proof}

\section{Energy estimates and the proof of the bootstrap improvement}
\label{sec:EnergyEstimates}

\subsection{The lower-order energy estimate}
Applying all the lower-order estimates yields an estimate for
\[
\mbL_{(e,\omega,\g,k,\phi)}:=\mbL_{(e,\omega)}+\mbL_{(\g,k)}+\mbL_{(\phi)}.
\]
\begin{prop} \label{prop:EstimateEnergyLow}
  Let $\s_p$, $\s_V$, $\s$, $k_0$, $k_1$, $(\S, h_{\refer})$, $(E_{i})_{i=1}^{n}$ and $V$ be as in Theorem~\ref{thm: big bang formation}
  and let $\rho_0 >0$. Then there is a standard constant $\ar_{\mbL}\in (0, \tfrac{1}{6n}]$ such that if Assumption~\ref{ass:Bootstrap}
  is satisfied with $t_{0}\leq 1$, $t_\rob<t_0$ and $\ar \in (0, \ar_{\mbL}]$, then, for any $t \in [t_{\rob}, t_0]$,
\begin{equation}
    \mbL_{(e,\omega,\g,k,\phi)} (t)^2 \leq C \mbL_{(e,\omega,\g,k,\phi)} (t_0)^2 + C \textstyle{\int}_t^{t_0}s^{-1+\s} \mbD(s)\left( \mbD(s) + s^\s \right) \md s.
\end{equation}
\end{prop}
In order to prove the proposition, we define the pointwise quantities
\begin{align*}
    Q_{m}^2
        := &  \textstyle{\sum}_{\abs {\bfI} = m}\left( t^{2(1 - 3\s)} \mfQ_{\bfI}+ t^{2(1-2\s)}\mfP_{\bfI} + t^2\mfK_{\bfI}\right)\\
    Q_{k_{0}+1}^2
        := &  \textstyle{\sum}_{\abs {\bfI} = k_{0}+1}\left( t^{2(1 - 3\s)} \mfQ_{\bfI}+ t^2\mfK_{\bfI}\right),
\end{align*}
where $m \in \{0,\dots,k_0\}$ and
\begin{subequations}\label{seq:mfQPKdef}
  \begin{align}
    \mfQ_{\bfI} := & \textstyle{\sum}_{i,I}(|E_{\bfI}(\dee{I}{i})|^2+|E_{\bfI}(\deo{I}{i})|^{2}),\label{eq:mfQdef}\\
    \mfP_{\bfI} := & \textstyle{\sum}_{I,J,K}|E_{\bfI}(\dega{IJK})|^2+\textstyle{\sum}_{I}|E_{\bfI}(\dep{I})|^2,\label{eq:mfPdef}\\
    \mfK_{\bfI} := & \textstyle{\sum}_{I,J}|E_{\bfI}(\dek{IJ})|^2+|E_{\bfI}(\dep{0})|^2.\label{eq:mfKdef} 
  \end{align}
\end{subequations}
Moreover, we define a decreasing sequence $(\delta_{m})_{m=0}^{k_0+1}$ as follows: $\delta_{k_0+1}=1$ and 
\begin{equation}
\delta_m = (B_{k_0} \sigma^{-2} + 1)^{k_0+1-m}
\end{equation}
for $0 \leq m \leq k_0$, where
\begin{equation}\label{eq:Bkzdef}
B_{k_0} := 4^{k_0+1} 
    \max \big\{ \| \bq_I + \bq_J - \bq_K\|_{C^{k_{0}+1}(\S_{t_0})}^2 ~\big|~
    I,J,K \in \{1,..,n\},\ I\neq J  \big\}.
\end{equation}
Note that $B_{k_0}\leq C\cdot 4^{k_{0}+1}\rho_0^2$ if the diagonal FRS initial data satisfies the FRS expansion-normalized bounds
of regularity $k_1$ for $\rho_0$ at $t_0$, where $C$ only depends on $(\S,h_{\refer})$ and $(E_{i})_{i=1}^{n}$. This is a consequence of
\eqref{eq: gamma bar p bar assumption} and Sobolev embedding. In particular, the $\delta_m$ are standard constants; see Notation~\ref{not:constant}. 
By an inductive argument, the following inequality also holds for $m \leq k_0$:
\begin{equation} \label{eq:DeltaInequality}
\delta_m \geq B_{k_0} \s^{-2} \textstyle{\sum}_{j = m+1}^{k_0+1} \delta_j;
\end{equation}
it holds for $m = k_0$, and if it holds for $m \in \{1,\dots,k_0\}$, it also holds for $m-1$:
\begin{align*}
\delta_{m-1}= \big(B_{k_0} \sigma^{-2} + 1 \big) \delta_m 
    &\geq B_{k_0} \sigma^{-2} \delta_m 
        + B_{k_0} \s^{-2} \textstyle{\sum}_{j = m+1}^{k_0+1} \delta_j 
    = B_{k_0} \s^{-2} \sum_{j = m}^{k_0+1} \delta_j.
\end{align*}
Finally, we define the weighted, energy-like quantity we wish to estimate:
\begin{align*}
\Q^2 := \textstyle{\sum}_{m=0}^{k_0+1} \delta_m Q_{m}^2.
\end{align*}

\begin{lemma} \label{lemma:EstimateForQ}
  Let $\s_p$, $\s_V$, $\s$, $k_0$, $k_1$, $(\S, h_{\refer})$, $(E_{i})_{i=1}^{n}$ and $V$ be as in Theorem~\ref{thm: big bang formation}
  and let $\rho_0 >0$. Assume that there are $t_{0}\leq 1$, $t_\rob<t_0$ and $\ar \in (0, \tfrac{1}{6n}]$ such that
  Assumption~\ref{ass:Bootstrap} is satisfied for this choice of parameters. Then there are standard constants
  $C$ and $\mC_{1}$ such that for any $t \in [t_{\rob}, t_0]$,
\begin{align}
\begin{split} \label{eq:inequalityforQ}
   \Q(t,x)^2
  \leq & \Q(t_0,x)^2 + C \textstyle{\int}_{t}^{t_0} s^{-1+\s} \mbD(s) (\mbD + s^{\s}) \\
&+ (\mC_{1}\ar - 2 \sigma) \textstyle{\sum}_{m=0}^{k_0+1} \delta_m \sum_{|\bfI| = m} 
    \int_{t}^{t_0} s^{-1+ 2(1-3\s)}\mfQ_{\bfI}(s,x) \md s \\
&+ (\mC_{1}\ar - 4 \sigma) \textstyle{\sum}_{m=0}^{k_0} \delta_m \sum_{|\bfI| = m} 
    \int_{t}^{t_0} s^{-1 + 2(1-2\s)}\mfP_{\bfI}(s,x) \md s.
\end{split}
\end{align}
\end{lemma}
\begin{proof}
  We write $\Q(t,x)^2 - \Q(t_0,x)^2 = \int_{t}^{t_0} -\d_s \left[ \Q(s,x)^2 \right] \md s$ and estimate the right hand side.
  To clarify the common structure of the terms to be estimated, it is convenient to consider two abstract functions, say $f$
  (one of $\dee{I}{i}$ etc.) and $\mfp$ (one of $1$, $\bq_I$, $-\bq_I$ or $\bq_I+\bq_J-\bq_K$, $I\neq J$), and a constant
  $\varkappa\in \{0,2,3\}$, and to calculate
  \begin{equation}\label{eq:qu to be est structure coeff abs}
    \begin{split}
       \textstyle{\int}_{t}^{t_0}  -\d_s (s^{1-\varkappa\s}E_{\bfI}f )^2 \md s  
      = &  \textstyle{\int}_{t}^{t_0} 2 s^{2(1-\varkappa\s)}
      \{ E_{\bfI} [ ( -s^{-1}\mfp - \d_s ) f] \cdot E_{\bfI} f+ s^{-1} [E_{\bfI},\mfp] (f)\cdot E_{\bfI}f \\
      & \phantom{\textstyle{\int}_{t}^{t_0} 2 s^{2(1-\varkappa\s)},}
      - s^{-1} (1-\varkappa\s- \mfp ) |E_{\bfI}f|^2\} \md s;
    \end{split}
  \end{equation}  
  note that $[E_\bfI,\d_t] = 0$ and that if $|\bfI| = 0$ there is no commutator term. If $f=\dep{0}$ or $f=\dek{IJ}$, then $\mfp=1$ and $\varkappa=0$,
  so that the last two terms in the integrand vanish. In the remaining cases, if $m\leq k_{0}+1$, then
  \begin{equation*}    
    \textstyle{\sum}_{|\bfI| = m } | [E_{\bfI},\mfp] (f) |^2
    \leq B_{k_0} \textstyle{\sum}_{|\bfJ| \leq m-1} | E_{\bfJ} f |^2
  \end{equation*}
  due to (\ref{eq:Bkzdef}) and (\ref{eq:commutatorraisedtom}) with $m$ and $\ell$ in (\ref{eq:commutatorraisedtom}) replaced by $2$ and
  $m$ respectively; note that if $J=K$, then $\bq_I+\bq_J-\bq_K=\bq_I$. Next, by Young's inequality,
  \begin{equation*}
    \begin{split}
      2 \textstyle{\sum}_{|\bfI| = m } | [E_{\bfI},\mfp] (f) |\cdot | E_{\bfI} f |
      \leq & \textstyle{\sum}_{|\bfI|=m} (  \s^{-1} | [E_{\bfI},\mfp](f) |^2+ \s  | E_{\bfI}f |^2  ) \\
    \leq & \s^{-1} B_{k_0} \textstyle{\sum}_{|\bfJ| \leq m-1} | E_{\bfJ} f |^2
        +  \s   \textstyle{\sum}_{|\bfI| = m}| E_{\bfI}f |^2.
    \end{split}
  \end{equation*}  
  As $B_{k_0} \s^{-1} \textstyle{\sum}_{j = m+1}^{k_0+1} \delta_j \leq \s \delta_m$, see \eqref{eq:DeltaInequality}, it follows that
  for $l\in\{k_{0},k_{0}+1\}$,
  \begin{equation*}
    \begin{split}
       \s^{-1} B_{k_0} \textstyle{\sum}_{m = 1}^{l}  \delta_m
      \textstyle{\sum}_{|\bfJ| \leq m-1} | E_{\bfJ}f|^2 
      = &  \s^{-1} B_{k_0} \textstyle{\sum}_{m = 1}^{l} \textstyle{\sum}_{j=0}^{m-1} \delta_m
      \textstyle{\sum}_{|\bfJ| = j}  
      | E_{\bfJ}f |^2 \\
      = &  \s^{-1} B_{k_0} \textstyle{\sum}_{m=0}^{l-1} \textstyle{\sum}_{j=m+1}^{l} \delta_j
      \textstyle{\sum}_{|\bfJ| = m}  
      |E_{\bfJ}f|^2 \\
      \leq & \s \textstyle{\sum}_{m=0}^{l-1} \delta_m \textstyle{\sum}_{|\bfJ| =m } 
      |E_{\bfJ}f|^2.
    \end{split}
  \end{equation*}
  Combining the above yields, for $l\in\{k_{0},k_{0}+1\}$,
  \begin{equation*}
    \begin{split}
      \textstyle{\sum}_{m = 0}^{l} \delta_m \textstyle{\sum}_{|\bfI| = m} 2| [E_{\bfI},\mfp] (f)
      | | E_{\bfI}f |
      \leq &
      2 \s \textstyle{\sum}_{m = 0}^{l} \delta_m\textstyle{\sum}_{|\bfI| = m}|E_{\bfI} f|^2 .
    \end{split}
  \end{equation*}
  Turning to the last term in the integrand on the right hand side of (\ref{eq:qu to be est structure coeff abs})
  (in case $\mfp\neq 1$), it contributes (ignoring the powers of $s$)
  \begin{equation*}
    \begin{split}
      - \textstyle{\sum}_{m=0}^{l}\delta_{m}\textstyle{\sum}_{|\bfI|=m}2(1-\varkappa\s- \mfp ) |E_{\bfI}f|^2
      \leq -2(5-\varkappa)\s \textstyle{\sum}_{m=0}^{l}\delta_{m}\textstyle{\sum}_{|\bfI|=m}|E_{\bfI}f|^2
    \end{split}
  \end{equation*}
  to the sum; note that, due to Lemma~\ref{le: first bound on $q_I$}, $(1-\varkappa\s)- \mfp>(5-\varkappa)\s$.
  Adding the last two estimates, the total contribution from the last two terms in the integrand on the right
  hand side of (\ref{eq:qu to be est structure coeff abs}) (for $\mfp\neq 1$) can be estimated from above by
  \[
  -2(4-\varkappa)\s \textstyle{\sum}_{m=0}^{l}\delta_{m}\textstyle{\sum}_{|\bfI|=m}\textstyle{\int}_{t}^{t_{0}}s^{-1+2(1-\varkappa\s)}|E_{\bfI}f|^2\md s.
  \]
  These terms give rise to $-2\s$ and $-4\s$ on the right hand side of (\ref{eq:inequalityforQ}).

  In order to estimate the contribution from the first term in the integrand on the right hand side of (\ref{eq:qu to be est structure coeff abs}),
  it is sufficient to appeal to Lemmata~\ref{lemma:EstFCLow}, \ref{lemma:EstSCLow}, \ref{lemma:EstFFLow}, \ref{lemma:EstSFeILow} and
  \ref{lemma:EstSFe0Low}. For example, if $f=\dee{I}{i}$, $\mfp=\bq_I$ and $\varkappa=3$, Lemma~\ref{lemma:EstFCLow} yields
  \begin{align*}
    & \textstyle{\sum}_{I,i}\textstyle{\int}_{t}^{t_0} 2 s^{2(1-3\s)} 
    \big| E_{\bfI} \big( \big( -\tfrac{\bq_{I}}{s} - \d_s \big)
    ( \dee{I}{i}) \big) \big| 
    | E_{\bfI} (\dee{I}{i})|\md s  \\
    \leq & C \ar \textstyle{\sum}_{I,i} \textstyle{\sum}_{|\bfJ| \leq |\bfI|} \textstyle{\int}_{t}^{t_0} s^{-1 + 2(1-3\s)}
    |E_{\bfJ} (\dee{I}{i})|^2 \md s+ C \textstyle{\int}_{t}^{t_0} s^{-1+\s} \mbD ( \mbD + s^{\s} ) \md s.
  \end{align*}
  In particular, it follows that 
  \begin{align*}
    & \textstyle{\sum}_{I,i}\textstyle{\sum}_{m = 0}^{k_0+1}  \delta_m \textstyle{\sum}_{|\bfI| = m} 
    \textstyle{\int}_{t}^{t_0} 2 s^{2(1-3\s)}
    \big| E_{\bfI} \big( \big( -\tfrac{\bq_{I}}{s} - \d_s \big)
    \big( \dee{I}{i} \big) \big) \big|
    \big|  E_{\bfI} (\dee{I}{i}) \big| \md s  \\
    \leq & C \ar \textstyle{\sum}_{I,i}\textstyle{\sum}_{m = 0}^{k_0+1} \delta_m \textstyle{\sum}_{|\bfI| = m}
    \textstyle{\int}_{t}^{t_0} s^{-1 + 2(1-3\s)} 
    |{E_{\bfI} (\dee{I}{i}})|^2 \md s + C \textstyle{\int}_{t}^{t_0} s^{-1+\s} \mbD \left( \mbD + s^{\s} \right) \md s.
  \end{align*}
  The remaining estimates are similar. The lemma follows. 
\end{proof}

\begin{proof}[Proof of Proposition~\ref{prop:EstimateEnergyLow}] 
By the assumptions of the proposition, we may make use
    of the conclusions of Lemma~\ref{lemma:EstimateForQ}.
Let $\ar_{\mbL} \in (0,\frac{1}{6n}]$ be such that $\mC_{1} \ar_{\mbL} \leq \s$,
where $\mC_{1}$ is the constant appearing in the statement of Lemma~\ref{lemma:EstimateForQ}.
For $\ar \in (0,\ar_{\mbL}]$, the last two terms on the right-hand side of (\ref{eq:inequalityforQ})
are then non-positive. Next, note that there is a standard constant $C_Q > 1$ such that 
\begin{equation*}
    C_Q^{-1} \mbL_{(e,\omega,\g,k,\phi)} (t)
        \leq \|\Q(t, \cdot)\|_{C^{0}(\S_t)}
        \leq C_Q \mbL_{(e,\omega,\g,k,\phi)} (t)
\end{equation*}
for any $t \in [t_{\rob}, t_0]$.
In particular, it follows now that  
\begin{equation*}
\Q(t,x)^2 \leq C_Q \mbL_{(e,\o,\g,k,\phi)}(t_0)^2 
+ C \textstyle{\int}_{t}^{t_0} s^{-1+\s} \mbD(s) (\mbD(s) + s^{\s}) \md s.
\end{equation*}
The desired inequality then follows by taking the supremum of both sides over $x \in \S$,
and then appealing to the inequality
$\mbL_{(e,\omega,\g,k,\phi)} (t) \leq C_Q \|\Q(t, \cdot)\|_{C^{0}(\S_t)}$.
\end{proof}

\subsection{The higher-order energy estimate}
We now continue by estimating
\[
\mbH_{(e,\o,\g,k,\phi)}^2:=\mbH_{(e,\o)}^2+\mbH_{(\g,k)}^2+\mbH_{(\phi)}^2.
\]
\begin{prop} \label{prop:EstimateEnergyHigh}
  Let $\s_p$, $\s_V$, $\s$, $k_0$, $k_1$, $(\S, h_{\refer})$, $(E_{i})_{i=1}^{n}$ and $V$ be as in Theorem~\ref{thm: big bang formation}
  and let $\rho_0 >0$. Then there is a standard constant $\ar_{\mbH}\in (0, \tfrac{1}{6n}]$ such that if $t_{0}\leq \tau_{H}$
  (see Lemma~\ref{lemma:AuxiliaryLapseEstimate}), $t_\rob<t_0$, $\ar \in (0, \ar_{\mbH}]$ and
  Assumption~\ref{ass:Bootstrap} is satisfied for this choice of parameters, then, for any $t \in [t_{\rob}, t_0]$,
\begin{align*}
\mbH_{(e,\o,\g,k,\phi)}(t)^2 
    &\leq \mbH_{(e,\o,\g,k,\phi)}(t_0)^2
    + C \textstyle{\int}_{t}^{t_0} s^{-1+\s} (\mbD(s) + s^{\s})(\mbD(s) + s^{3\s}) \md s.
\end{align*}
\end{prop}

We prove the above proposition via the following energy estimate:
\begin{lemma} \label{lemma:EstimateEnergyHigh}
  Let $\s_p$, $\s_V$, $\s$, $k_0$, $k_1$, $(\S, h_{\refer})$, $(E_{i})_{i=1}^{n}$ and $V$ be as in Theorem~\ref{thm: big bang formation}.
  Let $\varkappa_{A}:=(n+1)(2+3\s)$ and $\rho_0 >0$. Then there is a standard constant $\mC_{2}$ such that if $t_{0}\leq \tau_{H}$
  (see Lemma~\ref{lemma:AuxiliaryLapseEstimate}), $t_\rob<t_0$, $\ar \in (0, \tfrac{1}{6n}]$ and
  Assumption~\ref{ass:Bootstrap} is satisfied for this choice of parameters, then, for any $t \in [t_{\rob}, t_0]$,
\begin{align}
\begin{split} \label{eq:HighEnergyEstimate}
\mbH_{(e,\o,\g,k,\phi)}(t)^2 
    \leq & \mbH_{(e,\o,\g,k,\phi)}(t_0)^2
    + \big(\mC_{2} \ar + 2(\varkappa_{A}-A)\big) \textstyle{\int}_{t}^{t_0}
        s^{-1} \mbH_{(e,\o,\g,k,\phi)}(t)^2 ds \\
    & + C \textstyle{\int}_{t}^{t_0} s^{-1+\s}  (\mbD(s) + s^{\s})(\mbD(s) + s^{3\s}) \md s.
\end{split}
\end{align}
\end{lemma}
\begin{proof}
The proof is similar to the proof of Proposition~\ref{prop:EnergyFC}. We begin by writing 
\begin{equation} \label{eq:FTCHigherOrderEnergy}
\mbH_{(e,\o,\g,k,\phi)}(t)^2 - \mbH_{(e,\o,\g,k,\phi)}(t_0)^2
= - \textstyle{\int}^{t}_{t_0} \d_s \left[\mbH_{(e,\o,\g,k,\phi)}(s)^2\right] \md s.
\end{equation}
The estimate for the energy of the components of the frame and co-frame follows immediately from Proposition~\ref{prop:EnergyFC}.
We therefore focus on
\begin{equation*}
  \begin{split}
    & -\d_t\big(\tfrac{1}{2}\| \dega{} \|_{H^{k_1}(\S_t)}^2
    + \| \dek{} \|_{H^{k_1}(\S_t)}^2   + \| \vdep \|_{H^{k_1}(\S_t)}^2
    + \| \dep{0} \|_{H^{k_1}(\S_t)}^2\big) 
    = \textstyle{\sum}_{i=1}^{4}\mfE_{i}.
  \end{split}
\end{equation*}
Here, 
\begin{equation*}
  \begin{split}
    \mfE_{i} := &\textstyle{\sum}_{I,J,K}  \langle E^{i,\g}_{IJK},
    \dega{IJK}  \rangle_{H^{k_1}(\S_t)} 
    +  \textstyle{\sum}_{I,J}  \langle E^{i,k}_{IJ},
    \dek{IJ}  \rangle_{H^{k_1}(\S_t)} \\
    & +\textstyle{\sum}_{I}  \langle E^{i,\phi}_I,
    \dep{I}  \rangle_{H^{k_1}(\S_t)} 
    + \langle E^{i,\phi}_{0},
    \dep{0}  \rangle_{H^{k_1}(\S_t)}, 
  \end{split}
\end{equation*}
where we use the shorthand notation $\bq_{IJK}:=\bq_I + \bq_J - \bq_K$ and
\begin{align*}
  E^{1,\g}_{IJK} := & \big(-\tfrac{\bq_{\uI\uJ\uK}}{t}-\d_t\big)
  (\dega{\uI\uJ\uK}) - 2 e_I (N) \chk_{JK} - 2 N e_I (k_{JK}), \\
  E^{1,k}_{IJ} := & 2\big(-\tfrac{1}{t}-\d_t\big) (\dek{IJ}) -
  2t^{-1} (N-1) \chk_{IJ} + 2e_I e_J (N)\\
  & - 2N  e_K (\g_{KIJ}) - 2N e_I (\g_{JKK}), \\
  E^{1,\phi}_{I} := & 2\big(-\tfrac{\bq_{\uI}}{t}-\d_t\big)(\dep{\uI})+2e_I(N) \chdtp
   + 2N e_I \dtp, \\
  E^{1,\phi}_{0} := & 2\big(-\tfrac{1}{t}-\d_t\big)
  (\dep{0}) - 2t^{-1} (N-1) \chdtp + 2N e_I e_I (\phi).
\end{align*}
Moreover, 
\begin{align*}
  E^{2,\g}_{IJK} :=  \tfrac{\bq_{\uI\uJ\uK}}{t}\dega{\uI\uJ\uK}, \ \
  E^{2,k}_{IJ} :=  \tfrac{2}{t} \dek{IJ}, \ \
  E^{2,\phi}_{I} :=  \tfrac{2\bq_{\uI}}{t}\dep{\uI},\ \ 
  E^{2,\phi}_{0} :=  \tfrac{2}{t} \dep{0}.
\end{align*}
In addition,
\begin{align*}  
  E^{3,\g}_{IJK} := &   2e_I (N) \chk_{JK}, \ \ \ 
   E^{3,k}_{IJ} :=  2t^{-1} (N-1) \chk_{IJ}, \ \ \
  E^{3,\phi}_{I} := - 2e_I(N) \chdtp, \\
  E^{3,\phi}_{0} := &   2t^{-1} (N-1) \chdtp.
\end{align*}
Finally,
\begin{align*}
  E^{4,\g}_{IJK} := & 2N e_I (k_{JK}), \ \
   E^{4,k}_{IJ} :=  -  2e_I e_J (N) + 2N  e_K (\g_{KIJ}) + 2N e_I (\g_{JKK}) , \\
  E^{4,\phi}_{I} := &  - 2N e_I \dtp, \ \
   E^{4,\phi}_{0} := - 2N e_I e_I (\phi).
\end{align*}
To estimate $\mfE_{1}$, we use the Cauchy-Schwarz inequality for the $H^{k_1}$-inner product (and for finite dimensional
sums); multiply both sides of the resulting inequality by $t^{2(A+1)}$; and appeal to 
Lemmata~\ref{lemma:EstimateSCHigh}, \ref{lemma:EstimateFFHigh},  
\ref{lemma:EstimateSFeIHigh} and \ref{lemma:EstimateSFe0High}. 
The conclusion is that 
\begin{equation*}
  \begin{split}
    t^{2(A+1)}\mfE_{1}\leq & t^{-1} C \ar  \big( \mbH_{(\g,k)}^2 + \mbH_{(\phi)}^2 \big) + C  t^{-1+\s} \mbD (\mbD + t^{\s}).
  \end{split}
\end{equation*}
To estimate $\mfE_{2}$, note, first, that the contribution from $E^{2,k}_{IJ}$ and $E^{2,\phi}_{0}$ is
\begin{equation*}
  \begin{split}
    & \tfrac{2}{t} \big( \|\dek{} \|_{H^{k_1}(\S_t)}^2    
    + \| \dep{0} \|_{H^{k_1}(\S_t)}^2 \big).
  \end{split}
\end{equation*}
Next, by (\ref{eq: gamma bar p bar assumption}) and (\ref{eq:noGNSInnerProduct}) with $\eta = \s$, combined with
Remark~\ref{remark:supremum iso Cz norm},
\begin{equation*}
  \begin{split}
    & \textstyle{\sum}_{I,J,K} \big \langle \tfrac{\bq_{IJK}}{t} \dega{IJK},\dega{IJK} \big \rangle_{H^{k_1}(\S_t)}  \\
    \leq & t^{-1} \big( \s + \textstyle{\max}_{I \neq J}
    \sup_{x\in\S} \bq_{IJK}(t,x) \big) \| \dega{} \|_{H^{k_1}(\S_t)}^2 
    + C t^{-1}
    \|\dega{} \|_{C^{k_{0}}(\S_t)} \| \dega{} \|_{H^{k_1}(\S_t)} \\
    \leq & t^{-1} \big(2 - 8 \s \big) \cdot \tfrac{1}{2} \|\dega{} \|_{H^{k_1}(\S_t)}^2 
    + C t^{-1} \| \dega{} \|_{C^{k_{0}}(\S_t)} \| \dega{} \|_{H^{k_1}(\S_t)}.
  \end{split}
\end{equation*}
Similarly, due to (\ref{eq:bqboundsb}),
\begin{equation*}
  \begin{split}
    & 2 \big| \textstyle{\sum}_{I } 
    \left  \langle \tfrac{\bq_I}{t} \dep{I},
    \dep{I}\right\rangle_{H^{k_1}(\S_t)} \big| \\
    \leq & 2t^{-1} \big( \s + \textstyle{\max}_{I} \| \bq_I \|_{C^{0}(\S_t)} \big)
    \| \vdep \|_{H^{k_1}(\S_t)}^2 
    + C t^{-1} 
    \|\vdep \|_{C^{k_{0}}(\S_t)} \| \vdep \|_{H^{k_1}(\S_t)} \\
    \leq & t^{-1} \big(2 - 8 \s \big) \| \vdep \|_{H^{k_1}(\S_t)}^2
    + C  t^{-1} \| \vdep \|_{C^{k_{0}}(\S_t)}
    \| \vdep \|_{H^{k_1}(\S_t)}.
  \end{split}
\end{equation*}
Thus
\begin{equation*}
  \begin{split}
    t^{2(A+1)}\mfE_{2}\leq & 2 t^{-1}\big( \mbH_{(\g,k)}^2 + \mbH_{(\phi)}^2 \big)
    + C t^{-1+\s} \mbD (\mbD + t^{\s}).
  \end{split}
\end{equation*}
Next, note that by the definition of $\chk$, and since $t\chdtp=t_{0}\bp_{1}=\bP_{1}$,
\begin{equation*}
  \begin{split}
    \mfE_{3} = & 2t^{-1} \big[  \textstyle{\sum}_{I,J}  \langle e_I (N) \bq_{J}, 
    \dega{IJJ}  \rangle_{H^{k_1}(\S_t)} 
    +\textstyle{\sum}_{J}  \langle t^{-1} (N-1) \bq_{J},
    \dek{JJ} \rangle_{H^{k_1}(\S_t)} \\
    & \phantom{t^{-1} \big[}- \textstyle{\sum}_{I}  \langle e_I (N) \bP_{1},
    \dep{I}  \rangle_{H^{k_1}(\S_t)} 
    +  \langle t^{-1} (N-1) \bP_{1} ,
    \dep{0} \rangle_{H^{k_1}(\S_t)} \big].
  \end{split}
\end{equation*}
Next we appeal to (\ref{eq:noGNSInnerProduct2}) with $\eta = \s/4$; $m=n+1$; and the following $\varphi_i$, $\psi_j$, $\pi_{ij}$:
\begin{itemize}
\item $\varphi_I = \bq_I$, $\psi_I = e_I(N)$, $\varphi_{n+1} = -\bP_1$ and $\psi_{n+1} = t^{-1} (N-1)$,
\item $\pi_{IJ} = \dega{J\uI\uI}$, $\pi_{In+1} = \dek{\uI\uI}$, $\pi_{n+1I} = \dep{I}$ and $\pi_{n+1n+1} = -\dep{0}$,
\end{itemize}
where $I,J \in \{1,..,n\}$. In particular, keeping (\ref{eq: gamma bar p bar assumption}) and (\ref{eq: scalar field bar assumption})
in mind,
\begin{equation*}
  \begin{split}
    t^{2(A+1)}\mfE_{3} \leq & 2t^{-1} \big( \s/4 + \| \textstyle{\sum}_{I}
    \bq_I^2 + \bP_1^2 \|_{C^{0}(\S_t)}^{1/2} \big)
    \mbH_{(N)} \big(\mbH_{(\g,k)}^2 + \mbH_{(\phi)}^2 \big)^{1/2} \\
    &+ C t^{A-1+\s}\mbL_{(N)} \big(\mbH_{(\g,k)} + \mbH_{(\phi)} \big) \\
    \leq & t^{-1} \big(Cr + (2 + \s)^2 \big)
    \big(\mbH_{(\g,k)}^2 + \mbH_{(\phi)}^2 \big)
    + C t^{-1+\s} \mbD (\mbD + t^{\s}),
  \end{split}
\end{equation*}
where we appealed to (\ref{eq:BoundEigenvalues}) and Lemmata~\ref{lemma:LowerOrderLapseEstimates} and \ref{lemma:HigherOrderLapseEstimate}.

Finally, note that $\mfE_{4}$ can be written
\begin{equation}\label{eq:mfEfour second version}
  \begin{split}
    \mfE_{4}
    = & 2(\left \langle N e_I (\dek{JK}), 
    \dega{IJK} \right \rangle_{H^{k_1}(\S_t)}
    + \left \langle \dek{JK}, 
    N e_I  (\dega{IJK}) \right \rangle_{H^{k_1}(\S_t)}) \\
    &+ 2(\left \langle N e_I (\dek{IJ}), 
    \dega{JKK} \right \rangle_{H^{k_1}(\S_t)}
    +  \left \langle \dek{IJ} , 
    Ne_I(\dega{JKK}) \right \rangle_{H^{k_1}(\S_t)} ) \\
    &- 2  (\left \langle e_I e_J (N), 
    \dek{IJ} \right \rangle_{H^{k_1}(\S_t)}
    +  \left \langle e_J(N), 
    e_I (\dek{IJ}) \right \rangle_{H^{k_1}(\S_t)} )  - 2\Xi(t) \\
    & - 2 ( \left \langle Ne_I (\dep{0}), 
    \dep{I} \right \rangle_{H^{k_1}(\S_t)}
    + \left \langle Ne_I (\dep{I}), 
    \dep{0} \right \rangle_{H^{k_1}(\S_t)}  )- 2\Z(t).
  \end{split}  
\end{equation}
On the other hand, appealing to Lemmata~\ref{lemma:EstimateForZ} and \ref{lemma:EstimateForXi},
\begin{equation*}
  \begin{split}
    2 t^{2(A+1)} ( | \Z |  + | \Xi | )\leq & (C \ar + 4n(1+\s)) t^{-1} \big(\mbH_{(\g,k)}^2 + \mbH_{(\phi)}^2 \big)\\
    & + C t^{-1 + \s}  (\mbD + t^{\s})(\mbD + t^{3\s}).
  \end{split}
\end{equation*}
The remaining terms can be estimated by appealing to the divergence theorem in the form of Lemma~\ref{lemma:DivergenceEstimate}.
As a preparation, note that 
\begin{align*}
    \|\rodiv_{h_{\refer}} (N e_I) \|_{C^{0}(\S_t)}
    \leq & \| e_I(N) \|_{C^{0}(\S)} 
    + \| N \|_{C^{0}(\S)} \| \rodiv_{h_{\refer}} (e_I) \|_{C^{0}(\S)} \\
    \leq & C t^{-1+3\s} (\mbD + t^{\s}),\\
    \| N e_I^i \|_{C^{1}(\S_t)}  \leq & C \|N \|_{C^{1}(\S_t)} \|e \|_{C^{1}(\S_t)}\leq C t^{-1+3\s} (\mbD + t^{\s})
\end{align*}
due to Lemmata~\ref{lemma:estimatingthedynamicalvariables} and \ref{lemma:LowerOrderLapseEstimates}. Moreover,
\begin{align*}
  \| N e_I^i \|_{H^{k_1} (\S_t)}
    \leq & C\big[ \big(1 +  \| N -1 \|_{C^{0}(\S_t)} \big) \| e \|_{H^{k_1}(\S_t)} 
        + \| N-1\|_{H^{k_1}(\S_t)} \|e \|_{C^{0}(\S_t)} \big] \\
    \leq & C t^{-A-1 + 2\s} (\mbD + t^{\s}),
\end{align*}
due to Lemmata~\ref{lemma:estimatingthedynamicalvariables} and \ref{lemma: products in L2}. With these estimates at hand, note that
\begin{equation*}
  \begin{split}
    & t^{2(A+1)}  | \langle N e_I (\dek{JK}),\dega{IJK}  \rangle_{H^{k_1}(\S_t)} 
    +  \langle \dek{JK},N e_I  (\dega{IJK})  \rangle_{H^{k_1}(\S_t)} | \\
    \leq & C t^{-1+\s} \mbD (\mbD + t^{\s}),
  \end{split}
\end{equation*}
where we appealed to Assumption~\ref{ass:Bootstrap} as well as Lemmata~\ref{lemma:estimatingthedynamicalvariables}
and \ref{lemma:DivergenceEstimate}. The second term on the right hand side of (\ref{eq:mfEfour second version})
satisfies the same bound. Next, 
\begin{equation*}
  \begin{split}
    & t^{2(A+1)}  |  \langle e_I e_J (N), 
    \dek{IJ}  \rangle_{H^{k_1}(\S_t)}  +  \langle e_J(N),
    e_I (\dek{IJ})  \rangle_{H^{k_1}(\S_t)} | \\
    \leq & C t^{-1+\s} \mbD (\mbD + t^{\s})
  \end{split}
\end{equation*}
due to similar arguments. 
Finally, a similar argument yields
\begin{equation*}
  \begin{split}
    & t^{2(A+1)} |  \langle Ne_I (\dep{0}), \dep{I}  \rangle_{H^{k_1}(\S_t)}
    +  \langle Ne_I (\dep{I}), \dep{0}  \rangle_{H^{k_1}(\S_t)}  | 
    \leq C t^{-1+\s} \mbD (\mbD + t^{\s}).
  \end{split}
\end{equation*}
To summarize,
\begin{equation*}
  t^{2(A+1)}\mfE_{4}\leq  t^{-1} \big( C \ar + 4n(1+\s) \big)
  ( \mbH_{(\g,k)}^2 + \mbH_{(\phi)}^2 )
  + C  t^{-1+\s} (\mbD + t^{\s}) (\mbD + t^{3\s}).
\end{equation*}
Combining the above estimates yields
\begin{equation*}
  \begin{split}
    & -t^{2(A+1)} \d_t (\tfrac{1}{2}\| \dega{} \|_{H^{k_1}(\S_t)}^2
    + \| \dek{} \|_{H^{k_1}(\S_t)}^2   + \| \vdep \|_{H^{k_1}(\S_t)}^2
    + \| \dep{0} \|_{H^{k_1}(\S_t)}^2) \\
    \leq & t^{-1} \big(\mC_{2} r + 4n(1+\s) + 6 + 4\s + \s^2 \big)
    ( \mbH_{(\g,k)}^2 + \mbH_{(\phi)}^2 )+ C  t^{-1+\s} (\mbD + t^{\s})(\mbD + t^{3\s}),
  \end{split}
\end{equation*}
where $\mC_{2}$ is a standard constant. Commuting $t^{2(A+1)}$ with the operator $\d_t$ and combining the result with
(\ref{eq:EnergyFC}) and \eqref{eq:FTCHigherOrderEnergy}, the lemma follows. 
\end{proof}

\begin{proof}[Proof of Proposition~\ref{prop:EstimateEnergyHigh}]
By the assumptions of the proposition, we may make use
    of the conclusions of Lemma~\ref{lemma:EstimateEnergyHigh}.
Let $\ar_{\mbH} \in (0,\frac{1}{6n}]$ be such that $\mC_{2} \ar_{\mbH} \leq \s$,
where $\mC_{2}$ is the constant appearing in the statement of Lemma~\ref{lemma:EstimateEnergyHigh}.
Then we find that if $\ar \in (0,\ar_{\mbH}]$,  the second term in \eqref{eq:HighEnergyEstimate}
is non-positive; recall that $A = 2(n+1)(1+2\s)$.
\end{proof}

\subsection{Proof of the bootstrap improvement}
\label{sec:proofs}
Finally, we prove Theorem~\ref{thm:BootstrapImprovement}.

\begin{proof}[Proof of Theorem~\ref{thm:BootstrapImprovement}]
Let $\tau_{\rob} := \tau_H$ where $\tau_H$ is as in the statement of Lemma~\ref{lemma:AuxiliaryLapseEstimate}.
Observe that $\tau_{\rob}$ is a standard constant. Moreover, the following statement now follows directly from  
Propositions~\ref{prop:EstimateEnergyLow} and \ref{prop:EstimateEnergyHigh} above:
if $t_0 \leq \tau_{\rob}$ and if Assumption~\ref{ass:Bootstrap} is satisfied
for $\ar_\rob:= \min \{ \ar_{\mbL}, \ar_{\mbH}, \frac{1}{6n}\}$ on $[t_{\rob}, t_0]$,
where we note that $\ar_\rob$ is a standard constant, then the following inequality holds:
\begin{equation*}
  \mbD(t)^2 \leq C \mbD(t_0)^2+ C \textstyle{\int}_{t}^{t_0} s^{-1+3\s} \md s + C \textstyle{\int}_{t}^{t_0} s^{-1+\s} \mbD(s)^2 \md s.
\end{equation*}
By Grönwall's lemma and the fact that $t_{0}\leq 1$, it follows that
\begin{align*}
  \mbD(t)^2  &\leq [C \mbD(t_0)^2 + C ( t_0^{3\s} - t^{3\s} ) ] 
        \exp[C ( t_0^{\s} - t^{\s}) ] \leq C [\mbD(t_0)^2 + t_0^{3\s} ].
\end{align*}
Since $\mbD(t_0)\leq t_0^\s$ by assumption, it follows that $\mbD(t) \leq C t_0^{\s}$. Combining this estimate
with (\ref{eq:LapseControlLow}), (\ref{eq:EstLapseHigh}) and $\ar \leq 1$ yields
\begin{equation} \label{eq:ImprovedBootstrap}
\mbD(t) + \mbL_{(N)}(t) + \mbH_{(N)}(t) \leq \mathcal{C} t_0^{\s}
\end{equation}
for any $t \in [t_{\rob}, t_0]$, where $\mathcal{C}$ is a standard constant.
\end{proof}


\section{Asymptotics and curvature blowup}
\label{sec:asymptotics}
With the global existence result at hand, we can continue 
with formulating conclusions about the resulting spacetimes.
In Theorem~\ref{thm:Asymptotics} below, we obtain asymptotic information
for the components of the expansion-normalized Weingarten map
with respect to the Fermi-Walker propagated frame,
and similarly for the expansion-normalized time-derivative of the scalar field.
This suffices to show that the Kretschmann scalar
as well as the spacetime Ricci tensor contracted with itself both blow up as $t^{-4}$.
However, one major caveat is that we do not obtain any information regarding
the eigenspaces of the expansion-normalized Weingarten map. In what follows, note that
on a constant-$t$ slice, the functions $\Phi_0$, $\Phi_1$ introduced in Definition~\ref{def:Phizodef}
are given by $\Phi_1 = t e_0(\phi)$ and $\Phi_0 = \phi - \ln(t) \Phi_1$. Moreover, the components of
the expansion-normalized Weingarten map with respect to the Fermi-Walker propagated frame are given
by $\mK_{I}^{~J}~:=~t k_{I}^{~J}~=~t k_{IJ}$.

\begin{thm} \label{thm:Asymptotics}
  Let $\s_p$, $\s_V$, $\s$, $k_0$, $k_1$, $(\S, h_\refer)$, $(E_{i})_{i=1}^{n}$ and $V$ be as in Theorem~\ref{thm: big bang formation}
  and let $\rho_0 > 0$. Then there exists a standard constant $\tau_2 \leq \tau_1$, where  $\tau_1$ is as in
  Theorem~\ref{thm:GlobalExistenceFermiWalker}, such that the following holds: If $t_0 < \tau_2$; if $\be_I^i$, $\bq_I$, $\bp_0$,
  $\bp_1\in C^{\infty}(\S,\rn{})$ form diagonal FRS initial data satisfying the non-degenerate FRS expansion-normalized
  bounds of regularity
  $k_1$ for $\rho_0$ at $t_0$ as well as (\ref{eq:hamconstraint}) and (\ref{eq:MC}); if $\hat{e}_{I}^{i}$, $\hat{k}_{IJ}$, $\hat{\phi}_0$,
  $\hat{\phi}_1\in C^{\infty}(\S,\rn{})$ are initial data to (\ref{eq: transport frame})--(\ref{eq:LapseEquation}) satisfying
  $\hat{k}_{II}=1/t_0$ and $\mbD(t_0)\leq t_0^\s$, then the corresponding solution to (\ref{eq: transport frame})--(\ref{eq:LapseEquation}),
  as given in Theorem~\ref{thm:GlobalExistenceFermiWalker}, has the following properties:

\noindent \textbf{Asymptotic data:}
There exists functions $\mrK_{I}^{~J} \in C^{k_{0}+1}(\S)$ 
which at every $x \in \S$ form the components of
a symmetric matrix with distinct eigenvalues $\mrp_I \in C^{k_{0}+1}(\S)$,
as well as functions $\mrPhi_0$, $\mrPhi_1 \in C^{k_{0}+1}(\S)$,
which satisfy the estimates
\begin{subequations} \label{eq:Asymptotics} 
\begin{align} 
  \textstyle{\sum}_{I,J} \| \mK_{I}^{~J}(t,\cdot)- \mrK_{I}^{~J} \|_{C^{k_{0}+1}(\S)}
  &\leq C t_0^{\s} t^{\s},\label{eq:FFConvergenceRate}\\
  \| \Phi_0(t,\cdot) - \mrPhi_0 \|_{C^{k_{0}+1}(\S)}
  + \| \Phi_1(t,\cdot) - \mrPhi_1 \|_{C^{k_{0}+1}(\S)}
  &\leq C t_0^{\s} t^{\s},\label{eq:Scalar field convergence rate}
\end{align}
\end{subequations}
for any $t \in (0,t_0]$.
Moreover, the eigenvalues $\mrp_I(x)$ of the matrix $(\mrK_{I}^{~J}(x))_{I,J}$ satisfy
the generalized Kasner conditions $\sum_{I} \mrp_I= 1$,
$\sum_{I} \mrp_I^2 + \mrPhi_1^2 = 1$ and the condition $\mrp_I + \mrp_J - \mrp_K <1$ ($I\neq J$) on $\S$, as well as the estimates
\begin{align}
\| p_I(t,\cdot) - \mrp_I \|_{C^{k_0 +1}(\S_t)} \leq C t_0^{\s} t^{\s}
\end{align}
for any $I$ and $t \in (0,t_0]$. Here $p_I(t,\cdot)$ denote the continuous curves of eigenvalues of the matrices $\mK_{I}^{~J}(t,\cdot)$,
such that $p_I(t_0,\cdot) = \bq_I$.

\noindent \textbf{Curvature blow-up:}
The Kretschmann scalar and the Ricci curvature contracted with itself,
respectively given by 
    $\mathfrak{K}_g := \roRiem_{g,\mu \nu \xi \rho} \roRiem_g^{\mu \nu \xi \rho}$,
    $\mathfrak{R}_g := \roRic_{g, \mu \nu} \roRic_{g}^{\mu \nu}$
satisfy the following estimates:
\begin{subequations}
\begin{align}
\left \| t^4 \mathfrak{K}_g(t,\cdot) - 
    4 \left[ \textstyle{\sum}_{I} \mrp_I^2 (1 - \mrp_I)^2 
        + \textstyle{\sum}_{I < J} \mrp_I^2 \mrp_J^2 \right]
    \right \|_{C^{k_{0}+1}(\S)}
\leq C t_0^{2\s} t^{2\s}, \label{eq:Asymptotics Kretschmann}\\
t^4  \| \mathfrak{R}_g(t,\cdot) - \mrPhi_1^2
     \|_{C^{k_{0}+1}(\S)}
\leq C t_0^{2\s} t^{2\s}.\label{eq:Asymptotics Ricci squared}
\end{align}
\end{subequations}
Finally, all causal geodesics are past incomplete, and $\mathfrak{K}_g$ blows up along all past inextendible
causal curves. 
\end{thm}
\begin{proof}
The proof is similar to the one for the analogous statements in \cite{GIJ}.
The main difference is that in this case the smallness parameter is $t_0$,
and we require the eigenvalues to remain simple all the way up to 0.

To begin with, let $\rho_0 > 0$ and assume that we have diagonal FRS initial data as in (\ref{eq:beiIetcindata}) and
initial data as in (\ref{eq: hat initial data}). Note that we here insist that the diagonal FRS initial data
satisfy the \textit{non-degenerate} FRS expansion-normalized bounds of regularity $k_1$ for $\rho_0$
at $t_0$, as well as (\ref{eq:hamconstraint}) and (\ref{eq:MC}). In the course of
the proof, we gradually impose stronger restrictions on the standard constant $\tau_2$.
But to begin, let $\tau_2 \leq \tau_1$ where $\tau_1$ is as in Theorem~\ref{thm:GlobalExistenceFermiWalker}.
This means that we have a past global solution satisfying \eqref{eq:AposterioriBoundsGlobalExistence}.
By letting $\tau_2>0$ be small enough that $t_0 < \tau_2$ implies that $\mathcal{C} t_0^{\s} \leq 1/(6n)$,
we may assume that Assumption~\ref{ass:Bootstrap} is satisfied for some $r \leq 1/(6n)$ on $(0, t_0]$.
In particular, we may make use of all the estimates of Section~\ref{sec:MainEstimates}.

We begin by demonstrating the statements regarding the asymptotic data. If $[t_1,t_2] \subset (0,t_0]$,
integrating \eqref{eq:EstimateFFLow} from $t_1$ to $t_2$ yields
\begin{align*}
 & \| t_2 \dek{IJ} (t_2,\cdot) 
    - t_1 \dek{IJ} (t_1,\cdot) \|_{C^{k_{0}+1}(\S)}\\
    = &  \| \textstyle{\int}_{t_1}^{t_2} 
    \d_s (s \dek{IJ})(s,\cdot) \md s \|_{C^{k_{0}+1}(\S)}\\
    \leq & \textstyle{\int}_{t_1}^{t_2}
    \| \d_s (s \dek{IJ})(s,\cdot) \md s \|_{C^{k_{0}+1}(\S)}
    \leq C \textstyle{\int}_{t_1}^{t_2} s^{-1 + \s} (\mbD(s) + s^{\s}) \md s.
\end{align*}
Since $\mbD(s) + s^{\s} \leq C t_0^{\s}$ due to \eqref{eq:AposterioriBoundsGlobalExistence},
\begin{equation}
\| t_2 (\dek{IJ}) (t_2,\cdot) 
    - t_1 (\dek{IJ}) (t_1,\cdot) \|_{C^{k_{0}+1}(\S)}
\leq C t_0^{\s} (t_2^{\s} - t_1^{\s}),
\end{equation}
for any $[t_1, t_2] \subset (0,t_0]$.
In particular, as $t_1 \chk_{IJ} (t_1, \cdot) = t_2 \chk_{IJ}(t_2, \cdot) = \bq_\uI \delta_{\uI J}$
for any $t_1, t_2 \in (0,t_0]$, 
it follows that $t k_{IJ} = \mK_{I}^{~J}$ converges in $C^{k_{0}+1}(\S)$
as $t \downarrow 0$, since it forms a Cauchy sequence in a complete space.
We denote the limit by $\mrK_{I}^{~J}$. Moreover, for any $[t_1, t_2] \subset (0,t_0]$,
\begin{equation} \label{eq:HolderEstimateFF}
\| \mK_{I}^{~J}(t_2, \cdot) - \mK_{I}^{~J}(t_1,\cdot) \|_{C^{k_{0}+1}(\S)} 
    \leq C t_0^{\s} (t_2^{\s} - t_1^{\s}).
\end{equation}
In fact, letting $\mK_{I}^{~J}(0,\cdot) := \mrK_{I}^{~J}(\cdot)$, it follows that $\mK_{I}^{~J} \in C^{0,\s}([0,t_0], C^{k_{0}+1}(\S))$,
that (\ref{eq:FFConvergenceRate}) holds, and that
\begin{align}
\|\bq_I \delta_{I}^{~J} - \mrK_{I}^{~J}  \|_{C^{k_{0}+1}(\S_t)} &\leq C t_0^{2\s}.
\end{align}
The argument to obtain the function $\mrPhi_1$ is similar. Let $[t_1, t_2] \subset (0, t_0]$. Integrating
\eqref{eq:EstSFe0Low} from $t_1$ to $t_2$ yields, omitting spatial arguments,
\begin{align}
\begin{split} \label{eq:HolderEstimateSF}
\| t_2 \dtp (t_2) - t_1 \dtp(t_1) \|_{C^{k_{0}+1}(\S)}
    &= \| \Phi_1(t_2) - \Phi_1(t_1) \|_{C^{k_{0}+1}(\S)}\leq C t_0^{\s} (t_2^{2\s} - t_1^{2\s}).
\end{split}
\end{align}
As above, there thus exists a function $\mrPhi_1 \in C^{k_0 +1}(\S)$ such that 
\begin{equation} \label{eq:ScalarFieldConvergenceRate}
\| \Phi_1(t) - \mrPhi_1 (\cdot) \|_{C^{k_{0}+1}(\S)}
\leq C t_0^{\s} t^{2\s}.
\end{equation}
On the other hand,  for any $[t_1, t_2] \subset (0,t_0]$,
\begin{align*}
  \begin{split}
    &\Phi_0(t_2) - \Phi_0(t_1)\\
    = & \textstyle{\int}_{t_1}^{t_2} \d_s \phi(s) \md s
        -\ln(t_2) \Phi_1(t_2) + \ln(t_1) \Phi_1(t_1) \\
    = & \textstyle{\int}_{t_1}^{t_2} s^{-1} \big(N \Phi_1 (s) - \mrPhi_1 \big) \md s
        - \ln(t_2) \big( \Phi_1(t_2) - \mrPhi_1 \big) 
        + \ln(t_1) \big( \Phi_1(t_1) - \mrPhi_1 \big),
\end{split}
\end{align*}
since $\d_{t}\phi=t^{-1}N\Phi_{1}$. Hence,
\begin{align}
  \begin{split}
    & \| \Phi_0(t_2) - \Phi_0(t_1) \|_{C^{k_{0}+1}(\S)} \\
    \leq & \textstyle{\int}_{t_1}^{t_2} s^{-1} 
        \big( \|(N-1) \Phi_1(s) \|_{C^{k_{0}+1}(\S)} 
        +  \| \Phi_1(s) - \mrPhi_1 \|_{C^{k_{0}+1}(\S)} \big) \md s \\
    & + \langle \ln(t_2) \rangle 
        \| \Phi_1(t_2) - \mrPhi_1 \|_{C^{k_{0}+1}(\S)} 
    + \langle \ln(t_1) \rangle 
        \| \Phi_1(t_1) - \mrPhi_1\|_{C^{k_{0}+1}(\S)}  \\
    \leq & C t_0^{\s} \textstyle{\int}_{t_1}^{t_2} s^{-1+\s}  \md s 
        + C t_0^{\s} {t_2}^{\s} + C t_0^{\s} t_1^{\s}\leq C t_0^{\s} \big( t_2^{\s} + t_1^{\s} \big),
\end{split}
\end{align}
due to \eqref{eq:ScalarFieldConvergenceRate}, Lemma~\ref{lemma:LowerOrderLapseEstimates} and the bound
$\langle \ln(t) \rangle t^{\s} \leq C$. Again, $\Phi_0(t,\cdot)$ converges to a limit in $C^{k_{0}+1}(\S)$
and (\ref{eq:Scalar field convergence rate}) holds, recalling \eqref{eq:ScalarFieldConvergenceRate}. 

Now let us consider the statements regarding the eigenvalues of $\mK_{I}^{~J}$,
which are also needed in order to prove the statements concerning curvature blowup.
To begin, we claim that there is a standard constant $\tau_{2}\leq\tau_{1}$ such that
if $t_0 \leq \tau_2$ and if $\q_1(t,x),\dots,\q_n(t,x)$ denote the eigenvalues
of the matrices with components $K_{I}^{~J} (t,x)$, then the eigenvalues are simple
and there is a standard constant $d_1$ such that the distances between distinct
eigenvalues are bounded from below by $d_{1}$ on $[0, t_0] \times \S$.
Moreover, they may be ordered from largest to smallest and so that
\begin{equation} \label{eq:HolderEstimateEigenvalues}
\|\q_I(t_2) - \q_I(t_1) \|_{C^{k_{0}+1}(\S)} \leq Ct_0^{\s} (t_2^{\s} - t_1^{\s}),
\end{equation}
for any $[t_1,t_2] \subset [0,t_0]$. In particular, letting $t_2 = t_0$, $t_1 = t$, this becomes
\begin{equation} \label{eq:HolderEstimateEigenvaluesInitial}
\|\hat{p}_I - \q_I(t) \|_{C^{k_{0}+1}(\S)} \leq Ct_0^{\s} (t_0^{\s} - t^{\s}),
\end{equation}
where $\hat{p}_I$ are the eigenvalues of the expansion-normalized Weingarten map of the initial data (\ref{eq: hat initial data}).
Note that the $\hat{p}_I$ can be assumed to satisfy $|\hat{p}_I-\hat{p}_J|>1/(2\rho_0)$ for $I\neq J$, assuming $\tau_2$ to be a small
enough standard constant. It follows that $\q_I \in C^{0,\s}([0,t_0],C^{k_{0}+1}(\S))$ and, due to (\ref{eq: bar pi diagonal conditions}) and
$\mbD(t_0)\leq t_0^{\s}$, that 
\begin{equation}
p_I + p_J - p_K \leq 1 - \s_p + C t_0^{\s}
\end{equation}
on $[0,t_{0}]\times\S$ for any $I \neq J$. In particular, if $\tau_2$ is a small enough standard constant, then $t_0 \leq \tau_2$ implies
$\mrp_I + \mrp_J - \mrp_K < 1$ on $\S$ for any $I \neq J$.

In order to prove the claim, we apply Lemma~\ref{lemma:nondegeneracy}. More specifically, let $\ell=k_{0}+1$, 
$C_{\ell}=(C_1+1)\rho_{0}$ (where $C_1$ is the constant associated with Sobolev embedding from $H^{k_1+2}$ to $C^{k_0+1}$), $K_{\ell}=C$
(where $C$ is the constant appearing in \eqref{eq:EstimateFFLow}), $\zeta_0=2\rho_0$,
$\alpha=\s$ and $L=2$. Then, due to the fact that $|\hat{p}_I-\hat{p}_J|>1/(2\rho_0)$ for $I\neq J$;
(\ref{eq: gamma bar p bar assumption}), (\ref{eq:HolderEstimateFF}) and $\mbD(t_0)\leq t_0^\s$; \eqref{eq:EstimateFFLow}; and
Lemma~\ref{lemma:LowerOrderHamiltonianConstraint}, if $\tau_{2}$ is a small enough standard constant and $t_{0}\leq\tau_{2}$,
the conditions of Lemma~\ref{lemma:nondegeneracy} are satisfied with $M_{IJ} = \mK_I^{~J}$ and $T_+=t_{0}$. The claim follows. 

Next, since $\mK_{I}^{~J} - \mK_{J}^{~I} = 0$ and $ \mK_{I}^{~I} = 1$ on $(0,t_0]\times\S$,
$\mrK$ is symmetric and has trace $1$. Thus the eigenvalues of $\mrK_{I}^{~J}$, which are
real and sum to one, must be the limits of the continuous curves $\q_I(t)$. In particular,
they are thus distinct. Next, 
\begin{equation}
\| 1 - \mK_{I}^{~J} \mK_{J}^{~I} - \Phi_1^2 \|_{C^{k_{0}+1}(\S)}
    \leq C t_0^{2\s} t^{2\s},
\end{equation}
due to (\ref{eq:almost asymptotic Hcon}). Taking the limit yields $\sum_{I} \mrp_I^2 + \mrPhi_1^2 = 1$. 

Finally, concerning the Kretschmann-scalar, note that
\begin{align}
\begin{split}
\mathfrak{K} = & \roRiem_g(e_I, e_J, e_K,e_L) \roRiem_g(e_I, e_J, e_K,e_L) \\
    &+ 4 \roRiem_g(e_0, e_I, e_0,e_J) \roRiem_g(e_0, e_I, e_0,e_J) \\
    &- 4 \roRiem_g(e_0, e_I, e_J,e_K) \roRiem_g(e_0, e_I, e_J,e_K).
\end{split}
\end{align}
Using the Gau\ss\ equations,
\begin{equation*}
\roRiem_g(e_I, e_J, e_K,e_L) = \roRiem_h(e_I, e_J, e_K,e_L) + k_{IJ} k_{KL} - k_{IK} k_{JL}.
\end{equation*}
Hence, using the symmetries of the Riemann curvature tensor,
\begin{align*}
\begin{split}
 & \|  t^4\roRiem_g(e_I, e_J, e_K,e_L) \roRiem_g(e_I, e_J, e_K,e_L) 
    - 2 \tr(\mK^2)^2 + 2 \tr(\mK^4) \|_{C^{k_{0}+1}(\S_t)} \leq C t^{2\s}(\mbD + t^{\s})^2,    
\end{split}
\end{align*}
see (\ref{eq:EstimateSpatialRiemann}). Next,  
\begin{equation*}
\roRiem_g(e_0, e_I, e_0,e_J) = \roRic_h(e_I, e_J) - e_I (\phi) e_J(\phi) 
    - \tfrac{2}{n-1} (V \circ \phi) \delta_{IJ}
    + \tr(k) k_{IJ} - k_{IM} k_{JM}
\end{equation*}
We can estimate the first three terms by (\ref{eq:EstimateSpatialRiemann}),
(\ref{eq:EstimateEnergyMomentumSF}) and the scheme: 
\begin{align*}
\begin{split}
  &\| t^4 \roRiem_g(e_0, e_I, e_0,e_J) \roRiem_g(e_0, e_I, e_0,e_J) \\
  &\phantom{\|} - \tr(\mK)^2 \tr(\mK^2) +2 \tr(\mK) \tr(\mK^3) 
        - \tr(\mK^4) \|_{C^{k_{0}+1}(\S_t)}\leq C t^{2\s}(\mbD + t^{\s})^2.
\end{split}
\end{align*} 
Finally, by (\ref{eq:EstimateNormalRiemann}),
\begin{align}
\begin{split}
\| t^4 \roRiem_g(e_0, e_I, e_J,e_K) \roRiem_g(e_0, e_I, e_J,e_K) \|_{C^{k_{0}+1}(\S_t)} 
    \leq C t^{2\s}(\mbD + t^{\s})^2.
\end{split}
\end{align} 
We thus gather form the above estimates that 
\begin{align*}
\begin{split}
&  \| t^4 \mathfrak{K}_g - 2 \tr(\mK^4) + 8 \tr(\mK) \tr(\mK^3) 
    - 2 \tr(\mK^2)^2 - 4 \tr(\mK)^2 \tr(\mK^2)\|_{C^{k_{0}+1}(\S_t)} \\
    \leq & C t^{2\s}(\mbD + t^{\s})^2.
\end{split}
\end{align*}
On the other hand,
\begin{align*}
  & 2\tr(\mK^4) - 8 \tr(\mK) \tr(\mK^3) 
    + 2 \tr(\mK^2)^2 + 4 \tr(\mK)^2 \tr(\mK^2)  \\
= & 2 \textstyle{\sum}_{I} \q_I^4 - 8 \textstyle{\sum}_{I} \q_I^3 + 2 \big( \textstyle{\sum}_{I} \q_I^2 \big)^2
    + 4 \textstyle{\sum}_{I} \q_I^2 
=  4 \textstyle{\sum}_{I} \big( \q_I^2(1- \q_I)^2 + \textstyle{\sum}_{J>I} \q_I^2 \q_J^2 \big).
\end{align*}
The estimate (\ref{eq:Asymptotics Kretschmann}) follows. Finally, due to (\ref{eq:roRicpot})
\begin{equation*}
\mathfrak{R}_g := \roRic_{g,\mu \nu} \roRic_{g}^{\mu \nu} 
= |\md \phi|_g^4 - \tfrac{4}{n-1} |\md \phi|_g^2 \big(V \circ \phi\big)
    + \tfrac{4(n+1)}{n-1)^2} \big(V \circ \phi \big)^2.
\end{equation*}
By the scheme and (\ref{eq:EstimateEnergyMomentumSF}), (\ref{eq:Asymptotics Ricci squared}) follows.

In order to prove that all causal geodesics are past incomplete, let $\g:\mJ\rightarrow M$ be a causal geodesic and define $f^{\mu}(s)$
by the relation
\[
  \g'(s)=\textstyle{\sum}_{\mu=0}^{n}f^{\mu}(s)e_{\mu}|_{\g(s)}.
\]
If $\theta$ denotes the mean curvatures of the leaves of the foliation, i.e. $\theta=t^{-1}$, then
\begin{equation}\label{eq:theta along gamma}
  (\theta\circ\g)'(s)=\g'(s)\theta=-\tfrac{1}{N\circ\g(s)}\theta^{2}\circ\g(s) f^{0}(s).
\end{equation}
On the other hand, since $f^{0}=-\langle e_{0}|_{\g},\g'\rangle$ and $\g''=0$, 
\begin{equation}\label{eq:fzero along gamma}
  \begin{split}
    (\tfrac{d}{ds}f^{0})(s) = & -\textstyle{\sum}_{\mu,\nu}f^{\mu}(s)f^{\nu}(s) \langle \nabla_{e_{\mu}}e_{0},e_{\nu}\rangle\circ\g(s)\\
    = & -\textstyle{\sum}_{J,K}f^{J}(s)f^{K}(s)k_{JK}\circ\g(s)-\textstyle{\sum}_{J}f^{0}(s)f^{J}(s)[e_{J}(\ln N)]\circ\g(s),
  \end{split}
\end{equation}
see (\ref{eq:nablaezeIez}). Let $h:=f^{0}\cdot \theta\circ\g$. Then (\ref{eq:theta along gamma}) and (\ref{eq:fzero along gamma}) yield
\begin{equation}\label{eq:h prime causal geodesic incompleteness}
  h'=-\tfrac{1}{N\circ\g}h^{2}
  -\theta\circ\g\textstyle{\sum}_{J,K}f^{J}f^{K}k_{JK}\circ\g-\theta\circ\g\textstyle{\sum}_{J}f^{0}f^{J}[e_{J}(\ln N)]\circ\g.
\end{equation}
Let $\fbar:=(f^{1},\dots,f^{n})$. Then, due to the causality of the curve, $|\fbar|\leq f^{0}$. Moreover,
\[
  \textstyle{\sum}_{J,K}f^{J}f^{K}(tk_{JK})\circ\g\geq \min\{p_{I}\circ\g\}\cdot |\fbar|^{2}\geq -(1-4\s)|\fbar|^{2}
  \geq -(1-4\s)(f^{0})^{2},
\]
assuming the standard constant $\tau_{2}$ to be small enough and that we only consider the subinterval of $\mJ$, say $\mJ_-$, such that
$t\circ\g\leq t_{0}\leq\tau_{2}$; the second inequality follows from (\ref{eq:bqboundsb}), (\ref{eq:HolderEstimateEigenvaluesInitial})
and the assumption concerning $\tau_{2}$. Similarly, assuming the standard constant $\tau_{2}$ to be small enough, we can assume
$|1-1/N\circ\g|\leq \s$ on $\mJ_-$. This means that the sum of the first two terms on the right hand side of
(\ref{eq:h prime causal geodesic incompleteness}) can be bounded from above by $-3\s h^{2}$. Turning to the third term on the right hand side,
note that, assuming the standard constant $\tau_{2}$ to be small enough,
\[
\big|\textstyle{\sum}_{J}f^{J}[e_{J}(\ln N)]\circ\g\big|\leq |\fbar|\cdot |\vec{e}(\ln N)|\leq \s\theta\circ\g \cdot f^{0}
\]
on $\mJ_-$. Summing up, we conclude that $h'\leq -2\s h^{2}$ on $\mJ_-$. Since $h>0$, due to the causality of the curve, we conclude that $h$ blows up in
finite parameter time to the past. This means that $\g$ is past incomplete. Since
all the $\S_t$ are Cauchy hypersurfaces, the curvature invariants blow up along all past inextendible causal curves. This concludes the proof.
\end{proof}

\section{The proof of the main theorem} \label{sec: finishing the proof}
\begin{proof}[Proof of Theorem~\ref{thm: big bang formation}]
Let $\s_p$, $\s_V$, $\s$, $k_0$, $k_1$, $V$, $(\S,h_{\refer})$ and $(E_i)_{i=1}^n$ be as in the statement of the theorem.
Fix $\zeta_0 > 0$. Our task is to demonstrate the existence of a $\zeta_1 >0$, with dependence as in
Remark~\ref{remark:Constant dependence Main Theorem}, such that if $\mfI$ are $\s_p$-admissible CMC initial data;
$|\bq_I-\bq_J|>\zeta_{0}^{-1}$ for $I\neq J$; the associated expansion-normalized initial data $(\S, \bmH, \bmK, \bP_0, \bP_1)$
satisfy \eqref{eq: Sobolev bound assumption}; and $\bth > \zeta_1$, then the conclusions of the theorem hold. 
Due to Proposition~\ref{prop: initial data estimate}, there is a $\rho_0>0$, depending only on $\zeta_0$, $\s$, $k_0$, $k_1$ and $(E_i)_{i=1}^n$,
and unique (up to a sign in case of the elements of the frame) $\be_I^i$, $\bq_I\in C^{\infty}(\S,\rn{})$
such that $\be_I^i$, $\bq_I$, $\bp_0$, $\bp_1$ form smooth diagonal FRS initial data satisfying the non-degenerate FRS
expansion-normalized bounds of regularity $k_1$ for $\rho_0$ at $t_0 = \bth^{-1}$. 

Theorem~\ref{thm:GlobalExistenceFermiWalker} then yields a standard constant $\tau_{1}<1$ such that if $t_0 < \tau_1$, there exists a
unique smooth solution to (\ref{eq: transport frame})--(\ref{eq:LapseEquation}) on $(0,t_+)$ corresponding to initial data
$\hat{e}_I^i=\be_I^i$, $\hat{\phi}_i=\bp_i$, $i=1,2$, and $\hat{k}_{IJ}=\tfrac{\bq_\uI}{t_0}\delta_{\uI J}$, with $t_0\in (0,t_{+})$, inducing
the correct initial data on $\S_{t_{0}}$. However, by Proposition~\ref{prop:FieldEquationsFermiWalker},
this solution corresponds to a solution of the Einstein-non-linear scalar field equations
with a potential $V$ existing on the interval $(0,t_+)$
such that the hypersurfaces $\S_t$ are CMC Cauchy hypersurfaces of mean curvature $t^{-1}$
and the metric is given by $g = -N^2 dt \otimes t + \omega^I \otimes \omega^I$.

Moreover, as a consequence of Theorem~\ref{thm:Asymptotics}, there exists a standard constant $0<\tau_2 \leq \tau_1$
such that if $t_0 < \tau_2$ then the corresponding solution satisfies the conclusions of Theorem~\ref{thm:Asymptotics}.
This yields most of the conclusions of Theorem~\ref{thm: big bang formation} if we choose $\zeta_{1}=\tau_{2}^{-1}$.
To prove the existence of the diffeomorphism
$\Psi$, note that the solution obtained in Proposition~\ref{prop:FieldEquationsFermiWalker} is a globally hyperbolic
development of the initial data. Due to \cite[Corollary~23.44, p.~418]{stab}, there is a maximal globally hyperbolic development
$(M,g,\phi)$ of the initial data. By the abstract properties of the maximal globally hyperbolic development, there is thus a map
$\Psi:(0,t_{+})\times\S\rightarrow M$, which is a smooth isometry (meaning that it preserves both the metric and the scalar field)
onto its image. Moreover, $\Psi(t_{0},p)=\iota(p)$. Assume now that $\Psi((0,t_{0}]\times\S)$ is not all of $J^{-}(\iota(\S))$.
Then there is a point $p\in M- \Psi((0,t_{0}]\times\S)$, to the past of $\iota(\S)$. This leads to a contradiction by an argument similar
  to the proof of \cite[Lemma~18.18, p.~204]{RinCauchy}. For a similar reason, $(M,g)$ is past $C^{2}$ inextendible. The theorem
  follows. 
\end{proof}

\begin{proof}[Proof of Theorem~\ref{thm:degenerate case}]
  Given the conclusions of Lemma~\ref{lemma: from qms to frs data}, the proof of Theorem~\ref{thm:degenerate case} is very similar to the proof
  of Theorem~\ref{thm: big bang formation}; we combine
  Theorems~\ref{thm:GlobalExistenceFermiWalker} and \ref{thm:Asymptotics}. However, there are two main differences. First, we cannot
  assume the eigenvalues to be distinct. Thus the corresponding arguments in Theorem~\ref{thm:Asymptotics} have to be modified. We leave the
  details of this modification to the reader. Second, in Theorem~\ref{thm:degenerate case}, we allow any starting time such that
  (\ref{eq:condition for Cauchy stability}) is satisfied to the past of the starting time. In order to prove that this is sufficient, it
  is enough to appeal to Cauchy stability; see Lemma~\ref{lemma:CauchystabEinstein}. However, there is one technical issue associated with
  taking this step: the norms involved in the Cauchy stability statement, see (\ref{seq:Cauchycriterion}), involve the time derivative
  of the second fundamental form and the second time derivative of the scalar field. In practice, it is therefore necessary to use
  the equations to take the step from the norms we actually control to the norms appearing in Lemma~\ref{lemma:CauchystabEinstein}.
  Moreover, this step involves estimating the difference of the corresponding lapse functions. The necessary steps are similar to ones
  already taken in these notes, and are left to the reader. 
\end{proof}

\appendix

\section{Sobolev inequalities} \label{sec:SobolevInequalities}

Let $(\S,h_{\refer})$ and $(E_{i})_{i=1}^{n}$ be as described in Subsection~\ref{ssection: The reference frame}.
Let $s \in \nn{}$ and $\psi \in C^\infty(\S)$. The $C^{s}$- and $H^{s}$-norms we use in this paper are defined
as follows:
\begin{subequations}\label{seq:Sobolev norms}
  \begin{align}
    \|\psi\|_{C^{s}(\S)}
    := & \textstyle{\sum}_{|\bfI|\leq s} \norm{E_\bfI \psi}_{C^{0}(\S)}, \\
    \|\psi\|_{H^s(\S)}
    := & \big(\textstyle{\int}_{\S}\textstyle{\sum}_{|\bfI|\leq s}|E_{\bfI}\psi|^{2}\mu_{h_{\refer}}\big)^{1/2}. \label{eq:psiHkdef}
  \end{align}
\end{subequations}
Here $\mu_{h_{\refer}}$ is the volume form associated with the reference metric $h_{\refer}$ and the bold index $\bfI$
refers to a sequence of indices $\bfI \in \{1, \hdots, n\}^l$, for some $l \in \nn{}_0$ and we define
$E_{\bfI} := E_{i_1}E_{i_2}\cdots E_{i_\ell}$ in case $\bfI = (i_1, \hdots, i_l)$. We also use the notation $\abs \bfI = l$.
\begin{remark}
  There is one exception to the above convention, namely in the case of geometric initial data, expansion-normalized initial data etc; cf.,
  e.g., (\ref{eq: Sobolev bound assumption}), (\ref{eq:mKchhconvratesfN}), (\ref{eq:smallnessdataonsing}) etc. In these cases, we use the
  geometrically defined Sobolev and $C^{k}$-norms associated with $(\S,h_{\refer})$. 
\end{remark}
The $E_{i}$ do not, in general, commute. In several settings this creates important differences with Sobolev norms defined using commuting
derivative operators. We therefore here discuss some of the properties of the norms (\ref{seq:Sobolev norms}). 

\begin{lemma} \label{lemma:EstimateCommutators}
  With $(\S,h_{\refer})$ and $(E_{i})_{i=1}^{n}$ as in Subsection~\ref{ssection: The reference frame}, let
  $\ell$, $m\in\nn{}$ and $\varphi$, $\psi \in C^{\infty}(\S)$. Then
  \begin{subequations}
  \begin{align}
    \textstyle{\sum}_{|\bfI| = \ell} | E_{\bfI} (\varphi \psi)   |^{m}
    &\leq 2^{m \ell} \textstyle{\sum}_{|\bfJ| \leq \ell } \left| E_{\bfJ}(\varphi)   \right|^m
    \textstyle{\sum}_{|\bfJ| \leq \ell } \left| E_{\bfJ}(\psi)   \right|^m ,\label{eq:productraisedtom}\\
    \textstyle{\sum}_{|\bfI| = \ell} |[E_\bfI, \varphi] (\psi)   |^m
    &\leq 2^{m \ell} \textstyle{\sum}_{|\bfJ| \leq \ell} \left| E_{\bfJ}(\varphi)   \right|^m
    \textstyle{\sum}_{|\bfJ| \leq \ell - 1} \left| E_{\bfJ}(\psi)   \right|^m. \label{eq:commutatorraisedtom}
  \end{align}
  \end{subequations}
\end{lemma}
\begin{remark}
  The result is the same if $|\cdot|$ is replaced by $\|\cdot\|_{C^{0}(\S)}$.
\end{remark}
\begin{proof}
  We prove the first estimate by induction. In case $\ell = 0$, the statement is obvious. Now assume that the statement holds
  for a given $\ell = k\in\nn{}$. Then, as $(A+B)^m \leq 2^{m-1} (A^m + B^m)$, 
\begin{align*}
\textstyle{\sum}_{|\bfI| = k+1} | E_{\bfI} (\varphi \psi)   |^{m} 
= & \textstyle{\sum}_{|\bfJ| = k} \textstyle{\sum}_{i}  | E_{\bfJ} (E_{i}(\varphi) \psi
    + \varphi E_{i}(\psi) )  |^{m} \\
\leq & 2^{m-1} \textstyle{\sum}_{|\bfJ| = k} \textstyle{\sum}_{i}
    (| E_{\bfJ} (E_{i}(\varphi) \psi)  |^m 
    +  | E_{\bfJ}((\varphi) E_{i}(\psi))  |^{m}) \\
\leq & 2^{m-1} 2^{m k} \big[ \textstyle{\sum}_{|\bfJ| \leq  k} \textstyle{\sum}_{i}
    | E_{\bfJ} E_{i}(\varphi)  |^m 
     \textstyle{\sum}_{|\bfJ| \leq k} | E_{\bfJ} (\psi)  |^m \\
& \phantom{2^{m-1} 2^{m k} \big[}+   \textstyle{\sum}_{|\bfJ| \leq  k}  | E_{\bfJ} (\varphi)  |^m 
     \textstyle{\sum}_{|\bfJ| \leq  k}  \textstyle{\sum}_{i}  | E_{\bfJ} E_{i} (\psi)  |^m \big] \\
\leq & 2^{m(k+1)}  \textstyle{\sum}_{|\bfJ| \leq k+1 } \left| E_{\bfJ}(\varphi)   \right|^m 
    \textstyle{\sum}_{|\bfJ| \leq k+1 } \left| E_{\bfJ}(\psi)   \right|^m.
\end{align*}
Assuming that the second statement holds for a given $\ell = k$, 
\begin{align*}
\textstyle{\sum}_{|\bfI| = k+1} | [E_{\bfI}, \varphi] (\psi)   |^{m} 
= & \textstyle{\sum}_{|\bfJ| = k} \textstyle{\sum}_{i}  | E_{\bfJ}(E_{i}(\varphi)\psi)+[E_{\bfJ},\varphi]E_{i}\psi|^{m} \\
\leq & 2^{m-1} \textstyle{\sum}_{|\bfJ| = k} \textstyle{\sum}_{i}\big(| E_{\bfJ} (E_{i}(\varphi) \psi)  |^m 
    +  | [E_{\bfJ},\varphi] (E_{i}(\psi))  |^{m}\big) \\
\leq & 2^{m-1} 2^{m k} \big[ \textstyle{\sum}_{|\bfJ| \leq  k} \textstyle{\sum}_{i}
    | E_{\bfJ} E_{i}(\varphi)  |^m 
     \textstyle{\sum}_{|\bfJ| \leq k} | E_{\bfJ} (\psi)  |^m \\
& \phantom{2^{m-1} 2^{m k} \big[}+   \textstyle{\sum}_{|\bfJ| \leq  k}  | E_{\bfJ} (\varphi)  |^m 
     \textstyle{\sum}_{|\bfJ| \leq  k-1}  \textstyle{\sum}_{i}  | E_{\bfJ} E_{i} (\psi)  |^m \big] \\
\leq & 2^{m(k+1)}  \textstyle{\sum}_{|\bfJ| \leq k+1 } \left| E_{\bfJ}(\varphi)   \right|^m 
    \textstyle{\sum}_{|\bfJ| \leq k } \left| E_{\bfJ}(\psi)  \right|^m.
\end{align*}
This concludes the proof.
\end{proof}
This leads to the following corollary.
\begin{cor} \label{cor:SobolevAlgebra}
  With $(\S,h_{\refer})$ and $(E_{i})_{i=1}^{n}$ as in Subsection~\ref{ssection: The reference frame}, let
  $\ell\in\nn{}$ and $\varphi, \psi \in C^{\infty}(\S)$. Then
\begin{equation} \label{eq:SobolevAlgebra}
\| \varphi \psi \|_{C^{\ell}(\S)} \leq 
    \big(2^{\ell+1} - 1\big)
    \| \varphi \|_{C^{\ell}(\S)}
    \|  \psi \|_{C^{\ell}(\S)}.
\end{equation}
\end{cor}
Next, we formulate a fundamental interpolation estimate. 
\begin{lemma} \label{le: interpolation}
  With $(\S,h_{\refer})$ and $(E_{i})_{i=1}^{n}$ as in Subsection~\ref{ssection: The reference frame}, let
  $k$, $l\in\nn{}$ be such that $0\leq l \leq k$ and $k\geq 1$. Then there is a constant $C$, depending
  only on $k$, $l$, $(\S,h_{\refer})$ and $(E_i)_{i=1}^n$, such that for all $\psi \in C^{\infty}(\S)$,  
\begin{equation}\label{eq:psiinterpol}
    \|\psi\|_{H^{l}(\S)}\leq C\|\psi\|_{L^{2}(\S)}^{1-l/k}\|\psi\|_{H^{k}(\S)}^{l/k}.
\end{equation}
\end{lemma}
\begin{remark}
Since $\S$ is closed, we can replace $L^2(\S)$ with $C^{0}(\S)$.
\end{remark}
\begin{proof}
Note, to begin with, that if $\psi\in C^{\infty}(\S)$, then
\begin{equation}\label{eq:EipsiEipsieq}
    E_{i}[\psi E_{i}(\psi)]=[E_{i}(\psi)]^{2}+\psi E_{i}^{2}(\psi).
\end{equation}
Next, let $f\in C^{\infty}(\S)$ and $X$ be a smooth vector field. Then
\begin{equation}\label{eq:Xfint}
\textstyle{\int}_{\S}X(f)\mu_{h_{\refer}}
    =\textstyle{\int}_{\S}\ml_{X}(f\mu_{h_{\refer}})-\textstyle{\int}_{\S}f\ml_{X}\mu_{h_{\refer}}
    =-\textstyle{\int}_{\S}f(\rodiv_{h_{\refer}}X)\mu_{h_{\refer}},
\end{equation}
where we used Cartan's magic formula, Stokes' theorem, 
the fact that $\partial\S = \varnothing$ and the fact that
$\ml_{X}\mu_{h}=(\rodiv_{h}X)\mu_{h}$ in the last step.
Integrating (\ref{eq:EipsiEipsieq}), the left-hand side becomes
\[
\textstyle{\int}_{\S}E_{i}[\psi E_{i}(\psi)]\mu_{h_{\refer}}
=-\textstyle{\int}_{\S}\psi E_{i}(\psi)(\rodiv_{h_{\refer}}E_{i})\mu_{h_{\refer}}.
\]
Note that this expression can be estimated, in absolute value, by
\[
C\textstyle{\int}_{\S}|\psi|\cdot|E_{i}(\psi)|\mu_{h_{\refer}}
\leq\frac{1}{2}\|E_{i}(\psi)\|_{L^{2}(\S)}^{2}+\frac{1}{2}C^{2}\|\psi\|_{L^{2}(\S)}^{2}
\]
for some constant $C$ depending only on $(\S,h_{\refer})$ and $(E_i)_{i=1}^n$.
Combining the above observations yields
\[
\|E_{i}(\psi)\|_{L^{2}(\S)}^{2}\leq \textstyle{\int}_{\S}|\psi|
\cdot|E_{i}^{2}(\psi)|\mu_{h_{\refer}}
+\frac{1}{2}\|E_{i}(\psi)\|_{L^{2}(\S)}^{2}
+\frac{1}{2}C^{2}\|\psi\|_{L^{2}(\S)}^{2}.
\]
Taking the second term on the right-hand side to the left-hand side
and summing over $i$ yields a constant $C$, depending only on $(\S,h_{\refer})$ and $(E_i)_{i=1}^n$,
such that 
  \[
  \|\psi\|_{H^{1}(\S)}
  \leq C\|\psi\|_{L^{2}(\S)}^{1/2}\|\psi\|_{H^{2}(\S)}^{1/2}
  \]
  for all $\psi\in C^{\infty}(\S)$; the relevant estimate for $\psi$ in $L^{2}(\S)$ is immediate.

Assume now that (\ref{eq:psiinterpol}) holds for some $k\geq 2$
and all $0\leq l\leq k$; we know this to be true for $k=2$.
We now wish to prove that the statement holds with $k$ replaced by $k+1$.
If $l=0$ or $l=k+1$, the statement is trivial, so we assume $1\leq l\leq k$.
Note that, for these integers, (\ref{eq:psiinterpol}) holds. On the other hand,
\[
    \|\psi\|_{H^{k}(\S)}\leq \|\psi\|_{L^{2}(\S)} 
    + \textstyle{\sum}_{i}\|E_{i}\psi\|_{H^{k-1}(\S)}.
\]
For this reason it is sufficient to estimate two expressions.
First, we need to estimate
\[
  \|\psi\|_{L^{2}(\S)}^{1-l/k}\|E_{i}\psi\|_{H^{k-1}(\S)}^{l/k}.
\]
In order to estimate the second factor, note that, due to the induction hypothesis,
\[
  \|E_{i}\psi\|_{H^{k-1}(\S)}\leq C\|E_{i}\psi\|_{L^{2}(\S)}^{1-(k-1)/k}\|E_{i}\psi\|_{H^{k}(\S)}^{(k-1)/k}.
\]
On the other hand, since $l\geq 1$,
\[
  \|E_{i}\psi\|_{L^{2}(\S)}\leq C\|\psi\|_{L^{2}(\S)}^{1-1/l}\|\psi\|_{H^{l}(\S)}^{1/l}.
\]
Combining the last two estimates yields
\begin{equation}\label{eq:interpolwt}
    \|\psi\|_{L^{2}(\S)}^{1-l/k}\|E_{i}\psi\|_{H^{k-1}(\S)}^{l/k}\leq C\|\psi\|_{L^{2}(\S)}^{(k+1-l)(k-1)/k^{2}}\|\psi\|_{H^{l}(\S)}^{1/k^{2}}\|\psi\|_{H^{k+1}(\S)}^{l(k-1)/k^{2}}.
\end{equation}
The second expression we need to estimate is
\[
  \|\psi\|_{L^{2}(\S)}^{1-l/k}\|\psi\|_{L^{2}(\S)}^{l/k}=\|\psi\|_{L^{2}(\S)}.
\]
Clearly, the right-hand side is bounded by the right-hand side of (\ref{eq:interpolwt}).
To conclude,
  \[
  \|\psi\|_{H^{l}(\S)}\leq C\|\psi\|_{L^{2}(\S)}^{(k+1-l)(k-1)/k^{2}}\|\psi\|_{H^{l}(\S)}^{1/k^{2}}\|\psi\|_{H^{k+1}(\S)}^{l(k-1)/k^{2}}.
  \]
If $\|\psi\|_{H^{l}(\S)}=0$, there is nothing to prove.
Otherwise, we divide by $\|\psi\|_{H^{l}(\S)}^{1/k^{2}}$.
This yields  (\ref{eq:psiinterpol}) with $k$ replaced by $k+1$.
This finishes the inductive step and proves the lemma.
\end{proof}

The main tool for deriving the higher-order estimates, in particular the ones employed in the scheme of Subsection~\ref{sec:scheme},
is the following.

\begin{lemma} \label{lemma: products in L2}
  With $(\S,h_{\refer})$ and $(E_{i})_{i=1}^{n}$ as in Subsection~\ref{ssection: The reference frame}, let
  $l_{i}\in\nn{}$, $i=1,\dots,j$, and $l=l_{1}+\dots+l_{j}$. If $|\bfI_{i}|=l_{i}$ and $\psi_{i}\in C^{\infty}(\S)$, $i=1,\dots,j$, then
\begin{equation}\label{eq:moser}
    \|E_{\bfI_{1}}\psi_{1}\cdots E_{\bfI_{j}}\psi_{j}\|_{L^2(\S)}
    \leq C\textstyle{\sum}_{i=1}^{j}\|\psi_{i}\|_{H^{l}(\S)}
        \prod_{m\neq i}\|\psi_{m}\|_{C^{0}(\S)},
  \end{equation}
where $C$ only depends on $l$, $(\S,h_{\refer})$ and $(E_i)_{i=1}^n$. In particular,
\begin{equation}\label{eq:moser2}
    \|\psi_1 \cdots \psi_j \|_{H^l(\S)}
    \leq C \textstyle{\sum}_{i=1}^{j}\|\psi_{i}\|_{H^{l}(\S)}
        \prod_{m\neq i}\|\psi_{m}\|_{C^{0}(\S)}.
\end{equation}
\end{lemma}
\begin{proof}
The first statement is a special case of \cite[Corollary~B.8, p.~238]{RinWave}.
The second statement is an immediate consequence of the first. 
\end{proof}

Sometimes we need improvements of Lemma~\ref{lemma: products in L2}. 
\begin{lemma}  \label{lemma:noGNS}
  With $(\S,h_{\refer})$ and $(E_{i})_{i=1}^{n}$ as in Subsection~\ref{ssection: The reference frame}, let
  $1\leq m \in \nn{}$ and $\varphi_i$, $\psi_i$, $\chi_{i}$, $\pi_{ij}\in C^{\infty}(\S)$, $i,j \in \{1,..,m\}$.
  Moreover, let $\kappa_0$ be the smallest integer strictly greater than $n/2$, and $\kappa_{0}\leq \ell\in\nn{}$.
  Then, for any $ \eta > 0$,
\begin{align} \label{eq:noGNSSum}
\begin{split}
  \big\| \textstyle{\sum}_{i=1}^m \varphi_i \psi_i \big\|_{H^{\ell} (\S)}
  \leq & \big(\eta + \big\|\textstyle{\sum}_{i=1}^m \varphi_i^2 \big\|_{C^{0}(\S)}^{1/2} \big)\|\psi \|_{H^{\ell}(\S)}\\
  & + C \langle \eta^{-1} \rangle^{\ell-1} \left\langle \| \varphi \|_{H^{\ell}(\S)} \right \rangle^\ell
  \| \psi \|_{C^{\kappa_0}(\S)},
\end{split}
\end{align}
using conventions similar to (\ref{seq:norms of indexed quantities}). Moreover, for any $\eta >0$
\begin{align} \label{eq:noGNSInnerProduct}
\begin{split}
 \big| \textstyle{\sum}_{i=1}^m \langle \varphi_i \psi_i, \chi_i \rangle_{H^{\ell}(\S)} \big|
    \leq & \big(\eta + \textstyle{\max_i} \| \varphi_i \|_{C^{0}(\S)}  \big)
      \| \psi \|_{H^{\ell}(\S)}\| \chi \|_{H^{\ell}(\S)} \\
    & + C \langle \eta^{-1} \rangle^{\ell-1}  \left\langle \| \varphi \|_{H^{\ell}(\S)} \right \rangle^\ell
        \| \psi \|_{C^{\kappa_0}(\S)}\| \chi \|_{H^{\ell}(\S)}.
\end{split}
\end{align}
Finally, for any $\eta >0$,
\begin{align} \label{eq:noGNSInnerProduct2}
  \begin{split}
    \big| \textstyle{\sum}_{i,j=1}^{m} 
    \langle \varphi_i \psi_j, \pi_{ij} \rangle_{H^{\ell}(\S)} \big|
    \leq & \big(\eta + \big\| \textstyle{\sum}_{i=1}^{m} 
        \varphi_i^2 \big\|_{C^{0}(\S)}^{1/2} \big)
      \| \psi \|_{H^{\ell}(\S)}\| \pi \|_{H^{\ell}(\S)}\\
    & + C\langle \eta^{-1} \rangle^{\ell-1}  \left\langle\| \varphi \|_{H^{\ell}(\S)} \right \rangle^\ell
         \| \psi \|_{C^{\kappa_0}(\S)}\| \pi \|_{H^{\ell}(\S)}.
\end{split}
\end{align}
\end{lemma}
\begin{remark}
  The constant $C$ only depends on $\ell$, $m$, $(\S,h_{\refer})$ and $(E_i)_{i=1}^n$.
\end{remark}
\begin{remark}\label{remark:supremum iso Cz norm}
  In case $\psi_i=\chi_i$ in (\ref{eq:noGNSInnerProduct}), the left hand side, with the absolute value sign removed,
  can be estimated by the right hand side with $\| \varphi_i \|_{C^{0}(\S)}$ replaced by
  $\sup_{x\in\S}\varphi_i(x)$.
\end{remark}
\begin{proof}
  We first show that for any $\varphi$, $\psi\in C^{\infty}(\S)$ and $\lambda > 0$,
\begin{align} \label{eq:noGNSAuxInequality}
\begin{split}
\textstyle{\sum}_{|\bfI| \leq l}  
    \| [E_{\bfI},\varphi] (\psi) \|_{L^2(\S)}^2 
    &\leq \lambda \| \psi \|_{H^{\ell}(\S)}^2
        + C \lambda^{-\ell+1}  \| \varphi \|_{H^{\ell}(\S)}^{2\ell} 
            \| \psi \|_{C^{0}(\S)}^2  \\
    &\quad + C \| \varphi \|_{H^{\ell}(\S)}^2 \| \psi \|_{C^{\kappa_0}(\S)}^2,
\end{split}
\end{align}
where $C$ depends only on $\ell, (\S,h_\refer)$ and the frame. To prove this statement, note that
the left hand side is bounded by a sum of terms of the form
\begin{equation}\label{eq:term to be estimated aux est}
  C\|E_{\bfI_{1}}(\varphi)E_{\bfI_{2}} (\psi) \|_{L^2(\S)}^2
\end{equation}
with $|\bfI_{1}|\geq 1$ and $|\bfI_{1}|+|\bfI_{2}|\leq \ell$, the $C$ and the number of terms depends only on $n$ and $\ell$. 
If $|\bfI_2| \leq \kappa_0$,
\[
\|E_{\bfI_{1}}(\varphi)E_{\bfI_{2}} (\psi) \|_{L^2(\S)}^2\leq C\|\psi\|_{C^{\kappa_{0}}(\S)}^{2}\|\varphi\|_{H^{\ell}(\S)}^{2}.
\]
These terms are included in the second term on the right hand side of (\ref{eq:noGNSAuxInequality}). 
If $|\bfI_{2}|\geq \kappa_{0}+1$,
\begin{equation*}
  \begin{split}
    \|E_{\bfI_{1}}(\varphi)E_{\bfI_{2}} (\psi) \|_{L^2(\S)}^2 \leq & C\|\psi\|_{H^{\ell-1}(\S)}^{2}\|\varphi\|_{C^{\ell-\kappa_{0}-1}(\S)}^{2}
    \leq  C \| \varphi \|_{H^{\ell}(\S)}^2 \|\psi \|_{C^{0}(\S)}^{2/\ell} \|\psi \|_{H^{\ell}(\S)}^{2(1-1/\ell)}
  \end{split}
\end{equation*}
where $C$ depends only on $\ell, (\S,h_\refer)$ and the frame, and we appealed to Sobolev embedding, (\ref{eq:psiinterpol})
and the fact that $C^{0}(\S) \subset L^{2}(\S)$ due to compactness of $\S$. Finally, appealing to Young's inequality
(with $p = \ell, q = \ell/(\ell-1)$) yields, for any $\lambda > 0$,
\begin{align*}
    C \| \varphi \|_{H^{\ell}(\S)}^2  \|\psi \|_{C^{0}(\S)}^{2/\ell} 
    \|\psi\|_{H^{\ell}(\S)}^{2(1-1/\ell)} 
    = & ( C \lambda^{-1+1/\ell} \| \varphi \|_{H^{\ell}(\S)}^2 
        \|\psi \|_{C^{0}(\S)}^{2/\ell} )
    (\lambda^{(1-1/\ell)} \| \psi \|_{H^{\ell}(\S)}^{2(1-1/\ell)})\\
    \leq & C \lambda^{-\ell+1} \| \varphi \|^{2\ell}_{H^{\ell}(\S)}
            \|\psi\|_{C^{0}(\S)}^2
        + \lambda \|\psi\|_{H^{\ell}(\S)}^2.
\end{align*}
Thus (\ref{eq:noGNSAuxInequality}) holds. To prove (\ref{eq:noGNSSum}), note that 
\begin{align}
\begin{split}
\label{eq:noGNSExtractingNoDerivatives}
\textstyle{\sum}_{|\bfI|\leq \ell} \big \|\textstyle{\sum}_{i} \varphi_i E_{\bfI} \psi_i \big\|_{L^{2}(\S)}^2
    &= \textstyle{\sum}_{|\bfI|\leq \ell} \textstyle{\int}_{\S} \big( \textstyle{\sum}_{i}
        \varphi_i E_{\bfI} \psi_i \big)^2 d \mu_{h_{\refer}} \\
    &\leq \textstyle{\sum}_{|\bfI|\leq \ell} \textstyle{\int}_{\S} \big( \textstyle{\sum}_{i} \varphi_i^2 \big)
        \big( \textstyle{\sum}_{i} |E_{\bfI} \psi_i|^2 \big) d \mu_{h_{\refer}} \\
    &\leq \big\|\textstyle{\sum}_{i} \varphi_i^2 \big\|_{C^{0}(\S)} 
    \|\psi \|_{H^{\ell}(\S)}^2.
\end{split}
\end{align}
Next, 
\begin{align*}
   \left\| \textstyle{\sum}_{i} \varphi_i \psi_i \right\|_{H^{\ell} (\S)}
    \leq & \big[ \textstyle{\sum}_{|\bfI| \leq \ell} 
        \big( \left\|\textstyle{\sum}_{i} \varphi_i E_{\bfI} \psi_i \right\|_{L^2(\S)}   
        + \left\| \textstyle{\sum}_{i} [E_{\bfI},\varphi_i] (\psi_i) \right\|_{L^2(\S)} 
        \big)^2 \big]^{1/2} \\
    \leq & \big(\textstyle{\sum}_{|\bfI| \leq \ell} 
        \left \|\textstyle{\sum}_{i} \varphi_i E_{\bfI} \psi_i \right\|_{L^2(\S)}^2 \big)^{1/2}
    + \sqrt{m} \big( \sum_{|\bfI| \leq \ell} \textstyle{\sum}_{i}
        \left\|  [E_{\bfI},\varphi_i] (\psi_i) \right\|_{L^2(\S)}^2 \big)^{1/2}
\end{align*}
by several applications of the triangle inequality for either the $\ell^2$-norm or the $L^2(\S)$-norm,
as well as the Cauchy-Schwartz inequality. The first term on the far right-hand side can be estimated
by appealing to (\ref{eq:noGNSExtractingNoDerivatives}), and the second term by appealing to
(\ref{eq:noGNSAuxInequality}) with $\lambda = \eta^2/m$, for each $i$ separately. Combining these
observations with the concavity of the square root yields (\ref{eq:noGNSSum}). 

Next, to prove (\ref{eq:noGNSInnerProduct}), note that 
\begin{align*}
   \big| \textstyle{\sum}_{i} \langle \varphi_i \psi_i, \chi_i \rangle_{H^{\ell}(\S)} \big|
    = & \big| \textstyle{\sum}_{i} \sum_{|\bfI| \leq \ell} 
        \textstyle{\int}_{\S}\{ \varphi_i E_{\bfI}(\psi_i) E_{\bfI}(\chi_i) 
        +  [E_{\bfI},\varphi_i] (\psi_i) \cdot E_{\bfI}(\chi_i)\} \mu_{h_{\refer}} \big| \\
    \leq & \big[ \big( \textstyle{\sum}_{i} 
        \|\varphi_i\|_{C^{0}(\S)}^2 \|\psi_i\|_{H^\ell(\S)}^2 \big)^{1/2}+ \big( \textstyle{\sum}_{i}  \sum_{|\bfI|\leq \ell} 
        \|[E_{\bfI},\varphi_i](\psi_i) \|_{L^{2}(\S)}^2 \big)^{1/2} \big]     \| \chi \|_{H^{\ell}(\S)}
\end{align*}
Extracting $\textstyle{\max}_{i} \| \varphi_i \|_{C^{0}(\S)}$ from the first term in the parenthesis and
appealing to  (\ref{eq:noGNSAuxInequality}) with $\lambda = \eta^2$ yields (\ref{eq:noGNSInnerProduct}).
The justification of Remark~\ref{remark:supremum iso Cz norm} is similar. 

Finally, to prove (\ref{eq:noGNSInnerProduct2}), estimate
\begin{align*}
 \big| \textstyle{\sum}_{i,j} 
    \langle \varphi_i \psi_j, \pi_{ij} \rangle_{H^{\ell}(\S)} \big|
    = & \big| \textstyle{\sum}_{i,j} \sum_{|\bfI| \leq \ell} 
        \textstyle{\int}_{\S}\{ \varphi_i E_{\bfI}(\psi_j) E_{\bfI}(\pi_{ij}) 
        +  [E_{\bfI},\varphi_i] (\psi_j)\cdot E_{\bfI}(\pi_{ij})\} \mu_{h_{\refer}} \big| \\
\leq & \big[  \big\| \textstyle{\sum}_{i} \varphi_i^2 \big\|_{C^{0}(\S)}^{1/2} \|\psi\|_{H^{\ell}(\S)}
    + \big(\textstyle{\sum}_{i,j}  \sum_{|\bfI|\leq \ell} 
    \|[E_{\bfI},\varphi_i](\psi_j) \|_{L^{2}(\S)}^2 \big)^{1/2} \big]\| \pi \|_{H^{\ell}(\S)}.
\end{align*}
Again, appealing to  (\ref{eq:noGNSAuxInequality}) with $\lambda = \eta^2/m$ yields
(\ref{eq:noGNSInnerProduct2}). 
\end{proof}

Finally, we require the following estimate for certain energy estimates.

\begin{lemma} \label{lemma:DivergenceEstimate}
  With $(\S,h_{\refer})$ and $(E_{i})_{i=1}^{n}$ as in Subsection~\ref{ssection: The reference frame}, let $\ell \in \mathbb{N}$,
  $\varphi$, $\psi \in C^{\infty}(\S)$ and $X = X^i E_i \in \mathfrak{X}(\S)$. Then there is a constant $C$, depending only on
  $\ell$, $(\S,h_{\refer})$ and $(E_i)_{i=1}^n$, such that 
\begin{align*}
  \begin{split}
    | \langle X(\varphi), \psi \rangle_{H^{\ell}(\S)} 
    +  \langle \varphi, X (\psi) \rangle_{H^{\ell}(\S)}  | 
    \leq & \|\varphi\|_{H^{\ell}(\S)} \| \psi\|_{H^{\ell}(\S)}  
    \cdot  \big( \|\rodiv_{h_{\refer}} (X)\|_{C^{0}(\S)}
    + C \textstyle{\sum}_{i} \|X^i\|_{C^{1}(\S)} \big) \\
    & + C ( \|\varphi\|_{H^{\ell}(\S)} \| \psi\|_{C^{1}(\S)} 
    + \|\varphi\|_{C^{1}(\S)} \| \psi\|_{H^{\ell}(\S)} ) 
    \textstyle{\sum}_{i} \|X^i\|_{H^{\ell}(\S)}.
  \end{split}
\end{align*}
\end{lemma}
\begin{proof}
  Note that 
  \begin{align*}
    &\langle  X(\varphi), \psi \rangle_{H^{\ell}(\S)} 
    +  \langle \varphi, X (\psi) \rangle_{H^{\ell}(\S)} \\
= &  \textstyle{\sum}_{|\bfI| \leq \ell} \textstyle{\int}_{\S} (X (E_{\bfI}(\varphi) E_{\bfI} (\psi))
    + [E_{\bfI},X](\varphi) E_{\bfI}(\psi)
    + [E_{\bfI},X](\psi) E_{\bfI}(\varphi) ) \mu_{h_\refer} .
\end{align*}
  However, appealing to (\ref{eq:Xfint}),
\begin{align*}
\Big |\textstyle{\sum}_{|\bfI| \leq \ell} \textstyle{\int}_{\S} X \big(E_{\bfI}(\varphi) E_{\bfI} (\psi)\big) 
    \mu_{h_{\refer}} \Big| 
&\leq \| \rodiv_{h_{\refer}}(X) \|_{C^{0}(\S)}
    \| \varphi \|_{H^{\ell(\S)}} \| \psi \|_{H^{\ell(\S)}}.
\end{align*}
On the other hand, $[E_{\bfI},X](\varphi)$ is a linear combination of terms $E_{\bfI_{1}}(X^{i})E_{\bfI_{2}}\varphi$,
where the coefficients are functions associated with the commutators of the elements of the frame; $|\bfI_1|+|\bfI_2|\leq |\bfI|+1$;
and if $|\bfI_1|+|\bfI_2|=|\bfI|+1$, then $|\bfI_i|\geq 1$, $i=1,2$. Due to this observation, a similar observation
with $\varphi$ and $\psi$ interchanged, and Lemma~\ref{lemma: products in L2}, the lemma follows. 
\end{proof}

\section{Regularity of eigenvalues}
Here we prove results, used in Section~\ref{sec:asymptotics}, regarding the regularity of eigenvalues of symmetric matrices,
depending on bounds on, and regularity of, the components of the matrix; see Kato's book \cite{kato} for a standard reference
on the topic.

Note first that if the components of a family of symmetric matrices is $C^{1}$-dependent on $p\in \S$, where $\S$ is a $C^{1}$-manifold,
then there exist (Lipschitz) continuous, ordered parametrizations of the eigenvalues,
say $\lambda_1 \leq \cdots \leq \lambda_n$, $\lambda_i \in C^{0,1}(\S)$,
even if the eigenvalues do not remain simple as $p\in \S$ varies.

Let $M\in C^{\infty}(\S,\Symn)$, where $\S$ is a smooth manifold and $\Symn$ denotes the symmetric $n\times n$ matrices.
Then it is sufficient to consider the characteristic polynomial $f(x,\lambda) = \det(M(x) - \lambda \Id)$ to be a function from
$\S\times\rn{}$ to $\rn{}$, smooth in $x$ and analytic in $\lambda$. The zeroes of $f$ at $x_0\in\S$ are of
course the eigenvalues of $M(x_0)$, say $(\lambda_I(x_0))_{I=1}^n$. If the eigenvalues at $x_0$ are simple, then, for any $I$,
\begin{equation}
(\d_{\lambda} f)(x_0,\lambda_I(x_0)) 
    = - \textstyle{\prod}_{J \neq I} (\lambda_J(x_0) - \lambda_I(x_0)) \neq 0.
\end{equation}
There are thus smooth functions $\lambda_{I}:U\rightarrow\mathbb{R}$, $I=1,\dots,n$, where $U\ni x_{0}$ is open, 
such that $f(x, \lambda_I(x)) = 0$ for all $I$ and $x\in U$. Moreover, if $X\in\mathfrak{X}(U)$,
\begin{equation}\label{eq:der of eigenvalues}
X|_{x_{0}}\lambda_I = - \frac{ X|_{x_0} \det(M - \lambda_I(x_0))}
    {\textstyle{\prod}_{J \neq I} (\lambda_J(x_0) - \lambda_I(x_0))}.
\end{equation}

\begin{lemma} \label{lemma:nondegeneracy}
  With $(\S,h_{\refer})$ and $(E_{i})_{i=1}^{n}$ as in Subsection~\ref{ssection: The reference frame}, let $\ell\in\nn{}$,
  and $C_{\ell}$, $K_{\ell}$, $\zeta_0$, $\alpha$ and $L$ be given strictly positive constants. Then there is a
  $\tau_+\in (0,1)$, depending only on $C_{\ell}$, $K_{\ell}$, $\zeta_0$, $\alpha$ and $n$
  such that the following holds. If $T_+\in (0,\tau_+)$, and
  \[
  M\in C^{0}(N,\Symn)\cap C^{\infty}(N_{\roint},\Symn),
  \]
  where $N:=\mI\times\S$, $N_{\roint}:=\mI_{\roint}\times \S$, $\mI:=[0,T_+]$ and $\mI_\roint:=(0,T_+]$, satisfies
  \begin{subequations}
    \begin{align}
      \textstyle{\max}_{I,J}\sup_{t\in\mI_{\roint}} \| M_{IJ}(t,\cdot) \|_{C^{\ell}(\S)} & \leq C_{\ell}, \\
      \textstyle{\max}_{I,J}\sup_{t\in\mI_{\roint}}[t^{1-\alpha}\| \partial_t M_{IJ}(t,\cdot) \|_{C^{\ell}(\S)}] &\leq K_{\ell} T_+^{\alpha},\\
      \textstyle{\max}_{I}\sup_{t\in\mI} \| \lambda_I(t,\cdot) \|_{C^{0}(\S)} &\leq L,\\
      \textstyle{\min}_{I\neq J} \inf_{x \in \S} | \lambda_I(T_+,x) - \lambda_J(T_+,x) |  & \geq \zeta_0^{-1},\label{eq:lambda difference lower bd}
    \end{align}
  \end{subequations}
  where $(\lambda_I)_{I=1}^n$ is the continuous, ordered parametrization of the eigenvalues of $M$, then there is a $\zeta>0$, depending
  only on $\zeta_0$, $n$ and $L$, such that
  \begin{equation}\label{eq:lower bound difference of eigenvalues}
    \textstyle{\min}_{I \neq J} \inf_{p \in N} |\lambda_{I}(p) - \lambda_J(p)|\geq\zeta^{-1}.
  \end{equation}
  Moreover, there is a constant $\Lambda_\ell$,
  depending only on $\ell$, $C_{\ell}$, $K_{\ell}$, $L$, $\zeta_{0}$, $(\Sigma, h_{\refer})$ and $(E_{i})_{i=1}^{n}$, such that,
  for any $I$ and $[s,t] \subset \mI$,
  \begin{equation}
    \| \lambda_I(t,\cdot) - \lambda_I(s,\cdot) \|_{C^{\ell}(\S)} \leq \Lambda_{\ell}T_+^{\alpha} (t^{\alpha} - s^{\alpha}).
  \end{equation}  
\end{lemma}
\begin{proof}
For a given $(s,x) \in N$, consider the discriminant $D_M (s,x)$ of $M(s,x)$. It is, up to a sign, the product of the difference of all the eigenvalues
and can be written as a homogeneous polynomial of degree $n(n-1)$ in the matrix components $M_{IJ}(s,x)$. On the other hand, due to
(\ref{eq:lambda difference lower bd}),  
\begin{equation}
D_M(T_+,x) = \textstyle{\prod}_{I < J} [\lambda_I(T_+,x) - \lambda_J(T_+,x)]^2 
    \geq \zeta_0^{-n(n-1)}
\end{equation}
for any $x \in \S$. As $D_M$ is smooth on $N_{\roint}$, it follows that, for any $x \in \S$,
\begin{equation*}
D_M(T_-,x) = D_M(T_+,x) - \textstyle{\int}_{T_-}^{T_+} \d_s D_M(r,x) \md r.
\end{equation*}
As $D_M$ is a homogeneous polynomial in the $M_{IJ}$, $|\d_s D_M(s,x)|\leq K T_+^{\alpha} s^{-1+\alpha}$, where $K$ depends only on $C_{\ell}$, $K_{\ell}$,
$\alpha$ and the dimension $n$. Thus
\begin{align*}
D_M(T_-,x) &\geq D_M(T_+,x) - \textstyle{\int}_{T_-}^{T_+} |\d_r D_M(r,x)| \md r  
    \geq \zeta_0^{-n(n-1)} - \tfrac{K}{\alpha} T_+^{2\alpha}.
\end{align*}
Let $\tau_+ = (\alpha \zeta_0^{-n(n-1)}/(2K))^{\frac{1}{2\alpha}}$. Then, if $T_+\leq\tau_+$,  $D_M \geq \zeta_0^{-n(n-1)}/2$ on $N$.
On the other hand, by the bound on the eigenvalues we may assume that $|\lambda_I-\lambda_J|\leq 2L$ on $N$ for all $I$, $J$.
This means that there is a $\zeta>0$, depending only on $\zeta_0$, $n$ and $L$ such that (\ref{eq:lower bound difference of eigenvalues})
holds. Since the $\lambda_I$ are distinct, they are smooth on $N_{\roint}$. Moreover, as above, 
\begin{equation}
(\d_t\lambda_I)(s,x) = - \frac{ \d_t|_{t=s} \det(M(\cdot,x) - \lambda_I(s,x))}
    {\prod_{J \neq I} [\lambda_J(s,x) - \lambda_I(s,x)]}.
\end{equation}
The numerator on the right-hand side is a linear combination of monomials 
in the $M_{KJ}(s,x)$ and the eigenvalue $\lambda_I(s,x)$,
each of which are multiplied by one term of the form $\d_t M_{KJ}(s,x)$.
In particular, the right-hand side is thus bounded in $C^{\ell}$
by a term of the form $KT_+^{\alpha} s^{-1+\alpha}$, where $K$ depends only on $\ell$, $C_{\ell}$, $K_{\ell}$, 
$L$, $\zeta_{0}$, $(\Sigma, h_{\refer})$ and $(E_{i})_{i=1}^{n}$. In order to arrive at this conclusion,
we use (\ref{eq:lower bound difference of eigenvalues}) and the fact that the $\lambda_{I}(t,\cdot)$ are
bounded in $C^{\ell}$, uniformly in $t\in\mI_{\roint}$; the latter statement follows by iteratively applying
derivatives to (\ref{eq:der of eigenvalues}). It follows that 
\begin{equation*}
\| \lambda_I(t,\cdot) - \lambda_I (s,\cdot) \|_{C^{\ell}(\S)}
    \leq \textstyle{\int}_{s}^{t} \| \d_r \lambda_I(r,\cdot) \|_{C^{\ell}(\S)} \md r
    \leq \tfrac{K}{\alpha}T_+^{\alpha}(t^{\alpha} - s^{\alpha})
\end{equation*}
for any $[s,t] \subset \mI$, and any $I$. This concludes the proof.
\end{proof}

\end{document}